\pdfoutput=1
\PassOptionsToPackage{table}{xcolor}
\documentclass{article}

\usepackage{microtype}
\usepackage{graphicx}
\usepackage{subcaption}
\usepackage{booktabs} 
\usepackage{mathtools}
\usepackage{hyperref}

\usepackage[preprint]{manuscript}

\usepackage{amsmath}
\usepackage{longtable}
\usepackage{amssymb}
\usepackage{mathtools}
\usepackage{amsthm}
\usepackage{xcolor}
\definecolor{softburgundy}{RGB}{144, 35, 62}
\usepackage{booktabs}
\usepackage{array}
\usepackage{tabularx}
\usepackage[table]{xcolor}

\providecommand{\TV}{\mathrm{TV}}
\providecommand{\KL}{\operatorname{KL}}
\providecommand{\X}{\mathcal{X}}
\providecommand{\calL}{\mathcal{L}}
\providecommand{\Pois}{\operatorname{Pois}}
\providecommand{\PP}{\mathbb{P}}
\providecommand{\E}{\mathbb{E}}
\providecommand{\dd}{\mathrm{d}}

\usepackage[capitalize,noabbrev]{cleveref}

\theoremstyle{plain}
\newtheorem{theorem}{Theorem}[section]
\newtheorem{proposition}[theorem]{Proposition}

\theoremstyle{definition}

\theoremstyle{remark}

\usepackage[textsize=tiny]{todonotes}

\papertitlerunning{Edit Ordering and Endpoint Couplings in Graph Flow Matching}

\begin{document}

\twocolumn[
  \papertitle{When Can Edit Ordering Preserve Endpoint Couplings
  \texorpdfstring{\\}{ }
  in Graph Flow Matching?}

  \begin{paperauthorlist}
    \paperauthor{Zilin Li}{polyu}
    \paperauthor{Yiran Zhu}{ustc}
    \paperauthor{Zixing Song}{polyu}
  \end{paperauthorlist}

  \paperaffiliation{polyu}{The Hong Kong Polytechnic University,
    Hong Kong SAR, China}
  \paperaffiliation{ustc}{University of Science and Technology of China,
    Hefei, China}

  \papercorrespondingauthor{Zixing Song}{szxpolyu@gmail.com}

  \paperkeywords{Graph Flow Matching, Graph Editing,
    Endpoint Coupling, Edit Ordering, Markov Projection}

  \vskip 0.3in
]

\printAffiliationsAndNotice{}

\begin{abstract}
Graph editing plays an important role in graph generation,
where structural changes must satisfy task-specific constraints.
Discrete flow matching models such editing processes through
conditional sequences of edge additions and deletions, yet
marginalizing these trajectories may preserve graph distributions
while losing the underlying source--target coupling.
This loss arises when trajectories towards different targets
share intermediate graphs, making their intended destinations
indistinguishable from the current graph and time alone.
In this work, we show that legal edit ordering offers an
augmentation-free mechanism to preserve this correspondence
without changing the required edits, their lengths, or the
prescribed endpoints.
We establish an exact characterization of when the resulting
Markov dynamics attain the minimum achievable paired error
and develop a theoretically guided approach to edit scheduling
under graph constraints.
Experiments on observed graph pairs and controlled editing tasks
derived from real graph topologies demonstrate that appropriate
edit ordering can substantially improve paired endpoint fidelity.
\end{abstract}

\section{Introduction}
\label{sec:introduction}

One way to generate graphs is to construct them progressively
through a sequence of local changes, such as adding or removing
edges \citep{gao2010survey}. In stochastic formulations, conditional paths describe
how initial graphs reach target graphs, while transition rates \citep{norris1997markov}
govern the changes over time. Flow matching \citep{lipman2023flow}, discrete flow matching
\citep{gat2024discrete}, and generator matching
\citep{holderrieth2025generator} provide frameworks
for learning these dynamics \citep{shaul2025flow}. Recent graph generative
models further address distributional fidelity and structural
validity \citep{vignac2023digress,qin2025defog,madeira2024generative}, while
edit-based flows support generation through insertion and deletion
operations \citep{havasi2025edit}. These advances provide ways
to learn graph distributions through local transitions.
In paired graph-editing tasks, an initial graph and its intended
edited version form a source--target pair \citep{heo2024epic}.
Preserving this correspondence requires each source
to reach its prescribed target, even when the overall
generated distribution is correct.

This distinction becomes relevant when conditional paths
from distinct graph pairs are marginalized into a single
Markov process on graph states. Each conditional path is
constructed for its target, while the resulting source-blind
model selects edits using only the current graph and time.
If paths leading to different targets meet at the same graph,
their next edits can become ambiguous, potentially changing
the original source--target correspondence. Related effects
have been studied in rectified flow, bridge matching, and
discrete Schr\"odinger bridges
\citep{liu2023rectified,debortoli2023augmented,kim2025ddsbm}.
Other approaches use learned representations or additional
trajectory information to distinguish intersecting paths
\citep{kim2024simulationfree,zhang2024trajectory}.

The order of edge additions and deletions
determines which intermediate graphs are visited.
The same required edits may be performed in different
legal orders while preserving the endpoints and path lengths. For example, deleting distinguishing edges first may
cause routes towards different targets to merge, whereas
adding target-specific edges first can keep them separate
when the graph constraints permit.
Previous studies have explored flexible discrete paths,
structured action trajectories, and learnable generation
orders \citep{shaul2025flow,kemper2026trajectory,patel2026insertion}.
Related combinatorial work investigates disjoint shortest
paths and their reconfiguration
\citep{botea2021counting,saito2024reconfiguration,gajjar2024reconfiguring}.
Building on these insights, we investigate a key
yet understudied
question in graph flow matching:
\emph{can legal graph-edit ordering preserve the prescribed
source--target coupling?}
We keep the paired endpoints, required edits, shortest
path lengths, legality constraints, and sampling budget
fixed, isolating the effect of edit ordering on endpoint fidelity.

\begin{figure*}[htbp]
  \centering
  \includegraphics[width=\textwidth]{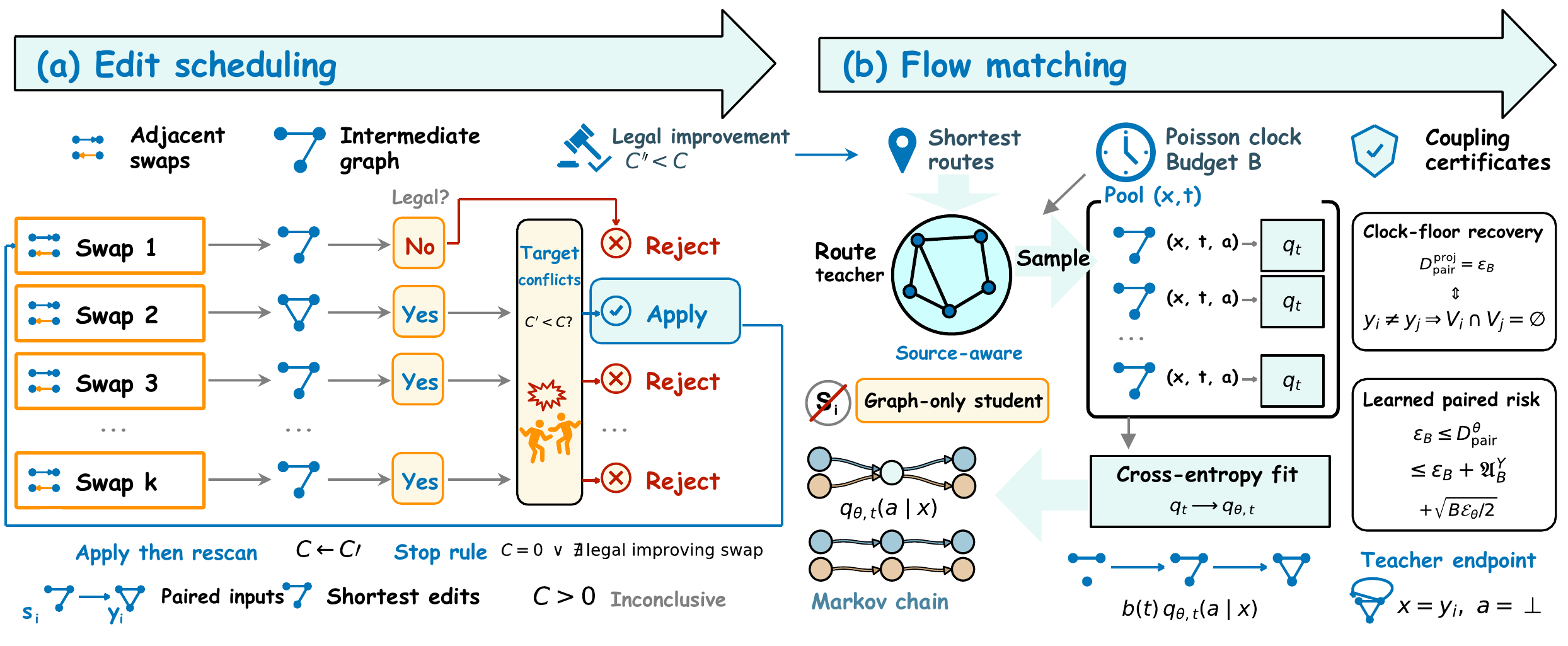}
  \caption{Overview of the proposed approach.
  (a) Legal adjacent swaps adjust edit orders to reduce
  conflicts between routes with different targets while
  preserving edit legality, endpoints, and path lengths.
  (b) The resulting routes define a source-aware teacher,
  whose pooled actions train a graph-only student.
  Structural recovery certificates and learned endpoint
  fidelity are evaluated separately.}
  \label{fig:edit_ordering_overview}
\end{figure*}

Here, we characterize when legal edit ordering can
preserve endpoint correspondence without augmenting the
generative state. Under a common Poisson clock, we prove
that the source-blind Markov projection achieves the
minimum attainable paired error if and only if routes
to different targets never share a graph state.
Routes to the same target may still merge.
This characterization leads to an exact route-selection
formulation and a legality-preserving edit scheduling
procedure (Figure~\ref{fig:edit_ordering_overview}).
We further identify capacity constraints under which
coupling recovery becomes impossible, and derive an
error bound relating structural ambiguity and learning
error to paired endpoint fidelity.
Experiments on observed graph pairs and controlled
editing tasks derived from real graph topologies
examine these theoretical findings in both ideal
Markov projections and learned graph-only models.

\paragraph{Contributions.}
Our main contributions are summarized as follows:
\begin{itemize}
  \item We show that legal graph-edit ordering alone can
  determine endpoint coupling preservation under Markov
  marginalization, with an exact recoverability criterion
  requiring no state augmentation.

  \item We formulate an exact constrained routing problem
  and develop legality-preserving edit scheduling with
  recovery certificates, while characterizing structural
  limits to coupling preservation.

  \item We bound paired endpoint error in terms of structural
  ambiguity and learning error, and demonstrate the benefits
  of edit ordering on observed graph pairs and controlled
  real-topology editing tasks.
\end{itemize}

\section{Method}\label{sec:method}
\subsection{Paired graph transformations and endpoint error}
\label{sec:setting}
We write $\mathbb P$ and $\mathbb E$ for probability and expectation,
$\mathcal L(Z)$ for the law of $Z$, and $\delta_z$ for a point mass at $z$.
For distributions $P,Q$ on a finite space $\Omega$, total variation (TV)
measures distributional discrepancy, while Kullback--Leibler (KL)
divergence compares their relative probability masses
\citep{cover2006elements,csiszar2011information}:
\begin{align*}
\|P-Q\|_{\TV}
&=\tfrac12\sum_{x\in\Omega}|P(x)-Q(x)|,\\
\KL(P\|Q)
&=\sum_{x\in\Omega}P(x)\log\frac{P(x)}{Q(x)}.
\end{align*}
We use the usual conventions for zero probabilities.
For sets $A$ and $B$, we write $A\setminus B$ for
set difference, $A\triangle B$ for symmetric difference,
and $|A|$ for cardinality. The indicator of an event
is written as $\mathbf 1\{\cdot\}$. We use
$\Pois(\lambda)$ for the Poisson distribution with
mean $\lambda$, and $\bot$ for a null edit that leaves
the current state unchanged.

Let $E_*$ be a finite set of identifiable edges, so each graph can be represented by the subset of its present edges. We write $\X\subseteq 2^{E_*}$ for the states satisfying the graph constraints, such as degree or edge-capacity limits. One legal edit adds or deletes a single edge while remaining in $\X$ \citep{diestel2017graph,korte2018combinatorial}. The source--target coupling is the joint distribution of the initial graph and its prescribed target \citep{lindvall1992lectures,thorisson2000coupling}:
\begin{equation}\label{eq:pair}
\pi=\sum_{i=1}^{m}p_i\delta_{(s_i,y_i)},\qquad
p_i>0,\quad \sum_i p_i=1,
\end{equation}

The $m$ pairs have weights $p_i$ and distinct sources, so
$(S,Y)\sim\pi$ satisfies $S=s_i\Rightarrow Y=y_i$;
targets may repeat. Set $D_i=s_i\setminus y_i$,
$I_i=y_i\setminus s_i$, and $K_i=|D_i|+|I_i|$.
Each pair admits a \emph{legal monotone route}
$\gamma_i=(x_{i,0},a_{i,0},\ldots,a_{i,K_i-1},x_{i,K_i})$
from $x_{i,0}=s_i$ to $x_{i,K_i}=y_i$, toggling each
edge in $D_i\cup I_i$ exactly once and no other edge.
For $0\le j<K_i$, action $a_{i,j}$ moves $x_{i,j}$
to $x_{i,j+1}$. These are shortest single-edge edit
routes. Downward-closed families admit delete-first
routes; otherwise monotone feasibility is assumed.
Only legal edit order varies; the pairs, edits, and
route lengths remain fixed.

We use a nonhomogeneous Poisson clock $N_t$, a counting process whose event rate at time $t$ is $b(t)>0$ almost everywhere on $[0,T]$ \citep{kingman1993poisson,bremaud1981point}. The clock is independent of $(S,Y)$; its cumulative intensity is $u(t)=\int_0^t b(r)\,\dd r$, and $B=u(T)\in(0,\infty)$ is the expected number of events by the terminal time. The source-aware route teacher starts at $s_i$ and performs its prescribed edit at each event, applying the null action $\bot$ after reaching $y_i$:
\begin{equation}\label{eq:route}
X_t^\gamma=x_{i,\min\{N_t,K_i\}},\quad
A_t^\gamma=
\begin{cases}
a_{i,N_t}, & N_t<K_i,\\
\bot, & N_t\ge K_i.
\end{cases}
\end{equation}
Our graph-only policies share this clock and take at most one legal edge edit per event. Their decisions use the current graph and time, without the source or editing history, as in a time-inhomogeneous Markov policy \citep{norris1997markov,anderson1991continuous}. Since route $i$ needs $K_i$ edits, the probability of too few clock events defines the clock-error floor:
\begin{equation}\label{eq:floor}
\varepsilon_B=\sum_i p_i\,\PP\{\Pois(B)<K_i\}.
\end{equation}
For every such policy $X_T$, $\PP(X_T\ne y_i\mid S=s_i)\ge\PP(N_T<K_i)$: reaching $y_i$ requires at least $K_i$ toggles, regardless of the policy. As $\pi$ is deterministic given $S$, the paired total variation is its mispairing probability, and hence
\begin{align}\label{eq:lb}
D_{\mathrm{pair}}
&:=\|\calL(S,X_T)-\pi\|_{\TV}\notag\\
&=\sum_i p_i\PP(X_T\ne y_i\mid S=s_i)
\ge\varepsilon_B.
\end{align}
We call equality in the paired-error bound \emph{clock-floor coupling recovery}: the pairing is attained whenever the clock supplies enough edits, while finite-budget failures remain.

\subsection{Coupling recoverability under Markov projection}
\label{sec:theory}
The source-aware teacher follows its sampled pair's prescribed edits. The \emph{source-blind Markov projection} instead pools teacher actions at the same graph $x$ and time $t$, giving
\begin{equation}\label{eq:q}
q_t(a\mid x)=\PP(A_t^\gamma=a\mid X_t^\gamma=x),
\end{equation}
at every teacher-occupied graph. The projection starts from the same source marginal and samples an edit from $q_t(\cdot\mid x)$ at each clock event. Define the one-time occupancy $p_t(x)=\PP(X_t^\gamma=x)$ and the successor graph $T_a(x)$ after action $a$, with $T_\bot(x)=x$. Averaging the teacher's transition flux over the current graph gives the finite-state forward equation \citep{norris1997markov,anderson1991continuous}:
\begin{equation}\label{eq:forward}
\dot p_t(y)=b(t)\left[\sum_{\substack{x,a:T_a(x)=y}}p_t(x)q_t(a\mid x)-p_t(y)\right].
\end{equation}
By uniqueness of this forward equation, the projection
preserves teacher one-time graph marginals, but not
necessarily the source-conditioned marginals or prescribed
coupling \citep{debortoli2023augmented,zhang2024trajectory}.

Let $\mathcal V_i=\{x_{i,0},\ldots,x_{i,K_i}\}$ be
route $i$'s visited states. Routes to the same target
may share states, since all their shortest continuations
remain target-directed. Sharing states across different
targets creates ambiguous edits. 

\begingroup
\makeatletter
\let\originalthmhead\thmhead
\def\thmhead#1#2#3{%
  \textcolor{black!60}{\rule{1ex}{1ex}}%
  \hspace{0.35em}%
  \originalthmhead{#1}{#2}{#3}%
}
\makeatother

\begin{theorem}[Exact endpoint recoverability criterion]\label{thm:endpoint}
Under the preceding assumptions, the source-blind Markov projection attains the clock floor $D_{\mathrm{pair}}^{\mathrm{proj}}=\varepsilon_B$ if and only if the selected legal shortest routes satisfy
\begin{equation}\label{eq:separator}
\mathcal V_i\cap\mathcal V_j=\varnothing
\qquad\text{for every }i,j\text{ with }y_i\ne y_j.
\end{equation}
Routes to a common target may merge, but routes to distinct targets must remain state-disjoint.
\end{theorem}
\endgroup

\begin{proof}
Suppose routes to distinct targets are state-disjoint. Every graph visited by a teacher route is associated with a unique target $y$, although several sources may share that target. Every non-null action supported by $q_t(\cdot\mid x)$ is an action of a monotone shortest route to $y$, so it decreases $|x\triangle y|$ by one. Its successor belongs to a route with the same target. At $x=y$ all supported actions are null. Starting from $s_i$, the projected process therefore reaches $y_i$ at its $K_i$th clock event and remains there. Its conditional error is $\PP(N_T<K_i)$, attaining the clock-error floor.

For necessity, suppose $x=x_{i,j}=x_{h,k}$
with $y_i\ne y_h$. At every interior time,
each prescribed action has positive pooled probability
because its Poisson route-position mass is positive.
At any $t\in(0,T)$, the blind projection can reach
$(X_t,N_t)=(x,j)$ from $s_i$ and $(x,k)$ from $s_h$,
both with positive probability. Conditional on $(x,t)$,
their future action laws agree, regardless of source
or past event count.

If the projection attained the clock floor, positivity
of all $p_i$ and the clock lower bounds
would force each source's conditional error to equal
$\PP(N_T<K_i)$. After either prefix, condition on
$L=N_T-N_t=\ell\ge\max\{K_i-j,K_h-k\}$,
which has positive probability by independent Poisson
increments. The terminal graph must then be $y_i$
or $y_h$, respectively, almost surely. Yet the identical future law given $(x,t,L=\ell)$
contradicts $y_i\ne y_h$, implying
$D_{\mathrm{pair}}^{\mathrm{proj}}>\varepsilon_B$.

\end{proof}

Endpoint recovery is weaker than pathwise imitation:
the latter requires the same next edit at every shared
state, whereas same-target routes may follow different
shortest continuations. The structure
is \emph{target-separated shortest routing}.

\paragraph{Legal-state graph and an obstruction.}\label{sec:cuts}
For each pair, the \emph{monotone interval} contains the legal graphs obtained by retaining the edges common to its endpoints and selecting subsets of the edges that differ:
\begin{equation}\label{eq:interval}
\mathcal I_i=\{x\in\X:s_i\cap y_i\subseteq x\subseteq s_i\cup y_i\}.
\end{equation}

Direct a legal single-edit step $x\to x'$ in
$\mathcal I_i$ whenever
$|x'\triangle y_i|=|x\triangle y_i|-1$.
The resulting acyclic graph $H_i$ has exactly the
legal shortest $s_i$--$y_i$ paths
\citep{diestel2017graph,korte2018combinatorial}.
Recovery requires one path per $H_i$, with no states
shared across different targets.

A \emph{cut} $C\subseteq\X$ intersects every
legal shortest route of each task
\citep{diestel2017graph,ahuja1993network}.
For $m$ distinct targets, $|C|<m$ prevents
target-separated routing, since each path would
require a different state in $C$. The capacity construction below illustrates this obstruction.

\subsection{Coupling-aware edit scheduling}
\label{sec:construction}
We use the recoverability criterion to formulate route selection and construct a legality-preserving edit scheduler, both under the constraints defining $\X$.

\paragraph{Exact finite-state formulation.}\label{sec:exact}
Let $\mathcal Y=\{y_i:1\le i\le m\}$ be the distinct targets, $W=\bigcup_i\mathcal I_i$ the union of legal monotone intervals, and $F_i$ the arcs of $H_i$. In the standard integer-flow formulation \citep{ahuja1993network,schrijver2003combinatorial,bertsimas1997linear}, $f^i_{uv}\in\{0,1\}$ indicates whether route $i$ uses arc $(u,v)\in F_i$, and $z_{y,x}\in\{0,1\}$ indicates whether state $x\in W$ is assigned to a route with target $y\in\mathcal Y$. The constraints are

\begin{equation}\label{eq:flow}
\begin{gathered}
\sum_{(x,v)\in F_i}\!f^i_{xv}
-\sum_{(u,x)\in F_i}\!f^i_{ux}
\!=\!\begin{cases}
1,&x=s_i,\\
-1,&x=y_i,\\
0,&\text{otherwise},
\end{cases}
\,(K_i>0)\\
f^i_{uv}\le z_{y_i,u},\quad
f^i_{uv}\le z_{y_i,v},\\
\sum_{y\in\mathcal Y}z_{y,x}\le1\quad(x\in W),
\quad z_{y_i,s_i}=z_{y_i,y_i}=1.
\end{gathered}
\end{equation}

Flow conservation selects a legal shortest route for each nonzero-length task, and the occupancy constraints prevent different targets from sharing states. For $K_i=0$, the endpoint constraint represents the zero-length route; same-target routes may merge.

\begingroup
\makeatletter
\let\originalthmhead\thmhead
\def\thmhead#1#2#3{%
  \textcolor{black!60}{\rule{1ex}{1ex}}%
  \hspace{0.35em}%
  \originalthmhead{#1}{#2}{#3}%
}
\makeatother

\begin{proposition}[Exact route selection for an arbitrary legal family]\label{prop:integer}
The integer-flow system above is feasible if and only if some selection of legal monotone shortest edit routes attains $D_{\mathrm{pair}}^{\mathrm{proj}}=\varepsilon_B$. An integral solution yields such routes by following the positive flow for each task.
\end{proposition}
\endgroup
\begin{proof}
Each $H_i$ decreases the nonnegative distance to $y_i$ along every arc and hence is acyclic. Unit integral flow conservation from $s_i$ to $y_i$ therefore gives an $s_i$--$y_i$ path, without cycles or disconnected flow components. The occupancy constraints prevent paths with different targets from sharing a graph. The recoverability criterion proves sufficiency. Conversely, set the flow variables to one on any target-separated shortest route family, zero elsewhere, and assign $z_{y,x}=1$ exactly when a route to target $y$ visits $x$. Every constraint is then satisfied.
\end{proof}

The state set $W$ can grow exponentially in $|s_i\triangle y_i|$. The exact formulation therefore provides a feasibility certificate and a small-instance reference for the scheduling procedure.

\paragraph{Recovery without identifying edges.}
The criterion does not rely on dedicated source identifiers. Let $R$ be any preserved collection of edges and let $a,b,c,d\notin R$ be distinct permitted edges. Consider
\begin{equation}\label{eq:cross}
\begin{aligned}
s_1=R\cup\{a,b\},&\quad y_1=R\cup\{c,d\},\\
s_2=R\cup\{a,c\},&\quad y_2=R\cup\{b,d\}.
\end{aligned}
\end{equation}
If the following states are legal (for instance under the edge-capacity constraint $|x|\le |R|+2$ with $R$ fixed), take the two routes, suppressing $R$:
\begin{equation}\label{eq:crosspaths}
\begin{aligned}
\{a,b\}&\to\{a\}\to\varnothing\to\{d\}\to\{c,d\},\\
\{a,c\}&\to\{c\}\to\{b,c\}\to\{b\}\to\{b,d\}.
\end{aligned}
\end{equation}
Both use exactly four mandatory edits, remain at or below their endpoint edge capacity, and have disjoint visited states; the recoverability criterion yields clock-floor recovery. Yet $s_1\cup y_1=s_2\cup y_2=R\cup\{a,b,c,d\}$: no edge is globally private to either task. The configurations along the routes distinguish their targets without dedicated identifying edges.

\paragraph{Legality-preserving scheduling.}\label{sec:algorithm}
For large state spaces, we improve legal shortest routes
by adjacent edit swaps \citep{korte2018combinatorial}.
Swapping $a_{i,j}$ and $a_{i,j+1}$ leaves the endpoint,
edits, and length unchanged and modifies only the
intermediate state after the first edit. Checking that
this state belongs to $\X$ certifies legality
with one oracle query.

To count conflicts between routes with distinct targets, let $n_y(x)$ be the number of routes to $y$ passing through graph $x$. Define the nonnegative integer conflict score
\begin{equation}\label{eq:collision}
C(\gamma)=\sum_{x\in\bigcup_i\mathcal V_i}
\ \sum_{\substack{y,z\in\mathcal Y\\y<z}} n_y(x)n_z(x).
\end{equation}
Here $<$ fixes an ordering of target labels.
$C(\gamma)=0$ exactly for target-separated routes and
thus certifies clock-floor recovery; positive $C$ is not
a paired-risk estimate. Strictly improving legal swaps
terminate at $C=0$ or a local optimum, preserving legal
shortest paths.

\begin{algorithm}[htbp]
\caption{Legality-preserving coupling-aware edit scheduling}\label{alg:swap}
\begin{algorithmic}[1]
\REQUIRE Paired graphs; legal shortest routes; legality oracle $\mathrm{Legal}(x)$.
\STATE Compute the target-labeled counts $n_y(x)$ and conflict score $C$.
\REPEAT
  \STATE $\mathrm{improved}\gets\mathrm{false}$.
  \FOR{each route $i$ and adjacent positions $j,j+1$}
    \STATE Compute the intermediate graph $x'$ after swapping the edits.
    \IF{\textcolor{softburgundy}{$\mathrm{Legal}(x')$ and the updated score is smaller}}
      \STATE Apply the swap and update $n_y(x)$ and $C$.
      \STATE $\mathrm{improved}\gets\mathrm{true}$; \textbf{break}.
    \ENDIF
  \ENDFOR
\UNTIL{not $\mathrm{improved}$ or \textcolor{softburgundy}{$C=0$}}
\STATE \textbf{return} legal shortest routes and the certificate $C$.
\end{algorithmic}
\end{algorithm}

The legality test can represent edge capacities, degree bounds, forbidden subgraphs, or connectivity, provided an initial legal monotone shortest route exists. Delete-first initialization works for downward-closed legal families; otherwise a starting route can be searched for in $H_i$ or supplied by the task. The integer-flow formulation distinguishes global feasibility from the result of local scheduling.

\subsection{Capacity constraints and unrecoverable pairings}
\label{sec:impossible}

Capacity constraints can force routes to distinct targets through a shared state, causing paired endpoint error even under the exact pooled policy. The following construction isolates the effect of capacity on edit ordering, with projected marginal error defined as $D_{\mathrm{marg}}=\|\calL(X_T^{\mathrm{proj}})-\calL(Y)\|_{\TV}$.

\begingroup
\makeatletter
\let\originalthmhead\thmhead
\def\thmhead#1#2#3{%
  \textcolor{black!60}{\rule{1ex}{1ex}}%
  \hspace{0.35em}%
  \originalthmhead{#1}{#2}{#3}%
}
\makeatother

\begin{theorem}[An arbitrary-size capacity separation]\label{thm:capacity}
Let $a_1,\ldots,a_m,b_1,\ldots,b_m$ be $2m$ distinct identifiable edges, $m\ge2$, with equiprobable paired tasks $s_i=\{a_i\}$ and $y_i=\{b_i\}$. Let $\X_C=\{x:|x|\le C\}$ and $\psi=\PP\{\Pois(B)\ge2\}=1-e^{-B}(1+B)$. At capacity $C=1$ every shortest route must delete $a_i$ before adding $b_i$, and the source-blind projection satisfies
\begin{equation}\label{eq:caplow}
D_{\mathrm{marg}}=1-\psi,\qquad D_{\mathrm{pair}}=1-\frac{\psi}{m}.
\end{equation}
At capacity $C\ge2$, the same pairs also admit legal add-first shortest routes, and for those routes
\begin{equation}\label{eq:caphigh}
D_{\mathrm{marg}}=D_{\mathrm{pair}}=1-\psi.
\end{equation}
At fixed $C\ge2$, comparing delete-first with add-first changes \emph{only} the edit orders: all sources, targets, required actions, path lengths, and clock budgets remain identical. The paired-error gap is exactly $(1-1/m)\psi$ at every $B>0$.
\end{theorem}
\endgroup

\begin{proof}

At $C=1$, every two-edit route deletes $a_i$
before adding $b_i$ and passes through $\varnothing$.
Source-independent timing and uniform pair weights
make the pooled addition at this state uniform over
$\{b_j\}_{j=1}^m$. Completion requires at least two
events (probability $\psi$), and its target is correct
with probability $1/m$. The correct-pair mass is
therefore $\psi/m$, while total target mass is $\psi$.

For $C\ge2$, add-first routes
$\{a_i\}\to\{a_i,b_i\}\to\{b_i\}$ are legal and
target-separated, so paired error is $1-\psi$.
Each target receives mass $\psi/m$, giving marginal
error $1-\psi$. Delete-first routes remain legal
and retain the capacity-one errors; the comparison
thus isolates edit order at fixed capacity.

\end{proof}

For the delete-first routes, the target-conditioned ambiguity defined next equals $\mathfrak A_B^Y=(1-1/m)\psi$: all routes meet at the empty graph, where their conditional action laws differ in TV by $1-1/m$. Consequently, the paired-error bound derived below is attained by the population-optimal student.

\subsection{Graph-only learning and endpoint evaluation}
\label{sec:learning}
To relate the ideal pooled policy to a trained graph-only student, we separate ambiguity across \emph{targets} from action-approximation error. The target-conditioned action law $q_t^y$ averages teacher edits among routes assigned to target $y$; its difference from $q_t$ defines the time-integrated ambiguity
\begin{align}\label{eq:qy}
q_t^y(a\mid x)
&=\PP(A_t^\gamma=a\mid X_t^\gamma=x,Y=y),\notag\\
\mathfrak A_B^Y
&=\int_0^T b(t)\,\E_{Y,X_t^\gamma}\Bigl[
\bigl\|q_t^Y(\cdot\mid X_t^\gamma)\notag\\
&\hspace{44pt}-q_t(\cdot\mid X_t^\gamma)\bigr\|_{\TV}\Bigr] \,\dd t.
\end{align}
Ambiguity vanishes for target-separated routes and does not penalize shortest continuations towards the same target. For a legal student action distribution $q_{\theta,t}$, including the null edit, we measure approximation error by the clock-averaged KL divergence \citep{cover2006elements,murphy2022probabilistic}:
\begin{equation}\label{eq:excess}
\mathfrak E_\theta=\frac1B\int_0^T b(t)\,\E_{X_t^\gamma}
\KL\bigl(q_t(\cdot\mid X_t^\gamma)\,\|\,q_{\theta,t}(\cdot\mid X_t^\gamma)\bigr)\,\dd t,
\end{equation}
We assume $\mathfrak E_\theta<\infty$ with the usual KL support condition.

\paragraph{Training a graph-only student.}
Let
$\mathcal G_{\mathrm{tr}}$ denote the training cohorts. A cohort
$g$ is a collection of comparable paired graph-editing tasks
sharing an identifiable edge universe $E_{*,g}$ and a public
legality family $\mathcal X_g$. Its set of pair indices is
$\mathcal I_g$. Each pair follows one legal shortest
teacher route; teacher data are pooled within, rather than
across, cohorts.

Write $\hat p_{i\mid g}>0$ for the normalized weights of
pairs in cohort $g$, with $\sum_{i\in\mathcal I_g}\hat p_{i\mid g}=1$.
At cumulative clock time $u\in[0,B]$, let
$w_{ij}(u)$ be the probability that teacher route $i$ is at
its $j$th state (given by the Poisson occupancy formula below). The
empirical teacher occupancy is
\begin{equation*}
\widehat\mu_{g,u}^{\mathrm{tr}}(x)
=\sum_{i\in\mathcal I_g}\hat p_{i\mid g}
\sum_{j=0}^{K_i}w_{ij}(u)\mathbf1\{x_{ij}=x\}.
\end{equation*}
The pooled teacher action law on occupied states is
\begin{align}\label{eq:empirical-pooled-kernel}
\widehat q_{g,u}^{\mathrm{tr}}(a\mid x)
&=\frac{1}{\widehat\mu_{g,u}^{\mathrm{tr}}(x)}
\sum_{i\in\mathcal I_g}\hat p_{i\mid g}\notag\\
&\quad\cdot\sum_{j=0}^{K_i}w_{ij}(u)
\mathbf1\{x_{ij}=x,\ a_{ij}=a\}.
\end{align}
Here $a_{i,K_i}=\bot$. This averages the teacher actions weighted by their state occupancies.

Sample cohort $g\sim\rho$ with
$\sum_{g\in\mathcal G_{\mathrm{tr}}}\rho_g=1$, then
$U\sim\mathrm{Unif}[0,B]$ and
$X\sim\widehat\mu_{g,U}^{\mathrm{tr}}$.
The student $q_\theta(a\mid x,u;\mathcal P_g)$ sees only
the graph, time, and public constraints
$\mathcal P_g$, not source, target, route, remaining
edits, or a cohort identifier.

We train the student by cross-entropy against the pooled
teacher distribution, the objective for fitting
categorical action probabilities
\citep{goodfellow2016deep,murphy2022probabilistic}:
\begin{align}\label{eq:empirical-pooled-ce}
\widehat{\mathcal L}_{\mathrm{CE}}^{\mathrm{tr}}(\theta)
&=\sum_{g\in\mathcal G_{\mathrm{tr}}}\rho_g
\mathbb E_{U,X\mid g}\Bigl[-\sum_a\notag\\
&\quad \widehat q_{g,U}^{\mathrm{tr}}(a\mid X)
\log q_\theta(a\mid X,U;\mathcal P_g)\Bigr]\\
&=\widehat H^{\mathrm{tr}}
+\widehat{\mathfrak E}_{\theta}^{\mathrm{tr}}.\notag
\end{align}
Here $\widehat H^{\mathrm{tr}}$ is the cohort-weighted
conditional entropy of the teacher, independent of $\theta$.
The remaining term is the student--teacher KL error:

\begin{equation*}
\widehat{\mathfrak E}_{\theta}^{\mathrm{tr}}
=\sum_{g\in\mathcal G_{\mathrm{tr}}}\rho_g\,
\mathbb E_{U,X\mid g}\Bigl[
\KL\bigl(\widehat q_{g,U}^{\mathrm{tr}}(\cdot\mid X)
\|q_\theta(\cdot\mid X,U;\mathcal P_g)\bigr)
\Bigr].
\end{equation*}

Cross-entropy supplies the
initial training objective and is subsequently combined
with a differentiable paired-endpoint loss to encourage the
correct target for each training source. Parameters may be shared across cohorts, although action pooling remains cohort-specific. Evaluation uses the same clock and sampling weights; held-out targets are used only for evaluation. We next bound endpoint error in terms of the population approximation error $\mathfrak E_\theta$.

\begingroup
\makeatletter
\let\originalthmhead\thmhead
\def\thmhead#1#2#3{%
  \textcolor{black!60}{\rule{1ex}{1ex}}%
  \hspace{0.35em}%
  \originalthmhead{#1}{#2}{#3}%
}
\makeatother

\begin{theorem}[Paired and marginal error guarantees]\label{thm:risk}
Under the paired-editing and common-clock assumptions above, for the population action law and paired distribution under consideration, the graph-only student satisfies
\begin{align}\label{eq:risk}
\varepsilon_B\le D_{\mathrm{pair}}^\theta
&\le\varepsilon_B+\mathfrak A_B^Y+\sqrt{B\mathfrak E_\theta/2},\\
D_{\mathrm{marg}}^\theta
&\le\varepsilon_B+\sqrt{B\mathfrak E_\theta/2}.\notag
\end{align}

Here $D_{\mathrm{pair}}^\theta
=\|\calL(S,X_T^\theta)-\pi\|_{\TV}$
and $D_{\mathrm{marg}}^\theta
=\|\calL(X_T^\theta)-\sum_i p_i\delta_{y_i}\|_{\TV}$.
If routes to distinct targets are state-disjoint and the student matches $q_t$, its paired risk equals $\varepsilon_B$.
\end{theorem}
\endgroup

\begin{proof}
Construct a target-aware process with the same initial pairs $(S,Y)$ and clock, but with action law $q_t^Y$. By the finite-state forward equation, its joint $(Y,X_t)$ marginals equal those of the route teacher: the transition flux conditional on $Y=y$ is obtained by averaging the teacher's source-specific action fluxes over its occupied states. Furthermore, every action supported by $q_t^y$ decreases Hamming distance to $y$, so from any of its associated sources $s_i$, this process reaches $y_i$ after exactly $K_i$ events and then remains there. Its paired error is therefore $\varepsilon_B$.

Maximally couple the target-aware and blind processes'
action marks using the same $(S,Y)$ and clock
\citep{lindvall1992lectures,thorisson2000coupling}.
The teacher-matched target-aware occupancy bounds
their disagreement probability, and hence terminal
paired TV, by $\mathfrak A_B^Y$.
For the blind projection and student starting from
the same source with the same clock, the
marked-path KL chain rule and projected--teacher
marginal identity give
\begin{align}\label{eq:pathkl}
&\KL(\mathbb P^{\rm proj}_{\rm marked}
\|\mathbb P^\theta_{\rm marked})\notag\\
&\quad=\int_0^T b(t)\,\E_{X_t^{\rm proj}}\bigl[
\KL(q_t(\cdot\mid X_t^{\rm proj})\notag\\
&\hspace{53pt}\|\,q_{\theta,t}(\cdot\mid X_t^{\rm proj}))
\bigr] \,\dd t\\
&\quad=B\mathfrak E_\theta.\notag
\end{align}
Pinsker's inequality converts this pathwise KL divergence
into a TV bound for $(S,X_T)$ of
$\sqrt{B\mathfrak E_\theta/2}$
\citep{csiszar2011information,tsybakov2009nonparametric}. The triangle inequality gives the paired upper bound. Marginal consistency of the blind projection with the teacher and direct application of the latter KL bound give the marginal upper bound. Finally, the clock lower bound follows from the required number of edits.
\end{proof}

For a fixed edit schedule, we can compute the projected endpoint risk directly. In cumulative clock time $u\in[0,B]$, we write $q_u(a\mid x)$ for the same kernel at $u=u(t)$. The Poisson occupancy probability of position $j$ on route $i$ is \citep{kingman1993poisson}
\begin{equation}\label{eq:weights}
w_{ij}(u)=\begin{cases}e^{-u}u^j/j!,&j<K_i,\\
\PP\{\Pois(u)\ge K_i\},&j=K_i.
\end{cases}
\end{equation}
Consequently, the projected kernel at a teacher-visited graph is explicitly
\begin{align}\label{eq:kernel}
q_u(a\mid x)
&=\frac{\displaystyle\sum_{i,j:x_{i,j}=x}
  p_iw_{ij}(u)\mathbf1\{a_{i,j}=a\}}
 {\displaystyle\sum_{i,j:x_{i,j}=x}p_iw_{ij}(u)},
\quad a_{i,K_i}=\bot.
\end{align}
At states with zero occupancy, the kernel may be defined arbitrarily. Define $r_i(u,x)=\PP(X_u^{\rm proj}=x\mid S=s_i)$, initialized by $r_i(0,x)=\mathbf1\{x=s_i\}$. With $T_a(x)$ denoting the legal successor and $T_\bot(x)=x$, the conditional state probabilities obey the finite-state forward system \citep{norris1997markov,anderson1991continuous}:
\begin{equation}\label{eq:ode}
\frac{\dd r_i(u,x)}{\dd u}
=\sum_{\substack{z,a:T_a(z)=x}}r_i(u,z)q_u(a\mid z)-r_i(u,x).
\end{equation}
Every successor with positive $q_u$ is a state on a selected teacher route, so the system closes on the finite union of their visited graphs. The exact projected paired risk is
\begin{equation}\label{eq:exactrisk}
D_{\mathrm{pair}}^{\rm proj}=1-\sum_i p_i r_i(B,y_i).
\end{equation}
Together these expressions compute the projected paired error of a selected legal schedule, including when $C>0$.

\section{Experiments}\label{sec:experiments}

\paragraph{Protocol.}
R1 uses observed-window source--target pairs from temporal
graph data: Email-Eu, Wiki-Talk~\citep{leskovec2014snap},
and Reddit body~\citep{rossi2020temporal}.
R2 uses TGBL-Coin-v2~\citep{huang2023temporal},
CollegeMsg~\citep{panzarasa2009patterns},
MathOverflow~\citep{paranjape2017motifs},
PROTEINS, and COLLAB~\citep{morris2020tudataset}.
Unlike R1, which uses observed source--target pairs, R2 employs controlled crossed-edit transformations on real graph topologies to isolate the effect of edit ordering. Each transformation requires two deletions and two additions ($K=4$), allowing different legal orders under a fixed edit budget. Additional results across varying edit lengths are provided in Appendix~\ref{app:extended-experiments}.
All order comparisons fix paired endpoints, required edits,
legality, a source-independent Poisson clock, and train/validation
splits. Learned-model comparisons use group-disjoint
validation and three matched route/model seeds \citep{poursafaei2022towards}.
The student receives only the current graph, clock time,
and public legality information, never the source, target,
or teacher route. Full construction, split, and clock protocols
are in Appendix~\ref{app:controlled-five-student}.

\paragraph{Numerical verification.}
The exact conditional forward ODE and independent event-driven
Monte Carlo agree on paired endpoint risks. Figure~\ref{fig:exp-exact-mc}
compares paired and marginal errors, conflict scores, and
numerical tolerances; further numerical and boundary checks are in Appendix~\ref{app:numerical-verification}.

\begin{figure}[htbp]
\centering
\includegraphics[width=\linewidth]{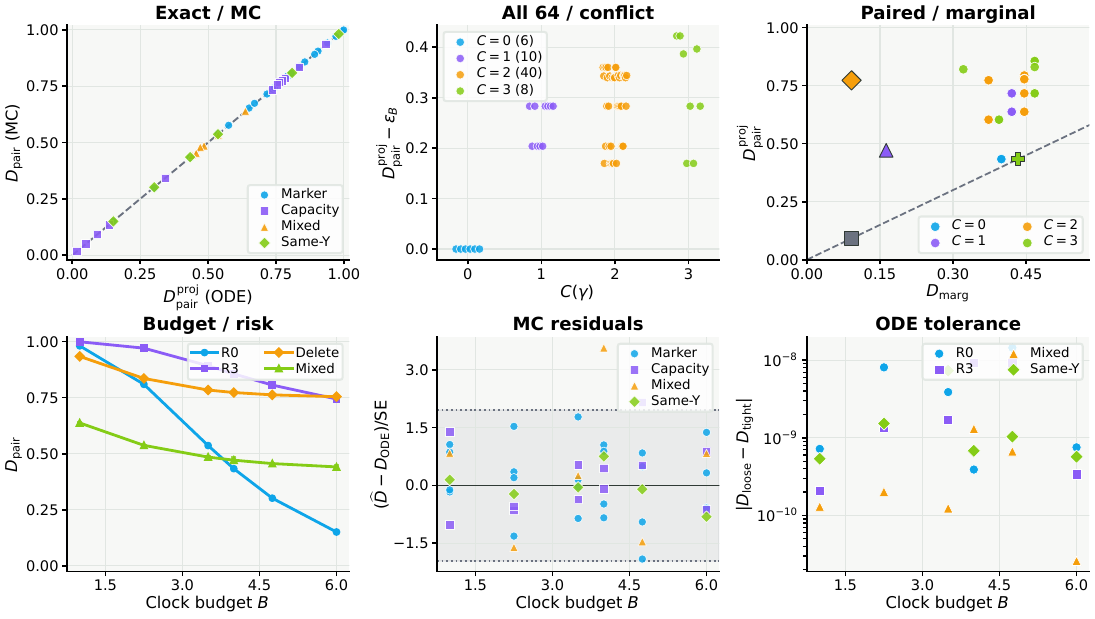}
\caption{
Exact endpoint-risk verification of the source-blind
Markov projection.
(a) Exact ODE versus Monte Carlo paired error.
(b) Excess paired error across 64 legal route pairs,
grouped by conflict score $C(\gamma)$.
(c) Paired error versus marginal TV.
(d) Paired error across Poisson clock budgets
with 95\% Monte Carlo confidence intervals.
(e) Standardized Monte Carlo residuals.
(f) ODE numerical-tolerance sensitivity.
}
\label{fig:exp-exact-mc}
\end{figure}

\paragraph{Observed-pair generalization (R1).}
Table~\ref{tab:exp-r1} reports held-out paired error.
Scheduling improves the ideal projection and the graph-only
student across the three observed datasets, with the largest
gains on Wiki-Talk and Reddit body.
Additional structural and ordering baselines are in Appendix~\ref{app:order-baselines}.

\begin{table}[htbp]
\centering
\small
\setlength{\tabcolsep}{1.5pt}
\renewcommand{\arraystretch}{1.05}
\caption{R1 observed-window paired error (\%, mean $\pm$ SD). Gain is Random minus Sched. in
percentage points.}
\label{tab:exp-r1}
\begin{tabular}{@{}lcc@{}}
\toprule
\multicolumn{3}{@{}l}{
\textit{A. Ideal Markov projection}} \\
\midrule
Dataset &
Random $\rightarrow$ Sched. &
Gain (pp) \\
\midrule
Email-Eu &
$21.87\pm0.21\rightarrow21.11\pm0.00$ &
$+0.76\;$ \\
Wiki-Talk &
$49.11\pm1.57\rightarrow34.22\pm0.22$ &
$+14.89\;\checkmark$ \\
Reddit body &
$36.47\pm0.72\rightarrow22.88\pm0.41$ &
$+13.59\;$ \\

\addlinespace[8pt]
\multicolumn{3}{@{}l}{
\textit{B. Graph-only Student}} \\
\midrule
Dataset &
Random $\rightarrow$ Sched. &
Gain (pp) \\
\midrule
Email-Eu &
$28.22\pm1.02\rightarrow27.46\pm1.09$ &
$+0.76\;$ \\
Wiki-Talk &
$39.49\pm1.31\rightarrow25.98\pm2.10$ &
$+13.51\;\checkmark$ \\
Reddit body &
$38.83\pm0.15\rightarrow28.77\pm0.15$ &
$+10.07\;$ \\
\bottomrule
\end{tabular}
\end{table}

\begingroup
\renewcommand{\thefootnote}{}
\footnotetext{%
\checkmark highlights selected R1 gains or near-floor R2 student errors.}
\endgroup

\paragraph{Graph-only training (R2).}
We compare ten fixed legal ordering controls on five
real-topology-derived controlled benchmarks, holding paired
endpoints, edit budgets, and Poisson clocks constant.
The graph-only student combines public prediction heads
with a GraphSAGE residual in PyTorch
\citep{hamilton2017inductive,paszke2019pytorch,xu2019powerful},
without source or target identities at inference.
Table~\ref{tab:exp-r2} reports all 50 full-validation
comparisons.
Zero-$C$ projections attain the $25\%$ floor; witness/C0
students approach it, while other students fare worse. Figure~\ref{fig:exp-training} shows the
training trajectories; detailed training protocols, diagnostics, and external comparisons
are in Appendices~\ref{app:controlled-five-student}
and~\ref{app:extended-experiments}.

\begin{table*}[htbp]
\centering
\begingroup
\fontsize{9.0}{11.0}\selectfont
\setlength{\tabcolsep}{3pt}
\renewcommand{\arraystretch}{1.10}
\setlength{\abovecaptionskip}{3pt}
\setlength{\belowcaptionskip}{0pt}
\caption{Graph-only student evaluation on five controlled benchmarks.}
\label{tab:exp-r2}
\resizebox{\textwidth}{!}{%
\begin{tabular}{@{}ll*{10}{c}@{}}
\toprule
\raisebox{3.7pt}{Dataset} & \raisebox{3.7pt}{Metric} & \raisebox{3.7pt}{Random} & \shortstack{Sched.\\Random} & \shortstack{Delete-\\first} & \shortstack{Sched.\\Delete} & \shortstack{Zero-$C$\\witness} & \shortstack{C0-\\first} & \shortstack{C0-\\last} & \raisebox{3.7pt}{C1} & \raisebox{3.7pt}{C2} & \raisebox{3.7pt}{C3} \\
\midrule
\rowcolor{blue!4}
\raisebox{14pt}[0pt][0pt]{\textbf{TGBL-Coin-v2}} & \shortstack[l]{\vphantom{$D_{\mathrm{pair}}^{\mathrm{proj}}$}$D_{\rm pair}^{\rm proj}$ (\%)\\[2.0pt]\vphantom{$D_{\mathrm{pair}}^{\mathrm{proj}}$}$D_{\rm pair}^{\theta}$ (\%)\\[2.0pt]\vphantom{$D_{\mathrm{pair}}^{\mathrm{proj}}$}KL (nats)} & \shortstack[c]{\vphantom{$D_{\mathrm{pair}}^{\mathrm{proj}}$}$58.64\pm0.52$\\[2.0pt]\vphantom{$D_{\mathrm{pair}}^{\mathrm{proj}}$}$66.16\pm0.07$\\[2.0pt]\vphantom{$D_{\mathrm{pair}}^{\mathrm{proj}}$}$0.639\pm0.050$} & \shortstack[c]{\vphantom{$D_{\mathrm{pair}}^{\mathrm{proj}}$}$38.05\pm1.19$\\[2.0pt]\vphantom{$D_{\mathrm{pair}}^{\mathrm{proj}}$}$65.63\pm0.24$\\[2.0pt]\vphantom{$D_{\mathrm{pair}}^{\mathrm{proj}}$}$0.695\pm0.015$} & \shortstack[c]{\vphantom{$D_{\mathrm{pair}}^{\mathrm{proj}}$}$78.53\pm0.00$\\[2.0pt]\vphantom{$D_{\mathrm{pair}}^{\mathrm{proj}}$}$68.91\pm1.41$\\[2.0pt]\vphantom{$D_{\mathrm{pair}}^{\mathrm{proj}}$}$0.188\pm0.015$} & \shortstack[c]{\vphantom{$D_{\mathrm{pair}}^{\mathrm{proj}}$}$25.00\pm0.00$\\[2.0pt]\vphantom{$D_{\mathrm{pair}}^{\mathrm{proj}}$}$40.09\pm8.74$\\[2.0pt]\vphantom{$D_{\mathrm{pair}}^{\mathrm{proj}}$}$0.860\pm0.212$} & \cellcolor{yellow!12}\shortstack[c]{\vphantom{$D_{\mathrm{pair}}^{\mathrm{proj}}$}$25.00\pm0.00$\\[2.0pt]\vphantom{$D_{\mathrm{pair}}^{\mathrm{proj}}$}$25.37\pm0.07\,\checkmark$\\[2.0pt]\vphantom{$D_{\mathrm{pair}}^{\mathrm{proj}}$}$0.002\pm0.000$} & \shortstack[c]{\vphantom{$D_{\mathrm{pair}}^{\mathrm{proj}}$}$25.00\pm0.00$\\[2.0pt]\vphantom{$D_{\mathrm{pair}}^{\mathrm{proj}}$}$25.43\pm0.12\,\checkmark$\\[2.0pt]\vphantom{$D_{\mathrm{pair}}^{\mathrm{proj}}$}$0.005\pm0.001$} & \shortstack[c]{\vphantom{$D_{\mathrm{pair}}^{\mathrm{proj}}$}$25.00\pm0.00$\\[2.0pt]\vphantom{$D_{\mathrm{pair}}^{\mathrm{proj}}$}$25.36\pm0.05\,\checkmark$\\[2.0pt]\vphantom{$D_{\mathrm{pair}}^{\mathrm{proj}}$}$0.004\pm0.001$} & \shortstack[c]{\vphantom{$D_{\mathrm{pair}}^{\mathrm{proj}}$}$62.50\pm0.00$\\[2.0pt]\vphantom{$D_{\mathrm{pair}}^{\mathrm{proj}}$}$29.03\pm0.34$\\[2.0pt]\vphantom{$D_{\mathrm{pair}}^{\mathrm{proj}}$}$0.432\pm0.015$} & \shortstack[c]{\vphantom{$D_{\mathrm{pair}}^{\mathrm{proj}}$}$68.31\pm0.00$\\[2.0pt]\vphantom{$D_{\mathrm{pair}}^{\mathrm{proj}}$}$59.91\pm0.14$\\[2.0pt]\vphantom{$D_{\mathrm{pair}}^{\mathrm{proj}}$}$0.142\pm0.006$} & \shortstack[c]{\vphantom{$D_{\mathrm{pair}}^{\mathrm{proj}}$}$78.53\pm0.00$\\[2.0pt]\vphantom{$D_{\mathrm{pair}}^{\mathrm{proj}}$}$67.66\pm0.98$\\[2.0pt]\vphantom{$D_{\mathrm{pair}}^{\mathrm{proj}}$}$0.201\pm0.032$} \\
\addlinespace[4.0pt]
\raisebox{14pt}[0pt][0pt]{\textbf{CollegeMsg}} & \shortstack[l]{\vphantom{$D_{\mathrm{pair}}^{\mathrm{proj}}$}$D_{\rm pair}^{\rm proj}$ (\%)\\[2.0pt]\vphantom{$D_{\mathrm{pair}}^{\mathrm{proj}}$}$D_{\rm pair}^{\theta}$ (\%)\\[2.0pt]\vphantom{$D_{\mathrm{pair}}^{\mathrm{proj}}$}KL (nats)} & \shortstack[c]{\vphantom{$D_{\mathrm{pair}}^{\mathrm{proj}}$}$58.38\pm2.26$\\[2.0pt]\vphantom{$D_{\mathrm{pair}}^{\mathrm{proj}}$}$64.70\pm3.20$\\[2.0pt]\vphantom{$D_{\mathrm{pair}}^{\mathrm{proj}}$}$0.628\pm0.042$} & \shortstack[c]{\vphantom{$D_{\mathrm{pair}}^{\mathrm{proj}}$}$37.83\pm5.97$\\[2.0pt]\vphantom{$D_{\mathrm{pair}}^{\mathrm{proj}}$}$61.48\pm7.22$\\[2.0pt]\vphantom{$D_{\mathrm{pair}}^{\mathrm{proj}}$}$0.693\pm0.026$} & \shortstack[c]{\vphantom{$D_{\mathrm{pair}}^{\mathrm{proj}}$}$78.53\pm0.00$\\[2.0pt]\vphantom{$D_{\mathrm{pair}}^{\mathrm{proj}}$}$68.87\pm1.42$\\[2.0pt]\vphantom{$D_{\mathrm{pair}}^{\mathrm{proj}}$}$0.190\pm0.014$} & \shortstack[c]{\vphantom{$D_{\mathrm{pair}}^{\mathrm{proj}}$}$25.00\pm0.00$\\[2.0pt]\vphantom{$D_{\mathrm{pair}}^{\mathrm{proj}}$}$33.78\pm4.78$\\[2.0pt]\vphantom{$D_{\mathrm{pair}}^{\mathrm{proj}}$}$0.758\pm0.172$} & \cellcolor{yellow!12}\shortstack[c]{\vphantom{$D_{\mathrm{pair}}^{\mathrm{proj}}$}$25.00\pm0.00$\\[2.0pt]\vphantom{$D_{\mathrm{pair}}^{\mathrm{proj}}$}$25.36\pm0.06\,\checkmark$\\[2.0pt]\vphantom{$D_{\mathrm{pair}}^{\mathrm{proj}}$}$0.002\pm0.000$} & \shortstack[c]{\vphantom{$D_{\mathrm{pair}}^{\mathrm{proj}}$}$25.00\pm0.00$\\[2.0pt]\vphantom{$D_{\mathrm{pair}}^{\mathrm{proj}}$}$25.42\pm0.12\,\checkmark$\\[2.0pt]\vphantom{$D_{\mathrm{pair}}^{\mathrm{proj}}$}$0.005\pm0.002$} & \shortstack[c]{\vphantom{$D_{\mathrm{pair}}^{\mathrm{proj}}$}$25.00\pm0.00$\\[2.0pt]\vphantom{$D_{\mathrm{pair}}^{\mathrm{proj}}$}$25.36\pm0.04\,\checkmark$\\[2.0pt]\vphantom{$D_{\mathrm{pair}}^{\mathrm{proj}}$}$0.004\pm0.001$} & \shortstack[c]{\vphantom{$D_{\mathrm{pair}}^{\mathrm{proj}}$}$62.50\pm0.00$\\[2.0pt]\vphantom{$D_{\mathrm{pair}}^{\mathrm{proj}}$}$29.05\pm0.32$\\[2.0pt]\vphantom{$D_{\mathrm{pair}}^{\mathrm{proj}}$}$0.431\pm0.015$} & \shortstack[c]{\vphantom{$D_{\mathrm{pair}}^{\mathrm{proj}}$}$68.31\pm0.00$\\[2.0pt]\vphantom{$D_{\mathrm{pair}}^{\mathrm{proj}}$}$59.91\pm0.14$\\[2.0pt]\vphantom{$D_{\mathrm{pair}}^{\mathrm{proj}}$}$0.142\pm0.006$} & \shortstack[c]{\vphantom{$D_{\mathrm{pair}}^{\mathrm{proj}}$}$78.53\pm0.00$\\[2.0pt]\vphantom{$D_{\mathrm{pair}}^{\mathrm{proj}}$}$67.66\pm0.98$\\[2.0pt]\vphantom{$D_{\mathrm{pair}}^{\mathrm{proj}}$}$0.201\pm0.032$} \\
\addlinespace[4.0pt]
\rowcolor{blue!4}
\raisebox{14pt}[0pt][0pt]{\textbf{MathOverflow}} & \shortstack[l]{\vphantom{$D_{\mathrm{pair}}^{\mathrm{proj}}$}$D_{\rm pair}^{\rm proj}$ (\%)\\[2.0pt]\vphantom{$D_{\mathrm{pair}}^{\mathrm{proj}}$}$D_{\rm pair}^{\theta}$ (\%)\\[2.0pt]\vphantom{$D_{\mathrm{pair}}^{\mathrm{proj}}$}KL (nats)} & \shortstack[c]{\vphantom{$D_{\mathrm{pair}}^{\mathrm{proj}}$}$58.60\pm0.87$\\[2.0pt]\vphantom{$D_{\mathrm{pair}}^{\mathrm{proj}}$}$66.10\pm0.31$\\[2.0pt]\vphantom{$D_{\mathrm{pair}}^{\mathrm{proj}}$}$0.652\pm0.018$} & \shortstack[c]{\vphantom{$D_{\mathrm{pair}}^{\mathrm{proj}}$}$38.03\pm1.63$\\[2.0pt]\vphantom{$D_{\mathrm{pair}}^{\mathrm{proj}}$}$61.96\pm6.34$\\[2.0pt]\vphantom{$D_{\mathrm{pair}}^{\mathrm{proj}}$}$0.732\pm0.042$} & \shortstack[c]{\vphantom{$D_{\mathrm{pair}}^{\mathrm{proj}}$}$78.53\pm0.00$\\[2.0pt]\vphantom{$D_{\mathrm{pair}}^{\mathrm{proj}}$}$68.91\pm1.41$\\[2.0pt]\vphantom{$D_{\mathrm{pair}}^{\mathrm{proj}}$}$0.189\pm0.013$} & \shortstack[c]{\vphantom{$D_{\mathrm{pair}}^{\mathrm{proj}}$}$25.00\pm0.00$\\[2.0pt]\vphantom{$D_{\mathrm{pair}}^{\mathrm{proj}}$}$44.16\pm17.90$\\[2.0pt]\vphantom{$D_{\mathrm{pair}}^{\mathrm{proj}}$}$0.850\pm0.236$} & \cellcolor{yellow!12}\shortstack[c]{\vphantom{$D_{\mathrm{pair}}^{\mathrm{proj}}$}$25.00\pm0.00$\\[2.0pt]\vphantom{$D_{\mathrm{pair}}^{\mathrm{proj}}$}$25.35\pm0.06\,\checkmark$\\[2.0pt]\vphantom{$D_{\mathrm{pair}}^{\mathrm{proj}}$}$0.002\pm0.000$} & \shortstack[c]{\vphantom{$D_{\mathrm{pair}}^{\mathrm{proj}}$}$25.00\pm0.00$\\[2.0pt]\vphantom{$D_{\mathrm{pair}}^{\mathrm{proj}}$}$25.42\pm0.12\,\checkmark$\\[2.0pt]\vphantom{$D_{\mathrm{pair}}^{\mathrm{proj}}$}$0.005\pm0.001$} & \shortstack[c]{\vphantom{$D_{\mathrm{pair}}^{\mathrm{proj}}$}$25.00\pm0.00$\\[2.0pt]\vphantom{$D_{\mathrm{pair}}^{\mathrm{proj}}$}$25.36\pm0.04\,\checkmark$\\[2.0pt]\vphantom{$D_{\mathrm{pair}}^{\mathrm{proj}}$}$0.004\pm0.000$} & \shortstack[c]{\vphantom{$D_{\mathrm{pair}}^{\mathrm{proj}}$}$62.50\pm0.00$\\[2.0pt]\vphantom{$D_{\mathrm{pair}}^{\mathrm{proj}}$}$29.03\pm0.33$\\[2.0pt]\vphantom{$D_{\mathrm{pair}}^{\mathrm{proj}}$}$0.432\pm0.015$} & \shortstack[c]{\vphantom{$D_{\mathrm{pair}}^{\mathrm{proj}}$}$68.31\pm0.00$\\[2.0pt]\vphantom{$D_{\mathrm{pair}}^{\mathrm{proj}}$}$59.91\pm0.15$\\[2.0pt]\vphantom{$D_{\mathrm{pair}}^{\mathrm{proj}}$}$0.142\pm0.006$} & \shortstack[c]{\vphantom{$D_{\mathrm{pair}}^{\mathrm{proj}}$}$78.53\pm0.00$\\[2.0pt]\vphantom{$D_{\mathrm{pair}}^{\mathrm{proj}}$}$67.66\pm0.97$\\[2.0pt]\vphantom{$D_{\mathrm{pair}}^{\mathrm{proj}}$}$0.201\pm0.031$} \\
\addlinespace[4.0pt]
\raisebox{14pt}[0pt][0pt]{\textbf{PROTEINS}} & \shortstack[l]{\vphantom{$D_{\mathrm{pair}}^{\mathrm{proj}}$}$D_{\rm pair}^{\rm proj}$ (\%)\\[2.0pt]\vphantom{$D_{\mathrm{pair}}^{\mathrm{proj}}$}$D_{\rm pair}^{\theta}$ (\%)\\[2.0pt]\vphantom{$D_{\mathrm{pair}}^{\mathrm{proj}}$}KL (nats)} & \shortstack[c]{\vphantom{$D_{\mathrm{pair}}^{\mathrm{proj}}$}$59.58\pm1.27$\\[2.0pt]\vphantom{$D_{\mathrm{pair}}^{\mathrm{proj}}$}$66.30\pm0.22$\\[2.0pt]\vphantom{$D_{\mathrm{pair}}^{\mathrm{proj}}$}$0.627\pm0.029$} & \shortstack[c]{\vphantom{$D_{\mathrm{pair}}^{\mathrm{proj}}$}$40.33\pm1.73$\\[2.0pt]\vphantom{$D_{\mathrm{pair}}^{\mathrm{proj}}$}$61.40\pm7.62$\\[2.0pt]\vphantom{$D_{\mathrm{pair}}^{\mathrm{proj}}$}$0.723\pm0.059$} & \shortstack[c]{\vphantom{$D_{\mathrm{pair}}^{\mathrm{proj}}$}$78.53\pm0.00$\\[2.0pt]\vphantom{$D_{\mathrm{pair}}^{\mathrm{proj}}$}$68.88\pm1.41$\\[2.0pt]\vphantom{$D_{\mathrm{pair}}^{\mathrm{proj}}$}$0.189\pm0.015$} & \shortstack[c]{\vphantom{$D_{\mathrm{pair}}^{\mathrm{proj}}$}$25.00\pm0.00$\\[2.0pt]\vphantom{$D_{\mathrm{pair}}^{\mathrm{proj}}$}$32.34\pm0.96$\\[2.0pt]\vphantom{$D_{\mathrm{pair}}^{\mathrm{proj}}$}$1.109\pm0.143$} & \cellcolor{yellow!12}\shortstack[c]{\vphantom{$D_{\mathrm{pair}}^{\mathrm{proj}}$}$25.00\pm0.00$\\[2.0pt]\vphantom{$D_{\mathrm{pair}}^{\mathrm{proj}}$}$25.35\pm0.07\,\checkmark$\\[2.0pt]\vphantom{$D_{\mathrm{pair}}^{\mathrm{proj}}$}$0.002\pm0.000$} & \shortstack[c]{\vphantom{$D_{\mathrm{pair}}^{\mathrm{proj}}$}$25.00\pm0.00$\\[2.0pt]\vphantom{$D_{\mathrm{pair}}^{\mathrm{proj}}$}$25.42\pm0.12\,\checkmark$\\[2.0pt]\vphantom{$D_{\mathrm{pair}}^{\mathrm{proj}}$}$0.005\pm0.001$} & \shortstack[c]{\vphantom{$D_{\mathrm{pair}}^{\mathrm{proj}}$}$25.00\pm0.00$\\[2.0pt]\vphantom{$D_{\mathrm{pair}}^{\mathrm{proj}}$}$25.36\pm0.04\,\checkmark$\\[2.0pt]\vphantom{$D_{\mathrm{pair}}^{\mathrm{proj}}$}$0.004\pm0.001$} & \shortstack[c]{\vphantom{$D_{\mathrm{pair}}^{\mathrm{proj}}$}$62.50\pm0.00$\\[2.0pt]\vphantom{$D_{\mathrm{pair}}^{\mathrm{proj}}$}$29.02\pm0.35$\\[2.0pt]\vphantom{$D_{\mathrm{pair}}^{\mathrm{proj}}$}$0.433\pm0.016$} & \shortstack[c]{\vphantom{$D_{\mathrm{pair}}^{\mathrm{proj}}$}$68.31\pm0.00$\\[2.0pt]\vphantom{$D_{\mathrm{pair}}^{\mathrm{proj}}$}$59.91\pm0.15$\\[2.0pt]\vphantom{$D_{\mathrm{pair}}^{\mathrm{proj}}$}$0.142\pm0.006$} & \shortstack[c]{\vphantom{$D_{\mathrm{pair}}^{\mathrm{proj}}$}$78.53\pm0.00$\\[2.0pt]\vphantom{$D_{\mathrm{pair}}^{\mathrm{proj}}$}$67.65\pm0.97$\\[2.0pt]\vphantom{$D_{\mathrm{pair}}^{\mathrm{proj}}$}$0.201\pm0.031$} \\
\addlinespace[4.0pt]
\rowcolor{blue!4}
\raisebox{14pt}[0pt][0pt]{\textbf{COLLAB}} & \shortstack[l]{\vphantom{$D_{\mathrm{pair}}^{\mathrm{proj}}$}$D_{\rm pair}^{\rm proj}$ (\%)\\[2.0pt]\vphantom{$D_{\mathrm{pair}}^{\mathrm{proj}}$}$D_{\rm pair}^{\theta}$ (\%)\\[2.0pt]\vphantom{$D_{\mathrm{pair}}^{\mathrm{proj}}$}KL (nats)} & \shortstack[c]{\vphantom{$D_{\mathrm{pair}}^{\mathrm{proj}}$}$58.69\pm1.49$\\[2.0pt]\vphantom{$D_{\mathrm{pair}}^{\mathrm{proj}}$}$65.26\pm1.79$\\[2.0pt]\vphantom{$D_{\mathrm{pair}}^{\mathrm{proj}}$}$0.643\pm0.029$} & \shortstack[c]{\vphantom{$D_{\mathrm{pair}}^{\mathrm{proj}}$}$38.51\pm1.33$\\[2.0pt]\vphantom{$D_{\mathrm{pair}}^{\mathrm{proj}}$}$60.58\pm8.95$\\[2.0pt]\vphantom{$D_{\mathrm{pair}}^{\mathrm{proj}}$}$0.729\pm0.075$} & \shortstack[c]{\vphantom{$D_{\mathrm{pair}}^{\mathrm{proj}}$}$78.53\pm0.00$\\[2.0pt]\vphantom{$D_{\mathrm{pair}}^{\mathrm{proj}}$}$68.86\pm1.43$\\[2.0pt]\vphantom{$D_{\mathrm{pair}}^{\mathrm{proj}}$}$0.191\pm0.017$} & \shortstack[c]{\vphantom{$D_{\mathrm{pair}}^{\mathrm{proj}}$}$25.00\pm0.00$\\[2.0pt]\vphantom{$D_{\mathrm{pair}}^{\mathrm{proj}}$}$33.11\pm0.93$\\[2.0pt]\vphantom{$D_{\mathrm{pair}}^{\mathrm{proj}}$}$0.921\pm0.047$} & \cellcolor{yellow!12}\shortstack[c]{\vphantom{$D_{\mathrm{pair}}^{\mathrm{proj}}$}$25.00\pm0.00$\\[2.0pt]\vphantom{$D_{\mathrm{pair}}^{\mathrm{proj}}$}$25.35\pm0.06\,\checkmark$\\[2.0pt]\vphantom{$D_{\mathrm{pair}}^{\mathrm{proj}}$}$0.002\pm0.000$} & \shortstack[c]{\vphantom{$D_{\mathrm{pair}}^{\mathrm{proj}}$}$25.00\pm0.00$\\[2.0pt]\vphantom{$D_{\mathrm{pair}}^{\mathrm{proj}}$}$25.42\pm0.12\,\checkmark$\\[2.0pt]\vphantom{$D_{\mathrm{pair}}^{\mathrm{proj}}$}$0.005\pm0.001$} & \shortstack[c]{\vphantom{$D_{\mathrm{pair}}^{\mathrm{proj}}$}$25.00\pm0.00$\\[2.0pt]\vphantom{$D_{\mathrm{pair}}^{\mathrm{proj}}$}$25.36\pm0.04\,\checkmark$\\[2.0pt]\vphantom{$D_{\mathrm{pair}}^{\mathrm{proj}}$}$0.004\pm0.001$} & \shortstack[c]{\vphantom{$D_{\mathrm{pair}}^{\mathrm{proj}}$}$62.50\pm0.00$\\[2.0pt]\vphantom{$D_{\mathrm{pair}}^{\mathrm{proj}}$}$29.03\pm0.34$\\[2.0pt]\vphantom{$D_{\mathrm{pair}}^{\mathrm{proj}}$}$0.432\pm0.015$} & \shortstack[c]{\vphantom{$D_{\mathrm{pair}}^{\mathrm{proj}}$}$68.31\pm0.00$\\[2.0pt]\vphantom{$D_{\mathrm{pair}}^{\mathrm{proj}}$}$59.91\pm0.15$\\[2.0pt]\vphantom{$D_{\mathrm{pair}}^{\mathrm{proj}}$}$0.142\pm0.006$} & \shortstack[c]{\vphantom{$D_{\mathrm{pair}}^{\mathrm{proj}}$}$78.53\pm0.00$\\[2.0pt]\vphantom{$D_{\mathrm{pair}}^{\mathrm{proj}}$}$67.65\pm0.97$\\[2.0pt]\vphantom{$D_{\mathrm{pair}}^{\mathrm{proj}}$}$0.201\pm0.032$} \\
\bottomrule
\end{tabular}%
}
\endgroup
\end{table*}

\begin{figure}[!t]
\centering
\includegraphics[width=\linewidth]{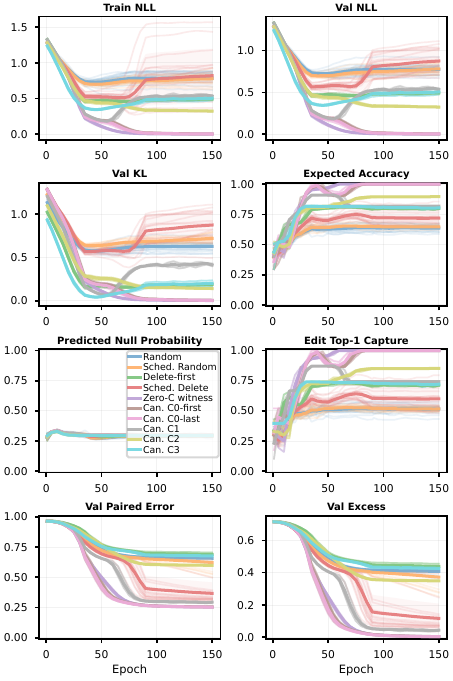}
\caption{
Aggregated graph-only student training trajectories
across five controlled benchmarks.
Above panels summarize training, prediction, and endpoint metrics. See details in Appendix~\ref{app:extended-experiments}.
}
\label{fig:exp-training}
\end{figure}

\paragraph{Ablations.}
We compare GraphSAGE with public-only and edge-MLP students,
and remove endpoint supervision, pooling, or the clock input
under the same controlled protocol \citep{you2020design}.

\begin{table}[htbp]
\centering
\small
\setlength{\tabcolsep}{1.5pt}
\caption{Architecture and component ablations on PROTEINS, sharing legal actions and
a Poisson budget. Clock excess is paired error above
the 25\% Poisson floor (pp).}
\label{tab:exp-ablation}
\begin{tabular}{@{}lcc@{}}
\toprule
Variant & $D_{\rm pair}^{\theta}$ (\%) & Clock excess (pp) \\
\midrule
GraphSAGE \textcolor{softburgundy}{(reference)} & $25.35\pm0.07$ & $0.35\pm0.07$ \\
Public-only & $29.56\pm0.06$ & $4.56\pm0.06$ \\
Edge-MLP & $27.40\pm0.05$ & $2.40\pm0.05$ \\
No endpoint loss & $46.87\pm0.42$ & $21.87\pm0.42$ \\
No pooling & $28.31\pm0.06$ & $3.31\pm0.06$ \\
No clock input & $33.26\pm0.05$ & $8.26\pm0.05$ \\
\bottomrule
\end{tabular}
\end{table}

\paragraph{Summary.}
Numerical verification separates marginal preservation
from endpoint recoverability, motivating paired-error
evaluation. R1 highlights the gap between structural
improvements and generalization to observed pairs.
R2 isolates ordering effects and shows that endpoint-trained
students can outperform teacher-induced Markov projections
without matching their action laws. Ablations identify
endpoint supervision and clock conditioning as most
influential, with gains from graph encoding
and pooling.

\section{Conclusion}

In this work, we investigate when legal edit ordering
preserves source--target couplings in graph flow matching.
Under a common Poisson clock, we establish an exact
recoverability criterion, develop a legality-preserving
scheduling procedure with recovery certificates, and derive
a paired-error bound linking structural ambiguity and
learning errors to endpoint fidelity. Numerical and
graph-only student experiments support these findings.
While our framework assumes deterministic pairings and
legal shortest edits, future work may extend it to
stochastic pairings, partially observed graphs, and
more general edit dynamics.

\section*{AI Use Statement}

AI tools were used to assist with manuscript preparation,
language refinement, and code development.
The authors take full responsibility for all technical
content, experimental results, and conclusions.

\section*{Impact Statement}

This work advances the understanding of graph generation
and editing. Its broader societal impacts depend on
downstream applications, particularly when graph editing
is applied to sensitive or safety-critical systems.

\bibliography{references}

@inproceedings{shaul2025flow,
  author    = {N. Shaul and I. Gat and M. Havasi and D. Severo
               and A. Sriram and P. Holderrieth and B. Karrer
               and Y. Lipman and R. T. Q. Chen},
  title     = {Flow Matching with General Discrete Paths:
               A Kinetic-Optimal Perspective},
  year      = {2025},
  booktitle = {International Conference on Learning Representations
               (ICLR 2025)}
}

@inproceedings{holderrieth2025generator,
  author    = {P. Holderrieth and M. Havasi and J. Yim and N. Shaul
               and I. Gat and T. Jaakkola and B. Karrer
               and R. T. Q. Chen and Y. Lipman},
  title     = {Generator Matching: Generative Modeling with
               Arbitrary {M}arkov Processes},
  year      = {2025},
  booktitle = {International Conference on Learning Representations
               (ICLR 2025)}
}

@inproceedings{qin2025defog,
  author    = {Y. Qin and M. Madeira and D. Thanou and P. Frossard},
  title     = {{DeFoG}: Discrete Flow Matching for Graph Generation},
  year      = {2025},
  pages     = {50269--50326},
  booktitle = {Proceedings of the 42nd International Conference
               on Machine Learning},
  volume    = {267},
  publisher = {PMLR}
}

@inproceedings{madeira2024generative,
  author    = {M. Madeira and C. Vignac and D. Thanou and P. Frossard},
  title     = {Generative Modelling of Structurally Constrained Graphs},
  year      = {2024},
  booktitle = {Advances in Neural Information Processing Systems
               (NeurIPS 2024)},
  volume    = {37}
}

@misc{havasi2025edit,
  author    = {M. Havasi and B. Karrer and I. Gat and R. T. Q. Chen},
  title     = {Edit Flows: Flow Matching with Edit Operations},
  year      = {2025},
  note      = {arXiv preprint arXiv:2506.09018}
}

@misc{debortoli2023augmented,
  author    = {V. De Bortoli and G.-H. Liu and T. Chen
               and E. A. Theodorou and W. Nie},
  title     = {Augmented Bridge Matching},
  year      = {2023},
  note      = {arXiv preprint arXiv:2311.06978}
}

@inproceedings{kim2025ddsbm,
  author    = {J. H. Kim and S. Kim and S. Moon and H. Kim
               and J. Woo and W. Y. Kim},
  title     = {Discrete Diffusion {Schr\"odinger} Bridge Matching
               for Graph Transformation},
  year      = {2025},
  booktitle = {International Conference on Learning Representations
               (ICLR 2025)}
}

@inproceedings{kim2024simulationfree,
  author    = {S. Kim and J. Yoo and J. Kim and Y. Cha
               and S. Kim and S. Hong},
  title     = {Simulation-Free Training of Neural {ODE}s on Paired Data},
  year      = {2024},
  booktitle = {Advances in Neural Information Processing Systems
               (NeurIPS 2024)},
  volume    = {37}
}

@inproceedings{zhang2024trajectory,
  author    = {X. Zhang and Y. Pu and Y. Kawamura and A. Loza
               and Y. Bengio and D. L. Shung and A. Tong},
  title     = {Trajectory Flow Matching with Applications to
               Clinical Time Series Modelling},
  year      = {2024},
  booktitle = {Advances in Neural Information Processing Systems
               (NeurIPS 2024)},
  volume    = {37}
}

@misc{kemper2026trajectory,
  author    = {N. Kemper and M. Lienen and D. L{\"u}dke
               and S. G{\"u}nnemann},
  title     = {Trajectory Matching: Discrete Flow Matching
               with Structured Actions},
  year      = {2026},
  note      = {arXiv preprint arXiv:2610.32602}
}

@inproceedings{patel2026insertion,
  author    = {D. Patel and B. Rozonoyer and G. Pandey
               and T. Naseem and R. F. Astudillo and A. McCallum},
  title     = {Insertion Based Sequence Generation with
               Learnable Order Dynamics},
  year      = {2026},
  pages     = {97406--97447},
  booktitle = {Proceedings of the 43rd International Conference
               on Machine Learning},
  volume    = {306},
  publisher = {PMLR}
}

@inproceedings{botea2021counting,
  author    = {A. Botea and M. Mattetti and A. Kishimoto
               and R. Marinescu and E. Daly},
  title     = {Counting Vertex-Disjoint Shortest Paths in Graphs},
  year      = {2021},
  booktitle = {Proceedings of the International Symposium
               on Combinatorial Search (SoCS 2021)},
  volume    = {12}
}

@article{saito2024reconfiguration,
  author    = {R. Saito and H. Eto and T. Ito and R. Uehara},
  title     = {Reconfiguration of Vertex-Disjoint Shortest Paths
               on Graphs},
  year      = {2024},
  journal   = {Journal of Graph Algorithms and Applications},
  volume    = {28},
  number    = {3},
  pages     = {87--101}
}

@article{gajjar2024reconfiguring,
  author    = {K. Gajjar and A. V. Jha and M. Kumar and A. Lahiri},
  title     = {Reconfiguring Shortest Paths in Graphs},
  year      = {2024},
  journal   = {Algorithmica},
  volume    = {86},
  pages     = {3309--3338}
}

@inproceedings{gat2024discrete,
  author    = {I. Gat and T. Remez and N. Shaul and F. Kreuk
               and R. T. Q. Chen and G. Synnaeve and Y. Adi
               and Y. Lipman},
  title     = {Discrete Flow Matching},
  year      = {2024},
  booktitle = {Advances in Neural Information Processing Systems
               (NeurIPS 2024)},
  volume    = {37}
}

@inproceedings{liu2023rectified,
  author    = {X. Liu and C. Gong and Q. Liu},
  title     = {Flow Straight and Fast: Learning to Generate and
               Transfer Data with Rectified Flow},
  year      = {2023},
  booktitle = {International Conference on Learning Representations
               (ICLR 2023)}
}

@article{gao2010survey,
  author    = {X. Gao and B. Xiao and D. Tao and X. Li},
  title     = {A Survey of Graph Edit Distance},
  journal   = {Pattern Analysis and Applications},
  year      = {2010},
  volume    = {13},
  number    = {1},
  pages     = {113--129}
}

@Book{norris1997markov,
  author    = {J. R. Norris},
  title     = {Markov Chains},
  publisher = {Cambridge University Press},
  year      = {1997}
}

@inproceedings{lipman2023flow,
  author    = {Y. Lipman and R. T. Q. Chen and H. Ben-Hamu
               and M. Nickel and M. Le},
  title     = {Flow Matching for Generative Modeling},
  year      = {2023},
  booktitle = {International Conference on Learning Representations
               (ICLR 2023)}
}

@inproceedings{vignac2023digress,
  author    = {C. Vignac and I. Krawczuk and A. Siraudin
               and B. Wang and V. Cevher and P. Frossard},
  title     = {{DiGress}: Discrete Denoising Diffusion
               for Graph Generation},
  year      = {2023},
  booktitle = {International Conference on Learning Representations
               (ICLR 2023)}
}

@inproceedings{heo2024epic,
  author    = {J. Heo and S. Lee and S. Ahn and D. Kim},
  title     = {{EPIC}: Graph Augmentation with Edit Path
               Interpolation via Learnable Cost},
  year      = {2024},
  pages     = {4116--4126},
  booktitle = {Proceedings of the 33rd International Joint
               Conference on Artificial Intelligence (IJCAI 2024)},
  publisher = {IJCAI}
}

@book{cover2006elements,
  author    = {Cover, Thomas M. and Thomas, Joy A.},
  title     = {Elements of Information Theory},
  edition   = {2},
  publisher = {John Wiley \& Sons},
  year      = {2006},
  doi       = {10.1002/047174882X}
}

@book{csiszar2011information,
  author    = {Csisz{\'a}r, Imre and K{\"o}rner, J{\'a}nos},
  title     = {Information Theory: Coding Theorems for Discrete Memoryless Systems},
  edition   = {2},
  publisher = {Cambridge University Press},
  year      = {2011}
}

@book{diestel2017graph,
  author    = {Diestel, Reinhard},
  title     = {Graph Theory},
  edition   = {5},
  series    = {Graduate Texts in Mathematics},
  volume    = {173},
  publisher = {Springer},
  year      = {2017},
  doi       = {10.1007/978-3-662-53622-3}
}

@book{korte2018combinatorial,
  author    = {Korte, Bernhard and Vygen, Jens},
  title     = {Combinatorial Optimization: Theory and Algorithms},
  edition   = {6},
  publisher = {Springer},
  year      = {2018},
  doi       = {10.1007/978-3-662-56039-6}
}

@book{lindvall1992lectures,
  author    = {Lindvall, Torgny},
  title     = {Lectures on the Coupling Method},
  publisher = {John Wiley \& Sons},
  year      = {1992}
}

@book{thorisson2000coupling,
  author    = {Thorisson, Hermann},
  title     = {Coupling, Stationarity, and Regeneration},
  publisher = {Springer},
  year      = {2000}
}

@book{kingman1993poisson,
  author    = {Kingman, J. F. C.},
  title     = {Poisson Processes},
  publisher = {Oxford University Press},
  year      = {1993}
}

@book{bremaud1981point,
  author    = {Br{\'e}maud, Pierre},
  title     = {Point Processes and Queues: Martingale Dynamics},
  publisher = {Springer},
  year      = {1981}
}

@book{anderson1991continuous,
  author    = {Anderson, William J.},
  title     = {Continuous-Time Markov Chains: An Applications-Oriented Approach},
  publisher = {Springer},
  year      = {1991},
  doi       = {10.1007/978-1-4612-3038-0}
}

@book{ahuja1993network,
  author    = {Ahuja, Ravindra K. and Magnanti, Thomas L. and Orlin, James B.},
  title     = {Network Flows: Theory, Algorithms, and Applications},
  publisher = {Prentice Hall},
  year      = {1993}
}

@book{schrijver2003combinatorial,
  author    = {Schrijver, Alexander},
  title     = {Combinatorial Optimization: Polyhedra and Efficiency},
  publisher = {Springer},
  year      = {2003}
}

@book{bertsimas1997linear,
  author    = {Bertsimas, Dimitris and Tsitsiklis, John N.},
  title     = {Introduction to Linear Optimization},
  publisher = {Athena Scientific},
  year      = {1997}
}

@book{goodfellow2016deep,
  author    = {Goodfellow, Ian and Bengio, Yoshua and Courville, Aaron},
  title     = {Deep Learning},
  publisher = {MIT Press},
  year      = {2016},
  url       = {https://www.deeplearningbook.org/}
}

@book{murphy2022probabilistic,
  author    = {Murphy, Kevin P.},
  title     = {Probabilistic Machine Learning: An Introduction},
  publisher = {MIT Press},
  year      = {2022}
}

@book{tsybakov2009nonparametric,
  author    = {Tsybakov, Alexandre B.},
  title     = {Introduction to Nonparametric Estimation},
  publisher = {Springer},
  year      = {2009},
  doi       = {10.1007/b13794}
}

@misc{leskovec2014snap,
  author       = {Leskovec, Jure and Krevl, Andrej},
  title        = {{SNAP Datasets}: {Stanford} Large Network Dataset Collection},
  year         = {2014},
  howpublished = {Stanford University},
  url          = {https://snap.stanford.edu/data/}
}

@inproceedings{paranjape2017motifs,
  author    = {Paranjape, Ashwin and Benson, Austin R. and Leskovec, Jure},
  title     = {Motifs in Temporal Networks},
  booktitle = {Proceedings of the Tenth ACM International Conference on Web Search and Data Mining},
  pages     = {601--610},
  year      = {2017},
  publisher = {ACM},
  doi       = {10.1145/3018661.3018731}
}

@misc{rossi2020temporal,
  author        = {Rossi, Emanuele and Chamberlain, Ben and Frasca, Fabrizio and Eynard, Davide and Monti, Federico and Bronstein, Michael},
  title         = {Temporal Graph Networks for Deep Learning on Dynamic Graphs},
  year          = {2020},
  eprint        = {2006.10637},
  archivePrefix = {arXiv},
  primaryClass  = {cs.LG},
  url           = {https://arxiv.org/abs/2006.10637}
}

@inproceedings{poursafaei2022towards,
  author    = {Poursafaei, Farimah and Huang, Shenyang and Pelrine, Kellin and Rabbany, Reihaneh},
  title     = {Towards Better Evaluation for Dynamic Link Prediction},
  booktitle = {Advances in Neural Information Processing Systems},
  volume    = {35},
  year      = {2022}
}

@inproceedings{huang2023temporal,
  author    = {Huang, Shenyang and Poursafaei, Farimah and Danovitch, Jacob and Fey, Matthias and Hu, Weihua and Rossi, Emanuele and Leskovec, Jure and Bronstein, Michael and Rabusseau, Guillaume and Rabbany, Reihaneh},
  title     = {Temporal Graph Benchmark for Machine Learning on Temporal Graphs},
  booktitle = {Advances in Neural Information Processing Systems},
  volume    = {36},
  year      = {2023}
}

@inproceedings{morris2020tudataset,
  author    = {Morris, Christopher and Kriege, Nils M. and Bause, Franka and Kersting, Kristian and Mutzel, Petra and Neumann, Marion},
  title     = {{TUDataset}: A Collection of Benchmark Datasets for Learning with Graphs},
  booktitle = {2020 International Conference on Machine Learning Workshop on Graph Representation Learning and Beyond},
  year      = {2020},
  eprint    = {2007.08663},
  url       = {https://arxiv.org/abs/2007.08663}
}

@inproceedings{hamilton2017inductive,
  author    = {Hamilton, Will and Ying, Zhitao and Leskovec, Jure},
  title     = {Inductive Representation Learning on Large Graphs},
  booktitle = {Advances in Neural Information Processing Systems},
  volume    = {30},
  year      = {2017}
}

@inproceedings{paszke2019pytorch,
  author    = {Paszke, Adam and Gross, Sam and Massa, Francisco and Lerer, Adam and Bradbury, James and Chanan, Gregory and Killeen, Trevor and Lin, Zeming and Gimelshein, Natalia and Antiga, Luca and Desmaison, Alban and Kopf, Andreas and Yang, Edward and DeVito, Zachary and Raison, Martin and Tejani, Alykhan and Chilamkurthy, Sasank and Steiner, Benoit and Fang, Lu and Bai, Junjie and Chintala, Soumith},
  title     = {{PyTorch}: An Imperative Style, High-Performance Deep Learning Library},
  booktitle = {Advances in Neural Information Processing Systems},
  volume    = {32},
  pages     = {8024--8035},
  year      = {2019}
}

@inproceedings{you2020design,
  author    = {You, Jiaxuan and Ying, Rex and Leskovec, Jure},
  title     = {Design Space for Graph Neural Networks},
  booktitle = {Advances in Neural Information Processing Systems},
  volume    = {33},
  year      = {2020}
}

@inproceedings{xu2019powerful,
  author    = {Xu, Keyulu and Hu, Weihua and Leskovec, Jure and Jegelka, Stefanie},
  title     = {How Powerful Are Graph Neural Networks?},
  booktitle = {International Conference on Learning Representations},
  year      = {2019},
  url       = {https://arxiv.org/abs/1810.00826}
}

@book{daley2003point,
  author    = {Daley, Daryl J. and Vere-Jones, David},
  title     = {An Introduction to the Theory of Point Processes, Volume I: Elementary Theory and Methods},
  edition   = {2},
  series    = {Probability and Its Applications},
  publisher = {Springer},
  year      = {2003},
  doi       = {10.1007/b97277}
}

@book{ethier1986markov,
  author    = {Ethier, Stewart N. and Kurtz, Thomas G.},
  title     = {Markov Processes: Characterization and Convergence},
  publisher = {John Wiley \& Sons},
  year      = {1986},
  doi       = {10.1002/9780470316658}
}

@book{coddington1955ode,
  author    = {Coddington, Earl A. and Levinson, Norman},
  title     = {Theory of Ordinary Differential Equations},
  publisher = {McGraw-Hill},
  year      = {1955}
}

@book{ford1962flows,
  author    = {Ford, Lester R., Jr. and Fulkerson, Delbert R.},
  title     = {Flows in Networks},
  publisher = {Princeton University Press},
  year      = {1962}
}

@book{bangjensen2009digraphs,
  author    = {Bang-Jensen, J{\o}rgen and Gutin, Gregory Z.},
  title     = {Digraphs: Theory, Algorithms and Applications},
  edition   = {2},
  publisher = {Springer},
  year      = {2009},
  doi       = {10.1007/978-1-84800-998-1}
}

@article{ito2011reconfiguration,
  author  = {Ito, Takehiro and Demaine, Erik D. and Harvey, Nicholas J. A. and Papadimitriou, Christos H. and Sideri, Martha and Uehara, Ryuhei and Uno, Yushi},
  title   = {On the Complexity of Reconfiguration Problems},
  journal = {Theoretical Computer Science},
  volume  = {412},
  number  = {12--14},
  pages   = {1054--1065},
  year    = {2011},
  doi     = {10.1016/j.tcs.2010.12.005}
}

@book{billingsley1995probability,
  author    = {Billingsley, Patrick},
  title     = {Probability and Measure},
  edition   = {3},
  publisher = {John Wiley \& Sons},
  year      = {1995}
}

@book{jacod2003limit,
  author    = {Jacod, Jean and Shiryaev, Albert N.},
  title     = {Limit Theorems for Stochastic Processes},
  edition   = {2},
  publisher = {Springer},
  year      = {2003},
  doi       = {10.1007/978-3-662-05265-5}
}

@book{gray2011entropy,
  author    = {Gray, Robert M.},
  title     = {Entropy and Information Theory},
  edition   = {2},
  publisher = {Springer},
  year      = {2011},
  doi       = {10.1007/978-1-4419-7970-4}
}

@inproceedings{luckie2013asrelationships,
  author    = {Luckie, Matthew and Huffaker, Bradley and claffy, k. and Dhamdhere, Amogh and Giotsas, Vasileios},
  title     = {{AS Relationships, Customer Cones, and Validation}},
  booktitle = {Proceedings of the 2013 Internet Measurement Conference},
  pages     = {243--256},
  year      = {2013},
  publisher = {ACM},
  doi       = {10.1145/2504730.2504735}
}

@inproceedings{jin2017reaction,
  author    = {Jin, Wengong and Coley, Connor and Barzilay, Regina and Jaakkola, Tommi},
  title     = {Predicting Organic Reaction Outcomes with {Weisfeiler-Lehman} Network},
  booktitle = {Advances in Neural Information Processing Systems},
  volume    = {30},
  year      = {2017}
}

@article{liu2017retrosynthetic,
  author    = {Liu, Bowen and Ramsundar, Bharath and Kawthekar, Prasad and Shi, Jade and Gomes, Joseph and Nguyen, Quang Luu and Ho, Stephen and Sloane, Jack and Wender, Paul and Pande, Vijay},
  title     = {Retrosynthetic Reaction Prediction Using Neural Sequence-to-Sequence Models},
  journal   = {ACS Central Science},
  volume    = {3},
  number    = {10},
  pages     = {1103--1113},
  year      = {2017},
  doi       = {10.1021/acscentsci.7b00303}
}

@inproceedings{jin2019graph2graph,
  author    = {Jin, Wengong and Yang, Kevin and Barzilay, Regina and Jaakkola, Tommi},
  title     = {Learning Multimodal Graph-to-Graph Translation for Molecule Optimization},
  booktitle = {International Conference on Learning Representations},
  year      = {2019},
  url       = {https://openreview.net/forum?id=Bkf8SNWM4}
}

@article{williams1992reinforce,
  author  = {Williams, Ronald J.},
  title   = {Simple Statistical Gradient-Following Algorithms for Connectionist Reinforcement Learning},
  journal = {Machine Learning},
  volume  = {8},
  number  = {3--4},
  pages   = {229--256},
  year    = {1992},
  doi     = {10.1007/BF00992696}
}

@inproceedings{perron2023cpsat,
  author    = {Perron, Laurent and Didier, Fr{\'e}d{\'e}ric and Gay, Steven},
  title     = {The {CP-SAT-LP} Solver},
  booktitle = {29th International Conference on Principles and Practice of Constraint Programming (CP 2023)},
  series    = {Leibniz International Proceedings in Informatics (LIPIcs)},
  volume    = {280},
  pages     = {3:1--3:2},
  year      = {2023},
  publisher = {Schloss Dagstuhl -- Leibniz-Zentrum f{\"u}r Informatik},
  doi       = {10.4230/LIPIcs.CP.2023.3}
}

@article{panzarasa2009patterns,
  author  = {Panzarasa, Pietro and Opsahl, Tore and Carley, Kathleen M.},
  title   = {Patterns and Dynamics of Users' Behavior and Interaction: Network Analysis of an Online Community},
  journal = {Journal of the American Society for Information Science and Technology},
  year    = {2009},
  volume  = {60},
  number  = {5},
  pages   = {911--932},
  doi     = {10.1002/asi.21015}
}

@inproceedings{luo2026simgfm,
  author    = {Luo, Chunyu and Luo, Yuankai and Wu, Xiao-Ming and Shi, Lei},
  title     = {{SimGFM}: Simplifying Discrete Flow Matching for Graph Generation},
  booktitle = {Proceedings of the 43rd International Conference on Machine Learning},
  series    = {Proceedings of Machine Learning Research},
  volume    = {306},
  pages     = {82637--82661},
  year      = {2026},
  publisher = {PMLR},
  url       = {https://proceedings.mlr.press/v306/luo26s.html}
}

@article{qin2025sparsediff,
  author    = {Qin, Yiming and Vignac, Cl{\'e}ment and Frossard, Pascal},
  title     = {{SparseDiff}: Sparse Discrete Diffusion for Scalable Graph Generation},
  journal   = {Transactions on Machine Learning Research},
  year      = {2025},
  url       = {https://openreview.net/forum?id=kuJ3lpxnVC}
}
\bibliographystyle{references}
\newpage

\onecolumn
\section*{Notation}
\label{sec:notation}

\begingroup
\footnotesize
\setlength{\LTpre}{5pt}
\setlength{\LTpost}{5pt}
\setlength{\tabcolsep}{4pt}
\renewcommand{\arraystretch}{1.38}

\begin{longtable}{@{}
>{\raggedright\arraybackslash}p{0.35\textwidth}
>{\raggedright\arraybackslash}p{0.61\textwidth}@{}}

\toprule
\textbf{Symbol} & \textbf{Meaning} \\
\midrule
\endfirsthead
\toprule
\textbf{Symbol} & \textbf{Meaning (continued)} \\
\midrule
\endhead
\bottomrule
\endfoot

\multicolumn{2}{@{}l}{
\textcolor{softburgundy}{\textit{Probability and set operations}}} \\
$\PP,\ \E$;\enspace $\calL(Z),\ \delta_z$
& Probability / expectation; law of $Z$ / Dirac mass at $z$. \\
$\Omega,\ P,Q$;\enspace $\|P-Q\|_{\TV},\ \KL(P\|Q)$
& Finite sample space / distributions; TV / KL divergence. \\
$\Pois(\lambda),\ \mathbf 1\{\cdot\}$;\enspace
$A\setminus B,\ A\triangle B,\ |A|$
& Poisson law / indicator; set difference / symmetric difference / size. \\
$\varnothing,\ \bot$
& Empty edge set; null (no-op) edit. \\

\addlinespace[8pt]
\multicolumn{2}{@{}l}{
\textcolor{softburgundy}{\textit{Paired graph editing}}} \\
$E_*,\ \X\subseteq 2^{E_*}$;\enspace $m,\ i,\ j$
& Edge universe / legal states; number of pairs / pair index / edit index. \\
$s_i,\ y_i;\ S,Y$;\enspace $\textcolor{softburgundy}{\pi},\ p_i$
& Pair-$i$ endpoints / random counterparts; prescribed coupling / pair weight. \\
$D_i,\ I_i$;\enspace $K_i=|D_i|+|I_i|$
& Deletions $s_i\setminus y_i$ / insertions $y_i\setminus s_i$; shortest edit count. \\
$\textcolor{softburgundy}{\gamma_i,\ \gamma}$;\enspace
$x_{i,j},\ a_{i,j}$
& Route $i$ / selected route family; state after $j$ edits / next action. \\
$\textcolor{softburgundy}{\mathcal V_i}$;\enspace
$T_a(x),\ T_\bot(x)$
& Visited states; successor after edit $a$ / identity null action. \\
$\mathcal I_i$;\enspace $H_i,\ F_i$
& Legal monotone interval (not $I_i$); shortest-edit DAG / its arcs. \\
$\mathcal Y,\ W$
& Distinct targets; union $\bigcup_i\mathcal I_i$ of legal intervals. \\

\addlinespace[8pt]
\multicolumn{2}{@{}l}{
\textcolor{softburgundy}{\textit{Poisson clock, projection, and endpoint risk}}} \\
$t,\ T,\ b(t),\ N_t$;\enspace $u(t),\ B=u(T)$
& Time / horizon / intensity / event count; cumulative intensity / budget. \\
$L,\ \ell$
& Future increment $N_T-N_t$ / realized count in the necessity proof. \\
$X_t^\gamma,\ A_t^\gamma$;\enspace $p_t(x),\ \dot p_t(x)$
& Route-teacher state / action; state occupancy / time derivative. \\
$\textcolor{softburgundy}{q_t(a\mid x)}$;\enspace $q_t^y(a\mid x)$
& Pooled source-blind action law; target-conditioned action law. \\
$q_{\theta,t},\ q_\theta(a\mid x,u;\mathcal P_g)$
& Learned action law; clock-time form with public cohort constraints. \\
$X_t^{\mathrm{proj}},\ X_t^\theta$;\enspace
$\textcolor{softburgundy}{\varepsilon_B}$
& Projected / student states; unavoidable Poisson clock-error floor. \\
$D_{\mathrm{pair}}$;\enspace
$\textcolor{softburgundy}{D_{\mathrm{pair}}^{\mathrm{proj}},\ D_{\mathrm{pair}}^\theta}$
& $\|\calL(S,X_T)-\pi\|_{\TV}$ (mispairing rate); projected / student errors. \\
$D_{\mathrm{marg}},\ D_{\mathrm{marg}}^\theta$
& Projected / student terminal marginal TV relative to $\calL(Y)$. \\
$\textcolor{softburgundy}{\mathfrak A_B^Y}$;\enspace
$\textcolor{softburgundy}{\mathfrak E_\theta}$
& Clock-integrated target-conditioned action TV; clock-averaged teacher--student KL. \\
$\theta$;\enspace
$\PP_{\mathrm{marked}}^{\mathrm{proj}},\ \PP_{\mathrm{marked}}^\theta$
& Student parameters; marked-path laws in the KL bound. \\
$w_{ij}(u)$;\enspace $q_u(a\mid x),\ r_i(u,x)$
& Poisson route-position weight; clock-time kernel / source-conditional projected state law. \\

\addlinespace[8pt]
\multicolumn{2}{@{}l}{
\textcolor{softburgundy}{\textit{Route selection, scheduling, and examples}}} \\
$f_{uv}^i,\ z_{y,x}$
& Binary route-arc flow / target-labeled state occupancy. \\
$n_y(x),\ \textcolor{softburgundy}{C(\gamma)}$
& Routes to $y$ visiting $x$ / cross-target state-conflict score. \\
$C\subseteq\X$;\enspace $C,\ \X_C$
& Shortest-route cut set; edge-capacity limit / legal states $\{x:|x|\le C\}$. \\
$R;\ a,b,c,d$;\enspace $a_i,b_i,\ \psi$
& Preserved backbone / example edges; capacity edges / $\PP\{\Pois(B)\ge2\}$. \\

\addlinespace[8pt]
\multicolumn{2}{@{}l}{
\textcolor{softburgundy}{\textit{Empirical graph-only training}}} \\
$\mathcal G_{\mathrm{tr}},\ g$;\enspace
$E_{*,g},\ \X_g,\ \mathcal I_g$
& Training cohorts / cohort; edge universe / legal states / pair indices. \\
$\hat p_{i\mid g},\ \rho_g$;\enspace
$\widehat\mu_{g,u}^{\mathrm{tr}}(x)$
& Within-cohort pair / cohort weights; teacher state occupancy. \\
$\widehat q_{g,u}^{\mathrm{tr}}(a\mid x)$
& Occupancy-weighted pooled teacher action law within cohort $g$. \\
$U,\ \mathcal P_g$;\enspace
$\textcolor{softburgundy}{\widehat{\mathcal L}_{\mathrm{CE}}^{\mathrm{tr}}(\theta)}$
& Uniform clock time / public legality; pooled-teacher cross-entropy. \\
$\widehat H^{\mathrm{tr}},\ \widehat{\mathfrak E}_{\theta}^{\mathrm{tr}}$
& Teacher conditional entropy / empirical action KL error. \\

\end{longtable}
\endgroup

\noindent\emph{Experimental regimes.}
R1 uses observed source--target graph pairs; R2 uses controlled crossed-edit tasks on real graph topologies. Reported paired errors are percentages, and ``pp'' denotes percentage points.


\clearpage
\onecolumn
\appendix

\section*{Appendix Contents}

\begingroup
\normalsize
\hypersetup{linkcolor=mydarkblue}
\setlength{\parindent}{0pt}
\setlength{\parskip}{0pt}
\setlength{\baselineskip}{15pt}

\newcommand{\appendixsectionentry}[2]{%
  \par\addvspace{9pt}%
  \noindent
  \hyperref[#1]{\textcolor{mydarkblue}{\bfseries\ref*{#1}}}%
  \hspace{0.85em}%
  \hyperref[#1]{\textcolor{black}{\bfseries #2}}%
  \nobreak\dotfill\pageref*{#1}\par%
}

\newcommand{\appendixsubsectionentry}[2]{%
  \par\addvspace{2pt}%
  \noindent\hspace*{2em}%
  \hyperref[#1]{\textcolor{mydarkblue}{\ref*{#1}}}%
  \hspace{0.75em}%
  \hyperref[#1]{\textcolor{black}{#2}}%
  \nobreak\dotfill\pageref*{#1}\par%
}
\appendixsectionentry{app:projection}{Markov Projection and Recoverability Proofs}
\appendixsubsectionentry{app:subsection-01}{Route Occupancies and the Clock-Floor Bound}
\appendixsubsectionentry{app:subsection-02}{Source- and Target-Conditioned Forward Equations}
\appendixsubsectionentry{app:endpoint}{Necessity and Sufficiency of Target Separation}
\appendixsubsectionentry{app:subsection-03}{Shared Targets and Pathwise Consistency}
\appendixsectionentry{app:certificates}{Routing and Scheduling Guarantees}
\appendixsubsectionentry{app:subsection-04}{Integral Route Selection and Zero-Length Cases}
\appendixsubsectionentry{app:subsection-05}{Adjacent-Swap Legality and Finite Termination}
\appendixsubsectionentry{app:subsection-06}{Marker-Free Route Enumeration}
\appendixsubsectionentry{app:subsection-07}{Capacity Obstructions and Recovery Certificates}
\appendixsectionentry{app:learning}{Learned-Policy Error Analysis}
\appendixsubsectionentry{app:subsection-08}{Target-Aware Coupling and Structural Ambiguity}
\appendixsubsectionentry{app:subsection-09}{Source-Action versus Target-Action Ambiguity}
\appendixsubsectionentry{app:subsection-10}{Marked-Path KL and Endpoint-Risk Bounds}
\appendixsubsectionentry{app:subsection-11}{Exact Conditional Forward Equations}
\appendixsectionentry{app:numerical-verification}{Numerical and Structural Verification}
\appendixsubsectionentry{app:subsection-12}{Poisson Clock and Markov-Projection Diagnostics}
\appendixsubsectionentry{app:subsection-13}{Synthetic Capacity and Marker-Free Checks}
\appendixsubsectionentry{app:subsection-14}{ODE--Monte Carlo Endpoint-Risk Verification}
\appendixsubsectionentry{app:subsection-15}{Boundary Cases and Target Aliasing}
\appendixsubsectionentry{app:subsection-16}{Exact Routing and Local-Scheduling Audits}
\appendixsubsectionentry{app:subsection-17}{Exhaustive Structural Phase Diagram}
\appendixsectionentry{app:benchmark-selection}{Datasets and Applicability Audits}
\appendixsubsectionentry{app:subsection-18}{Original Benchmark Screening}
\appendixsubsectionentry{app:subsection-19}{Observed-Pair Structural Diagnostics}
\appendixsubsectionentry{app:r2-controlled}{Multi-Regime Controlled Graph Edits}
\appendixsubsectionentry{app:r3-molecular}{Molecular-Pair Applicability}
\appendixsectionentry{app:order-baselines}{Ordering Baselines and Observed-Pair Evaluation}
\appendixsubsectionentry{app:subsection-20}{Matched Legal-Ordering Baselines}
\appendixsubsectionentry{app:subsection-21}{Structural and Computational Comparisons}
\appendixsubsectionentry{app:r1-generalization}{Observed-Pair Generalization Protocol}
\appendixsectionentry{app:controlled-five-student}{Unified Controlled Benchmark and Training Protocol}
\appendixsubsectionentry{app:subsection-22}{Four-Edit Construction and Benchmark Selection}
\appendixsubsectionentry{app:subsection-23}{Data Splits, Legality, and Clock Calibration}
\appendixsubsectionentry{app:subsection-24}{Teacher Schedules and Ideal Projection}
\appendixsubsectionentry{app:subsection-25}{GraphSAGE Student and Empirical Objectives}
\appendixsubsectionentry{app:subsection-26}{Model Selection and Endpoint Evaluation}
\appendixsectionentry{app:extended-experiments}{Extended Student Experiments}
\appendixsubsectionentry{app:subsection-27}{Full-Validation Diagnostics and Training Curves}
\appendixsubsectionentry{app:subsection-28}{Five-Axis Scaling Sensitivity}
\appendixsubsectionentry{app:subsection-29}{Edit-Depth Sensitivity}
\appendixsubsectionentry{app:subsection-31}{Checkpoint-Selection-Disjoint Validation}
\appendixsubsectionentry{app:subsection-32}{External-Backbone Ordering Comparison}
\endgroup

\clearpage


\section{Markov Projection and Recoverability Proofs}
\label{app:projection}

This appendix gives the detailed probability calculations underlying
Section~\ref{sec:setting} and the recoverability theorem in
Section~\ref{sec:theory}.  In particular, it expands the Poisson-clock
identities in \eqref{eq:route}--\eqref{eq:lb}, the pooled generator in
\eqref{eq:q}--\eqref{eq:forward}, and both directions of
Theorem~\ref{thm:endpoint}.  We work throughout with the \emph{fixed}
source--target coupling \eqref{eq:pair}, legal shortest single-edge
routes, and one clock independent of the source. These assumptions define the setting in which
endpoint recoverability is characterized.

\subsection{Route Occupancies and the Clock-Floor Bound}
\label{app:subsection-01}

\paragraph{Why the minimum edit count is $K_i$.}
Represent a graph by the set of its present edges \citep{diestel2017graph,gao2010survey}.  Fix pair $i$ and
write its required deletions and insertions as in
Section~\ref{sec:setting}:
\begin{align}
D_i &= s_i\setminus y_i,\notag\\
I_i &= y_i\setminus s_i,\notag\\
D_i\cap I_i &= \varnothing,\notag\\
s_i\triangle y_i &= D_i\cup I_i,\notag\\
|s_i\triangle y_i| &= |D_i|+|I_i|=K_i.
\label{eq:appA-minlength}
\end{align}
The third identity holds because an edge cannot simultaneously belong
to $s_i\setminus y_i$ and $y_i\setminus s_i$.  The fourth follows by
the definition of symmetric difference. The fifth follows because
the union is disjoint.  Each edge in $s_i\triangle y_i$ has a different
presence/absence value at the source and target, so it must be toggled
an odd number of times.  In particular, it must be toggled at least
once.  One legal non-null action toggles at most one edge; a null
action toggles none.  Thus no policy can reach $y_i$ from $s_i$ in
fewer than $K_i$ events \citep{diestel2017graph,korte2018combinatorial}.

For completeness, define the edge-Hamming distance to the target by
$d_i(x)=|x\triangle y_i|$.  A single toggle can correct at most one
incorrect edge, so for any legal action $a$,
\begin{align}
d_i(T_a(x)) &\ge d_i(x)-1,\notag\\
d_i(T_\bot(x)) &=d_i(x).
\label{eq:appA-dist-step}
\end{align}
Applying the first inequality successively after $n$ clock events
gives
\begin{align}
d_i(X_{\text{after }0\text{ events}})&=d_i(s_i)=K_i,\notag\\
d_i(X_{\text{after }1\text{ event}})&\ge K_i-1,\notag\\
d_i(X_{\text{after }2\text{ events}})&\ge K_i-2,\notag\\
d_i(X_{\text{after }n\text{ events}})&\ge K_i-n.
\label{eq:appA-distance-induction}
\end{align}
The last line follows by induction: if the inequality holds at step
$n$, then \eqref{eq:appA-dist-step} gives
$d_i(X_{n+1})\ge d_i(X_n)-1\ge K_i-(n+1)$.
When $n<K_i$, the right-hand side is positive; hence the current
graph cannot equal $y_i$.  This argument does not assume that the
policy follows the prescribed route, or that it never undoes an edit \citep{diestel2017graph,korte2018combinatorial}.

\paragraph{From the nonhomogeneous clock to position probabilities.}
Recall that $u(t)=\int_0^t b(r)\,\dd r$, $u(0)=0$, and $B=u(T)$.
For a nonhomogeneous Poisson clock, the number of events on an
interval has a Poisson distribution whose mean is the integrated
intensity, and disjoint intervals have independent counts
\citep{kingman1993poisson,bremaud1981point,daley2003point}.  
First put $g_j(u)=\PP(N_t=j)$ when $u=u(t)$.  The Poisson mass is
\begin{equation}
 g_j(u)=e^{-u}\frac{u^j}{j!},\qquad j=0,1,2,\ldots.
 \label{eq:appA-poisson-pmf}
\end{equation}
One can see the individual factors directly from the standard
Poisson counting equations \citep{kingman1993poisson,daley2003point}.  For zero events,
\begingroup\makeatletter\def\tagform@#1{}\def\theHequation{appAhidden1}\makeatother
\begin{align}
\frac{\dd g_0(u)}{\dd u}&=-g_0(u),\notag\\
g_0(0)&=1,\notag\\
\frac{\dd}{\dd u}\bigl(e^u g_0(u)\bigr)&=0,\notag\\
e^u g_0(u)&=1,\notag\\
g_0(u)&=e^{-u}.
\label{eq:appA-poisson-zero}
\end{align}
\endgroup\addtocounter{equation}{-1}
For $j\ge1$, a count of $j$ either stays at $j$ or arrives from
$j-1$ at the next event, yielding
\begingroup\makeatletter\def\tagform@#1{}\def\theHequation{appAhidden2}\makeatother
\begin{align}
\frac{\dd g_j(u)}{\dd u}&=g_{j-1}(u)-g_j(u),\notag\\
\frac{\dd}{\dd u}\bigl(e^u g_j(u)\bigr)&=e^u g_{j-1}(u).
\label{eq:appA-poisson-recursion}
\end{align}
\endgroup\addtocounter{equation}{-1}
Assuming $g_{j-1}(u)=e^{-u}u^{j-1}/(j-1)!$, as already established
for $j=1$ by the zero-count case, the second line becomes
\begin{align}
\frac{\dd}{\dd u}\bigl(e^u g_j(u)\bigr)
 &=e^u e^{-u}\frac{u^{j-1}}{(j-1)!},\notag\\
\notag  &=\frac{u^{j-1}}{(j-1)!}.
\end{align}
Integrating from $0$ to $u$ and using $g_j(0)=0$ for $j\ge1$ gives
\begin{align}
e^u g_j(u)-e^0g_j(0)
 &=\int_0^u\frac{v^{j-1}}{(j-1)!}\,\dd v,\notag\\
e^u g_j(u)
 &=\frac{u^j}{j!},\notag\\
\notag g_j(u)
 &=e^{-u}\frac{u^j}{j!},
\end{align}
which proves \eqref{eq:appA-poisson-pmf} for all $j$ by induction \citep{daley2003point}.
Equivalently, in physical time $t$, the clock distribution is
\begingroup\makeatletter\def\tagform@#1{}\def\theHequation{appAhidden3}\makeatother
\begin{equation}
\PP(N_t=j)=e^{-u(t)}\frac{u(t)^j}{j!}.
\label{eq:appA-clockphysical}
\end{equation}
\endgroup\addtocounter{equation}{-1}
The manipulations in physical time hold almost everywhere when
$b$ is integrable; working in cumulative time $u$ avoids any need
to assume that $b$ is differentiable \citep{kingman1993poisson}.

On the chosen teacher route, $j$ events have produced $x_{i,j}$
when $j<K_i$.  Once the $K_i$th edit occurs, further events produce
only $\bot$ and the graph remains $x_{i,K_i}=y_i$ \citep{norris1997markov}.  Therefore its
route \emph{position} is $\min\{N_t,K_i\}$, not always $N_t$.
Define
\begin{equation}
w_{ij}(t)=\PP\bigl(X_t^\gamma=x_{i,j}\mid S=s_i\bigr).
\label{eq:appA-wdefinition}
\end{equation}
For every $0\le j<K_i$ the position event is exactly
$\{N_t=j\}$, so
\begingroup\makeatletter\def\tagform@#1{}\def\theHequation{appAhidden4}\makeatother
\begin{equation}
w_{ij}(t)=e^{-u(t)}\frac{u(t)^j}{j!},\qquad j<K_i.
\label{eq:appA-wpre}
\end{equation}
\endgroup\addtocounter{equation}{-1}
At the terminal position, the event is $\{N_t\ge K_i\}$;
therefore
\begin{align}
w_{i,K_i}(t)
 &=\PP(N_t\ge K_i),\notag\\
 &=1-\PP(N_t<K_i),\notag\\
 &=1-\sum_{j=0}^{K_i-1}\PP(N_t=j),\notag\\
 &=1-e^{-u(t)}\sum_{j=0}^{K_i-1}\frac{u(t)^j}{j!}.
\label{eq:appA-wterminal}
\end{align}
An empty sum is zero.  In particular, if $K_i=0$ then
$w_{i,0}(t)=1$ for every $t$: a zero-length pair begins at its
correct target and every action is null.  Normalization is explicit:
\begingroup\makeatletter\def\tagform@#1{}\def\theHequation{appAhidden5}\makeatother
\begin{align}
\sum_{j=0}^{K_i}w_{ij}(t)
 &=\sum_{j=0}^{K_i-1}\PP(N_t=j)+\PP(N_t\ge K_i),\notag\\
 &=\PP(N_t<K_i)+\PP(N_t\ge K_i),\notag\\
 &=1.
\label{eq:appA-wnormalize}
\end{align}
\endgroup\addtocounter{equation}{-1}

\paragraph{No repeated state along a monotone route.}
Suppose two positions $j<k$ on route $i$ gave the same graph \citep{diestel2017graph}.
Between those positions the route toggles $k-j$ \emph{distinct}
mandatory edges, each exactly once.  In particular, at least one edge
changes its membership between positions $j$ and $k$ and is not
toggled back.  Thus $x_{i,j}\ne x_{i,k}$, a contradiction.  Hence
\begin{equation}
x_{i,j}=x_{i,k}\quad\Longrightarrow\quad j=k.
\label{eq:appA-norepeat}
\end{equation}
This is why \eqref{eq:appA-wdefinition} coincides with the route
\emph{position} probability rather than a sum over repeated visits.
Since $b(t)>0$ almost everywhere and $0<B<\infty$, for every
$0<t<T$ we have $0<u(t)<B$.  Consequently,
\begin{align}
w_{ij}(t)&>0&&\text{for }0\le j<K_i,\notag\\
w_{i,K_i}(t)&>0&&\text{for }K_i>0,\notag\\
w_{i,0}(t)&=1&&\text{for }K_i=0.
\label{eq:appA-positivepositions}
\end{align}
The terminal inequality follows because the Poisson event
$\{N_t\ge K_i\}$ has positive mass for every finite $K_i$ and
$u(t)>0$ \citep{kingman1993poisson,daley2003point}.  This positivity is used in the necessity proof of Appendix~\ref{app:endpoint}.

\paragraph{Detailed clock-floor calculation.}
Let $X_T$ be the endpoint of \emph{any} policy that starts from
$S$, shares the specified independent clock, and makes at most one
edge toggle per event.  Define
\begingroup\makeatletter\def\tagform@#1{}\def\theHequation{appAhidden6}\makeatother
\begin{equation}
 e_i=\PP(X_T\ne y_i\mid S=s_i).
 \label{eq:appA-ei}
\end{equation}
\endgroup\addtocounter{equation}{-1}
By \eqref{eq:appA-distance-induction}, the event $\{N_T<K_i\}$
is a subset of the failure event $\{X_T\ne y_i\}$ when $S=s_i$ \citep{kingman1993poisson}.
Split the failure event into the two disjoint possibilities
$N_T<K_i$ and $N_T\ge K_i$:
\begin{align}
e_i
 &=\PP(X_T\ne y_i,N_T<K_i\mid S=s_i)\notag\\
 &\quad+\PP(X_T\ne y_i,N_T\ge K_i\mid S=s_i),\notag\\
 &=\PP(N_T<K_i\mid S=s_i)\notag\\
 &\quad+\PP(X_T\ne y_i,N_T\ge K_i\mid S=s_i),\notag\\
 &=\PP(N_T<K_i)\notag\\
 &\quad+\PP(X_T\ne y_i,N_T\ge K_i\mid S=s_i),\notag\\
 &\ge\PP(N_T<K_i),\notag\\
 &=\PP\{\Pois(B)<K_i\}.
\label{eq:appA-floor-source}
\end{align}
The third equality uses independence of $N_T$ and $S$;
the inequality discards only a nonnegative probability.
For $K_i\ge1$, the right-hand side can be written without
probability notation as
\begingroup\makeatletter\def\tagform@#1{}\def\theHequation{appAhidden7}\makeatother
\begin{align}
\PP\{\Pois(B)<K_i\}
 &=\sum_{n=0}^{K_i-1}\PP(N_T=n),\notag\\
 &=\sum_{n=0}^{K_i-1}e^{-B}\frac{B^n}{n!},\notag\\
 &=e^{-B}\sum_{n=0}^{K_i-1}\frac{B^n}{n!}.
\label{eq:appA-floor-explicit}
\end{align}
\endgroup\addtocounter{equation}{-1}
It equals zero for $K_i=0$.

We prove why paired total variation equals the
weighted mispairing probability, a property specific to the
\emph{deterministic} prescribed pairing.  Let
$\mu_i(z)=\PP(X_T=z\mid S=s_i)$.  Since sources are distinct,
\begin{align}
\PP(S=s_i,X_T=z)&=p_i\mu_i(z),\notag\\
\notag \pi(S=s_i,Y=z)&=p_i\mathbf1\{z=y_i\}.
\end{align}
Using the finite-space definition of total variation
\citep{cover2006elements,csiszar2011information},
\begin{align}
D_{\mathrm{pair}}
 &=\frac12\sum_i\sum_z
 \left|p_i\mu_i(z)-p_i\mathbf1\{z=y_i\}\right|,\notag\\
 &=\frac12\sum_i p_i\sum_z
 \left|\mu_i(z)-\mathbf1\{z=y_i\}\right|,\notag\\
 &=\frac12\sum_i p_i\left(
 |\mu_i(y_i)-1|+\sum_{z\ne y_i}|\mu_i(z)|\right),\notag\\
 &=\frac12\sum_i p_i\left(
 1-\mu_i(y_i)+\sum_{z\ne y_i}\mu_i(z)\right),\notag\\
 &=\frac12\sum_i p_i\left(
 1-\mu_i(y_i)+1-\mu_i(y_i)\right),\notag\\
 &=\sum_i p_i\bigl(1-\mu_i(y_i)\bigr),\notag\\
 &=\sum_i p_i\PP(X_T\ne y_i\mid S=s_i),\notag\\
 &=\sum_i p_i e_i.
\label{eq:appA-tv-full}
\end{align}
The fourth line uses $0\le\mu_i(y_i)\le1$ and
$\mu_i(z)\ge0$; the fifth uses $\sum_z\mu_i(z)=1$ \citep{cover2006elements}.
Combining \eqref{eq:appA-floor-source} and
\eqref{eq:appA-tv-full} one term at a time yields
\begingroup\makeatletter\def\tagform@#1{}\def\theHequation{appAhidden8}\makeatother
\begin{align}
D_{\mathrm{pair}}
 &=\sum_i p_i e_i,\notag\\
 &\ge\sum_i p_i\PP\{\Pois(B)<K_i\},\notag\\
 &=\varepsilon_B,
\label{eq:appA-floor-final}
\end{align}
\endgroup\addtocounter{equation}{-1}
precisely the lower bound in \eqref{eq:lb}.  The same argument
also gives the following \emph{equality test}:
\begin{align}
D_{\mathrm{pair}}=\varepsilon_B
&\iff \sum_i p_i\Bigl(e_i-\PP(N_T<K_i)\Bigr)=0,\notag\\
&\iff e_i-\PP(N_T<K_i)=0\quad\text{for all }i,\notag\\
&\iff \PP(X_T\ne y_i,N_T\ge K_i\mid S=s_i)=0
\quad\text{for all }i.
\label{eq:appA-floor-equality}
\end{align}
The middle equivalence uses $p_i>0$ and the nonnegativity of
each summand.  This equality test is the critical bridge between
the numerical clock floor and the structural separation theorem.

\subsection{Source- and Target-Conditioned Forward Equations}
\label{app:subsection-02}

\paragraph{The teacher's source-conditioned equation.}
The basic forward-equation rule for a finite-state continuous-time
jump process counts probability \emph{entering} a state minus
probability \emph{leaving} it \citep{norris1997markov,anderson1991continuous,ethier1986markov}.
Here we derive that rule directly from the route-position weights
so that its application to the non-Markovian mixture is explicit.

Fix $S=s_i$ and set
\begingroup\makeatletter\def\tagform@#1{}\def\theHequation{appAhidden9}\makeatother
\begin{equation}
p_t^{s_i}(z)=\PP(X_t^\gamma=z\mid S=s_i).
\label{eq:appA-sourceoccupancy}
\end{equation}
\endgroup\addtocounter{equation}{-1}
By \eqref{eq:appA-norepeat} and the route-position law,
\begin{equation}
p_t^{s_i}(z)=\sum_{j=0}^{K_i}w_{ij}(t)
\mathbf1\{x_{i,j}=z\}.
\label{eq:appA-sourceoccupancy-sum}
\end{equation}
For $K_i\ge1$, write $w_j(t)$ for $w_{ij}(t)$ within this
paragraph.  Differentiating the explicit Poisson masses gives,
for almost every $t$,
\begin{align}
\dot w_0(t)&=-b(t)w_0(t),\notag\\
\dot w_j(t)&=b(t)\bigl(w_{j-1}(t)-w_j(t)\bigr),
 &&1\le j<K_i,\notag\\
\dot w_{K_i}(t)&=b(t)w_{K_i-1}(t).
\label{eq:appA-position-odes}
\end{align}
For the middle line, the chain and product rules give every
algebraic step:
\begingroup\makeatletter\def\tagform@#1{}\def\theHequation{appAhidden10}\makeatother
\begin{align}
\dot w_j(t)
 &=b(t)\frac{\dd}{\dd u}
 \left(e^{-u}\frac{u^j}{j!}\right)_{u=u(t)},\notag\\
 &=b(t)\left(-e^{-u}\frac{u^j}{j!}
 +e^{-u}\frac{ju^{j-1}}{j!}\right)_{u=u(t)},\notag\\
 &=b(t)\left(e^{-u}\frac{u^{j-1}}{(j-1)!}
 -e^{-u}\frac{u^j}{j!}\right)_{u=u(t)},\notag\\
 &=b(t)\bigl(w_{j-1}(t)-w_j(t)\bigr).
\label{eq:appA-position-middle-steps}
\end{align}
\endgroup\addtocounter{equation}{-1}
The terminal derivative follows by differentiating
\eqref{eq:appA-wterminal}, where consecutive terms cancel:
\begingroup\makeatletter\def\tagform@#1{}\def\theHequation{appAhidden11}\makeatother
\begin{align}
\frac{\dd}{\dd u}\left(1-e^{-u}
\sum_{j=0}^{K_i-1}\frac{u^j}{j!}\right)
 &=e^{-u}\sum_{j=0}^{K_i-1}\frac{u^j}{j!}
   -e^{-u}\sum_{j=1}^{K_i-1}\frac{u^{j-1}}{(j-1)!},\notag\\
 &=e^{-u}\frac{u^{K_i-1}}{(K_i-1)!},\notag\\
\dot w_{K_i}(t)&=b(t)w_{K_i-1}(t).
\label{eq:appA-position-terminal-steps}
\end{align}
\endgroup\addtocounter{equation}{-1}
For $K_i=0$, the single weight is identically one and its
derivative is zero; the resulting forward equation below is
still valid because the only action is $\bot$ \citep{anderson1991continuous,norris1997markov}.

Differentiating \eqref{eq:appA-sourceoccupancy-sum} and
substituting \eqref{eq:appA-position-odes},
\begin{align}
\dot p_t^{s_i}(z)
 &=\sum_{j=0}^{K_i}\dot w_{ij}(t)
    \mathbf1\{x_{i,j}=z\},\notag\\
 &=b(t)\left[
 -w_{i0}\mathbf1\{x_{i,0}=z\}
 +\sum_{j=1}^{K_i-1}(w_{i,j-1}-w_{ij})
    \mathbf1\{x_{i,j}=z\}
 +w_{i,K_i-1}\mathbf1\{x_{i,K_i}=z\}
 \right],\notag\\
 &=b(t)\left[
 \sum_{j=0}^{K_i-1}w_{ij}\mathbf1\{x_{i,j+1}=z\}
 -\sum_{j=0}^{K_i-1}w_{ij}\mathbf1\{x_{i,j}=z\}
 \right],\notag\\
 &=b(t)\sum_{j=0}^{K_i-1}w_{ij}(t)
 \Bigl(\mathbf1\{T_{a_{i,j}}(x_{i,j})=z\}
       -\mathbf1\{x_{i,j}=z\}\Bigr).
\label{eq:appA-sourceflow-direct}
\end{align}
The third line collects positive $w_{ij}$ terms
(the incoming flux) and negative $w_{ij}$ terms
(the outgoing flux); no stochastic approximation is involved.

For teacher-occupied states define
\begingroup\makeatletter\def\tagform@#1{}\def\theHequation{appAhidden12}\makeatother
\begin{equation}
q_t^{s_i}(a\mid x)
 =\PP(A_t^\gamma=a\mid X_t^\gamma=x,S=s_i).
\label{eq:appA-sourcekernel}
\end{equation}
\endgroup\addtocounter{equation}{-1}
Because the route does not revisit a graph,
\begingroup\makeatletter\def\tagform@#1{}\def\theHequation{appAhidden13}\makeatother
\begin{align}
q_t^{s_i}(a_{i,j}\mid x_{i,j})&=1,&&j<K_i,\notag\\
q_t^{s_i}(\bot\mid y_i)&=1.
\label{eq:appA-source-delta}
\end{align}
\endgroup\addtocounter{equation}{-1}
The terminal null edit contributes the same $p_t^{s_i}(z)$
to the incoming and outgoing terms \citep{anderson1991continuous,ethier1986markov}.  Consequently
\eqref{eq:appA-sourceflow-direct} can be written as
\begin{equation}
\dot p_t^{s_i}(z)
=b(t)\left[
\sum_{x,a:T_a(x)=z}p_t^{s_i}(x)q_t^{s_i}(a\mid x)
-p_t^{s_i}(z)\right].
\label{eq:appA-source-forward-full}
\end{equation}
The initial value is $p_0^{s_i}(z)=\mathbf1\{z=s_i\}$ \citep{norris1997markov}.

\paragraph{How conditioning and pooling fit together.}
The source law is $\mu(s_i)=p_i$ because the sources in
\eqref{eq:pair} are distinct.  The law of total probability \citep{murphy2022probabilistic}
immediately gives
\begingroup\makeatletter\def\tagform@#1{}\def\theHequation{appAhidden14}\makeatother
\begin{align}
p_t(x)
 &=\PP(X_t^\gamma=x),\notag\\
 &=\sum_{i=1}^m\PP(S=s_i)
             \PP(X_t^\gamma=x\mid S=s_i),\notag\\
 &=\sum_{i=1}^m p_i p_t^{s_i}(x),\notag\\
 &=\sum_{i=1}^m p_i\sum_{j=0}^{K_i}
        w_{ij}(t)\mathbf1\{x_{i,j}=x\}.
\label{eq:appA-pooled-occupancy}
\end{align}
\endgroup\addtocounter{equation}{-1}
Similarly, the \emph{joint} state--action probability is
\begingroup\makeatletter\def\tagform@#1{}\def\theHequation{appAhidden15}\makeatother
\begin{align}
\PP(X_t^\gamma=x,A_t^\gamma=a)
 &=\sum_i p_i\PP(X_t^\gamma=x,A_t^\gamma=a\mid S=s_i),\notag\\
 &=\sum_i p_i\sum_{j=0}^{K_i}w_{ij}(t)
   \mathbf1\{x_{i,j}=x\}\mathbf1\{a_{i,j}=a\},
\label{eq:appA-pooled-joint}
\end{align}
\endgroup\addtocounter{equation}{-1}
where $a_{i,K_i}=\bot$.  Whenever $p_t(x)>0$, the definition of
a conditional probability gives the kernel in \eqref{eq:q}:
\begin{align}
q_t(a\mid x)
 &=\frac{\PP(X_t^\gamma=x,A_t^\gamma=a)}
        {\PP(X_t^\gamma=x)},\notag\\
 &=\frac{\displaystyle\sum_i p_i\sum_{j=0}^{K_i}
        w_{ij}(t)\mathbf1\{x_{i,j}=x,a_{i,j}=a\}}
        {\displaystyle\sum_i p_i\sum_{j=0}^{K_i}
        w_{ij}(t)\mathbf1\{x_{i,j}=x\}}.
\label{eq:appA-pooled-kernel-full}
\end{align}
For an occupied $x$, summing the numerator over actions gives
the denominator, so $\sum_a q_t(a\mid x)=1$ \citep{murphy2022probabilistic}.
At an unoccupied state the conditional law is immaterial to the
teacher and may be extended arbitrarily as a legal probability
kernel \citep{ethier1986markov}.  The identity that performs the entire pooling step is
\begin{align}
p_t(x)q_t(a\mid x)
 &=\PP(X_t^\gamma=x,A_t^\gamma=a),\notag\\
 &=\sum_i p_i p_t^{s_i}(x)q_t^{s_i}(a\mid x).
\label{eq:appA-pooled-flux}
\end{align}
Multiply \eqref{eq:appA-source-forward-full} by $p_i$, sum over
$i$, and apply \eqref{eq:appA-pooled-flux} separately to its
incoming and outgoing terms:
\begin{align}
\dot p_t(z)
 &=\sum_i p_i\dot p_t^{s_i}(z),\notag\\
 &=b(t)\sum_i p_i\left[
   \sum_{x,a:T_a(x)=z}p_t^{s_i}(x)q_t^{s_i}(a\mid x)
   -p_t^{s_i}(z)\right],\notag\\
 &=b(t)\left[
   \sum_{x,a:T_a(x)=z}\sum_i p_i p_t^{s_i}(x)
       q_t^{s_i}(a\mid x)
   -\sum_i p_i p_t^{s_i}(z)\right],\notag\\
 &=b(t)\left[
   \sum_{x,a:T_a(x)=z}p_t(x)q_t(a\mid x)-p_t(z)
   \right].
\label{eq:appA-pooled-forward-steps}
\end{align}
The last line is the forward equation
\eqref{eq:forward} in the main paper \citep{norris1997markov,ethier1986markov}.

\paragraph{Why this proves marginal preservation.}
Define the blind projection to start with the same initial
source \emph{marginal} $p_0$ and, at time $t$ and state $x$,
to draw action $a$ with probability $q_t(a\mid x)$ at an event
of the same clock.  If $\bar p_t(z)=\PP(X_t^{\mathrm{proj}}=z)$,
the finite-state forward rule gives
\begin{equation}
\dot{\bar p}_t(z)=b(t)\left[
\sum_{x,a:T_a(x)=z}\bar p_t(x)q_t(a\mid x)
-\bar p_t(z)\right].
\label{eq:appA-proj-forward}
\end{equation}
Their initial laws agree:
\begingroup\makeatletter\def\tagform@#1{}\def\theHequation{appAhidden16}\makeatother
\begin{align}
\bar p_0(z)
 &=\sum_i p_i\mathbf1\{s_i=z\},\notag\\
p_0(z)
 &=\sum_i p_i\mathbf1\{x_{i,0}=z\},\notag\\
\bar p_0(z)&=p_0(z),
\label{eq:appA-initial-same}
\end{align}
\endgroup\addtocounter{equation}{-1}
because $x_{i,0}=s_i$ \citep{norris1997markov,anderson1991continuous}.

To state uniqueness explicitly, let $d_t=\bar p_t-p_t$
be the finite vector of coordinate differences and set
$\|d_t\|_1=\sum_z|d_t(z)|$.  Subtract
\eqref{eq:appA-pooled-forward-steps} from
\eqref{eq:appA-proj-forward} and integrate:
\begingroup\makeatletter\def\tagform@#1{}\def\theHequation{appAhidden17}\makeatother
\begin{align}
d_t(z)
 &=\int_0^t b(r)\left[
 \sum_{x,a:T_a(x)=z}d_r(x)q_r(a\mid x)-d_r(z)
 \right]\dd r.
\label{eq:appA-difference-integral}
\end{align}
\endgroup\addtocounter{equation}{-1}
Summing absolute values and using the triangle inequality,
\begin{align}
\|d_t\|_1
 &\le\int_0^t b(r)\left[
 \sum_z\sum_{x,a:T_a(x)=z}|d_r(x)|q_r(a\mid x)
 +\sum_z|d_r(z)|\right]\dd r,\notag\\
 &=\int_0^t b(r)\left[
 \sum_x|d_r(x)|\sum_a q_r(a\mid x)
 +\|d_r\|_1\right]\dd r,\notag\\
 &=2\int_0^t b(r)\|d_r\|_1\,\dd r.
\label{eq:appA-gronwall-bound}
\end{align}
Here each legal action has exactly one successor, which justifies
the second line, and $\sum_aq_r(a\mid x)=1$ justifies the third.
The nonnegative function $\|d_t\|_1$ has zero initial value and
satisfies the integral inequality
\eqref{eq:appA-gronwall-bound}.  The standard integral Gronwall
argument implies $\|d_t\|_1=0$ for every $t$ because
$\int_0^T b(r)\dd r=B<\infty$
\citep{norris1997markov,anderson1991continuous,coddington1955ode}.
Hence, coordinate by coordinate,
\begin{equation}
\boxed{\PP(X_t^{\mathrm{proj}}=z)
=\PP(X_t^\gamma=z)\quad\text{for all }z,t.}
\label{eq:appA-marginalidentity}
\end{equation}
This is a statement about the \emph{unconditional}
distribution of the present graph at each time.  It does not
identify the two processes' source-conditioned laws, their
complete path laws, or their terminal source--target coupling.
That distinction is central to the Markovian-projection issue
\citep{debortoli2023augmented,zhang2024trajectory,
gat2024discrete,holderrieth2025generator}.

\paragraph{Conditioning on the target instead of the source.}
Repeated targets are permitted by \eqref{eq:pair}, so write
\begingroup\makeatletter\def\tagform@#1{}\def\theHequation{appAhidden18}\makeatother
\begin{equation}
\nu(y)=\PP(Y=y)=\sum_{i:y_i=y}p_i,\qquad y\in\mathcal Y.
\label{eq:appA-targetmass}
\end{equation}
\endgroup\addtocounter{equation}{-1}
For $\nu(y)>0$, set
\begingroup\makeatletter\def\tagform@#1{}\def\theHequation{appAhidden19}\makeatother
\begin{align}
p_t^y(x)&=\PP(X_t^\gamma=x\mid Y=y),\notag\\
q_t^y(a\mid x)
 &=\PP(A_t^\gamma=a\mid X_t^\gamma=x,Y=y).
\label{eq:appA-target-conditional}
\end{align}
\endgroup\addtocounter{equation}{-1}
By dividing the joint law by $\nu(y)$,
\begingroup\makeatletter\def\tagform@#1{}\def\theHequation{appAhidden20}\makeatother
\begin{align}
p_t^y(x)
 &=\frac{\PP(Y=y,X_t^\gamma=x)}{\nu(y)},\notag\\
 &=\frac1{\nu(y)}\sum_{i:y_i=y}p_i
     \sum_{j=0}^{K_i}w_{ij}(t)\mathbf1\{x_{i,j}=x\}.
\label{eq:appA-targetoccupancy}
\end{align}
\endgroup\addtocounter{equation}{-1}
The target-conditioned state--action flux is similarly
\begingroup\makeatletter\def\tagform@#1{}\def\theHequation{appAhidden21}\makeatother
\begin{align}
\nu(y)p_t^y(x)q_t^y(a\mid x)
 &=\PP(Y=y,X_t^\gamma=x,A_t^\gamma=a),\notag\\
 &=\sum_{i:y_i=y}p_i\sum_{j=0}^{K_i}w_{ij}(t)
 \mathbf1\{x_{i,j}=x,a_{i,j}=a\}.
\label{eq:appA-targetflux}
\end{align}
\endgroup\addtocounter{equation}{-1}
The same incoming-minus-outgoing calculation as
\eqref{eq:appA-pooled-forward-steps}, now summing only over
$i$ with $y_i=y$, gives
\begingroup\makeatletter\def\tagform@#1{}\def\theHequation{appAhidden22}\makeatother
\begin{align}
\nu(y)\dot p_t^y(z)
 &=b(t)\left[
 \sum_{x,a:T_a(x)=z}\nu(y)p_t^y(x)q_t^y(a\mid x)
 -\nu(y)p_t^y(z)\right],\notag\\
\dot p_t^y(z)
 &=b(t)\left[
 \sum_{x,a:T_a(x)=z}p_t^y(x)q_t^y(a\mid x)
 -p_t^y(z)\right].
\label{eq:appA-target-forward}
\end{align}
\endgroup\addtocounter{equation}{-1}
The second line divides by the fixed positive number $\nu(y)$ \citep{anderson1991continuous}.
The process whose policy receives $Y=y$ and uses $q_t^y$
starts with the same conditional initial law
\begingroup\makeatletter\def\tagform@#1{}\def\theHequation{appAhidden23}\makeatother
\begin{equation}
p_0^y(x)=\frac1{\nu(y)}\sum_{i:y_i=y}
       p_i\mathbf1\{s_i=x\}.
\label{eq:appA-targetinitial}
\end{equation}
\endgroup\addtocounter{equation}{-1}
By the same finite-state uniqueness argument, this target-aware
process preserves $p_t^y$ for each target, and thus preserves the
joint \emph{target--graph} marginals
\begin{equation}
\calL(Y,X_t^{\mathrm{target}})
=\calL(Y,X_t^\gamma).
\label{eq:appA-targetjoint}
\end{equation}
It need not preserve $\calL(S,X_t^\gamma)$ because knowing $Y$
is generally weaker than knowing the original source $S$.
We use the exact target-conditioned law again in
Appendix~\ref{app:subsection-03}; this separation of conditioned
and unconditioned dynamics is standard in coupling and bridge
constructions \citep{lindvall1992lectures,thorisson2000coupling,
debortoli2023augmented}.

\subsection{Necessity and Sufficiency of Target Separation}
\label{app:endpoint}

We prove the two directions of Theorem~\ref{thm:endpoint} separately.
For clarity, ``attaining the floor'' means
$D_{\mathrm{pair}}^{\mathrm{proj}}=\varepsilon_B$ for the
\emph{fixed-route} Markov projection generated by
\eqref{eq:q}.  It does not concern an independently retrained
student that is free to depart from that kernel \citep{debortoli2023augmented,zhang2024trajectory}.

\paragraph{Sufficiency: target separation forces shortest progress.}
Assume the structural condition \eqref{eq:separator}:
\begin{equation}
\boxed{\mathcal V_i\cap\mathcal V_h=\varnothing
\quad\text{whenever }y_i\ne y_h.}
\label{eq:appA-suff-assumption}
\end{equation}
For every teacher-visited graph $x$, define its associated target
$y(x)$ to be the target of \emph{any} route that visits it.  This is
well defined: if routes $i$ and $h$ both visit $x$, then
$x\in\mathcal V_i\cap\mathcal V_h$, so the assumed separation
forces $y_i=y_h$.  There may be several possible sources for $x$,
but there is only one associated target.

Take a nonterminal teacher route position
$x=x_{h,k}\ne y_h=y(x)$, $k<K_h$.  The route follows a
mandatory edit towards its own target.  We check the precise
change in edit distance:
\begin{align}
|x_{h,k}\triangle y_h|&=K_h-k,\notag\\
|x_{h,k+1}\triangle y_h|&=K_h-(k+1),\notag\\
|x_{h,k+1}\triangle y_h|
 &=|x_{h,k}\triangle y_h|-1.
\label{eq:appA-shortest-descent}
\end{align}
The first two equalities hold because after $k$ and $k+1$
distinct mandatory toggles, respectively, that many required
mismatches have been removed and none reintroduced \citep{diestel2017graph,korte2018combinatorial}.
If $q_t(a\mid x)>0$, the numerator in
\eqref{eq:appA-pooled-kernel-full} has at least one positive
term.  Thus there exist some $h,k$ with $x_{h,k}=x$ and
$a=a_{h,k}$.  By the definition of $y(x)$, $y_h=y(x)$.
By \eqref{eq:appA-shortest-descent},
\begingroup\makeatletter\def\tagform@#1{}\def\theHequation{appAhidden24}\makeatother
\begin{equation}
|T_a(x)\triangle y(x)|=|x\triangle y(x)|-1
\quad\text{if }x\ne y(x).
\label{eq:appA-pooled-descent}
\end{equation}
\endgroup\addtocounter{equation}{-1}
Also, $T_a(x)=x_{h,k+1}\in\mathcal V_h$ and so the next
state remains a teacher-visited state associated with the
\emph{same} target.  At $x=y(x)$ every route associated with
that target has finished, because a shortest monotone route
can visit its target only at its terminal position.  Thus
\begingroup\makeatletter\def\tagform@#1{}\def\theHequation{appAhidden25}\makeatother
\begin{equation}
q_t(\bot\mid x)=1\qquad\text{if }x=y(x).
\label{eq:appA-absorbing-target}
\end{equation}
\endgroup\addtocounter{equation}{-1}

Now start the projected process from $S=s_i$.  Initially
$s_i\in\mathcal V_i$ and $y(s_i)=y_i$.  The two facts just
proved show inductively that all reachable projected states
retain associated target $y_i$; at every event before reaching
$y_i$, the distance to $y_i$ drops by exactly one, while after
reaching it the state does not move.  The induction may be
written event by event, with $Z_n$ denoting the graph after
$n$ events:
\begingroup\makeatletter\def\tagform@#1{}\def\theHequation{appAhidden26}\makeatother
\begin{align}
|Z_0\triangle y_i|&=|s_i\triangle y_i|=K_i,\notag\\
|Z_1\triangle y_i|&=K_i-1&&\text{if }K_i\ge1,\notag\\
|Z_2\triangle y_i|&=K_i-2&&\text{if }K_i\ge2,\notag\\
|Z_n\triangle y_i|&=\max\{K_i-n,0\}
&&\text{for every }n\ge0.
\label{eq:appA-distance-exact}
\end{align}
\endgroup\addtocounter{equation}{-1}
The last line is proved by induction on $n$: the case $n=0$
was calculated; if the distance is positive, the next action
reduces it by one; if it is zero, only $\bot$ is permitted and
it remains zero.  A graph has distance zero from $y_i$ exactly
when it \emph{equals} $y_i$.  Therefore
\begingroup\makeatletter\def\tagform@#1{}\def\theHequation{appAhidden27}\makeatother
\begin{align}
\{X_T^{\mathrm{proj}}=y_i\}\quad
 &\text{equals the event}\quad\{N_T\ge K_i\}
 \qquad\text{under }S=s_i,\notag\\
\PP(X_T^{\mathrm{proj}}\ne y_i\mid S=s_i)
 &=\PP(N_T<K_i\mid S=s_i),\notag\\
 &=\PP\{\Pois(B)<K_i\}.
\label{eq:appA-suff-conditional}
\end{align}
\endgroup\addtocounter{equation}{-1}
The statement is also true when $K_i=0$: the initial graph is
already the target and every action is null.  Averaging over the
original weights gives, one equality at a time,
\begin{align}
D_{\mathrm{pair}}^{\mathrm{proj}}
 &=\sum_i p_i\PP(X_T^{\mathrm{proj}}\ne y_i\mid S=s_i),\notag\\
 &=\sum_i p_i\PP\{\Pois(B)<K_i\},\notag\\
 &=\varepsilon_B.
\label{eq:appA-suff-result}
\end{align}
This proves the \emph{if} direction of
Theorem~\ref{thm:endpoint}.  The proof requires agreement on
\emph{target} at a shared state, not agreement on the next edit.

\paragraph{Necessity, Step 1: prescribed teacher actions have positive probability.}
For the reverse direction, suppose two selected routes towards
different targets share a graph:
\begingroup\makeatletter\def\tagform@#1{}\def\theHequation{appAhidden28}\makeatother
\begin{equation}
x=x_{i,j}=x_{h,k},\qquad y_i\ne y_h,
\qquad 0\le j\le K_i,\quad 0\le k\le K_h.
\label{eq:appA-shared-state}
\end{equation}
\endgroup\addtocounter{equation}{-1}
The integers $j$ and $k$ may differ: we do \emph{not}
assume the two routes reach $x$ after the same number of edits \citep{norris1997markov,anderson1991continuous}.
By \eqref{eq:appA-positivepositions}, each of the teacher
route-position weights is strictly positive at every
$0<t<T$.  Fix any prescribed action $a_{i,r}$ from a
nonterminal position $x_{i,r}$, where $0\le r<K_i$.  The
numerator of the pooled kernel contains the positive term
$p_iw_{ir}(t)$, whereas the denominator is the total
nonnegative teacher occupancy at that graph.  Thus
\begin{align}
q_t(a_{i,r}\mid x_{i,r})
 &=\frac{\sum_{v,\ell}p_vw_{v\ell}(t)
    \mathbf1\{x_{v,\ell}=x_{i,r},a_{v,\ell}=a_{i,r}\}}
    {p_t(x_{i,r})},\notag\\
 &\ge\frac{p_iw_{ir}(t)}{p_t(x_{i,r})},\notag\\
 &>0.
\label{eq:appA-prescribed-positive}
\end{align}
The last inequality uses $p_i>0$, $w_{ir}(t)>0$,
$p_t(x_{i,r})\ge p_iw_{ir}(t)>0$.
At a terminal route position, the analogous calculation gives
\begingroup\makeatletter\def\tagform@#1{}\def\theHequation{appAhidden29}\makeatother
\begin{align}
q_t(\bot\mid y_i)
 &\ge\frac{p_iw_{i,K_i}(t)}{p_t(y_i)},\notag\\
 &>0.
\label{eq:appA-null-positive}
\end{align}
\endgroup\addtocounter{equation}{-1}
This matters when one shared position is terminal for its route
but not for the other \citep{kingman1993poisson}.

\paragraph{Necessity, Step 2: a prescribed prefix can actually occur.}
Choose a deterministic interior time $t\in(0,T)$.
For $j\ge1$, write $E_{i,j}(t)$ for the event that the blind
projection, started at $s_i$, experiences exactly $j$ events by
$t$ and executes, in order, the route's first $j$ actions:
\begingroup\makeatletter\def\tagform@#1{}\def\theHequation{appAhidden30}\makeatother
\begin{equation}
(a_{i,0},a_{i,1},\ldots,a_{i,j-1}).
\label{eq:appA-prefix-list}
\end{equation}
\endgroup\addtocounter{equation}{-1}
This event implies $X_t^{\mathrm{proj}}=x_{i,j}$ and $N_t=j$ \citep{kingman1993poisson}.
For $j=0$ define $E_{i,0}(t)=\{N_t=0\}$.
We now verify that these prefixes have \emph{positive} probability,
not just that their actions are in the support of a one-time
teacher distribution \citep{norris1997markov}.

Conditional on exactly $j\ge1$ events in $(0,t)$, their ordered
times $0<t_1<\cdots<t_j<t$ have an everywhere-positive density
where $b(t_r)>0$.  In an unconditioned integral, the joint density
of the event times and no other events in $(0,t)$ is
$e^{-u(t)}\prod_{r=1}^j b(t_r)$ on that ordered simplex
\citep{kingman1993poisson,bremaud1981point}.  At event time
$t_r$, conditional on following the first $r-1$ prescribed
edits, the present graph is $x_{i,r-1}$; the probability of
executing the next required edit is
$q_{t_r}(a_{i,r-1}\mid x_{i,r-1})$.
The product rule for successive conditional probabilities \citep{daley2003point,bremaud1981point}
therefore yields the explicit integral
\begingroup\makeatletter\def\tagform@#1{}\def\theHequation{appAhidden31}\makeatother
\begin{align}
&\PP(E_{i,j}(t)\mid S=s_i)\notag\\
&=e^{-u(t)}
\int_{0<t_1<\cdots<t_j<t}
 \prod_{r=1}^j
 \left[b(t_r)q_{t_r}(a_{i,r-1}\mid x_{i,r-1})\right]
 \,\dd t_1\cdots\dd t_j,\notag\\
&>0.
\label{eq:appA-prefix-positive}
\end{align}
\endgroup\addtocounter{equation}{-1}
The integral is strictly positive because the ordered simplex
has positive measure; $b>0$ almost everywhere; and every
required kernel probability is positive at interior times by
\eqref{eq:appA-prescribed-positive}.  For $j=0$,
\begingroup\makeatletter\def\tagform@#1{}\def\theHequation{appAhidden32}\makeatother
\begin{align}
\PP(E_{i,0}(t)\mid S=s_i)
 &=\PP(N_t=0\mid S=s_i),\notag\\
 &=e^{-u(t)},\notag\\
 &>0.
\label{eq:appA-prefixzero-positive}
\end{align}
\endgroup\addtocounter{equation}{-1}
The argument permits $j=K_i$ and therefore also permits
the prefix to end exactly at $y_i$ \citep{bremaud1981point,kingman1993poisson}.
Applying the same construction to route $h$ shows
\begingroup\makeatletter\def\tagform@#1{}\def\theHequation{appAhidden33}\makeatother
\begin{equation}
\PP(E_{i,j}(t)\mid S=s_i)>0,
\qquad
\PP(E_{h,k}(t)\mid S=s_h)>0.
\label{eq:appA-both-prefixes}
\end{equation}
\endgroup\addtocounter{equation}{-1}
On the two events the graph at time $t$ is the \emph{same}
$x$, even though the sources and the past event counts can differ.

\paragraph{Necessity, Step 3: floor equality leaves no failures with enough events.}
Assume for contradiction that the projected process attains the
clock floor,
\begingroup\makeatletter\def\tagform@#1{}\def\theHequation{appAhidden34}\makeatother
\begin{equation}
D_{\mathrm{pair}}^{\mathrm{proj}}=\varepsilon_B.
\label{eq:appA-necess-contradict-assumption}
\end{equation}
\endgroup\addtocounter{equation}{-1}
Apply the equality characterization
\eqref{eq:appA-floor-equality} specifically to
$X_T=X_T^{\mathrm{proj}}$.  For each source $s_v$ we obtain
\begin{equation}
\PP(X_T^{\mathrm{proj}}\ne y_v,\ N_T\ge K_v\mid S=s_v)=0.
\label{eq:appA-no-extra-error}
\end{equation}
There is no cancellation between different tasks here: every
conditional excess failure probability is nonnegative and every
pair weight is strictly positive, so an average excess of zero
forces \emph{each} conditional excess to vanish \citep{cover2006elements}.

\paragraph{Necessity, Step 4: choose the same future event count.}
Let $L=N_T-N_t$ be the number of events after the chosen
time $t$.  Independent increments of the Poisson clock \citep{kingman1993poisson,daley2003point} give
\begin{align}
\calL(L\mid N_t, S,\text{past marks})
 &=\calL(L),\notag\\
L&\sim\Pois(B-u(t)),\notag\\
\PP(L=\ell)
 &=e^{-(B-u(t))}\frac{(B-u(t))^\ell}{\ell!}
 >0\quad(\ell=0,1,2,\ldots).
\label{eq:appA-future-poisson}
\end{align}
Choose one integer
\begingroup\makeatletter\def\tagform@#1{}\def\theHequation{appAhidden35}\makeatother
\begin{equation}
\ell\ge\max\{K_i-j,\ K_h-k\}.
\label{eq:appA-futurecount-choice}
\end{equation}
\endgroup\addtocounter{equation}{-1}
Such an integer always exists because both lengths are finite.
By the independence in \eqref{eq:appA-future-poisson},
\begingroup\makeatletter\def\tagform@#1{}\def\theHequation{appAhidden36}\makeatother
\begin{align}
&\PP(E_{i,j}(t),L=\ell\mid S=s_i)\notag\\
&\quad=\PP(E_{i,j}(t)\mid S=s_i)
        \PP(L=\ell),\notag\\
&\quad>0,
\label{eq:appA-positive-both-future}
\end{align}
\endgroup\addtocounter{equation}{-1}
and identically with $(h,k)$ in place of $(i,j)$ \citep{kingman1993poisson}.
On the first of these positive-probability events,
\begingroup\makeatletter\def\tagform@#1{}\def\theHequation{appAhidden37}\makeatother
\begin{align}
N_T&=N_t+L,\notag\\
   &=j+\ell,\notag\\
   &\ge K_i.
\label{eq:appA-suff-events-i}
\end{align}
\endgroup\addtocounter{equation}{-1}
On the second,
\begingroup\makeatletter\def\tagform@#1{}\def\theHequation{appAhidden38}\makeatother
\begin{align}
N_T&=k+\ell,\notag\\
   &\ge K_h.
\label{eq:appA-suff-events-h}
\end{align}
\endgroup\addtocounter{equation}{-1}
By \eqref{eq:appA-no-extra-error}, the first event therefore
requires $X_T^{\mathrm{proj}}=y_i$ almost surely, whereas the
second requires $X_T^{\mathrm{proj}}=y_h$ almost surely.

\paragraph{Necessity, Step 5: the same state and time have the same future law.}
The projected policy chooses each future action using only the
\emph{current graph} and \emph{current time}; it is not given
$S$, the target, or the previous event count.  Moreover, future
Poisson event times are independent of the previous events \citep{norris1997markov,kingman1993poisson}.
Consequently, if its present state at the deterministic time $t$
is $x$ and we additionally condition on $L=\ell$, its terminal
law has a common kernel, say
\begingroup\makeatletter\def\tagform@#1{}\def\theHequation{appAhidden39}\makeatother
\begin{equation}
H_{t,\ell}(z\mid x)
=\PP(X_T^{\mathrm{proj}}=z\mid X_t^{\mathrm{proj}}=x,L=\ell).
\label{eq:appA-common-future-kernel}
\end{equation}
\endgroup\addtocounter{equation}{-1}
This is a probability distribution in $z$.  It does not depend
on which source led to $x$ or on how many events were used
before $t$: the future event times, and at each event the
source-blind action probabilities, depend only on $t$, $x$,
and the evolving future graph.  This is the Markov
property used here \citep{norris1997markov,anderson1991continuous,ethier1986markov}.
In particular,
\begin{align}
&\PP(X_T^{\mathrm{proj}}=z\mid
  E_{i,j}(t),S=s_i,L=\ell)
  =H_{t,\ell}(z\mid x),\notag\\
&\PP(X_T^{\mathrm{proj}}=z\mid
  E_{h,k}(t),S=s_h,L=\ell)
  =H_{t,\ell}(z\mid x).
\label{eq:appA-common-future-two-prefixes}
\end{align}
The previous step and \eqref{eq:appA-common-future-two-prefixes}
force
\begin{align}
H_{t,\ell}(y_i\mid x)&=1,\notag\\
H_{t,\ell}(y_h\mid x)&=1.
\label{eq:appA-two-pointmasses}
\end{align}
But $y_i\ne y_h$, so the events $\{X_T=y_i\}$ and
$\{X_T=y_h\}$ are disjoint.  A probability distribution
cannot give mass one to both:
\begin{align}
H_{t,\ell}(y_i\mid x)+H_{t,\ell}(y_h\mid x)
 &\le\sum_z H_{t,\ell}(z\mid x),\notag\\
\notag  & =1,
\end{align}
whereas \eqref{eq:appA-two-pointmasses} makes the left-hand
side equal to $2$.  We conclude
\begin{equation}
\mathcal V_i\cap\mathcal V_h\ne\varnothing,
\ y_i\ne y_h
\quad\Longrightarrow\quad
D_{\mathrm{pair}}^{\mathrm{proj}}>\varepsilon_B.
\label{eq:appA-necessity-result}
\end{equation}
The strict inequality follows because the clock lower bound is
always valid and equality has just been ruled out.
Combining \eqref{eq:appA-necessity-result} with
\eqref{eq:appA-suff-result} proves exactly
Theorem~\ref{thm:endpoint}, including cross-depth intersections,
terminal/nonterminal intersections, repeated targets, and
zero-length routes.  The proof uses the common positive Poisson
clock twice: to make all teacher prefixes reachable with positive
probability, and to ensure a positive probability of any prescribed
finite number of future events \citep{kingman1993poisson,bremaud1981point}.

\subsection{Shared Targets and Pathwise Consistency}
\label{app:subsection-03}

\paragraph{Why paths towards a common target can merge.}
Theorem~\ref{thm:endpoint} only forbids intersections between
routes whose \emph{targets differ}.  To understand why, fix a target
$y$ and consider any teacher-visited graph $x$ on any shortest
route to that target.  Define $d_y(x)=|x\triangle y|$.  For every
non-null edit $a$ supported by the target-conditioned action
law $q_t^y(\cdot\mid x)$, there is a route to $y$ that prescribes
$a$ at $x$.  Thus the edit corrects one mismatch:
\begin{equation}
q_t^y(a\mid x)>0,\ a\ne\bot
\quad\Longrightarrow\quad
d_y(T_a(x))=d_y(x)-1.
\label{eq:appA-target-descent}
\end{equation}
At $x=y$, every shortest route ending at $y$ has finished, so
\begin{equation}
q_t^y(\bot\mid y)=1.
\label{eq:appA-target-stop}
\end{equation}
The successor of a supported edit is another state on a teacher
route to $y$; therefore this reasoning applies repeatedly after
every event \citep{norris1997markov}.  For a target-aware policy that uses $q_t^y$, an
initial source $s_i$ with $y_i=y$ has distance $K_i$ from $y$.
After $n$ events, its distance must be
\begingroup\makeatletter\def\tagform@#1{}\def\theHequation{appAhidden40}\makeatother
\begin{equation}
d_y(Z_n)=\max\{K_i-n,0\}.
\label{eq:appA-target-conditioned-progress}
\end{equation}
\endgroup\addtocounter{equation}{-1}
To verify this equation, start with $d_y(Z_0)=K_i$; when it is
positive, \eqref{eq:appA-target-descent} subtracts exactly one;
when it is zero, \eqref{eq:appA-target-stop} keeps it zero \citep{diestel2017graph}.
Hence the process reaches $y$ at its $K_i$th event and remains
there, even if intermediate actions switch between the
continuations of different source routes:
\begingroup\makeatletter\def\tagform@#1{}\def\theHequation{appAhidden41}\makeatother
\begin{align}
\PP(X_T^{\mathrm{target}}\ne y_i\mid S=s_i)
 &=\PP(N_T<K_i),\notag\\
 &=\PP\{\Pois(B)<K_i\}.
\label{eq:appA-target-floor}
\end{align}
\endgroup\addtocounter{equation}{-1}
This is consistent with the target--graph marginal identity
\eqref{eq:appA-targetjoint}.  Conditioning on $Y$ can remove
cross-target ambiguity, but \emph{does not imply} that the
source-conditioned teacher trajectory is reproduced
\citep{debortoli2023augmented,lindvall1992lectures,
thorisson2000coupling}.

Under the separation assumption \eqref{eq:separator}, any
teacher-occupied state $x$ belongs to only one target class
$y(x)$.  Conditioning on that already-determined target does
not remove or change any terms in the pooled kernel:
\begingroup\makeatletter\def\tagform@#1{}\def\theHequation{appAhidden42}\makeatother
\begin{align}
q_t(a\mid x)
 &=\PP(A_t^\gamma=a\mid X_t^\gamma=x),\notag\\
 &=\PP(A_t^\gamma=a\mid X_t^\gamma=x,Y=y(x)),\notag\\
 &=q_t^{y(x)}(a\mid x).
\label{eq:appA-pooled-equals-target}
\end{align}
\endgroup\addtocounter{equation}{-1}
The middle equality holds because
$\PP(Y=y(x)\mid X_t^\gamma=x)=1$ \citep{cover2006elements}.
Thus the \emph{source-blind} projection inherits the target-aware
shortest-progress property without actually observing $Y$.
This gives an alternative view of the sufficiency proof in
Appendix~\ref{app:endpoint} \citep{lindvall1992lectures}.

\paragraph{Exact path imitation is a different, stronger demand.}
For every route state $x_{i,j}$ define its prescribed action
$a_{i,j}$, including $a_{i,K_i}=\bot$ at the terminal state.
Suppose we want a deterministic graph-only action map $f(x)$ to
reproduce \emph{every selected ordered route}, not merely arrive
at the correct final graph.  The necessary consistency condition is
\begin{equation}
x_{i,j}=x_{h,k}
\quad\Longrightarrow\quad
a_{i,j}=a_{h,k}
\quad\text{for every }i,h,j,k.
\label{eq:appA-action-consistency}
\end{equation}
\emph{Necessity:} if $f$ exactly executes route $i$, then
$f(x_{i,j})=a_{i,j}$.  If it also exactly executes route $h$,
then $f(x_{h,k})=a_{h,k}$.  When the two states are the same,
the left-hand sides are equal; therefore the right-hand sides
must agree.  There is no use of a source index in $f$ \citep{norris1997markov}.

\emph{Sufficiency:} assume \eqref{eq:appA-action-consistency}.
At each state visited by at least one selected route, set
$f(x)$ equal to the common action prescribed there.  The
condition ensures that $f(x)$ is well defined.  Define it
arbitrarily by a legal action at all other states.  Starting at
$s_i=x_{i,0}$, induction gives
\begingroup\makeatletter\def\tagform@#1{}\def\theHequation{appAhidden43}\makeatother
\begin{align}
Z_0&=x_{i,0},\notag\\
f(Z_0)&=a_{i,0},\notag\\
Z_1&=T_{a_{i,0}}(x_{i,0})=x_{i,1},\notag\\
f(Z_1)&=a_{i,1},\notag\\
Z_2&=T_{a_{i,1}}(x_{i,1})=x_{i,2},\notag\\
Z_j&=x_{i,j}\quad(0\le j\le K_i).
\label{eq:appA-route-imitation-induction}
\end{align}
\endgroup\addtocounter{equation}{-1}
At the endpoint, $f(y_i)=\bot$, so the process also reproduces
the prescribed terminal hold.  Thus
\eqref{eq:appA-action-consistency} is necessary and sufficient
for deterministic, time-independent, source-blind reproduction
of all selected routes.  The same consistency is needed to
reproduce the \emph{deterministic conditional action laws} of
all source-aware teachers by a common time-dependent kernel,
because every route position has positive probability at every
interior clock time by \eqref{eq:appA-positivepositions}.
This is a pathwise requirement, not the endpoint requirement of
Theorem~\ref{thm:endpoint} \citep{zhang2024trajectory,shaul2025flow}.

\paragraph{A concrete shared-target example.}
Take four distinct permitted edges $a,b,c,d$ and the two
pairs
\begingroup\makeatletter\def\tagform@#1{}\def\theHequation{appAhidden44}\makeatother
\begin{equation}
\begin{array}{ll}
s_1=\{a,b,c\}, & y_1=\varnothing,\\
s_2=\{a,b,d\}, & y_2=\varnothing.
\end{array}
\label{eq:appA-sametarget-example}
\end{equation}
\endgroup\addtocounter{equation}{-1}
Assume the displayed subsets are legal; for instance, all are
legal in a downward-closed edge-set family.  Choose the two
shortest routes
\begingroup\makeatletter\def\tagform@#1{}\def\theHequation{appAhidden45}\makeatother
\begin{align}
\{a,b,c\}&\xrightarrow{-c}\{a,b\}
\xrightarrow{-b}\{a\}
\xrightarrow{-a}\varnothing,\notag\\
\{a,b,d\}&\xrightarrow{-d}\{a,b\}
\xrightarrow{-a}\{b\}
\xrightarrow{-b}\varnothing.
\label{eq:appA-sametarget-routes}
\end{align}
\endgroup\addtocounter{equation}{-1}
They share the intermediate graph $\{a,b\}$.  At that graph
the first teacher prescribes deletion of $b$, whereas the
second prescribes deletion of $a$.  Therefore
\eqref{eq:appA-action-consistency} fails: no deterministic
graph-only rule can reproduce both \emph{ordered} paths.
Nevertheless, both next edits decrease the number of edges
remaining to the \emph{same} target $\varnothing$ from two to
one.  Whichever edit is sampled, exactly one further deletion
reaches $\varnothing$.  The projected process therefore
completes each task at the third event, giving
\begingroup\makeatletter\def\tagform@#1{}\def\theHequation{appAhidden46}\makeatother
\begin{equation}
D_{\mathrm{pair}}^{\mathrm{proj}}
=\sum_{i=1}^2p_i\PP\{\Pois(B)<3\}=\varepsilon_B.
\label{eq:appA-sametarget-exact}
\end{equation}
\endgroup\addtocounter{equation}{-1}
\emph{disagreement on an intermediate action can destroy path
imitation without destroying source--target endpoint recovery} \citep{lindvall1992lectures,debortoli2023augmented}.

\paragraph{Scope of the recovered coupling.}
The theorem concerns equality with the unavoidable
\emph{finite-budget} floor, not zero error for arbitrary $B$.
It concerns a \emph{fixed teacher route family and its own
pooled Markov projection}, not all possible source-blind
policies trained with a separate objective.  Finally,
same-target route merging is harmless for the theorem precisely
because the prescribed conditional target is deterministic and
all permitted continuations on its teacher routes are shortest;
changing either assumption requires a new analysis.

\clearpage

\section{Routing and Scheduling Guarantees}
\label{app:certificates}

Appendix~\ref{app:projection} establishes that the
source-blind Markov projection attains the unavoidable
clock-error floor if and only if routes leading to
distinct targets remain state-disjoint.
This characterization identifies the structural
condition for endpoint recoverability, but does not
determine whether such routes can be selected under
the prescribed graph constraints \citep{korte2018combinatorial,ito2011reconfiguration}.
We now turn to the combinatorial problem of constructing
and certifying legal shortest routes that satisfy
this condition.

This section expands the combinatorial claims used in
Section~\ref{sec:theory} and in the coupling-aware scheduling
construction of Section~\ref{sec:construction}.  In particular,
we prove the exact integer-route correspondence in
Proposition~\ref{prop:integer}, analyze each accepted step of
Algorithm~\ref{alg:swap}, verify the marker-free construction
\eqref{eq:cross}--\eqref{eq:crosspaths} by explicit finite
classification, and give supplementary capacity and cut
certificates for Theorem~\ref{thm:capacity}.  All claims concern
\emph{legal monotone shortest edit routes} in the finite state
family $\X$ defined in Section~\ref{sec:setting}.  

\subsection{Integral Route Selection and Zero-Length Cases}
\label{app:subsection-04}

\paragraph{From mandatory edges to a directed state graph.}
Fix pair $i$ and recall from Section~\ref{sec:setting} that
\begingroup\makeatletter\def\tagform@#1{}\def\theHequation{appBhidden1}\makeatother
\begin{align}
 D_i&=s_i\setminus y_i,\notag\\
 I_i&=y_i\setminus s_i,\notag\\
 K_i&=|D_i|+|I_i|=|s_i\triangle y_i|.
 \label{eq:appB-edit-count}
\end{align}
\endgroup\addtocounter{equation}{-1}
The distinction between $I_i$ (the set of insertions) and
$\mathcal I_i$ (the legal monotone interval) is important.
Equation~\eqref{eq:interval} defines the latter by
\begingroup\makeatletter\def\tagform@#1{}\def\theHequation{appBhidden2}\makeatother
\begin{equation}
 \mathcal I_i
 =\{x\in\X:s_i\cap y_i\subseteq x\subseteq s_i\cup y_i\}.
 \label{eq:appB-interval}
\end{equation}
\endgroup\addtocounter{equation}{-1}
Thus, if $x\in\mathcal I_i$, every edge common to $s_i$
and $y_i$ is still present in $x$, and every edge absent
from both endpoints is still absent in $x$.  The only
coordinates that may differ are the $K_i$ mandatory edits \citep{gao2010survey,diestel2017graph}.
This is the usual edge-set representation of graph editing.

For a graph state $x$, define its remaining edit distance to
target $y_i$ as $d_i(x)=|x\triangle y_i|$.  On the interval,
we may count the incorrect edges without double counting: 
\begingroup\makeatletter\def\tagform@#1{}\def\theHequation{appBhidden3}\makeatother
\begin{align}
 d_i(x)
 &= |x\setminus y_i|+|y_i\setminus x|\notag\\
 &= |D_i\cap x|+|I_i\setminus x|.
 \label{eq:appB-distance-identity}
\end{align}
\endgroup\addtocounter{equation}{-1}
In the second line, $x\setminus y_i=D_i\cap x$ because
$x\subseteq s_i\cup y_i$, while
$y_i\setminus x=I_i\setminus x$ because
$s_i\cap y_i\subseteq x$.  At the source and target,
respectively, this yields
\begingroup\makeatletter\def\tagform@#1{}\def\theHequation{appBhidden4}\makeatother
\begin{align}
 d_i(s_i)
 &=|D_i\cap s_i|+|I_i\setminus s_i|\notag\\
 &=|D_i|+|I_i|=K_i,\notag\\
 d_i(y_i)
 &=|D_i\cap y_i|+|I_i\setminus y_i|\notag\\
 &=0+0=0.
 \label{eq:appB-distance-ends}
\end{align}
\endgroup\addtocounter{equation}{-1}
Every admissible arc $(x,v)$ of $H_i$ is a legal single-edge
toggle satisfying $d_i(v)=d_i(x)-1$, as defined beneath
\eqref{eq:interval}.  It therefore corrects exactly one
mandatory coordinate.  If a directed walk contains $r$
arcs, repeated substitution gives
\begingroup\makeatletter\def\tagform@#1{}\def\theHequation{appBhidden5}\makeatother
\begin{align}
 d_i(x_1)&=d_i(x_0)-1,\notag\\
 d_i(x_2)&=d_i(x_1)-1=d_i(x_0)-2,\notag\\
 d_i(x_r)&=d_i(x_0)-r.
 \label{eq:appB-distance-walk}
\end{align}
\endgroup\addtocounter{equation}{-1}
No directed cycle is possible: a cycle of length $r>0$
would imply both $d_i(x_r)=d_i(x_0)$ (same state) and
$d_i(x_r)=d_i(x_0)-r$ (strict descent).  These two equalities
would force $r=0$, a contradiction.  Hence $H_i$ is a
directed acyclic graph (DAG) \citep{bangjensen2009digraphs}.  Any directed path from $s_i$
to $y_i$ has length exactly $K_i$, since
\begin{align}
 0=d_i(y_i)&=d_i(s_i)-r\notag\\
 &=K_i-r,\notag\\
 r&=K_i.
 \label{eq:appB-path-length}
\end{align}
Every such path is therefore a legal monotone shortest
route; conversely, each legal monotone shortest route
follows arcs of $H_i$ \citep{korte2018combinatorial,bangjensen2009digraphs}.  This correspondence uses only the
state-membership predicate $x\in\X$, not downward closure
or any special graph topology \citep{korte2018combinatorial}.

\paragraph{Binary feasibility system.}
Write $F_i$ for the arcs of $H_i$,
$W=\bigcup_i\mathcal I_i$, and
$\mathcal Y=\{y_i:1\le i\le m\}$.  As in
\eqref{eq:flow}, let $f^i_{uv}\in\{0,1\}$ indicate
whether task $i$ selects arc $(u,v)\in F_i$.  Let
$z_{y,x}\in\{0,1\}$ indicate whether state $x\in W$
is assigned to target label $y\in\mathcal Y$ \citep{ahuja1993network,schrijver2003combinatorial}.
For $K_i>0$, introduce the signed flow balance \citep{ford1962flows,ahuja1993network}
\begingroup\makeatletter\def\tagform@#1{}\def\theHequation{appBhidden6}\makeatother
\begin{equation}
 b_i(x)=
 \sum_{(x,v)\in F_i}f^i_{xv}
 -\sum_{(u,x)\in F_i}f^i_{ux}.
 \label{eq:appB-balance}
\end{equation}
\endgroup\addtocounter{equation}{-1}
The integer system in the main text requires, for
$K_i>0$,
\begin{align}
 b_i(s_i)&=1,\notag\\
 b_i(y_i)&=-1,\notag\\
 b_i(x)&=0\quad(x\notin\{s_i,y_i\}),
 \label{eq:appB-flow-balance}
\end{align}
and also, for every task and every selected-arc candidate,
\begin{align}
 f^i_{uv}&\le z_{y_i,u},\notag\\
 f^i_{uv}&\le z_{y_i,v},\notag\\
 \sum_{y\in\mathcal Y}z_{y,x}&\le1
   \quad(x\in W),\notag\\
 z_{y_i,s_i}&=1,\notag\\
 z_{y_i,y_i}&=1.
 \label{eq:appB-flow-occupancy}
\end{align}
Here $\sum_y z_{y,x}\le1$ forbids assigning a single
graph state to two distinct targets, but it places no
restriction on how many routes with the \emph{same}
target may use that state.  This target-labeled form is
crucial because the theorem in the main text requires
target separation, not source separation
\citep{botea2021counting,ahuja1993network}.

\paragraph{Why one binary unit flow is exactly one path.}
We now prove the forward implication in
Proposition~\ref{prop:integer}   Fix $i$ with $K_i>0$ and
assume the binary variables satisfy
\eqref{eq:appB-flow-balance}.  Since $d_i(s_i)=K_i$
is the largest possible distance on $\mathcal I_i$,
no arc of $H_i$ enters $s_i$ \citep{bangjensen2009digraphs}.  Therefore
\begingroup\makeatletter\def\tagform@#1{}\def\theHequation{appBhidden7}\makeatother
\begin{align}
 \sum_{(u,s_i)\in F_i}f^i_{us_i}&=0,\notag\\
 b_i(s_i)&=
 \sum_{(s_i,v)\in F_i}f^i_{s_iv}=1,\notag\\
 \sum_{(s_i,v)\in F_i}f^i_{s_iv}&=1.
 \label{eq:appB-source-flow}
\end{align}
\endgroup\addtocounter{equation}{-1}
Because every $f^i_{uv}$ is either zero or one, there
is exactly one selected outgoing arc at the source \citep{ford1962flows}.
Follow it from $s_i$ to some state $v_1$.  If $v_1\ne y_i$,
the arriving arc contributes at least one unit to its
incoming flow.  Conservation at $v_1$ \citep{ahuja1993network} then gives
\begingroup\makeatletter\def\tagform@#1{}\def\theHequation{appBhidden8}\makeatother
\begin{align}
 0=b_i(v_1)
 &=\operatorname{out}_i(v_1)-\operatorname{in}_i(v_1),\notag\\
 \operatorname{out}_i(v_1)&=\operatorname{in}_i(v_1)\ge1.
 \label{eq:appB-continue}
\end{align}
\endgroup\addtocounter{equation}{-1}
Thus at least one selected outgoing arc exists.  Repeat
this step.  Each new arc reduces $d_i$ by one, so no state
can be revisited and the walk has at most $K_i$ arcs.
It cannot stop at a state other than $y_i$, because such a
state has positive incoming flow and zero balance, and
therefore must have positive outgoing flow.  It must end
at $y_i$.  By \eqref{eq:appB-path-length}, this selected
path has exactly $K_i$ arcs.

To exclude \emph{additional} selected arcs, subtract one
unit from $f^i$ on each arc of the path just found \citep{ford1962flows,schrijver2003combinatorial}.
Call the residual binary/nonnegative arc values $h_{uv}$.
Subtracting a complete source--target path changes the
original balances by $-1$ at the source, $+1$ at the target,
and zero elsewhere.  Comparing with
\eqref{eq:appB-flow-balance} shows
\begin{equation}
 \sum_{(x,v)\in F_i}h_{xv}
 -\sum_{(u,x)\in F_i}h_{ux}=0
 \quad\text{for every }x.
 \label{eq:appB-residual-balance}
\end{equation}
Suppose a residual arc exists.  Among all tails $u$ of
arcs with $h_{uv}>0$, select one for which $d_i(u)$ is
maximal.  Every possible predecessor $w\to u$ has
$d_i(w)=d_i(u)+1$.  If $h_{wu}>0$, then $w$ would be
the tail of a positive residual arc with greater
distance, contradicting maximality.  Hence the residual
incoming flow at $u$ is zero.  The residual outgoing
flow is positive by the choice of $u$.  Consequently
\begin{align*}
 \operatorname{out}_h(u)&>0,\\
 \operatorname{in}_h(u)&=0,\\
 \operatorname{out}_h(u)-\operatorname{in}_h(u)&>0,
\end{align*}
contradicting \eqref{eq:appB-residual-balance} \citep{bangjensen2009digraphs}.
Thus every residual arc is zero, and the original binary
flow consists of exactly the one path found above.
This is the DAG version of integral
flow decomposition \citep{ahuja1993network,schrijver2003combinatorial}.

\paragraph{Occupancy constraints and both directions of exactness.}
If a selected route $i$ visits $x$, either $x$ is one of
its endpoints or it touches a selected arc \citep{ahuja1993network,botea2021counting}.  At an
endpoint, \eqref{eq:appB-flow-occupancy} explicitly sets
$z_{y_i,x}=1$.  At an interior point, at least one
selected incident arc has $f^i_{uv}=1$, and the inequalities
$f^i_{uv}\le z_{y_i,u},z_{y_i,v}$ again force
$z_{y_i,x}=1$.  If a second route $h$ with $y_h\ne y_i$
visited that same $x$, we would also have $z_{y_h,x}=1$.
Then
\begin{align*}
 \sum_{y\in\mathcal Y}z_{y,x}
 &\ge z_{y_i,x}+z_{y_h,x}\\
 &=1+1=2\\
 &>1,
\end{align*}
contradicting the per-state capacity constraint.
The selected paths must therefore satisfy the
separation criterion \eqref{eq:separator}.  By
Theorem~\ref{thm:endpoint}, their source-blind Markov
projection reaches the clock floor
$D_{\mathrm{pair}}^{\mathrm{proj}}=\varepsilon_B$.

For the reverse implication, suppose a selected family
$\gamma=(\gamma_1,\ldots,\gamma_m)$ of legal shortest
routes satisfies \eqref{eq:separator}.  Define
\begingroup\makeatletter\def\tagform@#1{}\def\theHequation{appBhidden9}\makeatother
\begin{equation}
 f^i_{uv}=
 \begin{cases}
 1,&(u,v)\text{ lies on }\gamma_i,\\
 0,&\text{otherwise}.
 \end{cases}
 \label{eq:appB-construct-flow}
\end{equation}
\endgroup\addtocounter{equation}{-1}
Each nonzero-length route has one outgoing and no
incoming selected arc at its source, one incoming and
no outgoing arc at its target, and equal incoming and
outgoing counts at every interior point.  Therefore
\eqref{eq:appB-flow-balance} holds.  Set
\begingroup\makeatletter\def\tagform@#1{}\def\theHequation{appBhidden10}\makeatother
\begin{equation}
 z_{y,x}=\begin{cases}
 1,&\text{some selected route with target }y
       \text{ visits }x,\\
 0,&\text{otherwise}.
 \end{cases}
 \label{eq:appB-construct-label}
\end{equation}
\endgroup\addtocounter{equation}{-1}
A selected arc's two endpoints belong to its route,
so both inequalities in
\eqref{eq:appB-flow-occupancy} hold.  Because routes
to different targets are disjoint, at most one target
label is assigned to any $x$, giving
$\sum_y z_{y,x}\le1$.  The chosen routes contain their
sources and targets, so their endpoint labels also
satisfy the system \citep{ford1962flows,schrijver2003combinatorial}.  Thus every target-separated route
family produces a feasible binary solution.  Together
with the preceding implication and
Theorem~\ref{thm:endpoint}, this proves the exact
``if and only if'' statement in
Proposition~\ref{prop:integer}.

\paragraph{Zero-length tasks.}
If $K_i=0$, then $D_i=I_i=\varnothing$ and
$s_i=y_i$.  There is one legal route, consisting solely
of the single state $s_i$ and zero arcs \citep{korte2018combinatorial}.  For that task,
applying a unit source--sink balance would incorrectly
require $b_i(s_i)=1$ and $b_i(s_i)=-1$ at the same state.
The main-text system therefore explicitly applies its
flow-balance equations only to tasks with $K_i>0$.
For $K_i=0$, the correct encoding is
\begin{align}
 f^i&\text{ has no selected arcs},\notag\\
 z_{y_i,s_i}&=1,\notag\\
 z_{y_i,y_i}&=z_{y_i,s_i}=1.
 \label{eq:appB-zero-length}
\end{align}
If any route to a different target also visits $s_i$,
the per-state label constraint forbids the simultaneous
assignments.  Routes with the same target can visit it
without conflict.  The clock error of this task is
\begin{align*}
 \PP\{\Pois(B)<K_i\}
 &=\PP\{\Pois(B)<0\}\\
 &=0.
\end{align*}
Hence it is essential that the zero-length route itself
be protected from cross-target collisions: it is not
permitted to disappear from the certificate merely
because it has no arcs.

\paragraph{Scope of the integer formulation.}
The proof relies on binary $f$ and $z$ variables \citep{schrijver2003combinatorial,bertsimas1997linear}.
It does \emph{not} establish that a fractional relaxation
of the multi-target system is integral.  The standard
single-commodity flow-decomposition result is only one
ingredient; the shared target-capacity restrictions couple
different commodities \citep{schrijver2003combinatorial,
bertsimas1997linear}.  Also, an interval with $K_i$
mandatory edges can contain up to $2^{K_i}$ states
before legality filtering.  Indeed each of the $K_i$
coordinates may be in its source or target status \citep{korte2018combinatorial},
giving
\begingroup\makeatletter\def\tagform@#1{}\def\theHequation{appBhidden11}\makeatother
\begin{align}
 |\mathcal I_i|&\le 2^{K_i},\notag\\
 |W|=\left|\bigcup_i\mathcal I_i\right|
 &\le \sum_i |\mathcal I_i|\notag\\
 &\le \sum_i2^{K_i}.
 \label{eq:appB-state-size}
\end{align}
\endgroup\addtocounter{equation}{-1}
Thus Proposition~\ref{prop:integer} is an exact finite
feasibility certificate, not a polynomial-time guarantee.

\subsection{Adjacent-Swap Legality and Finite Termination}
\label{app:subsection-05}

\paragraph{Commuting two different mandatory toggles.}
Consider a selected shortest route $\gamma_i$ and two
adjacent edits at positions $j$ and $j+1$.  Let $x$ be
the legal state just before those edits.  The two actions
toggle edges $e$ and $f$.  They are distinct because a
monotone shortest route toggles every mandatory edge
exactly once (Section~\ref{sec:setting}) \citep{gao2010survey,diestel2017graph}.  Toggling
one edge is the symmetric-difference operation, so
\begin{align}
 T_e(x)&=x\triangle\{e\},\notag\\
 T_f(T_e(x))
 &=(x\triangle\{e\})\triangle\{f\},\notag\\
 &=x\triangle(\{e\}\triangle\{f\}),\notag\\
 &=x\triangle\{e,f\},
 \label{eq:appB-forward-swap}
\end{align}
and reversing their order gives
\begin{align}
 T_f(x)&=x\triangle\{f\},\notag\\
 T_e(T_f(x))
 &=(x\triangle\{f\})\triangle\{e\},\notag\\
 &=x\triangle(\{f\}\triangle\{e\}),\notag\\
 &=x\triangle\{e,f\}.
 \label{eq:appB-reverse-swap}
\end{align}
We used associativity and commutativity of symmetric
difference \citep{diestel2017graph}; the edits need not both be insertions or
both be deletions.  Equations
\eqref{eq:appB-forward-swap}--\eqref{eq:appB-reverse-swap}
show that the two edit orders reach exactly the same
state after two operations \citep{diestel2017graph,
korte2018combinatorial}.  Every state before the pair
of edits is unchanged.  Every state after the pair is
unchanged.  Only the state after the \emph{first} of
the two edits can differ \citep{ito2011reconfiguration}: 
\begingroup\makeatletter\def\tagform@#1{}\def\theHequation{appBhidden12}\makeatother
\begin{align}
 x_{\mathrm{old}}&=T_e(x)=x\triangle\{e\},\notag\\
 x_{\mathrm{new}}&=T_f(x)=x\triangle\{f\}.
 \label{eq:appB-intermediate}
\end{align}
\endgroup\addtocounter{equation}{-1}
The original route already certifies
$x\in\X$ and $T_f(T_e(x))\in\X$.
Therefore the entire swapped route is legal exactly
when the only new intermediate state satisfies
\begin{equation}
 \boxed{x_{\mathrm{new}}\in\X}.
 \label{eq:appB-single-oracle}
\end{equation}
This is one call to the state-membership oracle \citep{ito2011reconfiguration,saito2024reconfiguration}.
The same argument proves that every mandatory edit
remains present exactly once, the route length stays
$K_i$, and both endpoints remain fixed.  An initially
legal monotone shortest route is required; the swap
rule does not create such a route from an infeasible
instance.  Related path-reconfiguration questions are
studied in combinatorial optimization
\citep{saito2024reconfiguration,gajjar2024reconfiguring}.

\paragraph{Expanding the conflict score into route pairs.}
Recall $n_y(x)$ and $C(\gamma)$ from
\eqref{eq:collision} \citep{botea2021counting}.  Because every monotone route
strictly decreases its distance to the target, it
cannot revisit any state \citep{bangjensen2009digraphs}.  Consequently, for each
$y$ and $x$,
\begingroup\makeatletter\def\tagform@#1{}\def\theHequation{appBhidden13}\makeatother
\begin{equation}
 n_y(x)=\sum_{i:y_i=y}\mathbf1\{x\in\mathcal V_i\}.
 \label{eq:appB-count-indicators}
\end{equation}
\endgroup\addtocounter{equation}{-1}
Substituting this identity into the conflict score,
then distributing the finite sums one step at a time,
produces
\begin{align}
 C(\gamma)
 &=\sum_x\sum_{y<z}n_y(x)n_z(x)\notag\\
 &=\sum_x\sum_{y<z}
   \left(\sum_{i:y_i=y}\mathbf1\{x\in\mathcal V_i\}\right)
   \left(\sum_{h:y_h=z}\mathbf1\{x\in\mathcal V_h\}\right)\notag\\
 &=\sum_x\sum_{y<z}
   \sum_{i:y_i=y}\sum_{h:y_h=z}
   \mathbf1\{x\in\mathcal V_i\}
   \mathbf1\{x\in\mathcal V_h\}\notag\\
 &=\sum_{i<h:\,y_i\ne y_h}\sum_x
   \mathbf1\{x\in\mathcal V_i\cap\mathcal V_h\}\notag\\
 &=\sum_{i<h:\,y_i\ne y_h}
   |\mathcal V_i\cap\mathcal V_h|.
 \label{eq:appB-conflict-pairs}
\end{align}
The fourth equality counts each unordered pair of
routes with different targets exactly once; the
chosen total order $y<z$ has no numerical effect \citep{korte2018combinatorial}.
The final equality counts the elements in each
finite intersection.  Therefore
\begin{align}
 C(\gamma)&\in\{0,1,2,\ldots\},\notag\\
 C(\gamma)=0
 &\Longleftrightarrow
 |\mathcal V_i\cap\mathcal V_h|=0
 \quad\text{for all }y_i\ne y_h\notag\\
 &\Longleftrightarrow
 \mathcal V_i\cap\mathcal V_h=\varnothing
 \quad\text{for all }y_i\ne y_h.
 \label{eq:appB-certificate-iff}
\end{align}
By Theorem~\ref{thm:endpoint}, this is also equivalent
to achieving the clock floor for the \emph{fixed-route
Markov projection}.  A positive score is a conflict
count, not an estimate of paired TV error.

\paragraph{One-swap score update.}
Suppose a proposed swap changes route $i$ only by
replacing $x_{\mathrm{old}}$ with
$x_{\mathrm{new}}$ in its visited-state set.  The two
states are distinct because $e\ne f$ \citep{diestel2017graph}.  Neither state
can occur at another position of $\gamma_i$:
the old and new intermediate states both have
remaining distance $K_i-(j+1)$, and a monotone route
has exactly one state at each distance.  Define the
other-target incidence at any state $v$ by
\begingroup\makeatletter\def\tagform@#1{}\def\theHequation{appBhidden14}\makeatother
\begin{equation}
 M_i(v)=\sum_{h:y_h\ne y_i}
      \mathbf1\{v\in\mathcal V_h\}.
 \label{eq:appB-other-target}
\end{equation}
\endgroup\addtocounter{equation}{-1}
Only pairwise intersections involving task $i$ can
change; every pair of other routes is untouched \citep{ito2011reconfiguration}.
For task $i$, the old intermediate state contributes
$M_i(x_{\mathrm{old}})$ conflicts, and the new state
contributes $M_i(x_{\mathrm{new}})$.  Subtraction gives
\begin{align}
 C(\gamma_{\mathrm{new}})
 &=C(\gamma_{\mathrm{old}})
   -M_i(x_{\mathrm{old}})
   +M_i(x_{\mathrm{new}}),\notag\\
 C(\gamma_{\mathrm{new}})-C(\gamma_{\mathrm{old}})
 &=M_i(x_{\mathrm{new}})-M_i(x_{\mathrm{old}}).
 \label{eq:appB-local-delta}
\end{align}
The last identity gives the conflict-score update.  It also makes the
algorithm's acceptance test transparent: after checking
\eqref{eq:appB-single-oracle}, the swap is accepted
only when the right-hand side of
\eqref{eq:appB-local-delta} is strictly negative.

\paragraph{A finite descent argument and its limits.}
Every visited-state set has $K_i+1$ elements \citep{bangjensen2009digraphs}, so
\begingroup\makeatletter\def\tagform@#1{}\def\theHequation{appBhidden15}\makeatother
\begin{align}
 |\mathcal V_i\cap\mathcal V_h|
 &\le\min\{|\mathcal V_i|,|\mathcal V_h|\}\notag\\
 &=\min\{K_i+1,K_h+1\}.
 \label{eq:appB-overlap-upper}
\end{align}
\endgroup\addtocounter{equation}{-1}
Summing over different-target pairs yields the
finite upper bound
\begingroup\makeatletter\def\tagform@#1{}\def\theHequation{appBhidden16}\makeatother
\begin{equation}
 0\le C(\gamma)\le
 \sum_{i<h:\,y_i\ne y_h}\min\{K_i+1,K_h+1\}.
 \label{eq:appB-score-bound}
\end{equation}
\endgroup\addtocounter{equation}{-1}
If a swap is accepted, the exact score decreases by
at least one because both old and new scores are
integers \citep{korte2018combinatorial}:
\begingroup\makeatletter\def\tagform@#1{}\def\theHequation{appBhidden17}\makeatother
\begin{align}
 C_{r+1}&<C_r,\notag\\
 C_{r+1}&\le C_r-1,\notag\\
 C_r&\le C_0-r.
 \label{eq:appB-strict-descent}
\end{align}
\endgroup\addtocounter{equation}{-1}
The last line follows by applying the second line
repeatedly $r$ times.  Since $C_r\ge0$, we obtain
\begin{align}
 0&\le C_r\le C_0-r,\notag\\
 r&\le C_0.
 \label{eq:appB-termination}
\end{align}
Thus Algorithm~\ref{alg:swap} makes at most
$C(\gamma_{\mathrm{initial}})$ accepted swaps.
After a complete unsuccessful scan of the finite list
of legal adjacent swaps, the algorithm terminates at a
local minimum for the particular swap neighborhood and
scan order \citep{ito2011reconfiguration,gajjar2024reconfiguring}.  If the returned score is zero,
\eqref{eq:appB-certificate-iff} certifies global
\emph{feasibility of the displayed routes} and the
clock-floor guarantee.  If it is positive, no global
infeasibility claim follows: a different legal route
family, or a sequence involving an initially
non-improving swap, may have score zero.  Likewise,
strictly reducing $C$ does not prove that the exact
paired risk decreases, because the risk also depends
on clock-weighted action choices.  This distinction
between constrained path feasibility and local
reconfiguration is important
\citep{botea2021counting,gajjar2024reconfiguring}.

\subsection{Marker-Free Route Enumeration}
\label{app:subsection-06}

\paragraph{What the example fixes.}
We verify the crossed construction in
\eqref{eq:cross}--\eqref{eq:crosspaths} without
assuming any edge is a dedicated source marker.
Let $R$ be disjoint from the four distinct edges
$a,b,c,d$.  It suffices to count the variable parts \citep{diestel2017graph},
because the map $U\mapsto R\cup U$ is one-to-one
when $U\subseteq\{a,b,c,d\}$:
\begingroup\makeatletter\def\tagform@#1{}\def\theHequation{appBhidden18}\makeatother
\begin{align}
 R\cup U=R\cup V
 &\Longrightarrow (R\cup U)\setminus R
               =(R\cup V)\setminus R\notag\\
 &\Longrightarrow U=V.
 \label{eq:appB-backbone-injective}
\end{align}
\endgroup\addtocounter{equation}{-1}
We assume precisely the edge-capacity-two legal family
on the variable edges (equivalently, capacity
$|R|+2$ with $R$ held fixed), so every set with at most
two variable edges is legal \citep{madeira2024generative}.  The variable endpoints
of the two tasks are
\begingroup\makeatletter\def\tagform@#1{}\def\theHequation{appBhidden19}\makeatother
\begin{align}
 s_1=\{a,b\},\quad &y_1=\{c,d\},\notag\\
 s_2=\{a,c\},\quad &y_2=\{b,d\}.
 \label{eq:appB-cross-endpoints}
\end{align}
\endgroup\addtocounter{equation}{-1}
Both tasks have two distinct deletions and two distinct
insertions.  Their minimum length is four:
\begingroup\makeatletter\def\tagform@#1{}\def\theHequation{appBhidden20}\makeatother
\begin{align}
 K_1&=|\{a,b\}|+|\{c,d\}|=2+2=4,\notag\\
 K_2&=|\{a,c\}|+|\{b,d\}|=2+2=4.
 \label{eq:appB-cross-length}
\end{align}
\endgroup\addtocounter{equation}{-1}

\paragraph{Counting eight legal orders for each task.}
Initially, each task already has two variable edges.
An insertion as the first edit would give three edges
and violate the capacity \citep{madeira2024generative,ito2011reconfiguration}.  The first edit must therefore
be one of the two deletions: exactly two choices.
After that deletion, one old edge remains and the
variable-edge count is one.  There are two disjoint
possibilities.

In the first possibility, the second edit deletes the
remaining old edge.  Both insertions can then be made
in either order.  This produces $2!=2$ completions
for each first-deletion choice.  In the second
possibility, the second edit inserts one of the two
new edges, giving two choices.  The capacity is then
full again, so the third edit must delete the last
old edge, and the fourth edit must insert the last
new edge.  This produces another two completions
for each first-deletion choice.  Thus the exact count is
\begin{align}
 \#\{\text{first deletions}\}&=2,\notag\\
 \#\{\text{delete-second completions}\}&=2,\notag\\
 \#\{\text{insert-second completions}\}&=2,\notag\\
 \#\{\text{legal routes per task}\}
 &=2(2+2)\notag\\
 &=2\cdot4\notag\\
 &=8.
 \label{eq:appB-eight-orders}
\end{align}
The alternatives are exhaustive: after the first
deletion, the second edit must either delete the
remaining old edge or insert one of the new edges \citep{korte2018combinatorial}.
Consequently the number of ordered route pairs is
\begingroup\makeatletter\def\tagform@#1{}\def\theHequation{appBhidden21}\makeatother
\begin{align}
 8\cdot8&=64.
 \label{eq:appB-sixtyfour}
\end{align}
\endgroup\addtocounter{equation}{-1}
Work on disjoint shortest routes and on
reconfiguring shortest paths gives the broader
combinatorial context
\citep{botea2021counting,saito2024reconfiguration}.

\paragraph{Listing the legal routes.}
Write $+e$ for inserting edge $e$ and $-e$ for
deleting edge $e$ \citep{gao2010survey,havasi2025edit}.  The eight legal action sequences
for each task are listed side by side below.  The
superscript $r$ is a route index, not a clock-event
number:
\begin{equation}
\renewcommand{\arraystretch}{1.25}
\begin{array}{c|l|l}
 r & \gamma_1^{(r)} & \gamma_2^{(r)}\\\hline
 1&(-a,+c,-b,+d)&(-a,+b,-c,+d)\\
 2&(-a,+d,-b,+c)&(-a,+d,-c,+b)\\
 3&(-a,-b,+c,+d)&(-a,-c,+b,+d)\\
 4&(-a,-b,+d,+c)&(-a,-c,+d,+b)\\
 5&(-b,+c,-a,+d)&(-c,+b,-a,+d)\\
 6&(-b,+d,-a,+c)&(-c,+d,-a,+b)\\
 7&(-b,-a,+c,+d)&(-c,-a,+b,+d)\\
 8&(-b,-a,+d,+c)&(-c,-a,+d,+b)
\end{array}
\label{eq:appB-route-catalog}
\end{equation}
Each listed sequence can be checked directly: its
first edit is a deletion, each prefix has at most
two edges, and all four mandatory edits occur exactly
once \citep{madeira2024generative}.  Conversely, \eqref{eq:appB-eight-orders}
proves there are no missing legal sequences.

\paragraph{Computing the six state-disjoint route pairs.}
For each route, let its three \emph{interior} visited
states be the three states reached after one, two,
and three edits.  We suppress braces for singleton
edge sets \citep{diestel2017graph}; a two-letter symbol such as $bc$ denotes
$\{b,c\}$, and $\varnothing$ denotes the empty set.
The endpoint states are always $\{a,b\}$ and $\{c,d\}$
for task 1, and $\{a,c\}$ and $\{b,d\}$ for task 2.
Reading the action lists in
\eqref{eq:appB-route-catalog}, the intermediate
states are exactly
\begin{equation}
\renewcommand{\arraystretch}{1.22}
\begin{array}{c|c|c}
 r & \text{task 1 interiors} & \text{task 2 interiors}\\\hline
 1&(b,bc,c)&(c,bc,b)\\
 2&(b,bd,d)&(c,cd,d)\\
 3&(b,\varnothing,c)&(c,\varnothing,b)\\
 4&(b,\varnothing,d)&(c,\varnothing,d)\\
 5&(a,ac,c)&(a,ab,b)\\
 6&(a,ad,d)&(a,ad,d)\\
 7&(a,\varnothing,c)&(a,\varnothing,b)\\
 8&(a,\varnothing,d)&(a,\varnothing,d)
\end{array}
\label{eq:appB-interior-catalog}
\end{equation}
For example, task 1 route 8 visits
\begingroup\makeatletter\def\tagform@#1{}\def\theHequation{appBhidden22}\makeatother
\begin{align}
 \{a,b\}&\to\{a\}\notag\\
 &\to\varnothing\notag\\
 &\to\{d\}\notag\\
 &\to\{c,d\},
 \label{eq:appB-route8-states}
\end{align}
\endgroup\addtocounter{equation}{-1}
and task 2 route 1 visits
\begingroup\makeatletter\def\tagform@#1{}\def\theHequation{appBhidden23}\makeatother
\begin{align}
 \{a,c\}&\to\{c\}\notag\\
 &\to\{b,c\}\notag\\
 &\to\{b\}\notag\\
 &\to\{b,d\}.
 \label{eq:appB-route1-states}
\end{align}
\endgroup\addtocounter{equation}{-1}
The first state set is
$\{ab,a,\varnothing,d,cd\}$ and the second is
$\{ac,c,bc,b,bd\}$.  No symbol appears in both
sets, so these routes are target-separated.

For all 64 combinations, define the compatible
route-index set
\begingroup\makeatletter\def\tagform@#1{}\def\theHequation{appBhidden24}\makeatother
\begin{equation}
 J_r=\{s\in\{1,\ldots,8\}:
       \mathcal V_1^{(r)}\cap\mathcal V_2^{(s)}
          =\varnothing\}.
 \label{eq:appB-compatible-def}
\end{equation}
\endgroup\addtocounter{equation}{-1}
To calculate $J_r$, compare the five visited states
for task 1 route $r$ with the five visited states for
each task 2 route $s$, including both endpoints.
The complete direct comparison of
\eqref{eq:appB-interior-catalog} gives
\begingroup\makeatletter\def\tagform@#1{}\def\theHequation{appBhidden25}\makeatother
\begin{equation}
\renewcommand{\arraystretch}{1.22}
\begin{array}{c|c@{\qquad}c|c}
 r&J_r&r&J_r\\\hline
 1&\{6,8\}&5&\varnothing\\
 2&\varnothing&6&\{1,3\}\\
 3&\{6\}&7&\varnothing\\
 4&\varnothing&8&\{1\}
\end{array}
\label{eq:appB-compatible-sets}
\end{equation}
\endgroup\addtocounter{equation}{-1}
For instance, the two task-2 routes compatible with
task-1 route 1 are routes 6 and 8.  The set $J_2$
is empty because each of task-2's eight paths shares
at least one state with task-1 route 2.  
Summing the eight cardinalities separately yields
\begin{align}
 \#\{\text{state-disjoint pairs}\}
 &=|J_1|+|J_2|+|J_3|+|J_4|\notag\\
 &\quad+|J_5|+|J_6|+|J_7|+|J_8|\notag\\
 &=2+0+1+0+0+2+0+1\notag\\
 &=6.
 \label{eq:appB-six-zero}
\end{align}
The six pairs, written as (task-1 route index,
task-2 route index), are
\begingroup\makeatletter\def\tagform@#1{}\def\theHequation{appBhidden26}\makeatother
\begin{equation}
 (1,6),\ (1,8),\ (3,6),\ (6,1),\ (6,3),\ (8,1).
 \label{eq:appB-six-list}
\end{equation}
\endgroup\addtocounter{equation}{-1}
Every one has conflict score zero by
\eqref{eq:appB-conflict-pairs}, and every other one
has positive score.  Hence the capacity-two crossed
instance admits clock-floor coupling recovery even
though it has no globally private variable edge.
Indeed both tasks use the same edge universe:
\begingroup\makeatletter\def\tagform@#1{}\def\theHequation{appBhidden27}\makeatother
\begin{align}
 s_1\cup y_1 &=R\cup\{a,b,c,d\},\notag\\
 s_2\cup y_2 &=R\cup\{a,b,c,d\}.
 \label{eq:appB-no-private}
\end{align}
\endgroup\addtocounter{equation}{-1}
The exact recoverability conclusion follows by
Theorem~\ref{thm:endpoint}, not by treating the
conflict score as a numerical paired-risk estimate.

\paragraph{Four fixed illustrative schedules.}
The independent numerical checks use four representative
fixed route pairs, denoted R0--R3 in the theory draft.
The route pairs are:
\begingroup\makeatletter\def\tagform@#1{}\def\theHequation{appBhidden28}\makeatother
\begin{equation}
\renewcommand{\arraystretch}{1.22}
\begin{array}{c|l|l}
 &\text{task 1 actions}&\text{task 2 actions}\\\hline
 \mathrm{R0}&(-b,-a,+d,+c)&(-a,+b,-c,+d)\\
 \mathrm{R1}&(-b,-a,+d,+c)&(-a,-c,+b,+d)\\
 \mathrm{R2}&(-a,-b,+c,+d)&(-a,-c,+d,+b)\\
 \mathrm{R3}&(-a,-b,+c,+d)&(-a,-c,+b,+d)
\end{array}
\label{eq:appB-R0R3-actions}
\end{equation}
\endgroup\addtocounter{equation}{-1}
Using the visited-state catalogs above, their
cross-target intersections are
\begingroup\makeatletter\def\tagform@#1{}\def\theHequation{appBhidden29}\makeatother
\begin{align}
 \mathcal V_1^{\mathrm{R0}}\cap
 \mathcal V_2^{\mathrm{R0}}&=\varnothing,\notag\\
 \mathcal V_1^{\mathrm{R1}}\cap
 \mathcal V_2^{\mathrm{R1}}&=\{\varnothing\},\notag\\
 \mathcal V_1^{\mathrm{R2}}\cap
 \mathcal V_2^{\mathrm{R2}}&=\{\varnothing,\{c\}\},\notag\\
 \mathcal V_1^{\mathrm{R3}}\cap
 \mathcal V_2^{\mathrm{R3}}&=
       \{\varnothing,\{b\},\{c\}\}.
 \label{eq:appB-R0R3-intersections}
\end{align}
\endgroup\addtocounter{equation}{-1}
Since these examples each have two routes to distinct
targets, \eqref{eq:appB-conflict-pairs} becomes simply
$C(\gamma)=|\mathcal V_1\cap\mathcal V_2|$.
Therefore the four conflict scores are, respectively,
\begin{align}
 C(\mathrm{R0})&=0,\notag\\
 C(\mathrm{R1})&=1,\notag\\
 C(\mathrm{R2})&=2,\notag\\
 C(\mathrm{R3})&=3.
 \label{eq:appB-R0R3-scores}
\end{align}
These scores are combinatorial certificates.  Their
clock-dependent paired errors, Monte Carlo checks,
and full exhaustive phase-diagram results belong to
Appendix~\ref{app:numerical-verification}, not to the
present proofs.  The fact that a legal edit path has
a fixed endpoint does not imply that its pooled
Markov projection preserves the prescribed pairing
\citep{debortoli2023augmented,shaul2025flow}.

\subsection{Capacity Obstructions and Recovery Certificates}
\label{app:subsection-07}

\paragraph{A state-cut certificate via counting.}
Recall the directed shortest-route graphs $H_i$
from Section~\ref{sec:theory}.  We use ``cut'' here in
the elementary \emph{hitting-set} sense: a state set
$Z\subseteq\X$ meets a path if that path visits at
least one state of $Z$.  This is the relevant
vertex-disjoint-path obstruction, separate from any
edge cut within an individual input graph \citep{ford1962flows,bangjensen2009digraphs}
\citep{diestel2017graph,botea2021counting,ahuja1993network}.

Let $\mathcal Y$ be the distinct target labels.  For
each $y\in\mathcal Y$, select one task index $i_y$
with $y_{i_y}=y$.  Suppose $Z$ meets \emph{every}
legal monotone shortest route of task $i_y$ for
every distinct target $y$.  Consider any selection
of routes for all tasks, and focus on the selected
representatives $\gamma_{i_y}$.  By the hitting
assumption, for each target $y$ we can choose a
state
\begingroup\makeatletter\def\tagform@#1{}\def\theHequation{appBhidden30}\makeatother
\begin{equation}
 v_y\in \mathcal V_{i_y}\cap Z.
 \label{eq:appB-cut-choice}
\end{equation}
\endgroup\addtocounter{equation}{-1}
If target separation holds, then for $y\ne z$,
\begingroup\makeatletter\def\tagform@#1{}\def\theHequation{appBhidden31}\makeatother
\begin{align}
 \mathcal V_{i_y}\cap\mathcal V_{i_z}
 &=\varnothing,\notag\\
 v_y&\ne v_z.
 \label{eq:appB-cut-injection}
\end{align}
\endgroup\addtocounter{equation}{-1}
The map $y\mapsto v_y$ is thus an injection from
the finite target set into $Z$.  A finite injection
cannot map more distinct inputs than the number of
available outputs \citep{diestel2017graph}.  Therefore
\begingroup\makeatletter\def\tagform@#1{}\def\theHequation{appBhidden32}\makeatother
\begin{equation}
 |\mathcal Y|\le |Z|.
 \label{eq:appB-cut-necessary}
\end{equation}
\endgroup\addtocounter{equation}{-1}
Equivalently, the \emph{strict} inequality
\begin{equation}
 \boxed{|Z|<|\mathcal Y|}
 \label{eq:appB-cut-obstruction}
\end{equation}
certifies that no target-separated shortest-route
family exists.  By Proposition~\ref{prop:integer}
and Theorem~\ref{thm:endpoint}, the binary system
\eqref{eq:flow} is then infeasible and no legal
shortest-route selection can attain the clock floor.
This proof also works when targets repeat because we
selected only one representative per \emph{distinct}
target; replacing $|\mathcal Y|$ by $m$ would generally
be invalid when repeated targets are allowed.

The cut test is a \emph{sufficient obstruction} to
recovery, not a necessary condition for infeasibility:
failure to find a small common hitting set does not
prove that target-separated routes exist \citep{botea2021counting,ahuja1993network}.  Moreover,
a cut certificate proves only structural nonattainability
of the clock floor; it does not give a numerical lower
bound on excess paired TV beyond the strict positivity
implied by Theorem~\ref{thm:endpoint}.  This is why
structural certificates and exact endpoint-risk ODEs
are evaluated separately in the main text.

\paragraph{The unit-capacity bottleneck.}
To connect this obstruction to Theorem~\ref{thm:capacity},
consider its $m\ge2$ equiprobable tasks with pairwise
distinct edges:
\begingroup\makeatletter\def\tagform@#1{}\def\theHequation{appBhidden33}\makeatother
\begin{align}
 s_i&=\{a_i\},\notag\\
 y_i&=\{b_i\},\notag\\
 K_i&=|\{a_i\}|+|\{b_i\}|=2.
 \label{eq:appB-cap-data}
\end{align}
\endgroup\addtocounter{equation}{-1}
At scalar edge capacity one, the legal state family is
$\X_1=\{x:|x|\le1\}$ \citep{madeira2024generative}.  Inserting $b_i$ immediately
would give
\begin{align*}
 |\{a_i\}\cup\{b_i\}|&=2,\\
 2&>1,
\end{align*}
so this first edit is illegal.  Deleting $a_i$ is the
only permissible first mandatory edit, and then
inserting $b_i$ is the only remaining mandatory edit.
The unique shortest route is therefore
\begingroup\makeatletter\def\tagform@#1{}\def\theHequation{appBhidden34}\makeatother
\begin{equation}
 \{a_i\}\longrightarrow\varnothing
          \longrightarrow\{b_i\}.
 \label{eq:appB-cap-forced-path}
\end{equation}
\endgroup\addtocounter{equation}{-1}
Every route visits $\varnothing$.  With the singleton
hitting set $Z=\{\varnothing\}$, we have
\begingroup\makeatletter\def\tagform@#1{}\def\theHequation{appBhidden35}\makeatother
\begin{align}
 |Z|&=1,\notag\\
 |\mathcal Y|&=m,\notag\\
 1&<m \qquad(m\ge2).
 \label{eq:appB-cap-pigeonhole}
\end{align}
\endgroup\addtocounter{equation}{-1}
By \eqref{eq:appB-cut-obstruction}, target-separated
routing is impossible.  In fact the exact conflict
score is also immediate.  The only state shared by
any two distinct routes is $\varnothing$, so
\begin{align}
 C(\gamma)
 &=\sum_{i<h}|\mathcal V_i\cap\mathcal V_h|\notag\\
 &=\sum_{i<h}|\{\varnothing\}|\notag\\
 &=\sum_{i<h}1\notag\\
 &=\frac{m(m-1)}{2}\notag\\
 &=\binom{m}{2}.
 \label{eq:appB-cap-score}
\end{align}

\paragraph{Why pooling at the bottleneck loses the source label.}
We derive the terminal \emph{laws} underlying
\eqref{eq:caplow}--\eqref{eq:caphigh}.  The Poisson
clock is independent of the uniform task index
\citep{kingman1993poisson,bremaud1981point}.
Write $u=u(t)>0$.  The teacher is at the empty state
for task $i$ exactly when one clock event has occurred:
\begingroup\makeatletter\def\tagform@#1{}\def\theHequation{appBhidden36}\makeatother
\begin{align}
 \PP\{S=s_i,X_t^\gamma=\varnothing\}
 &=\PP\{S=s_i,N_t=1\}\notag\\
 &=\PP\{S=s_i\}\PP\{N_t=1\}\notag\\
 &=\frac1m\,e^{-u}u.
 \label{eq:appB-cap-empty-weight}
\end{align}
\endgroup\addtocounter{equation}{-1}
At that state, the teacher's next action is $+b_i$.
Conditioning on the occupied graph gives
\begin{align}
 q_t(+b_j\mid\varnothing)
 &=\frac{\PP\{X_t^\gamma=\varnothing,
                    A_t^\gamma=+b_j\}}
         {\PP\{X_t^\gamma=\varnothing\}}\notag\\
 &=\frac{\frac1m e^{-u}u}
         {\sum_{i=1}^m\frac1m e^{-u}u}\notag\\
 &=\frac{\frac1m e^{-u}u}{e^{-u}u}\notag\\
 &=\frac1m.
 \label{eq:appB-cap-uniform-kernel}
\end{align}
This shows explicitly which information is lost:
the source-specific mark $+b_i$ becomes a uniformly
chosen $+b_j$ at the common state $\varnothing$.
The continuation is Markovian and source-blind, as in
\eqref{eq:q} \citep{norris1997markov,anderson1991continuous}.

Let
\begingroup\makeatletter\def\tagform@#1{}\def\theHequation{appBhidden37}\makeatother
\begin{equation}
 \psi=\PP\{\Pois(B)\ge2\}.
 \label{eq:appB-psi-def}
\end{equation}
\endgroup\addtocounter{equation}{-1}
The first two Poisson masses \citep{kingman1993poisson,bremaud1981point} give
\begingroup\makeatletter\def\tagform@#1{}\def\theHequation{appBhidden38}\makeatother
\begin{align}
 \PP\{\Pois(B)=0\}&=e^{-B},\notag\\
 \PP\{\Pois(B)=1\}&=Be^{-B},\notag\\
 \PP\{\Pois(B)<2\}&=e^{-B}+Be^{-B},\notag\\
 \psi&=1-e^{-B}-Be^{-B}\notag\\
 &=1-e^{-B}(1+B).
 \label{eq:appB-psi-expand}
\end{align}
\endgroup\addtocounter{equation}{-1}
For a given source $s_i$, after zero events the process
is at $\{a_i\}$, after exactly one it is at
$\varnothing$, and after at least two it is at a
uniformly chosen target $\{b_j\}$.  The conditional
terminal law is therefore \citep{norris1997markov,kingman1993poisson}
\begin{align}
 \calL(X_T^{\mathrm{proj}}\mid S=s_i)
 &=e^{-B}\delta_{\{a_i\}}
   +Be^{-B}\delta_{\varnothing}\notag\\
 &\quad+\psi\frac1m\sum_{j=1}^m\delta_{\{b_j\}}.
 \label{eq:appB-cap-conditional-law}
\end{align}
The only term in this law assigning mass to the
prescribed target $y_i=\{b_i\}$ is $j=i$ in the final
sum.  Hence
\begingroup\makeatletter\def\tagform@#1{}\def\theHequation{appBhidden39}\makeatother
\begin{align}
 \PP\{X_T^{\mathrm{proj}}=y_i\mid S=s_i\}
 &=\psi\cdot\frac1m,\notag\\
 \PP\{X_T^{\mathrm{proj}}\ne y_i\mid S=s_i\}
 &=1-\frac\psi m.
 \label{eq:appB-cap-paired-conditional}
\end{align}
\endgroup\addtocounter{equation}{-1}
By the paired-TV identity \eqref{eq:lb}, averaging
these equal conditional errors with weights $1/m$
gives
\begin{align}
 D_{\mathrm{pair}}^{\mathrm{proj}}
 &=\sum_{i=1}^m\frac1m\left(1-\frac\psi m\right)\notag\\
 &=\left(\sum_{i=1}^m\frac1m\right)
              \left(1-\frac\psi m\right)\notag\\
 &=1-\frac\psi m.
 \label{eq:appB-cap-paired-check}
\end{align}

For the marginal calculation, average
\eqref{eq:appB-cap-conditional-law} over the source:
\begingroup\makeatletter\def\tagform@#1{}\def\theHequation{appBhidden40}\makeatother
\begin{align}
 \calL(X_T^{\mathrm{proj}})
 &=e^{-B}\frac1m\sum_i\delta_{\{a_i\}}
   +Be^{-B}\delta_{\varnothing}\notag\\
 &\quad+\psi\frac1m\sum_j\delta_{\{b_j\}}.
 \label{eq:appB-cap-marginal-law}
\end{align}
\endgroup\addtocounter{equation}{-1}
All $a_i$ and $b_j$ are distinct, and
$\varnothing$ is distinct from every singleton.
The target marginal is
$\calL(Y)=\frac1m\sum_j\delta_{\{b_j\}}$.
Using the finite-state definition of TV distance \citep{cover2006elements,csiszar2011information},
separating the disjoint source, empty, and target
supports gives
\begin{align}
 2D_{\mathrm{marg}}
 &=e^{-B}+Be^{-B}
   +\sum_{j=1}^m\left|\frac\psi m-\frac1m\right|\notag\\
 &=e^{-B}+Be^{-B}
   +\sum_{j=1}^m\frac{1-\psi}{m}\notag\\
 &=e^{-B}(1+B)+(1-\psi)\notag\\
 &=(1-\psi)+(1-\psi)\notag\\
 &=2(1-\psi),\notag\\
 D_{\mathrm{marg}}&=1-\psi.
 \label{eq:appB-cap-marginal-check}
\end{align}
The resulting pair and marginal formulas agree with
\eqref{eq:caplow}.  They expose the mechanism: the target marginal is correct \emph{on
completion}, while the prescribed source--target
association is lost with probability $1-1/m$ given
completion.  This is exactly why marginal matching
alone cannot certify coupling recovery
\citep{lindvall1992lectures,thorisson2000coupling}.

\paragraph{What extra capacity makes possible.}
At scalar capacity $C\ge2$, the add-first routes \citep{madeira2024generative}
\begingroup\makeatletter\def\tagform@#1{}\def\theHequation{appBhidden41}\makeatother
\begin{equation}
 \{a_i\}\longrightarrow\{a_i,b_i\}
             \longrightarrow\{b_i\}
 \label{eq:appB-cap-add-first}
\end{equation}
\endgroup\addtocounter{equation}{-1}
are legal.  Their three states are different across
all distinct indices $i$, because the $2m$ edges are
pairwise distinct.  Therefore for $i\ne h$,
\begingroup\makeatletter\def\tagform@#1{}\def\theHequation{appBhidden42}\makeatother
\begin{align}
 \{a_i\}&\ne\{a_h\},\notag\\
 \{a_i,b_i\}&\ne\{a_h,b_h\},\notag\\
 \{b_i\}&\ne\{b_h\},\notag\\
 \mathcal V_i\cap\mathcal V_h&=\varnothing.
 \label{eq:appB-cap-add-separated}
\end{align}
\endgroup\addtocounter{equation}{-1}
By Theorem~\ref{thm:endpoint}, each source reaches its
own target on every clock trajectory with at least
two events.  Consequently
\begin{align}
 \PP\{X_T^{\mathrm{proj}}=y_i\mid S=s_i\}
 &=\PP\{N_T\ge2\}\notag\\
 &=\psi,\notag\\
 D_{\mathrm{pair}}^{\mathrm{proj}}
 &=\sum_i\frac1m(1-\psi)\notag\\
 &=1-\psi.
 \label{eq:appB-cap-add-paired}
\end{align}
The terminal marginal again assigns $\psi/m$ to
each distinct target \citep{cover2006elements} and the remaining mass $1-\psi$
to states outside the target support.  Applying the
same two-support TV calculation as in
\eqref{eq:appB-cap-marginal-check} yields
\begingroup\makeatletter\def\tagform@#1{}\def\theHequation{appBhidden43}\makeatother
\begin{equation}
 D_{\mathrm{marg}}=1-\psi.
 \label{eq:appB-cap-add-marginal}
\end{equation}
\endgroup\addtocounter{equation}{-1}
Delete-first routes remain legal when capacity is
increased from one to two, and their pooled action law
is unchanged because only the route-occupied states
enter \eqref{eq:q}.  Thus at the \emph{same} fixed
capacity $C\ge2$, compare two legal orders of the
same mandatory edits:
\begin{align}
 D_{\mathrm{pair}}^{\mathrm{delete}}
 &=1-\frac\psi m,\notag\\
 D_{\mathrm{pair}}^{\mathrm{add}}
 &=1-\psi,\notag\\
 D_{\mathrm{pair}}^{\mathrm{delete}}
       -D_{\mathrm{pair}}^{\mathrm{add}}
 &=\left(1-\frac\psi m\right)-(1-\psi)\notag\\
 &=\psi-\frac\psi m\notag\\
 &=\boxed{\left(1-\frac1m\right)\psi}.
 \label{eq:appB-cap-gap-check}
\end{align}
Here $\psi>0$ for $B>0$, and $1-1/m>0$ for
$m\ge2$, so the difference is strictly positive.
The calculations give the capacity theorem's
\emph{order-only} comparison without changing endpoints,
edit length, clock, or legal-state family.  They also
show why the singleton state cut obstructs clock-floor
recovery at capacity one but does not obstruct the
add-first routes when capacity is at least two.

\paragraph{Three different kinds of certificate.}
The theory distinguishes the following logical claims \citep{ahuja1993network,botea2021counting,ford1962flows}.
First, the existence of a target-separated route family
is certified by a feasible binary solution to
\eqref{eq:flow} or by exhibiting the routes directly.
Second, nonexistence can be certified by a suitable
hitting set satisfying \eqref{eq:appB-cut-obstruction},
although not every infeasible instance need admit such
a small common cut.  Third, the scheduler's returned
score $C(\gamma)=0$ is a certificate for the \emph{given}
selected routes, while $C(\gamma)>0$ certifies only
that those particular routes still intersect.  Numerical
comparison of an exact solver and local scheduling,
and phase-diagram tests over constrained families,
are kept in Appendix~\ref{app:numerical-verification}.
None of these combinatorial certificates is a direct
replacement for the clock-dependent paired error
formula \eqref{eq:exactrisk} when the conflict score
is positive.

\clearpage

\section{Learned-Policy Error Analysis}
\label{app:learning}

Appendix~\ref{app:certificates} establishes when legal shortest
routes can be selected and scheduled to satisfy the target-separation
criterion of Theorem~\ref{thm:endpoint}.  In particular, the
zero-conflict certificate $C(\gamma)=0$ in
\eqref{eq:collision} concerns the \emph{ideal} source-blind
Markov projection, not the accuracy of an independently trained
policy.  We now connect those structural guarantees to the
population action law learned by a graph-only student.  The
argument separates two distinct sources of error: the information
lost when target-conditioned route actions are pooled and the
additional discrepancy between the pooled law and its learned
approximation \citep{debortoli2023augmented,holderrieth2025generator}.  This supplies the detailed proof of the paired and
marginal guarantees in Theorem~\ref{thm:risk} and the exact risk
evaluation formulas \eqref{eq:weights}--\eqref{eq:exactrisk}.

Throughout this appendix, $\pi$, $p_i$, $(s_i,y_i)$, $K_i$,
$N_t$, $b(t)$, $u(t)$, and $B$ are exactly those of
Section~\ref{sec:setting}.  The teacher's prescribed route is
$\gamma_i$, its action at position $j$ is $a_{i,j}$, and
$a_{i,K_i}=\bot$.  The common Poisson clock is independent of
$(S,Y)$, and every compared model makes at most one legal edge
toggle per clock event.  The graph-only kernel observes the
current graph and time but not $S$, $Y$, the selected route, or
the edit history.  All distributions and divergences in the
proof are \emph{population} quantities for a fixed comparable
pairing law; no generalization inequality is inferred merely from
finite-sample cross-entropy training
\citep{norris1997markov,anderson1991continuous,
        debortoli2023augmented,holderrieth2025generator}.

\subsection{Target-Aware Coupling and Structural Ambiguity}
\label{app:subsection-08}

\paragraph{Conditioning first on the intended target.}
Write $\mathcal Y=\{y_i:1\le i\le m\}$ for the distinct
targets, and define their probabilities by
\begingroup\makeatletter\def\tagform@#1{}\def\theHequation{appChidden1}\makeatother
\begin{equation}
 \nu(y)=\PP(Y=y)=\sum_{i:y_i=y}p_i.
 \label{eq:appC-target-weight}
\end{equation}
\endgroup\addtocounter{equation}{-1}

Every $\nu(y)$ is strictly positive for $y\in\mathcal Y$
because each $p_i$ is positive \citep{billingsley1995probability}.  For every $t$, let
\begingroup\makeatletter\def\tagform@#1{}\def\theHequation{appChidden2}\makeatother
\begin{align}
 p_t^y(x)
 &=\PP(X_t^\gamma=x\mid Y=y),\notag\\
 q_t^y(a\mid x)
 &=\PP(A_t^\gamma=a\mid X_t^\gamma=x,Y=y).
 \label{eq:appC-target-defs}
\end{align}
\endgroup\addtocounter{equation}{-1}

Here $q_t^y$ is defined at states with
$p_t^y(x)>0$; values assigned at other states never enter
an expectation under $p_t^y$ \citep{billingsley1995probability}.  Conditioning on the finite
partition $\{S=s_i:y_i=y\}$ gives \citep{billingsley1995probability}
\begin{align}
 \nu(y)p_t^y(x)
 &=\PP(Y=y,X_t^\gamma=x)\notag\\
 &=\sum_{i:y_i=y}\PP(S=s_i,X_t^\gamma=x)\notag\\
 &=\sum_{i:y_i=y}p_i
   \PP(X_t^\gamma=x\mid S=s_i).
 \label{eq:appC-target-occupancy-1}
\end{align}
Let $w_{ij}(u)$ be the Poisson route-position probability
in \eqref{eq:weights}, with $u=u(t)$.  Since the route
teacher is at $x_{i,j}$ when its truncated event count is
$j$, including the absorbing endpoint position $j=K_i$,
\begingroup\makeatletter\def\tagform@#1{}\def\theHequation{appChidden3}\makeatother
\begin{equation}
 \PP(X_t^\gamma=x\mid S=s_i)
 =\sum_{j=0}^{K_i}w_{ij}(u(t))
   \mathbf1\{x_{i,j}=x\}.
 \label{eq:appC-target-occupancy-2}
\end{equation}
\endgroup\addtocounter{equation}{-1}

Substitution into \eqref{eq:appC-target-occupancy-1}
yields the entirely explicit expression
\begin{equation}
 \nu(y)p_t^y(x)
 =\sum_{i:y_i=y}p_i\sum_{j=0}^{K_i}
    w_{ij}(u(t))\mathbf1\{x_{i,j}=x\}.
 \label{eq:appC-target-occupancy}
\end{equation}

An identical expansion of the joint event involving the
action mark gives \citep{billingsley1995probability}
\begin{align}
 &\nu(y)p_t^y(x)q_t^y(a\mid x)\notag\\
 &=\PP(Y=y,X_t^\gamma=x,A_t^\gamma=a)\notag\\
 &=\sum_{i:y_i=y}p_i\sum_{j=0}^{K_i}
  w_{ij}(u(t))
  \mathbf1\{x_{i,j}=x,a_{i,j}=a\}.
 \label{eq:appC-target-flux}
\end{align}
Thus $q_t^y$ averages route actions associated with the
\emph{same} target only.  The pooled source-blind kernel
\eqref{eq:q} makes the additional average over $y$:
\begingroup\makeatletter\def\tagform@#1{}\def\theHequation{appChidden4}\makeatother
\begin{align}
 p_t(x)q_t(a\mid x)
 &=\PP(X_t^\gamma=x,A_t^\gamma=a)\notag\\
 &=\sum_{y\in\mathcal Y}
    \PP(Y=y,X_t^\gamma=x,A_t^\gamma=a)\notag\\
 &=\sum_{y\in\mathcal Y}
    \nu(y)p_t^y(x)q_t^y(a\mid x),
 \label{eq:appC-pooling-flux}
\end{align}
\endgroup\addtocounter{equation}{-1}

where $p_t(x)=\sum_y\nu(y)p_t^y(x)$ \citep{cover2006elements,murphy2022probabilistic}.
These are conditional-mixture identities;
their relationship to Markov projection and conditional
flow matching is discussed in
\citet{debortoli2023augmented,zhang2024trajectory,
        shaul2025flow,gat2024discrete}.

\paragraph{A target-aware comparison process.}
Define $\bar X_t$ by drawing the original
$(S,Y)\sim\pi$ and the same exogenous clock, setting
$\bar X_0=S$, and sampling at each event from
$q_t^Y(\cdot\mid\bar X_{t-})$.  This process is a
mathematical comparison; the learned student is \emph{not}
allowed to observe $Y$.  By the target-conditioned
forward equation proved in \citep{ethier1986markov,anderson1991continuous}
Appendix~\ref{app:subsection-02}, the joint one-time
laws of $(Y,\bar X_t)$ equal those of
$(Y,X_t^\gamma)$.  Explicitly, for every $y,x,t$,
\begin{equation}
 \PP(Y=y,\bar X_t=x)=\nu(y)p_t^y(x).
 \label{eq:appC-target-aware-marginal}
\end{equation}
This is a statement about $(Y,X_t)$; it does \emph{not}
claim that the target-aware process reproduces each
source-conditioned teacher trajectory \citep{debortoli2023augmented,zhang2024trajectory}.

We next show that the comparison process reaches each
prescribed endpoint at the earliest possible event count.  Fix a pair $i$, condition on
$(S,Y)=(s_i,y_i)$, and define the remaining edit distance
\begingroup\makeatletter\def\tagform@#1{}\def\theHequation{appChidden5}\makeatother
\begin{equation}
 d_i(x)=|x\triangle y_i|.
 \label{eq:appC-target-distance}
\end{equation}
\endgroup\addtocounter{equation}{-1}

Initially, by the mandatory-edit definition,
\begingroup\makeatletter\def\tagform@#1{}\def\theHequation{appChidden6}\makeatother
\begin{equation}
 d_i(s_i)=|s_i\triangle y_i|=K_i.
 \label{eq:appC-initial-distance}
\end{equation}
\endgroup\addtocounter{equation}{-1}

Suppose $d_i(x)>0$ and
$q_t^{y_i}(a\mid x)>0$ at a reachable state $x$.
By \eqref{eq:appC-target-flux}, some teacher route to
$y_i$ visits $x$ and prescribes $a$ there.  Every such
route is monotone shortest, so its next non-null edit
removes exactly one edge of the symmetric difference
with $y_i$ \citep{diestel2017graph,korte2018combinatorial}.  Consequently,
\begingroup\makeatletter\def\tagform@#1{}\def\theHequation{appChidden7}\makeatother
\begin{align}
 d_i(T_a(x))
 &=|T_a(x)\triangle y_i|\notag\\
 &=|x\triangle y_i|-1\notag\\
 &=d_i(x)-1.
 \label{eq:appC-one-step-decrease}
\end{align}
\endgroup\addtocounter{equation}{-1}

The successor $T_a(x)$ is legal and itself lies on a
teacher route to $y_i$, so the same argument applies
at the next event \citep{diestel2017graph}.  If $d_i(x)=0$, then $x=y_i$ \citep{diestel2017graph}.
Every monotone shortest teacher route to $y_i$ visiting
$y_i$ has already finished, and its action is $\bot$.
Therefore
\begingroup\makeatletter\def\tagform@#1{}\def\theHequation{appChidden8}\makeatother
\begin{equation}
 q_t^{y_i}(\bot\mid y_i)=1,
 \qquad T_\bot(y_i)=y_i.
 \label{eq:appC-target-absorbing}
\end{equation}
\endgroup\addtocounter{equation}{-1}

Induction over the number $n$ of clock events now yields
\begingroup\makeatletter\def\tagform@#1{}\def\theHequation{appChidden9}\makeatother
\begin{align}
 d_i(\bar X_0)&=K_i,\notag\\
 d_i(\bar X_{\tau_1})&=(K_i-1)_+,\notag\\
 d_i(\bar X_{\tau_2})&=(K_i-2)_+,\notag\\
 &\ \vdots\notag\\
 d_i(\bar X_{\tau_n})&=(K_i-n)_+,
 \label{eq:appC-distance-induction}
\end{align}
\endgroup\addtocounter{equation}{-1}

where $r_+=\max\{r,0\}$ and $\tau_n$ denotes the
$n$th event time, whenever that event occurs by $T$ \citep{kingman1993poisson}.
The same expression covers $K_i=0$:
then the process starts at its target and uses only
null actions.  As $d_i(x)=0$ exactly when $x=y_i$,
\begingroup\makeatletter\def\tagform@#1{}\def\theHequation{appChidden10}\makeatother
\begin{align}
 \{\bar X_T=y_i\}
 &=\{d_i(\bar X_T)=0\}\notag\\
 &=\{(K_i-N_T)_+=0\}\notag\\
 &=\{N_T\ge K_i\}.
 \label{eq:appC-target-success-event}
\end{align}
\endgroup\addtocounter{equation}{-1}

Independence of the clock and the pair implies
\begingroup\makeatletter\def\tagform@#1{}\def\theHequation{appChidden11}\makeatother
\begin{align}
 \PP(\bar X_T\ne y_i\mid S=s_i)
 &=\PP(N_T<K_i\mid S=s_i)\notag\\
 &=\PP(N_T<K_i)\notag\\
 &=\PP\{\Pois(B)<K_i\}.
 \label{eq:appC-target-error-per-source}
\end{align}
\endgroup\addtocounter{equation}{-1}

Using the deterministic-pairing TV identity in
\eqref{eq:lb} and summing the conditional errors,
\begin{align}
 D_{\mathrm{pair}}^{\mathrm{aware}}
 &=\sum_i p_i
   \PP(\bar X_T\ne y_i\mid S=s_i)\notag\\
 &=\sum_i p_i\PP\{\Pois(B)<K_i\}\notag\\
 &=\varepsilon_B.
 \label{eq:appC-target-aware-floor}
\end{align}
Thus target awareness exactly attains the unavoidable
clock floor, even when several source paths to the same
target merge and prescribe different intermediate edits.

\paragraph{A maximal coupling of action marks.}
We compare $\bar X$ with the source-blind Markov projection
$X^{\mathrm{proj}}$, which uses $q_t$ from \eqref{eq:q} \citep{lindvall1992lectures,thorisson2000coupling}.
Both processes use the \emph{same} initial pair, event times,
state representation, and action alphabet (including $\bot$).
As long as the two graph states agree at $x$, couple their
next action marks by a maximal coupling of the categorical
laws $q_t^Y(\cdot\mid x)$ and $q_t(\cdot\mid x)$
\citep{lindvall1992lectures,thorisson2000coupling}.

For distributions $P$ and $Q$ on a finite
action set,
\begingroup\makeatletter\def\tagform@#1{}\def\theHequation{appChidden12}\makeatother
\begin{align}
 \|P-Q\|_{\TV}
 &=\frac12\sum_a|P(a)-Q(a)|\notag\\
 &=\frac12\sum_a
    [P(a)+Q(a)-2\min\{P(a),Q(a)\}]\notag\\
 &=1-\sum_a\min\{P(a),Q(a)\}.
 \label{eq:appC-max-coupling-identity}
\end{align}
\endgroup\addtocounter{equation}{-1}

A maximal coupling first chooses an identical action
with total probability $\sum_a\min\{P(a),Q(a)\}$,
and allocates its remaining probability mass to the
nonoverlapping residual distributions \citep{lindvall1992lectures,thorisson2000coupling}.  Hence the
probability of different marks at a common state is
exactly $\|P-Q\|_{\TV}$.  This construction gives
\begingroup\makeatletter\def\tagform@#1{}\def\theHequation{appChidden13}\makeatother
\begin{equation}
 h_t(y,x)=\|q_t^y(\cdot\mid x)
                   -q_t(\cdot\mid x)\|_{\TV}
 \label{eq:appC-local-alias}
\end{equation}
\endgroup\addtocounter{equation}{-1}

for the conditional disagreement probability at time
$t$, target $y$, and common state $x$ \citep{lindvall1992lectures}.

Let $\tau$ be the time of the first differing marks
(and set $\tau=\infty$ if none occurs).  The two graph
states agree strictly before $\tau$.  Given the history
immediately before an event at time $t<\tau$, the chance
that this event causes the first disagreement is
$h_t(Y,\bar X_{t-})$.  Since clock events have intensity
$b(t)$ and their occurrence is independent of the
current mark draw, the conditional expected number of
\emph{first} disagreements in $[t,t+\dd t]$ is
\begingroup\makeatletter\def\tagform@#1{}\def\theHequation{appChidden14}\makeatother
\begin{equation}
 b(t)\mathbf1\{t<\tau\}
 h_t(Y,\bar X_{t-})\,\dd t.
 \label{eq:appC-first-disagreement-rate}
\end{equation}
\endgroup\addtocounter{equation}{-1}

There can be at most one first disagreement \citep{jacod2003limit,bremaud1981point}.  Taking
expectations and integrating the compensator of the
common marked Poisson clock
\citep{kingman1993poisson,bremaud1981point} yields
\begin{align}
 \PP(\tau\le T)
 &=\E\int_0^T b(t)\mathbf1\{t<\tau\}
      h_t(Y,\bar X_{t-})\,\dd t\notag\\
 &\le\E\int_0^T b(t)
      h_t(Y,\bar X_{t-})\,\dd t\notag\\
 &=\int_0^T b(t)
      \E_{Y,\bar X_t}h_t(Y,\bar X_t)\,\dd t\notag\\
 &=\int_0^T b(t)
      \E_{Y,X_t^\gamma}h_t(Y,X_t^\gamma)\,\dd t\notag\\
 &=\mathfrak A_B^Y.
 \label{eq:appC-first-disagreement-bound}
\end{align}
In the third line, left limits can be replaced by
fixed-time values inside a Lebesgue integral because
a Poisson process has no event at a prespecified time
with positive probability.  The fourth line uses the
$(Y,X_t)$ marginal equality
\eqref{eq:appC-target-aware-marginal}; the last line
is precisely the definition \eqref{eq:qy}.

On $\{\tau>T\}$, the processes have taken identical
edits at all clock events and therefore have identical
terminal graphs.  The elementary coupling inequality
then gives \citep{lindvall1992lectures,thorisson2000coupling}
\begingroup\makeatletter\def\tagform@#1{}\def\theHequation{appChidden15}\makeatother
\begin{align}
 &\bigl\|\calL(S,\bar X_T)
       -\calL(S,X_T^{\mathrm{proj}})\bigr\|_{\TV}
       \notag\\
 &\qquad\le\PP(\bar X_T\ne X_T^{\mathrm{proj}})
       \notag\\
 &\qquad\le\PP(\tau\le T)
       \notag\\
 &\qquad\le\mathfrak A_B^Y.
 \label{eq:appC-paired-coupling-TV}
\end{align}
\endgroup\addtocounter{equation}{-1}

Applying the TV triangle inequality with the prescribed
law $\pi$ and then \citep{gray2011entropy,cover2006elements} \eqref{eq:appC-target-aware-floor},
\begin{align}
 D_{\mathrm{pair}}^{\mathrm{proj}}
 &=\|\calL(S,X_T^{\mathrm{proj}})-\pi\|_{\TV}
       \notag\\
 &\le\|\calL(S,X_T^{\mathrm{proj}})
              -\calL(S,\bar X_T)\|_{\TV}
        +\|\calL(S,\bar X_T)-\pi\|_{\TV}
       \notag\\
 &\le\mathfrak A_B^Y+\varepsilon_B.
 \label{eq:appC-projection-alias-risk}
\end{align}

This is the structural part of the upper bound in
Theorem~\ref{thm:risk}.  It is an \emph{upper bound}
on the loss caused by pooling targets, not a claim that
every disagreement in teacher actions creates an endpoint
error.  In particular, when the routes are target-separated
(or $C(\gamma)=0$ in \eqref{eq:collision}),
$q_t^Y=q_t$ on all relevant states, so
$\mathfrak A_B^Y=0$ and Theorem~\ref{thm:endpoint}
recovers equality in the clock-floor bound.

\subsection{Source-Action versus Target-Action Ambiguity}
\label{app:subsection-09}

\paragraph{Source-conditioned actions are deterministic.}
Fix an interior time $t\in(0,T)$ and a graph $x$ with
$p_t(x)>0$.  Because a legal monotone shortest route
strictly decreases its distance to its target, it cannot
visit the same graph twice \citep{diestel2017graph,korte2018combinatorial}.  Therefore, once both the
source $S=s_i$ and graph $x$ are specified, there is at
most one route position $j$ with $x_{i,j}=x$.  At such
a position the next mark is the prescribed action
$a_{i,j}$ (or $\bot$ if $j=K_i$).  Thus
\begin{equation}
 q_t^{s_i}(a\mid x)
 =\mathbf1\{a=a_{i,j}\}
 \quad\text{when }x=x_{i,j}.
 \label{eq:appC-source-point-mass}
\end{equation}
This deterministic conditional action law does not
imply that $q_t(a\mid x)$ is deterministic: several
sources can assign different actions to the same state.

Define the posterior probability of source $i$ at
$(t,x)$ by \citep{billingsley1995probability}
\begingroup\makeatletter\def\tagform@#1{}\def\theHequation{appChidden16}\makeatother
\begin{align}
 \alpha_i(t,x)
 &=\PP(S=s_i\mid X_t^\gamma=x)\notag\\
 &=\frac{\PP(S=s_i,X_t^\gamma=x)}{p_t(x)}
       \notag\\
 &=\frac{p_i\sum_{j=0}^{K_i}
          w_{ij}(u(t))\mathbf1\{x_{i,j}=x\}}
          {p_t(x)}.
 \label{eq:appC-source-posterior}
\end{align}
\endgroup\addtocounter{equation}{-1}

Every $\alpha_i$ is nonnegative, and
$\sum_i\alpha_i=1$.  If $a_i(x)$ denotes the action
of source $i$ at its unique visit to $x$, then
\begin{align}
 q_t(a\mid x)
 &=\sum_i\PP(S=s_i\mid X_t^\gamma=x)
              q_t^{s_i}(a\mid x)\notag\\
 &=\sum_i\alpha_i(t,x)
              \mathbf1\{a_i(x)=a\}.
 \label{eq:appC-action-source-mixture}
\end{align}
Sources not visiting $x$ have $\alpha_i=0$ and do
not contribute to this expression \citep{cover2006elements}.

\paragraph{Exact source-action Gini identity.}
For a point mass $\delta_c$ and an arbitrary action
law $Q$, direct evaluation of TV gives
\begin{align}
 \|\delta_c-Q\|_{\TV}
 &=\frac12\sum_a|\mathbf1\{a=c\}-Q(a)|
       \notag\\
 &=\frac12\bigl[|1-Q(c)|+
                    \sum_{a\ne c}Q(a)\bigr]\notag\\
 &=\frac12\bigl[(1-Q(c))+(1-Q(c))\bigr]
       \notag\\
 &=1-Q(c).
 \label{eq:appC-dirac-TV}
\end{align}
Applying this to \eqref{eq:appC-source-point-mass},
then averaging over the source posterior and using
\eqref{eq:appC-action-source-mixture},
\begin{align}
 &\E_{S\mid X_t^\gamma=x}
 \|q_t^S(\cdot\mid x)-q_t(\cdot\mid x)\|_{\TV}
       \notag\\
 &=\sum_i\alpha_i(t,x)
       [1-q_t(a_i(x)\mid x)]\notag\\
 &=\sum_i\alpha_i(t,x)
    -\sum_i\alpha_i(t,x)q_t(a_i(x)\mid x)\notag\\
 &=1-\sum_a q_t(a\mid x)
      \sum_{i:a_i(x)=a}\alpha_i(t,x)\notag\\
 &=1-\sum_a q_t(a\mid x)^2.
 \label{eq:appC-source-gini}
\end{align}
The last line is a finite-distribution identity.  It measures disagreement with
\emph{source-specific} action prescriptions and is also
the Gini impurity of the pooled categorical action
law \citep{murphy2022probabilistic,cover2006elements}.

\paragraph{Conditioning on target groups reduces ambiguity.}
At the same occupied $(t,x)$, define
\begingroup\makeatletter\def\tagform@#1{}\def\theHequation{appChidden17}\makeatother
\begin{equation}
 \beta_y(t,x)=\PP(Y=y\mid X_t^\gamma=x)
 =\sum_{i:y_i=y}\alpha_i(t,x).
 \label{eq:appC-target-posterior}
\end{equation}
\endgroup\addtocounter{equation}{-1}

The target-conditioned mark law is the mixture \citep{billingsley1995probability,murphy2022probabilistic}
\begin{equation}
 q_t^y(\cdot\mid x)
 =\sum_{i:y_i=y}
    \frac{\alpha_i(t,x)}{\beta_y(t,x)}
    q_t^{s_i}(\cdot\mid x)
 \label{eq:appC-target-source-mixture}
\end{equation}
when $\beta_y(t,x)>0$.  Equivalently,
\begingroup\makeatletter\def\tagform@#1{}\def\theHequation{appChidden18}\makeatother
\begin{equation}
 q_t(\cdot\mid x)=\sum_y
      \beta_y(t,x)q_t^y(\cdot\mid x).
 \label{eq:appC-target-action-mixture}
\end{equation}
\endgroup\addtocounter{equation}{-1}

The \emph{local} target-action ambiguity is
\begingroup\makeatletter\def\tagform@#1{}\def\theHequation{appChidden19}\makeatother
\begin{equation}
 G_Y(t,x)=\sum_y\beta_y(t,x)
 \|q_t^y(\cdot\mid x)-q_t(\cdot\mid x)\|_{\TV}.
 \label{eq:appC-local-target-ambiguity}
\end{equation}
\endgroup\addtocounter{equation}{-1}

Its clock-weighted population average equals
\begingroup\makeatletter\def\tagform@#1{}\def\theHequation{appChidden20}\makeatother
\begin{align}
 \mathfrak A_B^Y
 &=\int_0^T b(t)\E_{Y,X_t^\gamma}
   \|q_t^Y-q_t\|_{\TV}\,\dd t\notag\\
 &=\int_0^T b(t)\sum_xp_t(x)
    \sum_y\beta_y(t,x)
       \|q_t^y-q_t\|_{\TV}\,\dd t\notag\\
 &=\int_0^T b(t)\sum_xp_t(x)G_Y(t,x)\,\dd t.
 \label{eq:appC-target-integrated}
\end{align}
\endgroup\addtocounter{equation}{-1}

Arguments $(\cdot\mid x)$ have been suppressed only
inside the norms in this display.

The triangle inequality for the $L^1$ norm and the
convex-mixture representation \citep{gray2011entropy,cover2006elements}
\eqref{eq:appC-target-source-mixture} imply
\begingroup\makeatletter\def\tagform@#1{}\def\theHequation{appChidden21}\makeatother
\begin{align}
 &\|q_t^y(\cdot\mid x)-q_t(\cdot\mid x)\|_{\TV}
       \notag\\
 &=\left\|
    \sum_{i:y_i=y}\frac{\alpha_i}{\beta_y}
          (q_t^{s_i}-q_t)
   \right\|_{\TV}\notag\\
 &\le\sum_{i:y_i=y}\frac{\alpha_i}{\beta_y}
      \|q_t^{s_i}-q_t\|_{\TV}.
 \label{eq:appC-conditional-TV-convexity}
\end{align}
\endgroup\addtocounter{equation}{-1}

Multiplying by $\beta_y$, summing over targets,
and then applying \eqref{eq:appC-source-gini} gives
\begin{align}
 G_Y(t,x)
 &\le\sum_y\sum_{i:y_i=y}\alpha_i
       \|q_t^{s_i}-q_t\|_{\TV}\notag\\
 &=\sum_i\alpha_i
       \|q_t^{s_i}-q_t\|_{\TV}\notag\\
 &=1-\sum_a q_t(a\mid x)^2.
 \label{eq:appC-target-less-source}
\end{align}
For comparison, define the clock-integrated source
quantity
\begingroup\makeatletter\def\tagform@#1{}\def\theHequation{appChidden22}\makeatother
\begin{equation}
 \mathfrak A_B^S
 =\int_0^T b(t)\E_{S,X_t^\gamma}
    \|q_t^S-q_t\|_{\TV}\,\dd t.
 \label{eq:appC-source-integrated}
\end{equation}
\endgroup\addtocounter{equation}{-1}

Integrating \eqref{eq:appC-target-less-source} against
$b(t)p_t(x)$ gives the exact inequality
\begin{equation}
 \boxed{0\le\mathfrak A_B^Y\le\mathfrak A_B^S}.
 \label{eq:appC-source-target-comparison}
\end{equation}

Equality need not hold.  If two sources
have the same target and prescribe different legal
shortest next actions at a shared graph, the pooled law
may be nondeterministic even though target conditioning
adds no information.  For instance, the conditional
categorical configuration
\begingroup\makeatletter\def\tagform@#1{}\def\theHequation{appChidden23}\makeatother
\begin{equation}
 q_t^{s_1}=\delta_a,\qquad
 q_t^{s_2}=\delta_b,\qquad
 \alpha_1=\alpha_2=\tfrac12,\quad y_1=y_2
 \label{eq:appC-same-target-example}
\end{equation}
\endgroup\addtocounter{equation}{-1}

with $a\ne b$ has
\begingroup\makeatletter\def\tagform@#1{}\def\theHequation{appChidden24}\makeatother
\begin{align}
 q_t&=\tfrac12\delta_a+\tfrac12\delta_b,
       \notag\\
 1-\sum_cq_t(c)^2
 &=1-\left(\tfrac14+\tfrac14\right)
       =\tfrac12,\notag\\
 G_Y(t,x)&=\|q_t^{y_1}-q_t\|_{\TV}=0.
 \label{eq:appC-same-target-discrepancy}
\end{align}
\endgroup\addtocounter{equation}{-1}

Both distinct actions may still reduce the distance to
the same target, as shown in Appendix~\ref{app:subsection-03}.
Thus penalizing all source-action disagreements would
incorrectly penalize permissible same-target mixing.

\paragraph{An algebraic within-/between-target identity.}
The distinction can also be seen without TV \citep{murphy2022probabilistic}.  For
$q_y(a)=q_t^y(a\mid x)$, $\beta_y=\beta_y(t,x)$,
and $q(a)=\sum_y\beta_yq_y(a)$, add and subtract
$\sum_y\beta_y\sum_aq_y(a)^2$:
\begingroup\makeatletter\def\tagform@#1{}\def\theHequation{appChidden25}\makeatother
\begin{align}
 1-\sum_aq(a)^2
 &=\sum_y\beta_y
      \left(1-\sum_aq_y(a)^2\right)\notag\\
 &\quad+\left[\sum_y\beta_y\sum_aq_y(a)^2
                  -\sum_aq(a)^2\right].
 \label{eq:appC-variance-split}
\end{align}
\endgroup\addtocounter{equation}{-1}

The expression in square brackets equals \citep{murphy2022probabilistic}
\begin{align}
 &\sum_y\beta_y\sum_aq_y(a)^2
            -\sum_aq(a)^2\notag\\
 &=\sum_a\left[\sum_y\beta_yq_y(a)^2
          -\left(\sum_y\beta_yq_y(a)\right)^2\right]
          \notag\\
 &=\frac12\sum_{y,z}\beta_y\beta_z
            \sum_a(q_y(a)-q_z(a))^2.
 \label{eq:appC-between-target-square}
\end{align}

Indeed, expansion of the last square gives two
identical copies of
$\sum_y\beta_y\sum_aq_y(a)^2$ minus twice
$\sum_aq(a)^2$.  Consequently, the action Gini term
contains both \emph{within-target} action variation and
\emph{between-target} variation.  The bound in
Theorem~\ref{thm:risk} needs the target-conditioned
\emph{TV} term in \eqref{eq:qy}, not this squared
Euclidean between-target expression \citep{murphy2022probabilistic}.

\paragraph{Sharpness in the capacity bottleneck.}
The capacity family of Theorem~\ref{thm:capacity}
provides a simple closed-form check.  Under delete-first
orders, every task starts at $\{a_i\}$, reaches
$\varnothing$ after one event, and then reaches
$\{b_i\}$ after its second event.  At cumulative time
$u$, the teacher is at $\varnothing$ exactly when
$N_u=1$ \citep{kingman1993poisson}.  Hence
\begingroup\makeatletter\def\tagform@#1{}\def\theHequation{appChidden26}\makeatother
\begin{equation}
 \PP(X_u^\gamma=\varnothing)
 =\PP\{\Pois(u)=1\}=ue^{-u}.
 \label{eq:appC-capacity-empty-weight}
\end{equation}
\endgroup\addtocounter{equation}{-1}

At this common state the target is uniformly distributed
among the $m$ distinct targets, so \citep{billingsley1995probability}
\begingroup\makeatletter\def\tagform@#1{}\def\theHequation{appChidden27}\makeatother
\begin{align}
 q_u^{y_i}(\cdot\mid\varnothing)&=\delta_{+b_i},
       \notag\\
 q_u(\cdot\mid\varnothing)
 &=\frac1m\sum_{j=1}^m\delta_{+b_j}.
 \label{eq:appC-capacity-kernels}
\end{align}
\endgroup\addtocounter{equation}{-1}

By \eqref{eq:appC-dirac-TV}, their TV discrepancy is
\begingroup\makeatletter\def\tagform@#1{}\def\theHequation{appChidden28}\makeatother
\begin{align}
 \|q_u^{y_i}-q_u\|_{\TV}
 &=1-q_u(+b_i\mid\varnothing)\notag\\
 &=1-\frac1m.
 \label{eq:appC-capacity-alias-at-empty}
\end{align}
\endgroup\addtocounter{equation}{-1}

All other teacher-visited states have target-consistent
actions, so their target-conditioned discrepancy is zero.
With $u(t)=\int_0^t b(r)\,\dd r$ and
$\dd u=b(t)\,\dd t$ \citep{daley2003point},
\begin{align}
 \mathfrak A_B^Y
 &=\left(1-\frac1m\right)
      \int_0^T b(t)u(t)e^{-u(t)}\,\dd t\notag\\
 &=\left(1-\frac1m\right)
      \int_0^B ue^{-u}\,\dd u\notag\\
 &=\left(1-\frac1m\right)
      \left[-(u+1)e^{-u}\right]_{u=0}^{u=B}\notag\\
 &=\left(1-\frac1m\right)
      \bigl[1-(B+1)e^{-B}\bigr]\notag\\
 &=\left(1-\frac1m\right)\psi.
 \label{eq:appC-capacity-alias-sharp}
\end{align}

Integration by parts gives
$\int ue^{-u}\,\dd u=-ue^{-u}+\int e^{-u}\,\dd u
=-(u+1)e^{-u}$.  Comparing with
\eqref{eq:caplow} and \eqref{eq:caphigh},
\begin{align}
 D_{\mathrm{pair}}^{\mathrm{delete}}-\varepsilon_B
 &=\left(1-\frac\psi m\right)-(1-\psi)\notag\\
 &=\left(1-\frac1m\right)\psi\notag\\
 &=\mathfrak A_B^Y.
 \label{eq:appC-capacity-sharpness}
\end{align}

Thus the first-disagreement upper bound can be
attained exactly when the learned policy equals the
population pooled kernel.  It need not be tight on
arbitrary graph-edit instances.

\subsection{Marked-Path KL and Endpoint-Risk Bounds}
\label{app:subsection-10}

\paragraph{Which path laws are compared.}
The second error source is imperfect approximation
of the source-blind kernel.  Let $q_{\theta,t}(a\mid x)$
be a legal action distribution, including $\bot$, with
the same exogenous clock as $q_t(a\mid x)$.
The population discrepancy in \eqref{eq:excess} is
\begin{equation}
 \mathfrak E_\theta
 =\frac1B\int_0^T b(t)
  \E_{X_t^\gamma}
     \KL(q_t(\cdot\mid X_t^\gamma)
             \|q_{\theta,t}(\cdot\mid X_t^\gamma))\,\dd t.
 \label{eq:appC-population-KL}
\end{equation}
We assume this quantity is finite \citep{gray2011entropy,csiszar2011information}.  In particular, for
almost every event time and every state occupied with
positive probability by the projection, a student action
probability must be positive whenever the projected
action probability is positive.  This support requirement
alone does not guarantee integrability of the logarithm \citep{cover2006elements,csiszar2011information}.
If it fails on a set of positive occupancy and clock measure,
$\mathfrak E_\theta=\infty$ and the ensuing finite
Pinsker estimate is not applicable
\citep{cover2006elements,csiszar2011information}.

Let $M=N_T$ be the random event count and
$0<\tau_1<\cdots<\tau_M\le T$ the event times.
The \emph{marked path} records the initial source,
the event times, and the complete action sequence: \citep{jacod2003limit,daley2003point}
\begin{equation}
 \mathcal H=(S,M,\tau_1,\ldots,\tau_M,
                         A_1,\ldots,A_M).
 \label{eq:appC-marked-record}
\end{equation}

The realized graph after $k$ events is determined by
$X_0=S$ and the recursion \citep{ethier1986markov}
\begingroup\makeatletter\def\tagform@#1{}\def\theHequation{appChidden29}\makeatother
\begin{equation}
 X_{\tau_k}=T_{A_k}(X_{\tau_{k-1}}),
 \quad X_{\tau_0}=S.
 \label{eq:appC-marked-recursion}
\end{equation}
\endgroup\addtocounter{equation}{-1}

We include null marks $\bot$ at all events after
completion: omitting them would change the likelihood
ratio of the full marked process.  Denote the laws of
$\mathcal H$ under the projection and the student by
$\PP_{\mathrm{marked}}^{\mathrm{proj}}$ and
$\PP_{\mathrm{marked}}^\theta$, respectively \citep{kingman1993poisson,bremaud1981point}.

\paragraph{The likelihood-ratio chain rule, one mark at a time.}
Both path laws assign probability $p_i$ to the same
initial source $s_i$ and the same nonhomogeneous
Poisson density to the clock realization \citep{jacod2003limit,kingman1993poisson}.  Conditional
on that source and those event times, the projection
assigns an action sequence the probability
\begingroup\makeatletter\def\tagform@#1{}\def\theHequation{appChidden30}\makeatother
\begin{equation}
 \prod_{k=1}^{M}
   q_{\tau_k}(A_k\mid X_{\tau_k-}),
 \label{eq:appC-proj-mark-likelihood}
\end{equation}
\endgroup\addtocounter{equation}{-1}

whereas the student assigns
\begingroup\makeatletter\def\tagform@#1{}\def\theHequation{appChidden31}\makeatother
\begin{equation}
 \prod_{k=1}^{M}
   q_{\theta,\tau_k}(A_k\mid X_{\tau_k-}).
 \label{eq:appC-student-mark-likelihood}
\end{equation}
\endgroup\addtocounter{equation}{-1}

The event-time and source probability factors cancel
in their Radon--Nikodym ratio \citep{jacod2003limit,billingsley1995probability}.  Taking its logarithm,
\begin{align}
 \log\frac{\dd\PP_{\mathrm{marked}}^{\mathrm{proj}}}
              {\dd\PP_{\mathrm{marked}}^\theta}(\mathcal H)
 &=\log\prod_{k=1}^{M}
     \frac{q_{\tau_k}(A_k\mid X_{\tau_k-})}
          {q_{\theta,\tau_k}(A_k\mid X_{\tau_k-})}
             \notag\\
 &=\sum_{k=1}^{M}
     \log\frac{q_{\tau_k}(A_k\mid X_{\tau_k-})}
          {q_{\theta,\tau_k}(A_k\mid X_{\tau_k-})}.
 \label{eq:appC-marked-log-ratio}
\end{align}
When $M=0$, the empty product is one, and the sum is
zero.  Under the projection, conditional on the graph
and time just before event $k$, the law of $A_k$ is \citep{jacod2003limit}
$q_{\tau_k}(\cdot\mid X_{\tau_k-})$.  Averaging the
$k$th summand therefore gives
\begingroup\makeatletter\def\tagform@#1{}\def\theHequation{appChidden32}\makeatother
\begin{align}
 &\E_{A_k\mid X_{\tau_k-},\tau_k}^{\mathrm{proj}}
   \log\frac{q_{\tau_k}(A_k\mid X_{\tau_k-})}
            {q_{\theta,\tau_k}(A_k\mid X_{\tau_k-})}
          \notag\\
 &=\sum_a q_{\tau_k}(a\mid X_{\tau_k-})
    \log\frac{q_{\tau_k}(a\mid X_{\tau_k-})}
             {q_{\theta,\tau_k}(a\mid X_{\tau_k-})}
          \notag\\
 &=\KL\bigl(q_{\tau_k}(\cdot\mid X_{\tau_k-})
       \|q_{\theta,\tau_k}(\cdot\mid X_{\tau_k-})\bigr).
 \label{eq:appC-one-mark-KL}
\end{align}
\endgroup\addtocounter{equation}{-1}

Taking the projection expectation in \citep{billingsley1995probability}
\eqref{eq:appC-marked-log-ratio} and using the tower
property of conditional expectation,
\begin{align}
 &\KL(\PP_{\mathrm{marked}}^{\mathrm{proj}}
          \|\PP_{\mathrm{marked}}^\theta)\notag\\
 &=\E^{\mathrm{proj}}
   \log\frac{\dd\PP_{\mathrm{marked}}^{\mathrm{proj}}}
            {\dd\PP_{\mathrm{marked}}^\theta}\notag\\
 &=\E^{\mathrm{proj}}\sum_{k=1}^{M}
   \KL\bigl(q_{\tau_k}(\cdot\mid X_{\tau_k-})
       \|q_{\theta,\tau_k}(\cdot\mid X_{\tau_k-})\bigr).
 \label{eq:appC-marked-chain}
\end{align}
This is the conditional chain rule
for the action sequence; keeping the random
clock in the record makes it applicable without
identifying paths with different event counts
\citep{cover2006elements,murphy2022probabilistic}.

\paragraph{From the sum over events to an integral.}
Define the predictable, nonnegative function
\begingroup\makeatletter\def\tagform@#1{}\def\theHequation{appChidden33}\makeatother
\begin{equation}
 g(t,x)=\KL(q_t(\cdot\mid x)
                 \|q_{\theta,t}(\cdot\mid x)).
 \label{eq:appC-KL-integrand}
\end{equation}
\endgroup\addtocounter{equation}{-1}

The compensator identity for the source-independent
Poisson clock states that \citep{jacod2003limit,daley2003point}
\begin{align}
 \E^{\mathrm{proj}}\sum_{k=1}^{M}
       g(\tau_k,X_{\tau_k-})
 &=\E^{\mathrm{proj}}\int_0^T
          b(t)g(t,X_{t-})\,\dd t\notag\\
 &=\int_0^T b(t)\E_{X_t^{\mathrm{proj}}}
          g(t,X_t^{\mathrm{proj}})\,\dd t.
 \label{eq:appC-Poisson-compensator}
\end{align}

The first equality is the point-process compensation
formula; the second uses the absence of a jump at a
fixed deterministic time and Tonelli's theorem for a
nonnegative integrand
\citep{kingman1993poisson,bremaud1981point,
        anderson1991continuous}.  By the unconditional
Markov-projection identity \eqref{eq:forward}, whose
detailed proof is in Appendix~\ref{app:subsection-02},
\begin{equation}
 \calL(X_t^{\mathrm{proj}})
 =\calL(X_t^\gamma).
 \label{eq:appC-marginal-identity-for-KL}
\end{equation}
Combining \eqref{eq:appC-marked-chain}--
\eqref{eq:appC-marginal-identity-for-KL},
\begin{align}
 &\KL(\PP_{\mathrm{marked}}^{\mathrm{proj}}
          \|\PP_{\mathrm{marked}}^\theta)\notag\\
 &=\int_0^T b(t)
      \E_{X_t^{\mathrm{proj}}}g(t,X_t^{\mathrm{proj}})
         \,\dd t\notag\\
 &=\int_0^T b(t)
      \E_{X_t^\gamma}g(t,X_t^\gamma)\,\dd t\notag\\
 &=B\left[\frac1B\int_0^T b(t)
      \E_{X_t^\gamma}g(t,X_t^\gamma)\,\dd t\right]
         \notag\\
 &=\boxed{B\mathfrak E_\theta}.
 \label{eq:appC-exact-path-KL}
\end{align}
This proves the marked-path identity
\eqref{eq:pathkl}.  The equality depends on the two
models having \emph{the same} independent clock, the
same initial source law, and state/time-conditioned
mark kernels; a source-dependent event rate would
introduce additional likelihood-ratio terms.

\paragraph{Data processing and Pinsker.}
The complete marked path determines the terminal pair
$(S,X_T)$.  Applying a deterministic mapping to a
probability measure cannot increase its total variation
distance (data processing), so \citep{gray2011entropy,cover2006elements}
\begingroup\makeatletter\def\tagform@#1{}\def\theHequation{appChidden34}\makeatother
\begin{align}
 &\bigl\|\calL(S,X_T^{\mathrm{proj}})
            -\calL(S,X_T^\theta)\bigr\|_{\TV}
         \notag\\
 &\qquad\le
   \|\PP_{\mathrm{marked}}^{\mathrm{proj}}
            -\PP_{\mathrm{marked}}^\theta\|_{\TV}.
 \label{eq:appC-path-data-processing}
\end{align}
\endgroup\addtocounter{equation}{-1}

Pinsker's inequality for probability measures with
finite relative entropy gives \citep{gray2011entropy,tsybakov2009nonparametric}
\begin{equation}
 \|P-Q\|_{\TV}
 \le\sqrt{\tfrac12\KL(P\|Q)}.
 \label{eq:appC-Pinsker}
\end{equation}

Applying this with the two marked laws, and then using
\eqref{eq:appC-exact-path-KL},
\begin{align}
 &\bigl\|\calL(S,X_T^{\mathrm{proj}})
            -\calL(S,X_T^\theta)\bigr\|_{\TV}
          \notag\\
 &\le\sqrt{\tfrac12\KL(
      \PP_{\mathrm{marked}}^{\mathrm{proj}}
         \|\PP_{\mathrm{marked}}^\theta)}\notag\\
 &=\sqrt{\tfrac12 B\mathfrak E_\theta}\notag\\
 &=\sqrt{B\mathfrak E_\theta/2}.
 \label{eq:appC-proj-student-TV}
\end{align}
Both facts are standard information-theoretic tools
\citep{cover2006elements,csiszar2011information,
        tsybakov2009nonparametric}.  The bound is on the
\emph{terminal joint source--output law}, not merely
its marginal output distribution \citep{cover2006elements,csiszar2011information}.

\paragraph{Proof of the paired-error bound in the Method.}
Write $P_T^\theta=\calL(S,X_T^\theta)$,
$P_T^{\mathrm{proj}}=\calL(S,X_T^{\mathrm{proj}})$,
and recall that $\pi=\calL(S,Y)$.  The triangle
inequality, followed by the two previously established
bounds, gives
\begin{align}
 D_{\mathrm{pair}}^\theta
 &=\|P_T^\theta-\pi\|_{\TV}\notag\\
 &\le\|P_T^\theta-P_T^{\mathrm{proj}}\|_{\TV}
           +\|P_T^{\mathrm{proj}}-\pi\|_{\TV}
           \notag\\
 &\le\sqrt{B\mathfrak E_\theta/2}
           +D_{\mathrm{pair}}^{\mathrm{proj}}\notag\\
 &\le\boxed{\sqrt{B\mathfrak E_\theta/2}
           +\varepsilon_B+\mathfrak A_B^Y}.
 \label{eq:appC-paired-upper-final}
\end{align}
Conversely, a graph-only student has at most one edge
toggle per event and starts at $s_i$ with
$|s_i\triangle y_i|=K_i$.  On $\{N_T<K_i\}$ it
cannot reach $y_i$; hence
\begingroup\makeatletter\def\tagform@#1{}\def\theHequation{appChidden35}\makeatother
\begin{align}
 D_{\mathrm{pair}}^\theta
 &=\sum_i p_i\PP(X_T^\theta\ne y_i\mid S=s_i)
          \notag\\
 &\ge\sum_i p_i\PP(N_T<K_i\mid S=s_i)
          \notag\\
 &=\sum_i p_i\PP\{\Pois(B)<K_i\}\notag\\
 &=\varepsilon_B.
 \label{eq:appC-paired-lower-final}
\end{align}
\endgroup\addtocounter{equation}{-1}

The lower and upper inequalities reproduce the first
line of \eqref{eq:risk} \emph{with the same assumptions
as the main-text theorem}.

\paragraph{Proof of the marginal-error bound.}
Let $\nu=\calL(Y)=\sum_i p_i\delta_{y_i}$, and
write $\mu_T^\gamma=\calL(X_T^\gamma)$ \citep{tsybakov2009nonparametric}.
The projection preserves the unconditional terminal
teacher marginal:
\begin{equation}
 \calL(X_T^{\mathrm{proj}})=\mu_T^\gamma.
 \label{eq:appC-terminal-marginal-match}
\end{equation}
The route teacher and target $Y$ live on the same
probability space \citep{lindvall1992lectures,thorisson2000coupling}.  For \emph{any} coupling of random
variables $U,V$, the elementary coupling inequality is
\begingroup\makeatletter\def\tagform@#1{}\def\theHequation{appChidden36}\makeatother
\begin{equation}
 \|\calL(U)-\calL(V)\|_{\TV}
 \le\PP(U\ne V).
 \label{eq:appC-elementary-coupling-inequality}
\end{equation}
\endgroup\addtocounter{equation}{-1}

For the teacher, given $S=s_i$, its terminal state
is $y_i$ exactly when at least $K_i$ events occurred \citep{kingman1993poisson}.
Therefore
\begin{align}
 \|\mu_T^\gamma-\nu\|_{\TV}
 &\le\PP(X_T^\gamma\ne Y)\notag\\
 &=\sum_i p_i\PP(N_T<K_i\mid S=s_i)\notag\\
 &=\varepsilon_B.
 \label{eq:appC-teacher-marginal-bound}
\end{align}
Notice that the last assertion is a bound, not
necessarily an equality: errors associated with
different sources may cancel after marginalization \citep{cover2006elements}.
Using marginal data processing in
\eqref{eq:appC-proj-student-TV}, followed by
\eqref{eq:appC-terminal-marginal-match} and
\eqref{eq:appC-teacher-marginal-bound},
\begin{align}
 D_{\mathrm{marg}}^\theta
 &=\|\calL(X_T^\theta)-\nu\|_{\TV}\notag\\
 &\le\|\calL(X_T^\theta)
          -\calL(X_T^{\mathrm{proj}})\|_{\TV}
       +\|\calL(X_T^{\mathrm{proj}})-\nu\|_{\TV}
       \notag\\
 &\le\sqrt{B\mathfrak E_\theta/2}
          +\|\mu_T^\gamma-\nu\|_{\TV}\notag\\
 &\le\sqrt{B\mathfrak E_\theta/2}+\varepsilon_B.
 \label{eq:appC-marginal-upper-final}
\end{align}
This is the second line of \eqref{eq:risk}.
In particular, the structural ambiguity term is not
needed for the marginal bound because pooling preserves
$\calL(X_t)$ exactly, irrespective of whether the
pairing survives.

\paragraph{Scope of the population bound.}
If $C(\gamma)=0$, then $\mathfrak A_B^Y=0$;
if additionally $q_{\theta,t}=q_t$ on the population
projection support, then $\mathfrak E_\theta=0$ \citep{cover2006elements}.
The preceding upper bound and the universal clock lower
bound then coincide, yielding
\begingroup\makeatletter\def\tagform@#1{}\def\theHequation{appChidden37}\makeatother
\begin{equation}
 D_{\mathrm{pair}}^\theta=\varepsilon_B.
 \label{eq:appC-zero-alias-zero-KL}
\end{equation}
\endgroup\addtocounter{equation}{-1}

When $C(\gamma)>0$, neither the numerical value of
$C(\gamma)$ nor the bound $\mathfrak A_B^Y$ determines
the exact endpoint error; its value is obtained from
the conditional forward equations in the next subsection.
Also, a student trained with paired-endpoint supervision
may attain a smaller paired error than the ideal pooled
projection, since the population imitation bound is not
an optimality statement about all source-blind policies \citep{murphy2022probabilistic}.

For completeness, the cross-entropy interpretation of
the approximation term can be verified algebraically \citep{gray2011entropy,murphy2022probabilistic}.
For categorical laws $P$ and $Q$, write
\begingroup\makeatletter\def\tagform@#1{}\def\theHequation{appChidden38}\makeatother
\begin{align}
 H(P,Q)&=-\sum_aP(a)\log Q(a),\notag\\
 H(P)&=-\sum_aP(a)\log P(a),\notag\\
 \KL(P\|Q)&=\sum_aP(a)\log\frac{P(a)}{Q(a)}.
 \label{eq:appC-entropy-definitions}
\end{align}
\endgroup\addtocounter{equation}{-1}

Then, by expanding the logarithmic ratio \citep{gray2011entropy,cover2006elements},
\begingroup\makeatletter\def\tagform@#1{}\def\theHequation{appChidden39}\makeatother
\begin{align}
 H(P)+\KL(P\|Q)
 &=-\sum_aP(a)\log P(a)
        +\sum_aP(a)[\log P(a)-\log Q(a)]\notag\\
 &=-\sum_aP(a)\log Q(a)\notag\\
 &=H(P,Q).
 \label{eq:appC-entropy-KL-decomposition}
\end{align}
\endgroup\addtocounter{equation}{-1}

Averaging this pointwise identity under the empirical
cohort/time/state sampling scheme in
\eqref{eq:empirical-pooled-ce} yields the main-text
decomposition
$\widehat{\mathcal L}_{\mathrm{CE}}^{\mathrm{tr}}
=\widehat H^{\mathrm{tr}}+
 \widehat{\mathfrak E}_{\theta}^{\mathrm{tr}}$.
The teacher entropy does not depend on $\theta$ \citep{murphy2022probabilistic,goodfellow2016deep}.
However, the \emph{empirical} KL term in that identity
is not automatically equal to the \emph{population}
$\mathfrak E_\theta$ in \eqref{eq:appC-population-KL}:
they use potentially different pair laws, cohorts, and
state/time occupancies.  Estimating their difference
requires separate generalization assumptions, which
are not part of Theorem~\ref{thm:risk}
\citep{goodfellow2016deep,murphy2022probabilistic}.

\subsection{Exact Conditional Forward Equations}
\label{app:subsection-11}

\paragraph{Computing the pooled kernel from any fixed routes.}
The previous subsection supplies a bound.  We now
derive the exact finite-state evaluation that appears
in \eqref{eq:weights}--\eqref{eq:exactrisk} and is used
for the numerical comparisons of
Section~\ref{sec:experiments}.  Work in cumulative
clock time $u\in[0,B]$, so $N_u$ is a unit-rate
Poisson process after the deterministic time change
\citep{kingman1993poisson,bremaud1981point}.
For a route with $K_i>0$, position $j<K_i$ is occupied
precisely when $N_u=j$ \citep{daley2003point,kingman1993poisson}.  Therefore
\begin{equation}
 w_{ij}(u)=e^{-u}\frac{u^j}{j!},
 \quad 0\le j<K_i.
 \label{eq:appC-weights-nonterminal}
\end{equation}

The terminal position is occupied when $N_u\ge K_i$:
\begingroup\makeatletter\def\tagform@#1{}\def\theHequation{appChidden40}\makeatother
\begin{align}
 w_{i,K_i}(u)
 &=\PP\{\Pois(u)\ge K_i\}\notag\\
 &=1-\PP\{\Pois(u)<K_i\}\notag\\
 &=1-\sum_{j=0}^{K_i-1}e^{-u}\frac{u^j}{j!}.
 \label{eq:appC-weights-terminal}
\end{align}
\endgroup\addtocounter{equation}{-1}

If $K_i=0$, the route contains only $s_i=y_i$ and
$w_{i0}(u)=\PP(N_u\ge0)=1$.  For every $K_i$,
\begingroup\makeatletter\def\tagform@#1{}\def\theHequation{appChidden41}\makeatother
\begin{align}
 \sum_{j=0}^{K_i}w_{ij}(u)
 &=\sum_{j=0}^{K_i-1}\PP(N_u=j)
           +\PP(N_u\ge K_i)\notag\\
 &=\PP(N_u<K_i)+\PP(N_u\ge K_i)\notag\\
 &=1.
 \label{eq:appC-weights-normalized}
\end{align}
\endgroup\addtocounter{equation}{-1}

Each teacher route contributes a normalized
state/action distribution at every clock time.

Define the full teacher state mass and joint state/action
mass, respectively, by
\begin{align}
 \mu_u(x)
 &=\sum_i p_i\sum_{j=0}^{K_i}
       w_{ij}(u)\mathbf1\{x_{i,j}=x\},
 \label{eq:appC-state-mass}\\
 J_u(x,a)
 &=\sum_i p_i\sum_{j=0}^{K_i}
       w_{ij}(u)\mathbf1\{x_{i,j}=x,a_{i,j}=a\}.
 \label{eq:appC-action-mass}
\end{align}
The law of total probability gives \citep{billingsley1995probability}
\begingroup\makeatletter\def\tagform@#1{}\def\theHequation{appChidden42}\makeatother
\begin{align}
 \sum_aJ_u(x,a)
 &=\sum_i p_i\sum_{j=0}^{K_i}w_{ij}(u)
      \mathbf1\{x_{i,j}=x\}
      \sum_a\mathbf1\{a_{i,j}=a\}\notag\\
 &=\sum_i p_i\sum_{j=0}^{K_i}w_{ij}(u)
      \mathbf1\{x_{i,j}=x\}\notag\\
 &=\mu_u(x).
 \label{eq:appC-joint-marginal}
\end{align}
\endgroup\addtocounter{equation}{-1}

At any $x$ with $\mu_u(x)>0$, Bayes' rule gives \citep{billingsley1995probability}
\begin{align}
 q_u(a\mid x)
 &=\PP(A_u^\gamma=a\mid X_u^\gamma=x)
          \notag\\
 &=\frac{\PP(X_u^\gamma=x,A_u^\gamma=a)}
          {\PP(X_u^\gamma=x)}\notag\\
 &=\frac{J_u(x,a)}{\mu_u(x)}.
 \label{eq:appC-kernel-Bayes}
\end{align}
Substituting \eqref{eq:appC-state-mass}--
\eqref{eq:appC-action-mass} reproduces
\eqref{eq:kernel}.  Moreover,
\begingroup\makeatletter\def\tagform@#1{}\def\theHequation{appChidden43}\makeatother
\begin{align}
 \sum_aq_u(a\mid x)
 &=\sum_a\frac{J_u(x,a)}{\mu_u(x)}\notag\\
 &=\frac{\sum_aJ_u(x,a)}{\mu_u(x)}\notag\\
 &=\frac{\mu_u(x)}{\mu_u(x)}\notag\\
 &=1.
 \label{eq:appC-kernel-normalized}
\end{align}
\endgroup\addtocounter{equation}{-1}

Every action with positive $q_u(a\mid x)$ is a legal
successor action of some teacher route \citep{ethier1986markov}.  For
$u\in(0,B)$, all visited nonterminal route positions
have positive Poisson weight, and terminal positions
also have positive tail weight.  Consequently, the
kernel is defined on the union
\begingroup\makeatletter\def\tagform@#1{}\def\theHequation{appChidden44}\makeatother
\begin{equation}
 W_\gamma=\bigcup_{i=1}^m\mathcal V_i
 \label{eq:appC-reachable-union}
\end{equation}
\endgroup\addtocounter{equation}{-1}

for all interior clock times.  At unoccupied states,
$q_u$ may be assigned any legal categorical law; the
ideal projected process cannot reach them from the
chosen initial sources.  Values assigned at the single
time $u=0$ likewise do not alter a continuous-time
Poisson evolution because the clock has zero probability
of an event exactly at that time \citep{kingman1993poisson}.

\paragraph{First-order probability balance over one clock step.}
Fix source $S=s_i$ and define
\begin{equation}
 r_i(u,x)=\PP(X_u^{\mathrm{proj}}=x\mid S=s_i),
 \qquad r_i(0,x)=\mathbf1\{x=s_i\}.
 \label{eq:appC-conditional-initial}
\end{equation}
For a small clock-time increment $h>0$, the unit-rate
Poisson law has \citep{daley2003point,kingman1993poisson}
\begingroup\makeatletter\def\tagform@#1{}\def\theHequation{appChidden45}\makeatother
\begin{align}
 \PP(N_{u+h}-N_u=0)
 &=e^{-h}=1-h+O(h^2),\notag\\
 \PP(N_{u+h}-N_u=1)
 &=he^{-h}=h+O(h^2),\notag\\
 \PP(N_{u+h}-N_u\ge2)
 &=1-e^{-h}(1+h)=O(h^2).
 \label{eq:appC-one-small-step-clock}
\end{align}
\endgroup\addtocounter{equation}{-1}

Condition on the graph $z$ immediately before the
increment.  With zero events, the graph stays at $z$;
with one event, a mark $a$ is drawn from $q_u(\cdot\mid z)$
to first order and the successor is $T_a(z)$ \citep{norris1997markov,ethier1986markov}.
After summing the mutually exclusive incoming cases,
\begingroup\makeatletter\def\tagform@#1{}\def\theHequation{appChidden46}\makeatother
\begin{align}
 r_i(u+h,x)
 &=\bigl(1-h\bigr)r_i(u,x)\notag\\
 &\quad+h\sum_{z\in W_\gamma}
          \sum_{a:T_a(z)=x}
             r_i(u,z)q_u(a\mid z)\notag\\
 &\quad+o(h).
 \label{eq:appC-small-step-balance}
\end{align}
\endgroup\addtocounter{equation}{-1}

For time-dependent kernels the $o(h)$ term also
contains their variation over the short interval;
bounded measurability suffices for the integral
forward equation, and the derivative holds at almost
every $u$ \citep{norris1997markov,anderson1991continuous}.
Subtracting $r_i(u,x)$ from both sides,
\begingroup\makeatletter\def\tagform@#1{}\def\theHequation{appChidden47}\makeatother
\begin{align}
 r_i(u+h,x)-r_i(u,x)
 &=-h r_i(u,x)\notag\\
 &\quad+h\sum_{z,a:T_a(z)=x}
        r_i(u,z)q_u(a\mid z)+o(h).
 \label{eq:appC-small-step-difference}
\end{align}
\endgroup\addtocounter{equation}{-1}

Dividing by $h$ and taking $h\downarrow0$ gives
\begin{equation}
 \frac{\dd r_i(u,x)}{\dd u}
 =\sum_{z,a:T_a(z)=x}r_i(u,z)q_u(a\mid z)
      -r_i(u,x).
 \label{eq:appC-exact-conditional-ODE}
\end{equation}
This is the forward system \eqref{eq:ode} \citep{norris1997markov,anderson1991continuous}.
The summation includes self-transitions $T_\bot(x)=x$;
thus the $-r_i(u,x)$ term is the rate of \emph{all}
events at $x$, while the incoming sum restores the
probability of null-event marks that leave $x$
unchanged.  There is no double subtraction.

\paragraph{Normalization, support, and uniqueness.}
Summing \eqref{eq:appC-exact-conditional-ODE} over all
$x\in W_\gamma$, and exchanging the finite sums,
\begingroup\makeatletter\def\tagform@#1{}\def\theHequation{appChidden48}\makeatother
\begin{align}
 \frac{\dd}{\dd u}\sum_xr_i(u,x)
 &=\sum_x\sum_{z,a:T_a(z)=x}
      r_i(u,z)q_u(a\mid z)-\sum_xr_i(u,x)\notag\\
 &=\sum_zr_i(u,z)\sum_aq_u(a\mid z)
       -\sum_xr_i(u,x)\notag\\
 &=\sum_zr_i(u,z)-\sum_xr_i(u,x)\notag\\
 &=0.
 \label{eq:appC-ODE-normalization}
\end{align}
\endgroup\addtocounter{equation}{-1}

Since the initial sum is one, it remains one.  Every
non-null mark supported by $q_u(\cdot\mid z)$ from a
visited state is the prescribed edge toggle of some
route; its successor lies in $W_\gamma$.  Therefore the
system is closed on that finite set and no probability
mass leaks to excluded states.  The generator has
nonnegative off-diagonal rates and rows summing to
zero, so its finite-state transition evolution
preserves nonnegativity \citep{ethier1986markov,norris1997markov}.  Bounded measurable generator
coefficients on the finite interval also give uniqueness
of the integral equation associated with
\eqref{eq:appC-exact-conditional-ODE}
\citep{norris1997markov,anderson1991continuous}.

\paragraph{Why pooled marginals do not determine paired risk.}
The weighted mixture of the conditional solutions is \citep{billingsley1995probability}
\begingroup\makeatletter\def\tagform@#1{}\def\theHequation{appChidden49}\makeatother
\begin{equation}
 \bar r(u,x)=\sum_i p_i r_i(u,x).
 \label{eq:appC-conditional-mixture}
\end{equation}
\endgroup\addtocounter{equation}{-1}

Multiplying \eqref{eq:appC-exact-conditional-ODE} by
$p_i$ and summing over $i$ yields
\begingroup\makeatletter\def\tagform@#1{}\def\theHequation{appChidden50}\makeatother
\begin{align}
 \frac{\dd\bar r(u,x)}{\dd u}
 &=\sum_i p_i\frac{\dd r_i(u,x)}{\dd u}\notag\\
 &=\sum_{z,a:T_a(z)=x}
       \left(\sum_i p_i r_i(u,z)\right)q_u(a\mid z)
       -\sum_i p_i r_i(u,x)\notag\\
 &=\sum_{z,a:T_a(z)=x}
       \bar r(u,z)q_u(a\mid z)-\bar r(u,x).
 \label{eq:appC-mixture-ODE}
\end{align}
\endgroup\addtocounter{equation}{-1}

Its initial distribution is
\begingroup\makeatletter\def\tagform@#1{}\def\theHequation{appChidden51}\makeatother
\begin{equation}
 \bar r(0,x)=\sum_i p_i\mathbf1\{x=s_i\},
 \label{eq:appC-mixture-initial}
\end{equation}
\endgroup\addtocounter{equation}{-1}

which is also the teacher's initial marginal \citep{norris1997markov}.
The teacher occupancy $\mu_u(x)$ in
\eqref{eq:appC-state-mass} satisfies the same forward
equation by \eqref{eq:forward}.  Finite-state
uniqueness therefore implies
\begingroup\makeatletter\def\tagform@#1{}\def\theHequation{appChidden52}\makeatother
\begin{equation}
 \bar r(u,x)=\mu_u(x)
 \quad\text{for every }u\in[0,B],\ x\in W_\gamma.
 \label{eq:appC-mixture-teacher-equality}
\end{equation}
\endgroup\addtocounter{equation}{-1}

Nevertheless, individual $r_i(u,x)$ need not equal
the source-conditioned teacher occupancies
$\sum_jw_{ij}(u)\mathbf1\{x_{i,j}=x\}$.  This
is exactly why a marginally correct projection may
have incorrect source--target pairings.

\paragraph{Exact paired and marginal endpoint errors.}
Since $S=s_i$ with probability $p_i$, the probability
of correct source--target correspondence is
\begingroup\makeatletter\def\tagform@#1{}\def\theHequation{appChidden53}\makeatother
\begin{align}
 \PP(X_B^{\mathrm{proj}}=Y)
 &=\sum_i\PP(S=s_i,
                  X_B^{\mathrm{proj}}=y_i)\notag\\
 &=\sum_i p_i
           \PP(X_B^{\mathrm{proj}}=y_i\mid S=s_i)
           \notag\\
 &=\sum_i p_i r_i(B,y_i).
 \label{eq:appC-correct-pair-probability}
\end{align}
\endgroup\addtocounter{equation}{-1}

The deterministic-coupling TV identity in
\eqref{eq:lb} now yields
\begin{align}
 D_{\mathrm{pair}}^{\mathrm{proj}}
 &=\PP(X_B^{\mathrm{proj}}\ne Y)\notag\\
 &=1-\PP(X_B^{\mathrm{proj}}=Y)\notag\\
 &=1-\sum_i p_i r_i(B,y_i).
 \label{eq:appC-exact-paired-risk}
\end{align}
This reproduces \eqref{eq:exactrisk} without treating
$C(\gamma)$ as a numerical proxy for risk.  For the
terminal marginal, let
\begingroup\makeatletter\def\tagform@#1{}\def\theHequation{appChidden54}\makeatother
\begin{equation}
 \nu(x)=\PP(Y=x)=\sum_{i:y_i=x}p_i.
 \label{eq:appC-target-marginal}
\end{equation}
\endgroup\addtocounter{equation}{-1}

By the definition of finite-state total variation \citep{gray2011entropy,cover2006elements},
\begin{align}
 D_{\mathrm{marg}}^{\mathrm{proj}}
 &=\frac12\sum_{x\in W_\gamma}
    \bigl|\PP(X_B^{\mathrm{proj}}=x)-\nu(x)\bigr|
       \notag\\
 &=\frac12\sum_{x\in W_\gamma}
    \left|\sum_i p_i r_i(B,x)-\nu(x)\right|\notag\\
 &=\frac12\sum_{x\in W_\gamma}
    |\mu_B(x)-\nu(x)|.
 \label{eq:appC-exact-marginal-risk}
\end{align}

All target states and all reachable projected states
belong to $W_\gamma$, so no probability contributions
are omitted by restricting the finite sum to this set \citep{anderson1991continuous}.
The marginal value depends only on teacher occupancies,
which are unchanged by Markov projection, whereas
\eqref{eq:appC-exact-paired-risk} depends on the
\emph{source-conditioned} solutions $r_i$.

\paragraph{Exact characterization and numerical evaluation.}
Equations \eqref{eq:appC-state-mass},
\eqref{eq:appC-kernel-Bayes},
\eqref{eq:appC-conditional-initial}, and
\eqref{eq:appC-exact-conditional-ODE} give an exact
finite-dimensional mathematical characterization for
any specified family of legal shortest routes.
A floating-point integration of these ODEs, however,
introduces solver tolerances and is a \emph{numerical
approximation}, not a closed-form identity \citep{norris1997markov,anderson1991continuous}.
Independent Monte Carlo agreement, tolerance sensitivity,
and comparisons across edit-order classes belong to
Appendix~\ref{app:numerical-verification}.

Finally, for a fixed routing family, the two measures
\begingroup\makeatletter\def\tagform@#1{}\def\theHequation{appChidden55}\makeatother
\begin{align}
 C(\gamma)
 &\quad\text{from }\eqref{eq:collision},\notag\\
 D_{\mathrm{pair}}^{\mathrm{proj}}
 &\quad\text{from }\eqref{eq:appC-exact-paired-risk}
 \label{eq:appC-certificate-versus-risk}
\end{align}
\endgroup\addtocounter{equation}{-1}
answer different questions.  The first is a finite
combinatorial certificate: $C(\gamma)=0$ if and only
if Theorem~\ref{thm:endpoint}'s separation criterion
holds.  The second measures the exact clock-dependent
endpoint error of the chosen projected dynamics,
including when $C(\gamma)>0$.  The derivations above
show how these structural and stochastic statements
fit together without treating either as a surrogate
for the other.

\clearpage
\section{Numerical and Structural Verification}
\label{app:numerical-verification}

\subsection{Poisson Clock and Markov-Projection Diagnostics}
\label{app:subsection-12}
We test the Poisson route occupancies and source-blind action kernel
using six equally weighted graph-editing tasks with edit lengths
$K_i\in\{0,1,2,4\}$, repeated targets, and edge capacity two.
Route-position probabilities follow \eqref{eq:weights} under the common
clock \citep{kingman1993poisson}, while the pooled kernel in
\eqref{eq:q} defines the Markov projection \citep{norris1997markov}.
Figure~\ref{fig:task2-protocol}
compares analytical and simulated clock statistics, pooled actions,
and one-time state distributions.

\begin{figure*}[htbp]
\centering
\includegraphics[width=\textwidth]{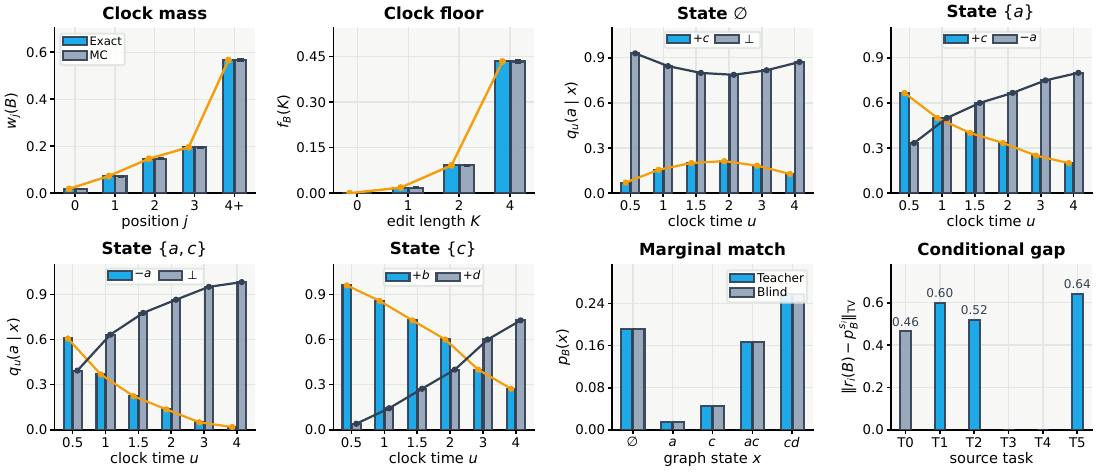}
\caption{Poisson-clock and Markov-projection diagnostics at $B=4$. (a,b) Route occupancies and clock floors; (c--f) pooled actions; (g,h) marginal agreement and source-conditional differences.}
\label{fig:task2-protocol}
\end{figure*}

The teacher and projection have the same unconditional state
marginals, although their source-conditioned distributions can differ.
Table~\ref{tab:br-task2-kernel} reports these discrepancies,
including the zero-edit case. 

\begin{table*}[htbp]
\centering
\caption{Clock floors and source-conditioned state-law discrepancies for six tasks at $B=4$. Here $p_B^{s_i}$ denotes the source-conditioned
terminal distribution of the route teacher.}
\label{tab:br-task2-kernel}
\small
\setlength{\tabcolsep}{7pt}
\begin{tabular}{cllrrrr}
\toprule
Task & Source $s_i$ & Target $y_i$ & $p_i$ & $K_i$ &
$\Pr\{\mathrm{Pois}(B)<K_i\}$ &
$\|r_i(B)-p_B^{s_i}\|_{\mathrm{TV}}$ \\
\midrule
T0 & $\varnothing$ & $\varnothing$ & $1/6$ & 0 & 0.0000 & 0.4646 \\
T1 & $\{a,b\}$ & $\{c,d\}$ & $1/6$ & 4 & 0.4335 & 0.6002 \\
T2 & $\{a,c\}$ & $\{b,d\}$ & $1/6$ & 4 & 0.4335 & 0.5172 \\
T3 & $\{d\}$ & $\{c,d\}$ & $1/6$ & 1 & 0.0183 & 0.0000 \\
T4 & $\{b,c\}$ & $\{b,d\}$ & $1/6$ & 2 & 0.0916 & 0.0000 \\
T5 & $\{a\}$ & $\{a,c\}$ & $1/6$ & 1 & 0.0183 & 0.6409 \\
\bottomrule
\end{tabular}
\end{table*}

\subsection{Synthetic Capacity and Marker-Free Checks}
\label{app:subsection-13}
The capacity family in Theorem~\ref{thm:capacity} compares
add-first and delete-first routes for
$m\in\{2,4,8,16\}$ at fixed edge capacity two. The marker-free
construction in~\eqref{eq:cross} has eight legal routes per task and
64 route pairs; the fixed schedules R0--R3 are specified Appendix~\ref{app:certificates}. Figure~\ref{fig:task1} displays
paired risks over the clock budget and the relation between
structural conflicts and target-conditioned aliasing.

\begin{figure*}[htbp]
\centering
\includegraphics[width=0.8\textwidth]{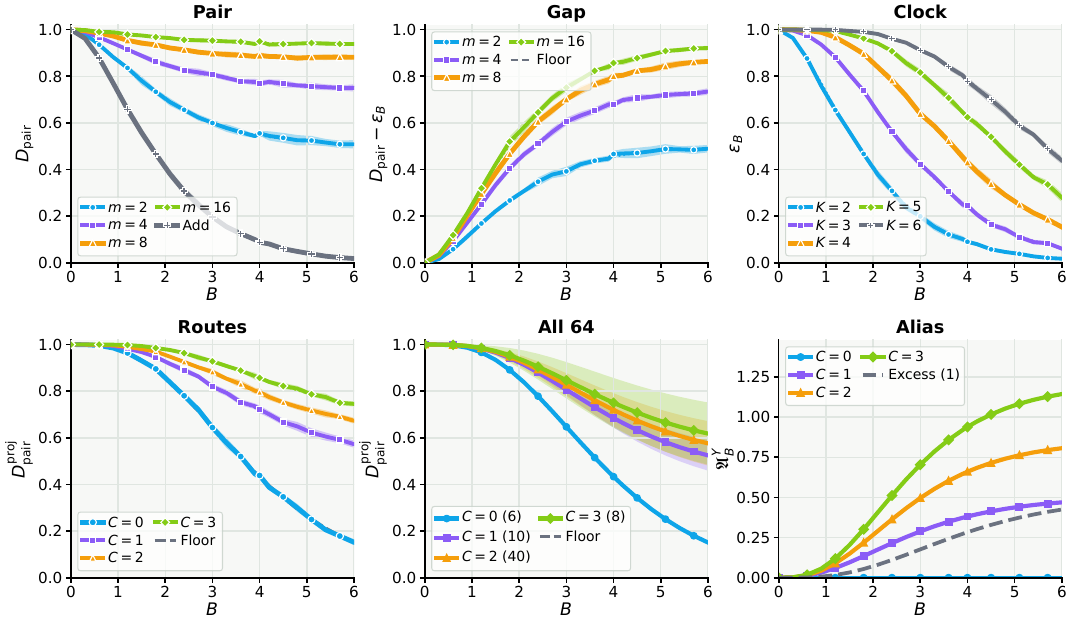}
\caption{Capacity and marker-free checks. Paired errors, Poisson floors, 64-route comparisons, and target-conditioned aliasing across clock budgets.}
\label{fig:task1}
\end{figure*}

Table~\ref{tab:br-task1-verification} compares exact projected
risks with independent event-driven simulations at $B=4$.
Add-first capacity routes reach the clock floor; delete-first
routes retain the same marginal error but have larger paired risk.
The 64 marker-free selections include six zero-conflict pairs.

\begin{table*}[t!]
\centering
\caption{Capacity and marker-free projected risks at $B=4$ (Monte Carlo mean $\pm$ SD).}
\label{tab:br-task1-verification}
\small
\setlength{\tabcolsep}{4pt}
\begin{tabular}{llrrrrrr}
\toprule
Family & Schedule & $C(\gamma)$ & $\varepsilon_B$ & $D_{\rm pair}^{\rm proj}$ & MC (mean $\pm$ SD) & $D_{\rm marg}$ & $\mathfrak{A}^{Y}_{B}$ \\
\midrule
Capacity $m=2$ & Add & 0 & 0.092 & 0.092 & $0.088\pm0.011$ & 0.092 & 0.000 \\
Capacity $m=2$ & Delete & 1 & 0.092 & 0.546 & $0.555\pm0.010$ & 0.092 & 0.454 \\
Capacity $m=4$ & Add & 0 & 0.092 & 0.092 & $0.089\pm0.009$ & 0.092 & 0.000 \\
Capacity $m=4$ & Delete & 6 & 0.092 & 0.773 & $0.768\pm0.014$ & 0.092 & 0.681 \\
Capacity $m=8$ & Add & 0 & 0.092 & 0.092 & $0.088\pm0.011$ & 0.092 & 0.000 \\
Capacity $m=8$ & Delete & 28 & 0.092 & 0.886 & $0.889\pm0.013$ & 0.092 & 0.795 \\
Capacity $m=16$ & Add & 0 & 0.092 & 0.092 & $0.093\pm0.006$ & 0.092 & 0.000 \\
Capacity $m=16$ & Delete & 120 & 0.092 & 0.943 & $0.950\pm0.006$ & 0.092 & 0.852 \\
\midrule
Marker-free & R0 & 0 & 0.433 & 0.433 & $0.439\pm0.011$ & 0.433 & 0.000 \\
Marker-free & R1 & 1 & 0.433 & 0.717 & $0.722\pm0.016$ & 0.433 & 0.381 \\
Marker-free & R2 & 2 & 0.433 & 0.794 & $0.796\pm0.016$ & 0.433 & 0.660 \\
Marker-free & R3 & 3 & 0.433 & 0.856 & $0.857\pm0.009$ & 0.433 & 0.938 \\
\bottomrule
\end{tabular}
\end{table*}

\subsection{ODE--Monte Carlo Endpoint-Risk Verification}
\label{app:subsection-14}
We evaluate the source-blind Markov projection by solving
the source-conditioned forward equations \eqref{eq:ode}.
For each fixed legal schedule, the pooled action kernel
is computed from the Poisson route occupancies, and the
resulting state probabilities are propagated to the terminal
clock budget. We independently estimate endpoint risks
using event-driven Monte Carlo with the same schedules,
initial source distribution, and Poisson clock. Each
configuration uses five seeds and 10,000 trajectories per
seed, with pointwise 95\% confidence intervals.

Figure~\ref{fig:exp-exact-mc} compares the exact ODE and
Monte Carlo estimates, including paired and marginal errors
across legal route selections and clock budgets. It also
examines conflict scores, standardized Monte Carlo
residuals, and sensitivity to ODE solver tolerances.
Table~\ref{tab:br-task3-exact-mc} reports the corresponding
endpoint risks at $B=4$ for conflict-free, conflicting,
capacity-limited, mixed-depth, and shared-target routes.

The marker-free R0 schedule attains the clock floor,
whereas R1--R3 incur increasing excess paired errors.
Under the capacity-four-task construction, add-first
and delete-first schedules preserve the same marginal
error but produce different paired risks. Shared-target
routes can merge without excess paired error, while
cross-target intersections introduce endpoint ambiguity.
Together, these comparisons verify the distinction
between marginal preservation and coupling recovery,
and connect the numerical results to the exact
recoverability criterion in Theorem~\ref{thm:endpoint}.

\begin{table*}[htbp]
\centering
\caption{Exact ODE and independent Monte Carlo endpoint-risk estimates at $B=4$. Monte Carlo: five seeds, 10,000 trajectories per seed; intervals are pointwise $95\%$ CIs. The capacity comparison fixes edge capacity two. Same-$Y$ permits shared states without excess paired error. The Mixed case shows a larger ODE--Monte Carlo deviation,
while the remaining estimates fall within their reported
95\% confidence intervals.}
\label{tab:br-task3-exact-mc}
\small
\setlength{\tabcolsep}{4.2pt}
\begin{tabular}{llrrrrrr}
\toprule
Family & Schedule & $C(\gamma)$ & $\varepsilon_B$ &
$D_{\rm pair}^{\rm proj}$ & MC $\pm$95\% CI &
$D_{\rm marg}$ & Excess \\
\midrule
Marker-free & R0 & 0 & 0.4335 & 0.4335 &
$0.4324\pm0.0043$ & 0.4335 & 0.0000 \\
Marker-free & R1 & 1 & 0.4335 & 0.7167 &
$0.7150\pm0.0040$ & 0.4335 & 0.2833 \\
Marker-free & R2 & 2 & 0.4335 & 0.7937 &
$0.7956\pm0.0035$ & 0.4335 & 0.3602 \\
Marker-free & R3 & 3 & 0.4335 & 0.8561 &
$0.8575\pm0.0031$ & 0.4335 & 0.4226 \\
\midrule
Capacity $m=4$ & Add & 0 & 0.0916 & 0.0916 &
$0.0921\pm0.0025$ & 0.0916 & 0.0000 \\
Capacity $m=4$ & Delete & 6 & 0.0916 & 0.7729 &
$0.7727\pm0.0037$ & 0.0916 & 0.6813 \\
\midrule
Protocol & Mixed & 4 & 0.1659 & 0.4714 &
$0.4794\pm0.0044$ & 0.1628 & 0.3055 \\
Protocol & Same-$Y$ & 0 & 0.4335 & 0.4335 &
$0.4351\pm0.0043$ & 0.4335 & 0.0000 \\
\bottomrule
\end{tabular}
\end{table*}

\subsection{Boundary Cases and Target Aliasing}
\label{app:subsection-15}
We examine the capacity boundary in Theorem~\ref{thm:capacity},
shared-target merging, cross-depth intersections, and the
aliasing term in Theorem~\ref{thm:risk}. Figure~\ref{fig:task4-boundaries}
compares exact paired and marginal errors, source-specific clock
floors, and the bound under controlled changes to the pooled
policy. At edge capacity one, add-first routes are illegal;
at capacity two, the matched add-first and delete-first routes
differ only in their order.

\begin{figure*}[htbp]
\centering
\includegraphics[width=\textwidth]{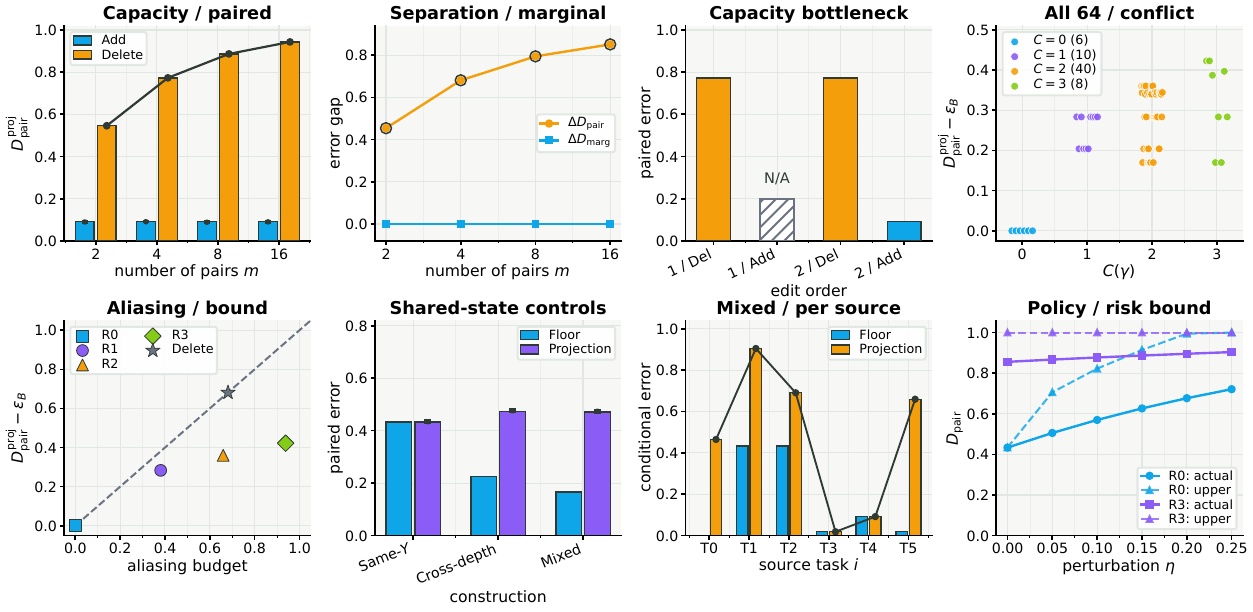}
\caption{Capacity, shared-target, and cross-depth boundary checks at $B=4$, including aliasing budgets and paired-error bounds for perturbed graph-only policies.}
\label{fig:task4-boundaries}
\end{figure*}

Table~\ref{tab:br-task4-boundaries} shows that the
capacity-family marginal error remains unchanged across
edit orders, while paired error depends on the selected
routes. At edge capacity two, add-first schedules attain
the clock floor, whereas delete-first schedules introduce
shared intermediate states and additional paired error.
As the number of tasks increases from $m=2$ to $m=16$,
the delete-first paired error rises from $0.5458$ to
$0.9432$. In these configurations, the excess paired
error equals $\mathfrak A_B^Y$, attaining the aliasing
term in Theorem~\ref{thm:risk}. At capacity one,
add-first routing becomes infeasible, making the
shared-state bottleneck unavoidable.

The marker-free schedules R0--R3 exhibit increasing
paired error with additional cross-target conflicts,
despite identical marginal errors. Same-target merging
attains the clock floor even when routes share states.
In contrast, cross-depth and mixed configurations
produce excess paired error through intersections
between routes with distinct targets. These results
distinguish permissible route merging from
target-conditioned ambiguity.

\begin{table*}[htbp]
\centering
\caption{Exact paired risks, marginal errors, and target-conditioned aliasing at $B=4$.}
\label{tab:br-task4-boundaries}
\small
\setlength{\tabcolsep}{5pt}
\begin{tabular}{llrrrrr}
\toprule
Family & Configuration & $C(\gamma)$ &
$\varepsilon_B$ & $D_{\rm pair}^{\rm proj}$ &
$D_{\rm marg}$ & $\mathfrak A_B^Y$ \\
\midrule
Capacity & $m=2$ Add & 0 & 0.0916 & 0.0916 &
0.0916 & 0.0000 \\
Capacity & $m=2$ Delete & 1 & 0.0916 & 0.5458 &
0.0916 & 0.4542 \\
Capacity & $m=4$ Add & 0 & 0.0916 & 0.0916 &
0.0916 & 0.0000 \\
Capacity & $m=4$ Delete & 6 & 0.0916 & 0.7729 &
0.0916 & 0.6813 \\
Capacity & $m=8$ Add & 0 & 0.0916 & 0.0916 &
0.0916 & 0.0000 \\
Capacity & $m=8$ Delete & 28 & 0.0916 & 0.8864 &
0.0916 & 0.7949 \\
Capacity & $m=16$ Add & 0 & 0.0916 & 0.0916 &
0.0916 & 0.0000 \\
Capacity & $m=16$ Delete & 120 & 0.0916 & 0.9432 &
0.0916 & 0.8516 \\
\midrule
Bottleneck & $m=4$, $C_{\rm edge}=1$ & 6 &
0.0916 & 0.7729 & 0.0916 & 0.6813 \\
\midrule
Marker-free & R0 & 0 & 0.4335 & 0.4335 &
0.4335 & 0.0000 \\
Marker-free & R1 & 1 & 0.4335 & 0.7167 &
0.4335 & 0.3809 \\
Marker-free & R2 & 2 & 0.4335 & 0.7937 &
0.4335 & 0.6596 \\
Marker-free & R3 & 3 & 0.4335 & 0.8561 &
0.4335 & 0.9383 \\
\midrule
Control & Same-$Y$ & 0 & 0.4335 & 0.4335 &
0.4335 & 0.0000 \\
Control & Cross-depth & 1 & 0.2259 & 0.4735 &
0.2259 & 0.2683 \\
Control & Mixed & 4 & 0.1659 & 0.4714 &
0.1628 & 0.5913 \\
\bottomrule
\end{tabular}
\end{table*}

\subsection{Exact Routing and Local-Scheduling Audits}
\label{app:subsection-16}
We compare the binary-flow oracle in
Proposition~\ref{prop:integer} with exhaustive legal-route
enumeration \citep{ahuja1993network}, then apply the
adjacent-swap scheduler in Algorithm~\ref{alg:swap}
\citep{korte2018combinatorial}. Table~\ref{tab:br-task5-oracle}
separates exact feasibility from local search outcomes.
Marker-free and cross-depth tasks admit zero-conflict
solutions, but some initializations terminate at positive
conflict. Capacity bottlenecks and restrictive graph
legality give infeasible cases.

\begin{table*}[htbp]
\centering
\caption{Exact route feasibility and local scheduling across seven structural cases at $B=4$. Flow 0? tests whether $C^*=0$ is feasible. Trapped counts initializations ending above $C^*$; $C_{\rm sched}$ and $D_{\rm pair}^{\rm sched}$ refer to the best-conflict run. Queries count legality-oracle calls, excluding initialization.}

\label{tab:br-task5-oracle}
\scriptsize
\setlength{\tabcolsep}{3.7pt}
\resizebox{\textwidth}{!}{%
\begin{tabular}{@{}l r l r c r r r r r r r@{}}
\toprule
Family & $m$ & Routes / task & Comb. & Flow 0? &
$C^*$ & $C_{\rm sched}$ & Trapped & Swaps & Queries &
$D_{\rm pair}^{\rm sched}$ & $D^*_{\rm risk}$ \\
\midrule
Capacity 2 & 4 & 2,2,2,2 & 16 & Yes &
0 & 0 & 0/16 & 0 & 0 & 0.0916 & 0.0916 \\
Capacity 1 & 4 & 1,1,1,1 & 1 & No &
6 & 6 & 0/1 & 0 & 4 & 0.7729 & 0.7729 \\
Marker-free & 2 & 8,8 & 64 & Yes &
0 & 0 & 24/64 & 0 & 0 & 0.4335 & 0.4335 \\
Same target & 2 & 8,8 & 64 & Yes &
0 & 0 & 0/64 & 0 & 0 & 0.4335 & 0.4335 \\
Cross-depth & 2 & 8,1 & 8 & Yes &
0 & 0 & 4/8 & 0 & 0 & 0.2259 & 0.2259 \\
Mixed $K$ & 6 & 1,8,8,1,1,1 & 64 & No &
2 & 2 & 32/64 & 1 & 9 & 0.3477 & 0.3477 \\
Graph legality & 2 & 1,1 & 1 & No &
1 & 1 & 0/1 & 0 & 2 & 0.5458 & 0.5458 \\
\bottomrule
\end{tabular}%
}

\vspace{2pt}\par\footnotesize
\end{table*}

\subsection{Exhaustive Structural Phase Diagram}
\label{app:subsection-17}
The exhaustive audit covers seven structural families,
16 configurations, and 340 legal route selections at $B=4$.
The families vary edge capacity, shared graph structure,
legality constraints, target identity, and edit depth,
covering both recoverable and structurally obstructed
settings. For each configuration, we enumerate all legal
shortest routes and determine whether a zero-conflict
selection exists using the exact routing formulation
and Theorem~\ref{thm:endpoint}.

We then apply the local scheduler to all 340 initial
selections, recording initial and final conflict scores,
zero-conflict recovery rates, local-search trapping,
and legality-oracle queries. This separates the
existence of a recoverable routing from the ability
of adjacent swaps to find one. To assess endpoint
behavior, we compute exact projected paired risks
and target-conditioned aliasing on 68 conflict-stratified
starts. These evaluations distinguish structural
conflict reduction from changes in endpoint fidelity. Figure~\ref{fig:task6-phase} summarizes the resulting
structural regimes, including oracle feasibility,
scheduler recovery, computational cost, and changes
in paired risk across the seven families. The results
show that zero-conflict feasibility does not guarantee
recovery by local scheduling, particularly in shared-skeleton
and cross-depth settings. Capacity and legality constraints
can prevent recovery altogether, while conflict reduction
does not necessarily imply a decrease in paired risk.

\begin{figure*}[htbp]
\centering
\includegraphics[width=\textwidth]{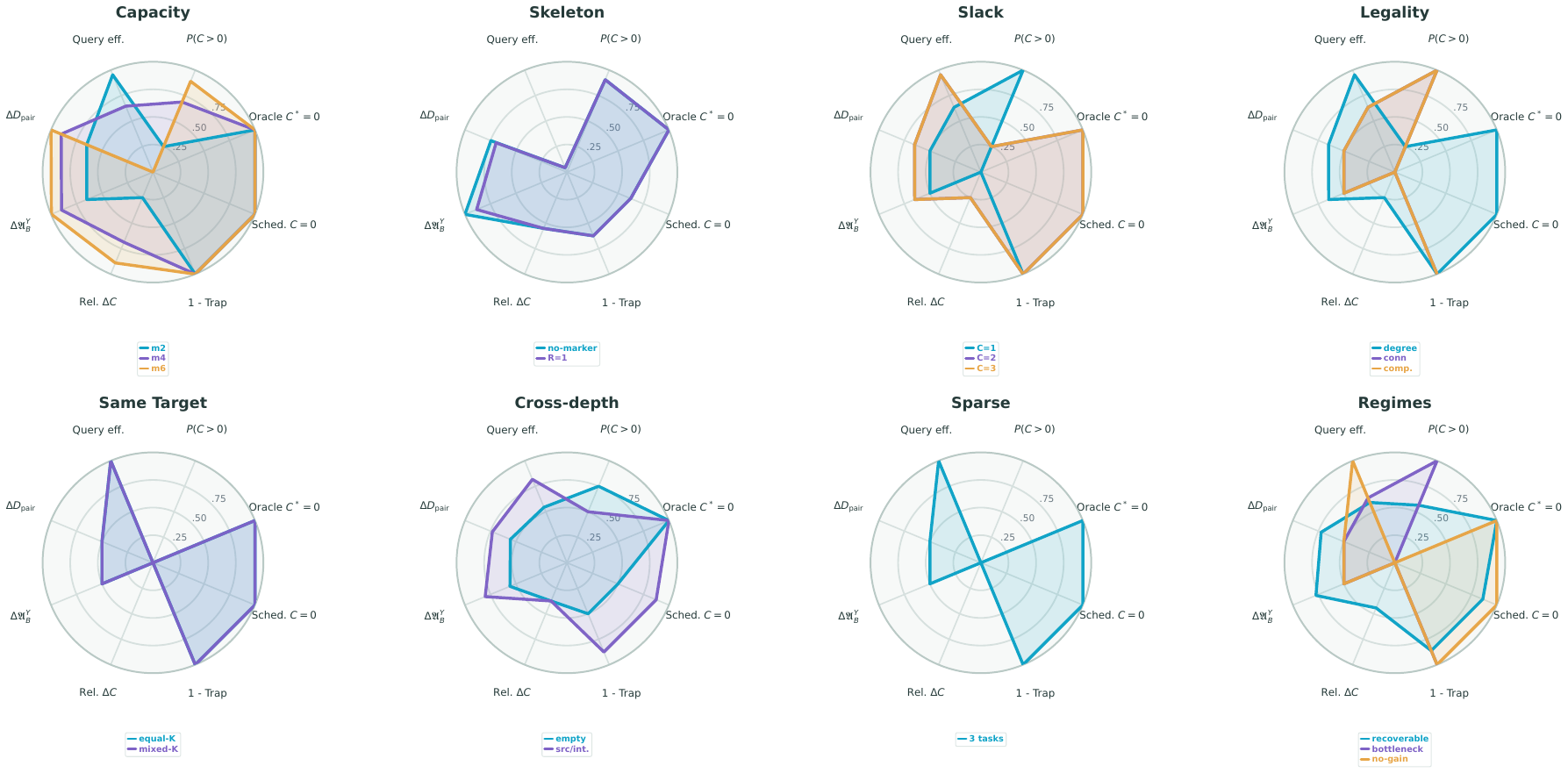}
\caption{Exhaustive structural audit at $B=4$. Eight panels compare conflict load, zero-conflict feasibility, scheduler recovery, query cost, aliasing, and paired-risk changes.}
\label{fig:task6-phase}
\end{figure*}

Table~\ref{tab:br-task6-phase} aggregates the results
by structural family, while
Table~\ref{tab:br-task6-config} reports the individual
configurations. In the capacity family, the scheduler
reaches zero conflict from every initialization, including
the six-task setting where $89.1\%$ of initial selections
contain conflicts. The shared-skeleton family behaves
differently: although zero-conflict routes exist,
the scheduler succeeds on only $62.5\%$ of its 128
initializations. The two shared-skeleton configurations
have identical recovery rates, indicating that adding
a preserved edge does not eliminate these local
search obstructions.

Cross-depth configurations also exhibit dependence
on initialization. The empty-state intersection
recovers in $50\%$ of runs, compared with $87.5\%$
for the source--intermediate intersection. Capacity
and graph-legality constraints yield different limits:
zero-conflict routing is feasible in two of the three
capacity-slack configurations and one of the three
graph-legality configurations. Same-target configurations
remain conflict-free across all 96 selections, consistent
with the target-separated criterion.

The conflict-stratified risk evaluations show reductions
in paired error across the sampled scheduler outcomes,
with no observed increases in these runs. Improvements
also occur when the scheduler stops above the globally
minimum conflict score, indicating that partial conflict
reduction can still improve endpoint fidelity.
However, the conflict objective is not a monotone
surrogate for endpoint risk: a separate legal adjacent
swap reduces $C$ from $3$ to $2$ while increasing
paired error from $0.603268$ to $0.637344$.
This distinction matters when comparing schedules
that remain structurally ambiguous after optimization.
The exact routing formulation therefore provides a
recoverability certificate, whereas local scheduling
supplies a feasible construction whose endpoint risk
should be evaluated separately.

\begin{table*}[htbp]
\centering
\caption{Family-level feasibility, scheduling, and projected-risk audit at $B=4$. Family values average over settings. Starts are exhaustive; Risk $n$ counts conflict-stratified risk evaluations. Changes are initial minus scheduled values. Trap denotes termination above $C^*$. Calls exclude route enumeration and MILP.}
\label{tab:br-task6-phase}
\scriptsize
\setlength{\tabcolsep}{2.6pt}
\resizebox{\textwidth}{!}{%
\begin{tabular}{lrrrrrrrrrrrrr}
\toprule
Family & Cfg. & Starts & Risk $n$ & $P(C>0)$
& Oracle-0 & $\overline{C^*}$ & Sched-0 & Trap
& $\Delta C_{\rm rel}$ & $\Delta\mathfrak A_B^Y$
& $\Delta D_{\rm pair}$ & Risk$\uparrow$ & Calls \\
\midrule
S1 Capacity & 3 & 84 & 16 & 0.609
& 1.000 & 0.00 & 1.000 & 0.000
& 0.609 & +0.2650 & +0.2650 & 0.000 & 2.56 \\
S2 Shared skeleton & 2 & 128 & 12 & 0.906
& 1.000 & 0.00 & 0.625 & 0.375
& 0.552 & +0.3376 & +0.1663 & 0.000 & 5.25 \\
S3 Capacity slack & 3 & 9 & 9 & 0.500
& 0.667 & 0.33 & 0.667 & 0.000
& 0.167 & +0.0757 & +0.0757 & 0.000 & 0.83 \\
S4 Graph legality & 3 & 6 & 6 & 0.750
& 0.333 & 0.67 & 0.333 & 0.000
& 0.083 & +0.0379 & +0.0379 & 0.000 & 1.42 \\
S5 Same target & 2 & 96 & 12 & 0.000
& 1.000 & 0.00 & 1.000 & 0.000
& 0.000 & +0.0000 & +0.0000 & 0.000 & 0.00 \\
S6 Cross-depth & 2 & 16 & 12 & 0.625
& 1.000 & 0.00 & 0.688 & 0.312
& 0.375 & +0.1373 & +0.1082 & 0.000 & 1.75 \\
S7 Sparse control & 1 & 1 & 1 & 0.000
& 1.000 & 0.00 & 1.000 & 0.000
& 0.000 & +0.0000 & +0.0000 & 0.000 & 0.00 \\
\bottomrule
\end{tabular}%
}
\end{table*}

The setting-level results further show how recovery rates
and paired-risk reductions vary within each structural family.

\begin{table*}[htbp]
\centering
\caption{Setting-level route enumeration and scheduling results at $B=4$. All starts enter the structural statistics. $\Delta D_{\rm pair}$ uses the conflict-stratified risk sample; positive values indicate reduced paired error.}
\label{tab:br-task6-config}
\scriptsize
\setlength{\tabcolsep}{3.5pt}
\resizebox{0.8\textwidth}{!}{%
\begin{tabular}{llcrrrrrr}
\toprule
Family & Config. & Routes/task & Starts & $C^*$
& $P(C>0)$ & Sched-0 & Trap
& $\Delta D_{\rm pair}$ \\
\midrule
S1 Capacity & $m=2$ & $[2,2]$ & 4 & 0
& 0.250 & 1.000 & 0.000 & $+0.11355$ \\
S1 Capacity & $m=4$ & $[2,2,2,2]$ & 16 & 0
& 0.688 & 1.000 & 0.000 & $+0.30281$ \\
S1 Capacity & $m=6$ & $[2,2,2,2,2,2]$ & 64 & 0
& 0.891 & 1.000 & 0.000 & $+0.37851$ \\
\midrule
S2 Shared skeleton & no-marker & $[8,8]$ & 64 & 0
& 0.906 & 0.625 & 0.375 & $+0.18554$ \\
S2 Shared skeleton & $R=1$ & $[8,8]$ & 64 & 0
& 0.906 & 0.625 & 0.375 & $+0.14715$ \\
\midrule
S3 Capacity slack & $C=1$ & $[1,1]$ & 1 & 1
& 1.000 & 0.000 & 0.000 & $+0.00000$ \\
S3 Capacity slack & $C=2$ & $[2,2]$ & 4 & 0
& 0.250 & 1.000 & 0.000 & $+0.11355$ \\
S3 Capacity slack & $C=3$ & $[2,2]$ & 4 & 0
& 0.250 & 1.000 & 0.000 & $+0.11355$ \\
\midrule
S4 Graph legality & degree & $[2,2]$ & 4 & 0
& 0.250 & 1.000 & 0.000 & $+0.11355$ \\
S4 Graph legality & conn. & $[1,1]$ & 1 & 1
& 1.000 & 0.000 & 0.000 & $+0.00000$ \\
S4 Graph legality & comp. & $[1,1]$ & 1 & 1
& 1.000 & 0.000 & 0.000 & $+0.00000$ \\
\midrule
S5 Same target & equal-$K$ & $[8,8]$ & 64 & 0
& 0.000 & 1.000 & 0.000 & $+0.00000$ \\
S5 Same target & mixed-$K$ & $[4,8]$ & 32 & 0
& 0.000 & 1.000 & 0.000 & $+0.00000$ \\
\midrule
S6 Cross-depth & empty & $[8,1]$ & 8 & 0
& 0.750 & 0.500 & 0.500 & $+0.04127$ \\
S6 Cross-depth & src/int. & $[8,1]$ & 8 & 0
& 0.500 & 0.875 & 0.125 & $+0.17522$ \\
\midrule
S7 Sparse control & 3 tasks & $[1,1,1]$ & 1 & 0
& 0.000 & 1.000 & 0.000 & $+0.00000$ \\
\bottomrule
\end{tabular}%
}

\end{table*}

\clearpage
\section{Datasets and Applicability Audits}
\label{app:benchmark-selection}

\subsection{Original Benchmark Screening}
\label{app:subsection-18}

We screen ten graph data sources under the original fixed-pair,
legal shortest-edit protocol. The sources span temporal interaction
networks \citep{leskovec2014snap,paranjape2017motifs,huang2023temporal},
observed graph snapshots, reaction editing, and controlled tasks on
TU graph topologies \citep{morris2020tudataset}.
Table~\ref{tab:task7-ten} reports dataset size, median edit length,
clock calibration, conflict prevalence, and recovery certificates.
The original shortlist consists of Email-Eu, Wiki-Talk, Reddit body,
PROTEINS, and COLLAB. It is distinct from the unified controlled
benchmark used for graph-only student training in
Table~\ref{tab:exp-r2}.

In the table, $C^+$ counts initially conflicting cohorts, Cert. gives
training/validation certificates, and New Cert. counts certificates
added by the supplementary oracle recheck. Oracle C/U records newly
certified and still-unresolved cohorts. Raw sizes use the respective
dataset units, and screening uses no test targets. The recheck resolves
the previously unresolved Email-Eu and Reddit cases; Wiki-Talk retains
one unresolved cohort after 24 restarts.

\begin{table*}[htbp]
\centering
\caption{Original ten-source applicability audit. Panel A reports dataset statistics and clock calibration; Panel B reports structural diagnostics and oracle rechecks. T/V denotes training/validation.}
\label{tab:task7-ten}
\small
\setlength{\tabcolsep}{4pt}
\renewcommand{\arraystretch}{1.14}

\textbf{A. Dataset statistics and clock calibration}\par\vspace{4pt}
\begin{tabular*}{\textwidth}{@{\extracolsep{\fill}}lrrrrrrr@{}}
\toprule
Dataset & Raw & Train & Val & Test src. & $K_{50}$ & $B_{\mathrm{mid}}$ & $\varepsilon_{\mathrm{val}}$ \\
\midrule
\multicolumn{8}{@{}l}{\textit{Temporal event networks}} \\
Email-Eu~\citep{leskovec2014snap} & 332,334 & 1,177 & 391 & 412 & 15 & 23.068 & 0.210 \\
Wiki-Talk~\citep{leskovec2014snap} & 7,833,140 & 563 & 163 & 157 & 2 & 4.325 & 0.239 \\
Reddit body~\citep{rossi2020temporal} & 286,561 & 1,865 & 632 & 687 & 6 & 9.368 & 0.203 \\
TGB Coin-v2~\citep{huang2023temporal} & 22,809,487 & 1,784 & 670 & 590 & 6 & 11.386 & 0.363 \\
TGB Flight-v2~\citep{huang2023temporal} & 67,169,570 & 2,047 & 684 & 720 & 21 & 31.914 & 0.161 \\
\addlinespace[3pt]
\multicolumn{8}{@{}l}{\textit{Natural observed snapshots}} \\
CAIDA AS Relationships~\citep{luckie2013asrelationships} & 122 & 1,632 & 618 & 738 & 2 & 4.561 & 0.254 \\
\addlinespace[3pt]
\multicolumn{8}{@{}l}{\textit{Native reaction editing}} \\
USPTO-15K~\citep{jin2017reaction} & 15,000 & 9,700 & 1,394 & 2,864 & 2 & 2.853 & 0.260 \\
USPTO-50K~\citep{liu2017retrosynthetic} & 50,016 & 25,963 & 3,255 & 4,955 & 1 & 1.394 & 0.249 \\
\addlinespace[3pt]
\multicolumn{8}{@{}l}{\textit{Controlled graph-edit tasks}} \\
\textbf{PROTEINS}~\citep{morris2020tudataset} & 1,113 & 324 & 110 & 114 & 4 & 5.109 & 0.250 \\
\textbf{COLLAB}~\citep{morris2020tudataset} & 5,000 & 260 & 94 & 92 & 4 & 5.109 & 0.250 \\
\bottomrule
\end{tabular*}

\vspace{12pt}

\textbf{B. Structural diagnostics and oracle recheck}
\par\vspace{5pt}

\begingroup
\footnotesize
\setlength{\tabcolsep}{2.5pt}
\renewcommand{\arraystretch}{1.12}

\begin{tabularx}{\linewidth}{
@{}lrrrrr>{\raggedright\arraybackslash\scriptsize}X@{}
}
\toprule
Dataset
& \shortstack{$K\ge2$\\T/V}
& \shortstack{$C^+$\\T/V}
& \shortstack{Cert.\\T/V (updated)}
& \shortstack{New Cert.\\T/V}
& \shortstack{Oracle\\C/U}
& Original gate \\
\midrule

Email-Eu
& 1,111/370 & 161/57
& \textbf{1,036/331} & 20/0 & 1/0
& HOLD-ORACLE \\

Wiki-Talk
& 356/106 & 466/139
& 143/46 & 0/0 & 0/1
& HOLD-ORACLE \\

Reddit body
& 1,781/615 & 1,436/502
& \textbf{790/281} & 250/107 & 10/0
& HOLD-ORACLE \\

TGB Coin-v2
& 1,646/660 & 575/130
& 884/501 & -- & --
& HOLD-ATTEST \\

TGB Flight-v2
& 2,038/684 & 26/6
& 1,944/684 & -- & --
& LOW-COLLISION \\

\addlinespace[3pt]

CAIDA AS Relationships
& 1,121/438 & 729/288
& 0/0 & -- & --
& HOLD-PILOT \\

\addlinespace[3pt]

USPTO-15K
& 7,816/1,139 & -- & --
& -- & --
& HOLD-CROSS-REACTION-IDENTITY \\

USPTO-50K
& 140/10 & -- & --
& -- & --
& HOLD-CROSS-REACTION-IDENTITY \\

\addlinespace[3pt]

\textbf{PROTEINS}
& 324/110 & 324/110
& 324/110 & -- & --
& CONTROLLED-VALID \\

\textbf{COLLAB}
& 260/94 & 260/94
& 260/94 & -- & --
& CONTROLLED-VALID \\

\bottomrule
\end{tabularx}
\endgroup
\end{table*}

\subsection{Observed-Pair Structural Diagnostics}
\label{app:subsection-19}

We next examine the original five-source structural audit conducted
before student training. Email-Eu, Wiki-Talk, and Reddit body use
observed-window occurrence-graph pairs (R1); PROTEINS and COLLAB use
controlled crossed edits on native graph topologies (R2).
The frozen protocol measures route lengths, initial cross-target
conflicts, and legality-preserving conflict descent.
Figure~\ref{fig:task8-structural} and
Table~\ref{tab:task8-pilot} report the corresponding diagnostics.

Here $C^+$ denotes initially conflicting training/validation pairs.
$\Delta C_{\mathrm{val}}$ is the fractional conflict reduction relative
to seeded random legal shortest routes. Zero-conflict validation
coverage is the fraction of validation pairs in cohorts reaching
$C=0$ after bounded multi-seed scheduling. Both are structural
metrics rather than paired endpoint errors.

\begin{figure*}[htbp]
\centering
\includegraphics[width=\textwidth]{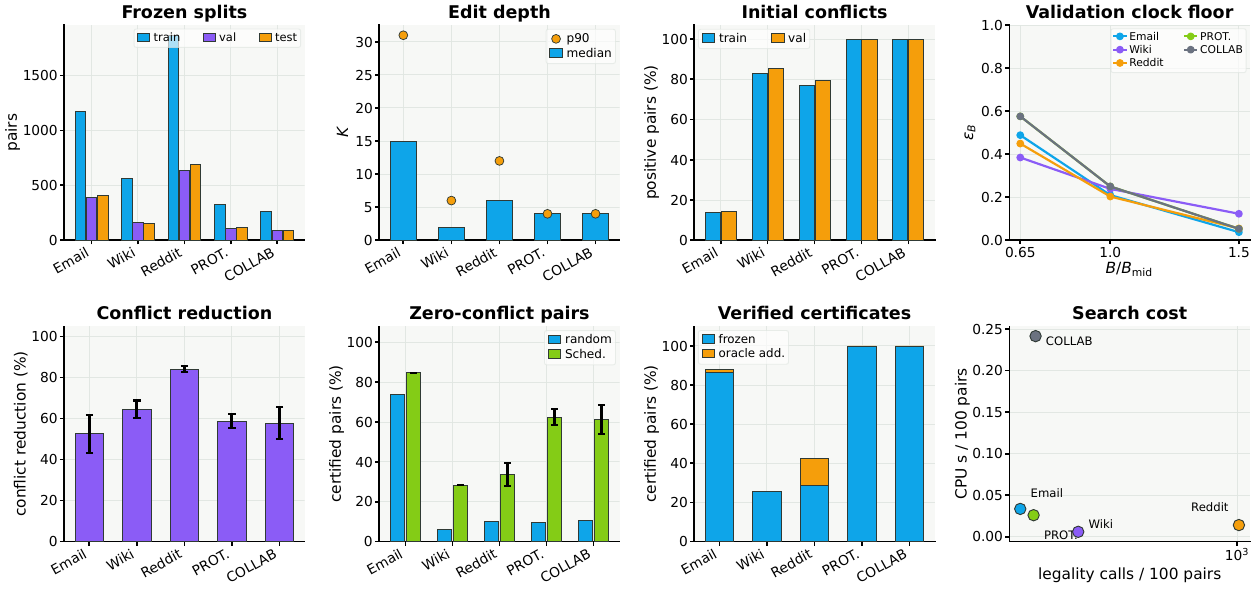}
\caption{Pre-student structural diagnostics under the original five-source protocol. Panels show edit depths, cross-target conflicts, scheduling recovery, and computational cost.}
\label{fig:task8-structural}
\end{figure*}
\begin{table*}[htbp]
\centering
\caption{Original pre-student R1/R2 structural audit with frozen train/validation/test splits. Conflict reduction and zero-conflict coverage are mean $\pm$ sample SD across route seeds.}
\label{tab:task8-pilot}
\scriptsize
\setlength{\tabcolsep}{4pt}
\renewcommand{\arraystretch}{1.12}
\begin{tabular*}{\textwidth}{@{\extracolsep{\fill}}llrrrrrrr@{}}
\toprule
Dataset & Scope & Train & Val & Test src. &
$K_{50}$ & $C^+$ T/V &
$\Delta C_{\mathrm{val}}$ &
Zero-conf. val. \\
\midrule
Email-Eu
& R1 observed
& 1,177 & 391 & 412 & 15 & 161/57
& $54.4\pm8.2\%$
& $84.7\pm0.0\%$ \\

Wiki-Talk
& R1 observed
& 563 & 163 & 157 & 2 & 466/139
& $66.3\pm2.5\%$
& $28.2\pm0.0\%$ \\

Reddit body
& R1 observed
& 1,865 & 632 & 687 & 6 & 1,436/502
& $83.0\pm1.1\%$
& $32.0\pm7.4\%$ \\

\addlinespace[3pt]

PROTEINS
& R2 controlled
& 324 & 110 & 114 & 4 & 324/110
& $57.5\pm4.0\%$
& $61.1\pm4.0\%$ \\

COLLAB
& R2 controlled
& 260 & 94 & 92 & 4 & 260/94
& $59.1\pm5.6\%$
& $63.0\pm3.9\%$ \\

\bottomrule
\end{tabular*}
\end{table*}

Across the five sources, mean validation conflict reductions range
from $54.4\%$ to $83.0\%$, while zero-conflict coverage ranges from
$28.2\%$ to $84.7\%$. Thus conflict reduction and complete recovery
differ under the fixed search budget. The two protocols use different pair constructions
and should be compared separately. The corresponding R1
endpoint results appear in Table~\ref{tab:exp-r1}.

\subsection{Multi-Regime Controlled Graph Edits}
\label{app:r2-controlled}

We evaluate five controlled regimes derived from Email-Eu, Wiki-Talk,
Reddit body, PROTEINS, and COLLAB. All targets are constructed by
public source-derived rules under frozen group splits. The regimes
are marker-free edits, motif repair, capacity bottlenecks, relaxed
replacements, and extended motif edits, with lengths
$K\in\{2,4,6\}$. Structural certificates concern individual
two-task cohorts, separate from the fixed $K=4$ training experiment
in Appendix~\ref{app:controlled-five-student} and
Table~\ref{tab:exp-r2}.

Table~\ref{tab:task9-r2} reports cohort counts and recoverability.
R2 T/V counts generated train/validation pairs (two per cohort);
Keep V is the fraction of eligible validation source anchors;
$R$ edges measures retained source-backbone edges. Cert/No/U counts
certified zero-conflict solutions, proven infeasibility, and unresolved
cases. $\Delta C$ is the reduction relative to random legal routes,
Sched-0 the fraction reaching zero conflict, and Recov.+ the recovery
rate among initially conflicting cohorts. Scheduling uses five route
seeds and at most 256 adjacent checks per cohort; each regime sets
its clock from training edit lengths to $\varepsilon_B=0.25$.

\begin{table*}[htbp]
\centering
\caption{Structural feasibility and scheduling across five controlled edit regimes. Results use source-derived pairs and frozen group splits; scheduling statistics average over five route seeds.}
\label{tab:task9-r2}
\scriptsize
\setlength{\tabcolsep}{3pt}
\renewcommand{\arraystretch}{1.09}
\resizebox{\textwidth}{!}{%
\begin{tabular}{@{}llrrrrrrrrr@{}}
\toprule
Dataset & Regime & $K$ & R2 T/V & Keep V (\%) &
$R$ edges (\%) & $C^+$ V (\%) & Cert/No/U &
$\Delta C$ V (\%) & Sched-0 V (\%) & Recov.+ (\%) \\
\midrule
Email-Eu & Marker-free & 4 & 2,240/744 & 95.1 & 99.4 & 89.7 & 372/0/0 & $60.1\pm2.4$ & $62.8\pm2.0$ & $58.5\pm2.6$ \\
Email-Eu & Motif repair & 4 & 848/272 & 34.8 & 85.7 & 92.4 & 107/29/0 & $41.4\pm3.3$ & $48.5\pm4.4$ & $44.2\pm5.5$ \\
Email-Eu & Capacity bottleneck & 2 & 2,160/704 & 90.0 & 83.1 & 100.0 & 0/352/0 & $0.0\pm0.0$ & $0.0\pm0.0$ & $0.0\pm0.0$ \\
Email-Eu & Relaxed replacement & 2 & 2,160/704 & 90.0 & 83.1 & 25.1 & 352/0/0 & $100.0\pm0.0$ & $100.0\pm0.0$ & $100.0\pm0.0$ \\
Email-Eu & Extended motif & 6 & 1,104/336 & 43.0 & 78.4 & 63.2 & 164/4/0 & $62.7\pm1.3$ & $77.5\pm2.2$ & $64.4\pm3.3$ \\
\addlinespace[3pt]
Wiki-Talk & Marker-free & 4 & 1,004/278 & 85.3 & 98.6 & 88.3 & 139/0/0 & $62.8\pm5.0$ & $66.5\pm4.2$ & $62.0\pm5.3$ \\
Wiki-Talk & Motif repair & 4 & 74/20 & 6.1 & 59.1 & 96.0 & 5/5/0 & $20.9\pm9.5$ & $24.0\pm8.9$ & $20.4\pm11.8$ \\
Wiki-Talk & Capacity bottleneck & 2 & 290/62 & 19.0 & 53.0 & 100.0 & 0/31/0 & $0.0\pm0.0$ & $0.0\pm0.0$ & $0.0\pm0.0$ \\
Wiki-Talk & Relaxed replacement & 2 & 290/62 & 19.0 & 53.0 & 25.2 & 31/0/0 & $100.0\pm0.0$ & $100.0\pm0.0$ & $100.0\pm0.0$ \\
Wiki-Talk & Extended motif & 6 & 70/18 & 5.5 & 44.9 & 66.7 & 9/0/0 & $21.0\pm9.2$ & $48.9\pm9.9$ & $22.8\pm11.5$ \\
\addlinespace[3pt]
Reddit body & Marker-free & 4 & 3,656/1,234 & 97.6 & 99.3 & 89.0 & 617/0/0 & $60.1\pm2.0$ & $63.2\pm1.9$ & $58.6\pm2.1$ \\
Reddit body & Motif repair & 4 & 424/108 & 8.5 & 70.6 & 95.6 & 21/33/0 & $14.3\pm3.6$ & $20.4\pm4.3$ & $16.7\pm3.2$ \\
Reddit body & Capacity bottleneck & 2 & 2,490/756 & 59.8 & 58.9 & 100.0 & 0/378/0 & $0.0\pm0.0$ & $0.0\pm0.0$ & $0.0\pm0.0$ \\
Reddit body & Relaxed replacement & 2 & 2,490/756 & 59.8 & 58.9 & 25.9 & 378/0/0 & $100.0\pm0.0$ & $100.0\pm0.0$ & $100.0\pm0.0$ \\
Reddit body & Extended motif & 6 & 514/126 & 10.0 & 53.1 & 76.5 & 58/5/0 & $23.0\pm3.9$ & $39.4\pm4.4$ & $20.6\pm6.9$ \\
\addlinespace[3pt]
PROTEINS & Marker-free & 4 & 324/110 & 100.0 & 97.1 & 88.0 & 55/0/0 & $61.1\pm4.3$ & $65.1\pm3.3$ & $60.3\pm3.9$ \\
PROTEINS & Motif repair & 4 & 288/90 & 81.8 & 95.3 & 88.4 & 45/0/0 & $62.7\pm5.5$ & $66.2\pm4.3$ & $61.7\pm5.2$ \\
PROTEINS & Capacity bottleneck & 2 & 324/110 & 100.0 & 95.1 & 100.0 & 0/55/0 & $0.0\pm0.0$ & $0.0\pm0.0$ & $0.0\pm0.0$ \\
PROTEINS & Relaxed replacement & 2 & 324/110 & 100.0 & 95.1 & 23.3 & 55/0/0 & $100.0\pm0.0$ & $100.0\pm0.0$ & $100.0\pm0.0$ \\
PROTEINS & Extended motif & 6 & 322/110 & 100.0 & 92.6 & 60.4 & 55/0/0 & $91.1\pm4.9$ & $94.5\pm3.4$ & $90.8\pm6.0$ \\
\addlinespace[3pt]
COLLAB & Marker-free & 4 & 260/94 & 100.0 & 99.6 & 90.2 & 47/0/0 & $57.5\pm7.4$ & $60.9\pm5.1$ & $56.4\pm7.3$ \\
COLLAB & Motif repair & 4 & 74/32 & 34.0 & 98.7 & 90.0 & 16/0/0 & $68.0\pm15.9$ & $71.2\pm16.3$ & $68.7\pm16.4$ \\
COLLAB & Capacity bottleneck & 2 & 128/48 & 51.1 & 98.9 & 100.0 & 0/24/0 & $0.0\pm0.0$ & $0.0\pm0.0$ & $0.0\pm0.0$ \\
COLLAB & Relaxed replacement & 2 & 128/48 & 51.1 & 98.9 & 26.7 & 24/0/0 & $100.0\pm0.0$ & $100.0\pm0.0$ & $100.0\pm0.0$ \\
COLLAB & Extended motif & 6 & 106/40 & 42.6 & 98.2 & 54.0 & 20/0/0 & $89.7\pm7.6$ & $93.0\pm5.7$ & $88.3\pm9.1$ \\
\bottomrule
\end{tabular}%
}
\end{table*}

Capacity bottlenecks have no zero-conflict solutions in these
cohorts, whereas relaxed replacements reach zero conflict throughout.
Motif-repair feasibility varies with topology: the Email-Eu,
Wiki-Talk, and Reddit body validation cohorts have, respectively,
$107/29$, $5/5$, and $21/33$ certified/infeasible cases.
Extended motifs show a wider range of recovery rates. These results
characterize structural feasibility and local scheduling, not
learned-student endpoint accuracy.

\subsection{Molecular-Pair Applicability}
\label{app:r3-molecular}

We examine native source--target SMILES pairs in the QED, LogP,
and DRD2 training corpora. The fixed-edge model in Section~\ref{sec:setting} requires a
consistent atom inventory and an identified correspondence
between source and target atoms before a shortest edit route
can be defined. Table~\ref{tab:task10-molecular-audit}
reports the sequential eligibility checks.

Atoms denotes matching element, isotope, and formal-charge
inventories; H inv. additionally checks hydrogen counts associated
with atom attributes. Det. $s$ measures sources with a unique
canonical target, and Pre-map counts one hydrogen-compatible
representative per deterministic source. Later gates require
externally verified atom identities: Binary checks nonaromatic
single-bond toggles with unchanged retained bonds, Legal checks
the strict chemistry predicate, and Route checks a legal monotone
shortest path. These gates are evaluated sequentially, so
passing an earlier check does not establish route feasibility.
The identity requirement ensures that bond edits refer to
consistent atom pairs across both endpoints.
This separates molecular attribute compatibility from the
existence of a valid graph-edit transformation.
Raw split is a separate five-seed heuristic check of
an 80/10/10 component-based partition, not an established split.

\begin{table*}[htbp]
\centering
\caption{Native molecular-pair applicability in train-only QED, LogP, and DRD2 corpora. Percentages use nonblank input lines as the denominator; dashes indicate unassessed stages.}
\label{tab:task10-molecular-audit}
\scriptsize
\setlength{\tabcolsep}{3.1pt}
\renewcommand{\arraystretch}{1.12}
\resizebox{\textwidth}{!}{%
\begin{tabular}{@{}lrrrrrrrrrrr@{}}
\toprule
Dataset & Pairs & Parsed (\%) & Atoms (\%) & H inv. (\%) & Det. $s$ (\%) & Pre-map (\%) & ID ext. (\%) & Binary (\%) & Legal (\%) & Route (\%) & Raw split \\
\midrule
QED~\citep{jin2019graph2graph} & 88,306 & 100.0 & 0.4 & 0.0 & 12.4 & 0.0 & -- & -- & -- & -- & 0/5 \\
LogP~\citep{jin2019graph2graph} & 99,909 & 100.0 & 0.2 & $<0.1$ & 16.9 & $<0.1$ & -- & -- & -- & -- & 5/5 \\
DRD2~\citep{jin2019graph2graph} & 34,404 & 100.0 & 0.3 & $<0.1$ & 31.7 & $<0.1$ & -- & -- & -- & -- & 0/5 \\
\bottomrule
\end{tabular}}
\end{table*}
\begin{figure*}[htbp]
\centering
\includegraphics[width=\textwidth]{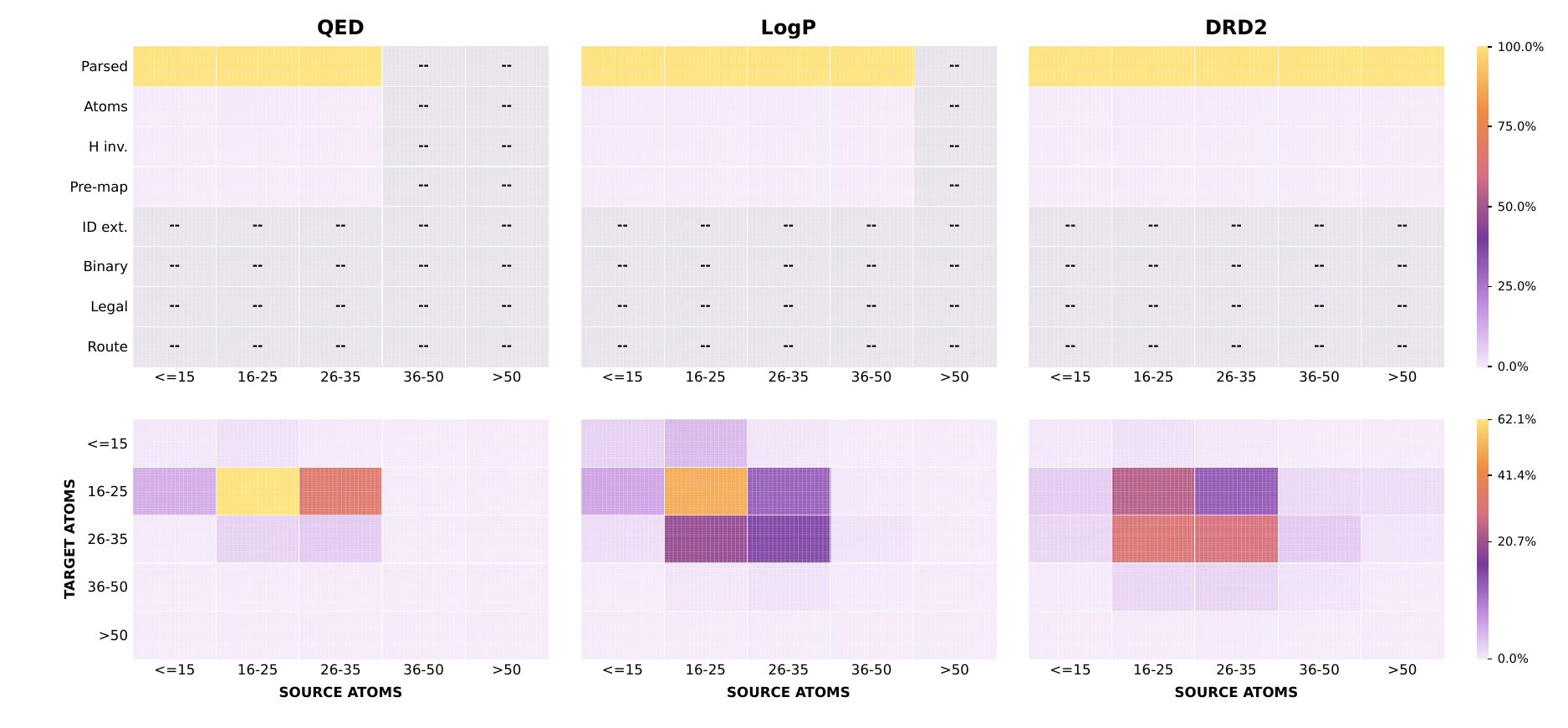}
\caption{Molecular-pair eligibility checks by source atom count (top) and source--target atom-count distributions (bottom).}
\label{fig:task10-molecular}
\end{figure*}

Only $0.2$--$0.4\%$ of input rows preserve the required atom-attribute
inventory, with fewer satisfying the hydrogen constraint. In the
absence of independently verified source--target atom maps,
identity-dependent editing checks remain unassessed.
Figure~\ref{fig:task10-molecular} shows the associated atom-count
diagnostics. The audit establishes the scope of direct molecular
applicability without measuring generative performance.

\clearpage
\section{Ordering Baselines and Observed-Pair Evaluation}
\label{app:order-baselines}

\subsection{Matched Legal-Ordering Baselines}
\label{app:subsection-20}

We compare legal edit-ordering controls on the original five-source,
multiregime benchmark in Appendix~\ref{app:r2-controlled}.
The unit is a two-task validation cohort; paired endpoints, mandatory
edits, legal-state constraints, shortest lengths, and the
train-calibrated Poisson clock are fixed within each regime.
The comparison includes random legal orders, fixed delete-first and
add-first orders, uniform sampling over feasible shortest orders,
restarts, greedy construction, neutral edits, strict conflict descent,
and target-aliasing minimization. The last two use the conflict
score in~\eqref{eq:collision} and the ambiguity term
in~\eqref{eq:qy}, respectively.
The Learned control adapts LoFlexMDM-style stochastic ordering
\citep{patel2026insertion} to legal edit permutations, using
train-only REINFORCE updates \citep{williams1992reinforce}.
It predicts edit order, not the graph-only action distribution
$q_\theta(a\mid x,u)$ of Section~\ref{sec:learning}.

Table~\ref{tab:task11-order-baselines}
reports the complete matched-order audit across five data sources.
Here $n$ is the number of validation cohorts; Valid is the percentage for which the method
produces legal shortest routes for both tasks, and $C=0$ is the
zero-conflict fraction conditional on validity. The same
prespecified validation subset is used for numerical risk evaluation
within each regime: $n_D$ counts eligible cohorts, and
$D^{\rm proj}$ is the exact projected paired error computed from
\eqref{eq:kernel}--\eqref{eq:exactrisk}. We report
$\Delta D=D_{\rm pair,Random}^{\rm proj}-D_{\rm pair,method}^{\rm proj}$
in percentage points, restricted to method-eligible cohorts.
Risk and conflict rates are mean $\pm$ sample SD across matched
route seeds; $D^{\rm proj}$ and Valid are percentages.
A dash denotes an unassessed entry rather than zero risk.

The Calls column counts legality-membership queries per cohort,
including failed probes and the allocated share of route setup;
CPU ms measures per-cohort processing, excluding endpoint-risk
ODE scoring. Selected ordering controls are query-matched to the scheduler
from the same random initialization, but not matched on CPU time.
The Oracle rows are separate CP-SAT feasibility checks
\citep{perron2023cpsat} on smaller selected sets, reporting
certified/infeasible/unresolved counts (CP C/N/U); they do not
report projected risk. All evaluated policies use the same
train-calibrated clock floor $\varepsilon_B=0.25$.

\begingroup
\setlength{\tabcolsep}{0.4pt}
\renewcommand{\arraystretch}{1.07}
\scriptsize
\begin{longtable}{@{}lllr r r r r r r r l@{}}

\caption{Matched legal-order baselines on five data sources under the original multiregime validation protocol.}\label{tab:task11-order-baselines}\\

\toprule

Data & Regime & Method & $n$ & Valid & $C=0$ & $n_D$ & $D^{\rm proj}$ & $\Delta D$ & Calls & CPU ms & CP C/N/U \\

\midrule

\endfirsthead

\multicolumn{12}{l}{\textit{Task 11 matched-order baselines (continued)}}\\

\toprule

Data & Regime & Method & $n$ & Valid & $C=0$ & $n_D$ & $D^{\rm proj}$ & $\Delta D$ & Calls & CPU ms & CP C/N/U \\

\midrule

\endhead

\midrule\multicolumn{12}{r}{\textit{Continued on next page}}\\

\endfoot

\bottomrule

\endlastfoot

Email-Eu & Marker-free & Random & 372 & 100.0 & $10.8\pm1.7$ & 5 & $61.9\pm6.0$ & $0.0\pm0.0$ & 36.0 & 0.34 & -- \\

 &  & Delete &  & 100.0 & $0.0\pm0.0$ & 5 & $62.5\pm0.0$ & $-0.6\pm6.0$ & 8.0 & 0.02 & -- \\

 &  & Add &  & 0.0 & \textemdash & 0 & \textemdash & \textemdash & 2.0 & 0.01 & -- \\

 &  & Uniform &  & 100.0 & $9.0\pm1.8$ & 5 & $60.2\pm11.3$ & $1.7\pm8.8$ & 36.0 & 0.35 & -- \\

 &  & Restart &  & 100.0 & $11.7\pm1.8$ & 5 & $61.9\pm6.0$ & $0.0\pm0.0$ & 50.2 & 0.45 & -- \\

 &  & Greedy &  & 100.0 & $91.6\pm0.6$ & 5 & $26.5\pm3.4$ & $35.4\pm7.9$ & 40.0 & 0.38 & -- \\

 &  & Neutral &  & 100.0 & $10.8\pm1.7$ & 5 & $64.5\pm5.8$ & $-2.7\pm4.1$ & 57.5 & 0.46 & -- \\

 &  & \textbf{Sched.} &  & 100.0 & $63.2\pm1.5$ & 5 & $41.7\pm3.9$ & $20.1\pm6.0$ & 65.5 & 0.48 & -- \\

 &  & Alias &  & 100.0 & $32.0\pm1.3$ & 5 & $52.3\pm5.2$ & $9.6\pm6.2$ & 57.3 & 0.60 & -- \\

 &  & Learned &  & 100.0 & $8.6\pm1.4$ & 5 & $55.6\pm7.8$ & $6.3\pm9.6$ & 36.0 & 1.03 & -- \\

 &  & Oracle & 12 & -- & -- & -- & -- & -- & -- & 3.526 & 12/0/0 \\

\addlinespace[3pt]

Email-Eu & Motif repair & Random & 136 & 100.0 & $8.8\pm2.1$ & 5 & $57.0\pm6.4$ & $0.0\pm0.0$ & 35.1 & 0.42 & -- \\

 &  & Delete &  & 76.5 & $0.0\pm0.0$ & 3 & $65.5\pm0.0$ & $-6.8\pm4.5$ & 7.0 & 0.07 & -- \\

 &  & Add &  & 0.0 & \textemdash & 0 & \textemdash & \textemdash & 2.0 & 0.01 & -- \\

 &  & Uniform &  & 100.0 & $6.3\pm2.5$ & 5 & $59.8\pm2.2$ & $-2.7\pm7.9$ & 35.1 & 0.43 & -- \\

 &  & Restart &  & 100.0 & $9.7\pm2.0$ & 5 & $57.0\pm6.4$ & $0.0\pm0.0$ & 49.3 & 0.62 & -- \\

 &  & Greedy &  & 100.0 & $72.9\pm2.7$ & 5 & $35.1\pm2.4$ & $22.0\pm5.6$ & 39.1 & 0.48 & -- \\

 &  & Neutral &  & 100.0 & $8.8\pm2.1$ & 5 & $58.1\pm5.5$ & $-1.1\pm4.0$ & 56.6 & 0.70 & -- \\

 &  & \textbf{Sched.} &  & 100.0 & $48.1\pm4.9$ & 5 & $48.9\pm7.9$ & $8.2\pm6.1$ & 64.7 & 0.77 & -- \\

 &  & Alias &  & 100.0 & $24.9\pm3.8$ & 5 & $51.4\pm7.7$ & $5.6\pm3.7$ & 56.5 & 0.69 & -- \\

 &  & Learned &  & 100.0 & $6.5\pm3.2$ & 5 & $59.7\pm5.2$ & $-2.6\pm4.4$ & 35.1 & 1.09 & -- \\

 &  & Oracle & 12 & -- & -- & -- & -- & -- & -- & 2.962 & 10/2/0 \\

\addlinespace[3pt]

Email-Eu & Tight cap. & Random & 352 & 100.0 & $0.0\pm0.0$ & 5 & $62.5\pm0.0$ & $0.0\pm0.0$ & 10.0 & 0.12 & -- \\

 &  & Delete &  & 100.0 & $0.0\pm0.0$ & 5 & $62.5\pm0.0$ & $0.0\pm0.0$ & 4.0 & 0.01 & -- \\

 &  & Add &  & 0.0 & \textemdash & 0 & \textemdash & \textemdash & 2.0 & 0.01 & -- \\

 &  & Uniform &  & 100.0 & $0.0\pm0.0$ & 5 & $62.5\pm0.0$ & $0.0\pm0.0$ & 10.0 & 0.12 & -- \\

 &  & Restart &  & 100.0 & $0.0\pm0.0$ & 5 & $62.5\pm0.0$ & $0.0\pm0.0$ & 16.0 & 0.19 & -- \\

 &  & Greedy &  & 100.0 & $0.0\pm0.0$ & 5 & $62.5\pm0.0$ & $0.0\pm0.0$ & 12.0 & 0.13 & -- \\

 &  & Neutral &  & 100.0 & $0.0\pm0.0$ & 5 & $62.5\pm0.0$ & $0.0\pm0.0$ & 20.0 & 0.20 & -- \\

 &  & \textbf{Sched.} &  & 100.0 & $0.0\pm0.0$ & 5 & $62.5\pm0.0$ & $0.0\pm0.0$ & 24.0 & 0.21 & -- \\

 &  & Alias &  & 100.0 & $0.0\pm0.0$ & 5 & $62.5\pm0.0$ & $0.0\pm0.0$ & 20.0 & 0.19 & -- \\

 &  & Learned &  & 100.0 & $0.0\pm0.0$ & 5 & $62.5\pm0.0$ & $0.0\pm0.0$ & 10.0 & 0.41 & -- \\

 &  & Oracle & 12 & -- & -- & -- & -- & -- & -- & 0.807 & 0/12/0 \\

\addlinespace[3pt]

Email-Eu & Relaxed repl. & Random & 352 & 100.0 & $74.9\pm2.2$ & 5 & $32.5\pm0.0$ & $0.0\pm0.0$ & 10.0 & 0.14 & -- \\

 &  & Delete &  & 100.0 & $0.0\pm0.0$ & 5 & $62.5\pm0.0$ & $-30.0\pm0.0$ & 4.0 & 0.01 & -- \\

 &  & Add &  & 100.0 & $100.0\pm0.0$ & 5 & $25.0\pm0.0$ & $7.5\pm0.0$ & 4.0 & 0.01 & -- \\

 &  & Uniform &  & 100.0 & $76.4\pm0.7$ & 5 & $40.0\pm7.5$ & $-7.5\pm7.5$ & 10.0 & 0.14 & -- \\

 &  & Restart &  & 100.0 & $74.9\pm2.2$ & 5 & $32.5\pm0.0$ & $0.0\pm0.0$ & 14.3 & 0.19 & -- \\

 &  & Greedy &  & 100.0 & $100.0\pm0.0$ & 5 & $25.0\pm0.0$ & $7.5\pm0.0$ & 12.0 & 0.15 & -- \\

 &  & Neutral &  & 100.0 & $74.9\pm2.2$ & 5 & $32.5\pm0.0$ & $0.0\pm0.0$ & 18.3 & 0.20 & -- \\

 &  & \textbf{Sched.} &  & 100.0 & $100.0\pm0.0$ & 5 & $25.0\pm0.0$ & $7.5\pm0.0$ & 22.3 & 0.22 & -- \\

 &  & Alias &  & 100.0 & $100.0\pm0.0$ & 5 & $25.0\pm0.0$ & $7.5\pm0.0$ & 18.3 & 0.20 & -- \\

 &  & Learned &  & 100.0 & $99.7\pm0.3$ & 5 & $25.0\pm0.0$ & $7.5\pm0.0$ & 10.0 & 0.47 & -- \\

 &  & Oracle & 12 & -- & -- & -- & -- & -- & -- & 1.238 & 12/0/0 \\

\addlinespace[3pt]

Email-Eu & Ext. motif & Random & 168 & 100.0 & $35.8\pm3.2$ & 5 & $45.5\pm5.7$ & $0.0\pm0.0$ & 118.0 & 1.26 & -- \\

 &  & Delete &  & 62.5 & $0.0\pm0.0$ & 4 & $63.7\pm0.0$ & $-22.1\pm5.1$ & 9.0 & 0.10 & -- \\

 &  & Add &  & 0.0 & \textemdash & 0 & \textemdash & \textemdash & 2.0 & 0.02 & -- \\

 &  & Uniform &  & 100.0 & $32.5\pm5.0$ & 5 & $43.0\pm5.9$ & $2.5\pm11.0$ & 118.0 & 1.28 & -- \\

 &  & Restart &  & 100.0 & $37.6\pm3.1$ & 5 & $45.5\pm5.7$ & $0.0\pm0.0$ & 136.3 & 1.49 & -- \\

 &  & Greedy &  & 100.0 & $86.5\pm2.4$ & 5 & $26.9\pm4.4$ & $18.6\pm9.0$ & 124.0 & 1.36 & -- \\

 &  & Neutral &  & 100.0 & $38.3\pm3.2$ & 5 & $47.3\pm4.7$ & $-1.8\pm1.8$ & 147.7 & 1.60 & -- \\

 &  & \textbf{Sched.} &  & 100.0 & $76.1\pm2.4$ & 5 & $30.7\pm4.2$ & $14.9\pm3.2$ & 159.8 & 1.74 & -- \\

 &  & Alias &  & 100.0 & $55.4\pm2.9$ & 5 & $37.1\pm6.6$ & $8.4\pm3.0$ & 147.7 & 1.60 & -- \\

 &  & Learned &  & 100.0 & $36.8\pm2.9$ & 5 & $43.2\pm8.0$ & $2.3\pm6.1$ & 118.0 & 2.43 & -- \\

 &  & Oracle & 12 & -- & -- & -- & -- & -- & -- & 21.012 & 12/0/0 \\

\addlinespace[3pt]

\midrule

Wiki-Talk & Marker-free & Random & 139 & 100.0 & $12.9\pm3.0$ & 5 & $52.8\pm8.8$ & $0.0\pm0.0$ & 36.0 & 0.22 & -- \\

 &  & Delete &  & 100.0 & $0.0\pm0.0$ & 5 & $67.5\pm0.0$ & $-14.7\pm8.8$ & 8.0 & 0.02 & -- \\

 &  & Add &  & 0.0 & \textemdash & 0 & \textemdash & \textemdash & 2.0 & 0.01 & -- \\

 &  & Uniform &  & 100.0 & $8.2\pm1.5$ & 5 & $57.3\pm5.1$ & $-4.5\pm12.7$ & 36.0 & 0.23 & -- \\

 &  & Restart &  & 100.0 & $14.5\pm2.1$ & 5 & $49.8\pm10.3$ & $3.0\pm3.9$ & 50.3 & 0.30 & -- \\

 &  & Greedy &  & 100.0 & $91.4\pm2.9$ & 5 & $27.1\pm4.7$ & $25.7\pm8.4$ & 40.0 & 0.24 & -- \\

 &  & Neutral &  & 100.0 & $12.9\pm3.0$ & 5 & $58.8\pm10.6$ & $-6.0\pm4.8$ & 57.5 & 0.77 & -- \\

 &  & \textbf{Sched.} &  & 100.0 & $65.2\pm3.3$ & 5 & $35.6\pm6.0$ & $17.2\pm6.8$ & 65.6 & 0.33 & -- \\

 &  & Alias &  & 100.0 & $31.9\pm3.8$ & 5 & $44.5\pm10.2$ & $8.3\pm9.5$ & 57.3 & 0.31 & -- \\

 &  & Learned &  & 100.0 & $12.1\pm4.3$ & 5 & $52.1\pm9.7$ & $0.7\pm14.2$ & 36.0 & 0.72 & -- \\

 &  & Oracle & 12 & -- & -- & -- & -- & -- & -- & 3.021 & 12/0/0 \\

\addlinespace[3pt]

Wiki-Talk & Motif repair & Random & 10 & 100.0 & $6.0\pm8.9$ & 5 & $61.5\pm8.5$ & $0.0\pm0.0$ & 34.2 & 0.30 & -- \\

 &  & Delete &  & 40.0 & $0.0\pm0.0$ & 3 & $62.5\pm0.0$ & $-2.4\pm12.0$ & 5.1 & 0.04 & -- \\

 &  & Add &  & 0.0 & \textemdash & 0 & \textemdash & \textemdash & 2.0 & 0.01 & -- \\

 &  & Uniform &  & 100.0 & $6.0\pm8.9$ & 5 & $61.4\pm5.6$ & $0.2\pm8.8$ & 34.2 & 0.31 & -- \\

 &  & Restart &  & 100.0 & $8.0\pm11.0$ & 5 & $61.8\pm8.9$ & $-0.3\pm0.6$ & 49.0 & 0.44 & -- \\

 &  & Greedy &  & 100.0 & $50.0\pm0.0$ & 5 & $42.2\pm1.9$ & $19.3\pm8.6$ & 38.2 & 0.33 & -- \\

 &  & Neutral &  & 100.0 & $6.0\pm8.9$ & 5 & $61.0\pm8.6$ & $0.5\pm3.1$ & 56.2 & 0.48 & -- \\

 &  & \textbf{Sched.} &  & 100.0 & $24.0\pm5.5$ & 5 & $52.3\pm6.3$ & $9.2\pm6.2$ & 64.2 & 0.52 & -- \\

 &  & Alias &  & 100.0 & $6.0\pm8.9$ & 5 & $56.6\pm7.6$ & $4.9\pm1.6$ & 56.0 & 0.47 & -- \\

 &  & Learned &  & 100.0 & $4.0\pm8.9$ & 5 & $60.0\pm4.2$ & $1.6\pm6.5$ & 34.2 & 0.81 & -- \\

 &  & Oracle & 10 & -- & -- & -- & -- & -- & -- & 2.813 & 5/5/0 \\

\addlinespace[3pt]

Wiki-Talk & Tight cap. & Random & 31 & 100.0 & $0.0\pm0.0$ & 5 & $62.5\pm0.0$ & $0.0\pm0.0$ & 10.0 & 0.09 & -- \\

 &  & Delete &  & 100.0 & $0.0\pm0.0$ & 5 & $62.5\pm0.0$ & $0.0\pm0.0$ & 4.0 & 0.01 & -- \\

 &  & Add &  & 0.0 & \textemdash & 0 & \textemdash & \textemdash & 2.0 & 0.01 & -- \\

 &  & Uniform &  & 100.0 & $0.0\pm0.0$ & 5 & $62.5\pm0.0$ & $0.0\pm0.0$ & 10.0 & 0.09 & -- \\

 &  & Restart &  & 100.0 & $0.0\pm0.0$ & 5 & $62.5\pm0.0$ & $0.0\pm0.0$ & 16.0 & 0.15 & -- \\

 &  & Greedy &  & 100.0 & $0.0\pm0.0$ & 5 & $62.5\pm0.0$ & $0.0\pm0.0$ & 12.0 & 0.10 & -- \\

 &  & Neutral &  & 100.0 & $0.0\pm0.0$ & 5 & $62.5\pm0.0$ & $0.0\pm0.0$ & 20.0 & 0.16 & -- \\

 &  & \textbf{Sched.} &  & 100.0 & $0.0\pm0.0$ & 5 & $62.5\pm0.0$ & $0.0\pm0.0$ & 24.0 & 0.17 & -- \\

 &  & Alias &  & 100.0 & $0.0\pm0.0$ & 5 & $62.5\pm0.0$ & $0.0\pm0.0$ & 20.0 & 0.15 & -- \\

 &  & Learned &  & 100.0 & $0.0\pm0.0$ & 5 & $62.5\pm0.0$ & $0.0\pm0.0$ & 10.0 & 0.34 & -- \\

 &  & Oracle & 12 & -- & -- & -- & -- & -- & -- & 0.690 & 0/12/0 \\

\addlinespace[3pt]

Wiki-Talk & Relaxed repl. & Random & 31 & 100.0 & $75.5\pm6.3$ & 5 & $35.5\pm6.7$ & $0.0\pm0.0$ & 10.0 & 0.10 & -- \\

 &  & Delete &  & 100.0 & $0.0\pm0.0$ & 5 & $62.5\pm0.0$ & $-27.0\pm6.7$ & 4.0 & 0.01 & -- \\

 &  & Add &  & 100.0 & $100.0\pm0.0$ & 5 & $25.0\pm0.0$ & $10.5\pm6.7$ & 4.0 & 0.01 & -- \\

 &  & Uniform &  & 100.0 & $69.7\pm9.8$ & 5 & $35.5\pm8.6$ & $0.0\pm11.9$ & 10.0 & 0.11 & -- \\

 &  & Restart &  & 100.0 & $75.5\pm6.3$ & 5 & $35.5\pm6.7$ & $0.0\pm0.0$ & 14.2 & 0.15 & -- \\

 &  & Greedy &  & 100.0 & $100.0\pm0.0$ & 5 & $25.0\pm0.0$ & $10.5\pm6.7$ & 12.0 & 0.11 & -- \\

 &  & Neutral &  & 100.0 & $75.5\pm6.3$ & 5 & $35.5\pm6.7$ & $0.0\pm0.0$ & 18.2 & 0.16 & -- \\

 &  & \textbf{Sched.} &  & 100.0 & $100.0\pm0.0$ & 5 & $25.0\pm0.0$ & $10.5\pm6.7$ & 22.2 & 0.18 & -- \\

 &  & Alias &  & 100.0 & $100.0\pm0.0$ & 5 & $25.0\pm0.0$ & $10.5\pm6.7$ & 18.2 & 0.17 & -- \\

 &  & Learned &  & 100.0 & $100.0\pm0.0$ & 5 & $25.0\pm0.0$ & $10.5\pm6.7$ & 10.0 & 0.38 & -- \\

 &  & Oracle & 12 & -- & -- & -- & -- & -- & -- & 1.227 & 12/0/0 \\

\addlinespace[3pt]

Wiki-Talk & Ext. motif & Random & 9 & 100.0 & $33.3\pm29.4$ & 5 & $44.6\pm14.7$ & $0.0\pm0.0$ & 114.7 & 0.81 & -- \\

 &  & Delete &  & 0.0 & \textemdash & 0 & \textemdash & \textemdash & 3.3 & 0.03 & -- \\

 &  & Add &  & 0.0 & \textemdash & 0 & \textemdash & \textemdash & 2.0 & 0.01 & -- \\

 &  & Uniform &  & 100.0 & $24.4\pm14.5$ & 5 & $45.4\pm4.0$ & $-0.8\pm15.8$ & 114.7 & 0.82 & -- \\

 &  & Restart &  & 100.0 & $37.8\pm31.0$ & 5 & $43.9\pm13.7$ & $0.7\pm1.6$ & 133.5 & 0.98 & -- \\

 &  & Greedy &  & 100.0 & $80.0\pm16.5$ & 5 & $28.8\pm5.3$ & $15.8\pm10.3$ & 120.7 & 0.87 & -- \\

 &  & Neutral &  & 100.0 & $33.3\pm29.4$ & 5 & $47.1\pm16.1$ & $-2.5\pm2.3$ & 145.0 & 1.03 & -- \\

 &  & \textbf{Sched.} &  & 100.0 & $57.8\pm21.4$ & 5 & $36.2\pm11.5$ & $8.4\pm8.5$ & 157.0 & 1.11 & -- \\

 &  & Alias &  & 100.0 & $40.0\pm24.3$ & 5 & $38.6\pm10.7$ & $6.0\pm5.7$ & 144.8 & 1.02 & -- \\

 &  & Learned &  & 100.0 & $28.9\pm18.6$ & 5 & $47.3\pm12.1$ & $-2.7\pm12.4$ & 114.7 & 1.60 & -- \\

 &  & Oracle & 9 & -- & -- & -- & -- & -- & -- & 13.095 & 9/0/0 \\

\addlinespace[3pt]

\midrule

Reddit body & Marker-free & Random & 617 & 100.0 & $10.7\pm0.6$ & 5 & $56.5\pm9.0$ & $0.0\pm0.0$ & 36.0 & 0.24 & -- \\

 &  & Delete &  & 100.0 & $0.0\pm0.0$ & 5 & $67.5\pm0.0$ & $-11.0\pm9.0$ & 8.0 & 0.02 & -- \\

 &  & Add &  & 0.0 & \textemdash & 0 & \textemdash & \textemdash & 2.0 & 0.01 & -- \\

 &  & Uniform &  & 100.0 & $9.6\pm1.0$ & 5 & $52.1\pm7.4$ & $4.4\pm5.0$ & 36.0 & 0.25 & -- \\

 &  & Restart &  & 100.0 & $11.6\pm0.8$ & 5 & $56.5\pm9.0$ & $0.0\pm0.0$ & 50.3 & 0.33 & -- \\

 &  & Greedy &  & 100.0 & $92.3\pm0.5$ & 5 & $27.6\pm3.6$ & $29.0\pm8.2$ & 40.0 & 0.27 & -- \\

 &  & Neutral &  & 100.0 & $10.7\pm0.6$ & 5 & $60.0\pm10.5$ & $-3.5\pm5.2$ & 57.5 & 0.34 & -- \\

 &  & \textbf{Sched.} &  & 100.0 & $64.2\pm1.8$ & 5 & $40.0\pm11.0$ & $16.5\pm11.0$ & 65.6 & 0.36 & -- \\

 &  & Alias &  & 100.0 & $31.2\pm2.0$ & 5 & $49.3\pm6.0$ & $7.2\pm5.4$ & 57.4 & 0.34 & -- \\

 &  & Learned &  & 100.0 & $10.2\pm3.4$ & 5 & $53.6\pm10.8$ & $2.9\pm15.5$ & 36.0 & 0.77 & -- \\

 &  & Oracle & 12 & -- & -- & -- & -- & -- & -- & 3.204 & 12/0/0 \\

\addlinespace[3pt]

Reddit body & Motif repair & Random & 54 & 100.0 & $1.9\pm1.9$ & 5 & $63.7\pm3.4$ & $0.0\pm0.0$ & 32.9 & 0.30 & -- \\

 &  & Delete &  & 35.2 & $0.0\pm0.0$ & 2 & $67.0\pm0.0$ & $-8.3\pm6.5$ & 4.9 & 0.04 & -- \\

 &  & Add &  & 0.0 & \textemdash & 0 & \textemdash & \textemdash & 2.0 & 0.01 & -- \\

 &  & Uniform &  & 100.0 & $3.0\pm1.7$ & 5 & $65.3\pm3.6$ & $-1.6\pm1.9$ & 32.9 & 0.30 & -- \\

 &  & Restart &  & 100.0 & $2.2\pm1.5$ & 5 & $63.7\pm3.4$ & $0.0\pm0.0$ & 47.1 & 0.45 & -- \\

 &  & Greedy &  & 100.0 & $34.8\pm1.5$ & 5 & $51.4\pm1.9$ & $12.3\pm4.4$ & 36.9 & 0.34 & -- \\

 &  & Neutral &  & 100.0 & $1.9\pm1.9$ & 5 & $67.4\pm2.9$ & $-3.7\pm3.8$ & 54.7 & 0.48 & -- \\

 &  & \textbf{Sched.} &  & 100.0 & $23.7\pm2.0$ & 5 & $55.9\pm6.2$ & $7.8\pm7.1$ & 62.8 & 0.54 & -- \\

 &  & Alias &  & 100.0 & $10.0\pm5.2$ & 5 & $60.9\pm2.0$ & $2.8\pm2.3$ & 54.7 & 0.48 & -- \\

 &  & Learned &  & 100.0 & $3.0\pm2.1$ & 5 & $66.7\pm3.1$ & $-3.0\pm5.1$ & 32.9 & 0.81 & -- \\

 &  & Oracle & 12 & -- & -- & -- & -- & -- & -- & 2.197 & 6/6/0 \\

\addlinespace[3pt]

Reddit body & Tight cap. & Random & 378 & 100.0 & $0.0\pm0.0$ & 5 & $62.5\pm0.0$ & $0.0\pm0.0$ & 10.0 & 0.09 & -- \\

 &  & Delete &  & 100.0 & $0.0\pm0.0$ & 5 & $62.5\pm0.0$ & $0.0\pm0.0$ & 4.0 & 0.01 & -- \\

 &  & Add &  & 0.0 & \textemdash & 0 & \textemdash & \textemdash & 2.0 & 0.01 & -- \\

 &  & Uniform &  & 100.0 & $0.0\pm0.0$ & 5 & $62.5\pm0.0$ & $0.0\pm0.0$ & 10.0 & 0.10 & -- \\

 &  & Restart &  & 100.0 & $0.0\pm0.0$ & 5 & $62.5\pm0.0$ & $0.0\pm0.0$ & 16.0 & 0.15 & -- \\

 &  & Greedy &  & 100.0 & $0.0\pm0.0$ & 5 & $62.5\pm0.0$ & $0.0\pm0.0$ & 12.0 & 0.10 & -- \\

 &  & Neutral &  & 100.0 & $0.0\pm0.0$ & 5 & $62.5\pm0.0$ & $0.0\pm0.0$ & 20.0 & 0.16 & -- \\

 &  & \textbf{Sched.} &  & 100.0 & $0.0\pm0.0$ & 5 & $62.5\pm0.0$ & $0.0\pm0.0$ & 24.0 & 0.17 & -- \\

 &  & Alias &  & 100.0 & $0.0\pm0.0$ & 5 & $62.5\pm0.0$ & $0.0\pm0.0$ & 20.0 & 0.15 & -- \\

 &  & Learned &  & 100.0 & $0.0\pm0.0$ & 5 & $62.5\pm0.0$ & $0.0\pm0.0$ & 10.0 & 0.33 & -- \\

 &  & Oracle & 12 & -- & -- & -- & -- & -- & -- & 0.646 & 0/12/0 \\

\addlinespace[3pt]

Reddit body & Relaxed repl. & Random & 378 & 100.0 & $75.5\pm2.3$ & 5 & $35.5\pm6.7$ & $0.0\pm0.0$ & 10.0 & 0.11 & -- \\

 &  & Delete &  & 100.0 & $0.0\pm0.0$ & 5 & $62.5\pm0.0$ & $-27.0\pm6.7$ & 4.0 & 0.01 & -- \\

 &  & Add &  & 100.0 & $100.0\pm0.0$ & 5 & $25.0\pm0.0$ & $10.5\pm6.7$ & 4.0 & 0.01 & -- \\

 &  & Uniform &  & 100.0 & $75.5\pm2.2$ & 5 & $31.0\pm3.4$ & $4.5\pm6.7$ & 10.0 & 0.11 & -- \\

 &  & Restart &  & 100.0 & $75.5\pm2.3$ & 5 & $35.5\pm6.7$ & $0.0\pm0.0$ & 14.2 & 0.37 & -- \\

 &  & Greedy &  & 100.0 & $100.0\pm0.0$ & 5 & $25.0\pm0.0$ & $10.5\pm6.7$ & 12.0 & 0.12 & -- \\

 &  & Neutral &  & 100.0 & $75.5\pm2.3$ & 5 & $35.5\pm6.7$ & $0.0\pm0.0$ & 18.2 & 0.17 & -- \\

 &  & \textbf{Sched.} &  & 100.0 & $100.0\pm0.0$ & 5 & $25.0\pm0.0$ & $10.5\pm6.7$ & 22.2 & 0.18 & -- \\

 &  & Alias &  & 100.0 & $100.0\pm0.0$ & 5 & $25.0\pm0.0$ & $10.5\pm6.7$ & 18.2 & 0.16 & -- \\

 &  & Learned &  & 100.0 & $99.8\pm0.2$ & 5 & $25.0\pm0.0$ & $10.5\pm6.7$ & 10.0 & 0.39 & -- \\

 &  & Oracle & 12 & -- & -- & -- & -- & -- & -- & 1.257 & 12/0/0 \\

\addlinespace[3pt]

Reddit body & Ext. motif & Random & 63 & 100.0 & $21.9\pm5.4$ & 5 & $45.9\pm5.7$ & $0.0\pm0.0$ & 105.6 & 0.77 & -- \\

 &  & Delete &  & 1.6 & $0.0\pm0.0$ & 0 & \textemdash & \textemdash & 3.9 & 0.04 & -- \\

 &  & Add &  & 0.0 & \textemdash & 0 & \textemdash & \textemdash & 2.0 & 0.01 & -- \\

 &  & Uniform &  & 100.0 & $22.9\pm7.0$ & 5 & $51.6\pm11.4$ & $-5.7\pm5.7$ & 105.6 & 0.78 & -- \\

 &  & Restart &  & 100.0 & $22.5\pm5.4$ & 5 & $45.9\pm5.7$ & $0.0\pm0.0$ & 126.1 & 0.95 & -- \\

 &  & Greedy &  & 100.0 & $67.6\pm4.3$ & 5 & $38.8\pm4.0$ & $7.0\pm6.1$ & 111.6 & 0.83 & -- \\

 &  & Neutral &  & 100.0 & $23.2\pm3.7$ & 5 & $47.1\pm6.1$ & $-1.3\pm2.0$ & 137.3 & 1.01 & -- \\

 &  & \textbf{Sched.} &  & 100.0 & $42.2\pm5.2$ & 5 & $42.8\pm4.9$ & $3.1\pm3.9$ & 149.5 & 1.10 & -- \\

 &  & Alias &  & 100.0 & $31.1\pm4.0$ & 5 & $43.9\pm5.4$ & $1.9\pm4.3$ & 137.3 & 1.00 & -- \\

 &  & Learned &  & 100.0 & $25.1\pm4.4$ & 5 & $45.4\pm3.6$ & $0.5\pm7.6$ & 105.6 & 1.57 & -- \\

 &  & Oracle & 12 & -- & -- & -- & -- & -- & -- & 8.716 & 12/0/0 \\

\addlinespace[3pt]

\midrule

PROTEINS & Marker-free & Random & 55 & 100.0 & $9.5\pm5.9$ & 5 & $52.3\pm5.3$ & $0.0\pm0.0$ & 36.0 & 0.78 & -- \\

 &  & Delete &  & 100.0 & $0.0\pm0.0$ & 5 & $66.1\pm0.0$ & $-13.7\pm5.3$ & 8.0 & 0.04 & -- \\

 &  & Add &  & 0.0 & \textemdash & 0 & \textemdash & \textemdash & 2.0 & 0.02 & -- \\

 &  & Uniform &  & 100.0 & $9.8\pm4.2$ & 5 & $56.1\pm3.7$ & $-3.7\pm5.7$ & 36.0 & 0.78 & -- \\

 &  & Restart &  & 100.0 & $10.9\pm5.5$ & 5 & $51.7\pm3.6$ & $0.7\pm4.4$ & 50.7 & 0.93 & -- \\

 &  & Greedy &  & 100.0 & $92.4\pm3.5$ & 5 & $26.1\pm2.4$ & $26.3\pm6.3$ & 40.0 & 0.84 & -- \\

 &  & Neutral &  & 100.0 & $9.5\pm5.9$ & 5 & $59.8\pm6.3$ & $-7.4\pm4.6$ & 57.6 & 0.95 & -- \\

 &  & \textbf{Sched.} &  & 100.0 & $61.1\pm6.4$ & 5 & $36.6\pm10.0$ & $15.7\pm7.8$ & 65.7 & 1.00 & -- \\

 &  & Alias &  & 100.0 & $26.9\pm5.9$ & 5 & $45.2\pm3.3$ & $7.2\pm2.4$ & 57.4 & 0.97 & -- \\

 &  & Learned &  & 100.0 & $10.9\pm3.9$ & 5 & $57.9\pm4.0$ & $-5.6\pm6.6$ & 36.0 & 2.14 & -- \\

 &  & Oracle & 12 & -- & -- & -- & -- & -- & -- & 3.660 & 12/0/0 \\

\addlinespace[3pt]

PROTEINS & Motif repair & Random & 45 & 100.0 & $9.8\pm3.0$ & 5 & $51.1\pm1.9$ & $0.0\pm0.0$ & 36.0 & 1.03 & -- \\

 &  & Delete &  & 100.0 & $0.0\pm0.0$ & 5 & $64.3\pm0.0$ & $-13.2\pm1.9$ & 8.0 & 0.23 & -- \\

 &  & Add &  & 0.0 & \textemdash & 0 & \textemdash & \textemdash & 2.0 & 0.02 & -- \\

 &  & Uniform &  & 100.0 & $11.1\pm2.7$ & 5 & $56.2\pm5.6$ & $-5.0\pm7.1$ & 36.0 & 1.03 & -- \\

 &  & Restart &  & 100.0 & $11.1\pm3.1$ & 5 & $51.2\pm2.5$ & $-0.1\pm1.4$ & 50.4 & 1.50 & -- \\

 &  & Greedy &  & 100.0 & $91.6\pm6.6$ & 5 & $29.7\pm5.2$ & $21.5\pm4.2$ & 40.0 & 1.18 & -- \\

 &  & Neutral &  & 100.0 & $9.8\pm3.0$ & 5 & $55.8\pm5.7$ & $-4.6\pm4.2$ & 57.6 & 1.72 & -- \\

 &  & \textbf{Sched.} &  & 100.0 & $61.3\pm8.7$ & 5 & $31.8\pm8.1$ & $19.4\pm7.5$ & 65.7 & 1.92 & -- \\

 &  & Alias &  & 100.0 & $31.1\pm10.5$ & 5 & $38.3\pm4.9$ & $12.8\pm4.8$ & 57.4 & 1.68 & -- \\

 &  & Learned &  & 100.0 & $9.8\pm3.4$ & 5 & $60.0\pm2.8$ & $-8.9\pm3.0$ & 36.0 & 2.42 & -- \\

 &  & Oracle & 12 & -- & -- & -- & -- & -- & -- & 3.744 & 12/0/0 \\

\addlinespace[3pt]

PROTEINS & Tight cap. & Random & 55 & 100.0 & $0.0\pm0.0$ & 5 & $62.5\pm0.0$ & $0.0\pm0.0$ & 10.0 & 0.24 & -- \\

 &  & Delete &  & 100.0 & $0.0\pm0.0$ & 5 & $62.5\pm0.0$ & $0.0\pm0.0$ & 4.0 & 0.03 & -- \\

 &  & Add &  & 0.0 & \textemdash & 0 & \textemdash & \textemdash & 2.0 & 0.02 & -- \\

 &  & Uniform &  & 100.0 & $0.0\pm0.0$ & 5 & $62.5\pm0.0$ & $0.0\pm0.0$ & 10.0 & 0.24 & -- \\

 &  & Restart &  & 100.0 & $0.0\pm0.0$ & 5 & $62.5\pm0.0$ & $0.0\pm0.0$ & 16.0 & 0.33 & -- \\

 &  & Greedy &  & 100.0 & $0.0\pm0.0$ & 5 & $62.5\pm0.0$ & $0.0\pm0.0$ & 12.0 & 0.27 & -- \\

 &  & Neutral &  & 100.0 & $0.0\pm0.0$ & 5 & $62.5\pm0.0$ & $0.0\pm0.0$ & 20.0 & 0.36 & -- \\

 &  & \textbf{Sched.} &  & 100.0 & $0.0\pm0.0$ & 5 & $62.5\pm0.0$ & $0.0\pm0.0$ & 24.0 & 0.39 & -- \\

 &  & Alias &  & 100.0 & $0.0\pm0.0$ & 5 & $62.5\pm0.0$ & $0.0\pm0.0$ & 20.0 & 0.35 & -- \\

 &  & Learned &  & 100.0 & $0.0\pm0.0$ & 5 & $62.5\pm0.0$ & $0.0\pm0.0$ & 10.0 & 0.74 & -- \\

 &  & Oracle & 12 & -- & -- & -- & -- & -- & -- & 0.794 & 0/12/0 \\

\addlinespace[3pt]

PROTEINS & Relaxed repl. & Random & 55 & 100.0 & $74.5\pm6.0$ & 5 & $34.0\pm9.8$ & $0.0\pm0.0$ & 10.0 & 0.29 & -- \\

 &  & Delete &  & 100.0 & $0.0\pm0.0$ & 5 & $62.5\pm0.0$ & $-28.5\pm9.8$ & 4.0 & 0.03 & -- \\

 &  & Add &  & 100.0 & $100.0\pm0.0$ & 5 & $25.0\pm0.0$ & $9.0\pm9.8$ & 4.0 & 0.02 & -- \\

 &  & Uniform &  & 100.0 & $75.3\pm8.6$ & 5 & $31.0\pm3.4$ & $3.0\pm8.6$ & 10.0 & 0.29 & -- \\

 &  & Restart &  & 100.0 & $74.5\pm6.0$ & 5 & $34.0\pm9.8$ & $0.0\pm0.0$ & 14.3 & 0.37 & -- \\

 &  & Greedy &  & 100.0 & $100.0\pm0.0$ & 5 & $25.0\pm0.0$ & $9.0\pm9.8$ & 12.0 & 0.31 & -- \\

 &  & Neutral &  & 100.0 & $74.5\pm6.0$ & 5 & $34.0\pm9.8$ & $0.0\pm0.0$ & 18.3 & 0.39 & -- \\

 &  & \textbf{Sched.} &  & 100.0 & $100.0\pm0.0$ & 5 & $25.0\pm0.0$ & $9.0\pm9.8$ & 22.3 & 0.43 & -- \\

 &  & Alias &  & 100.0 & $100.0\pm0.0$ & 5 & $25.0\pm0.0$ & $9.0\pm9.8$ & 18.3 & 0.39 & -- \\

 &  & Learned &  & 100.0 & $100.0\pm0.0$ & 5 & $25.0\pm0.0$ & $9.0\pm9.8$ & 10.0 & 0.91 & -- \\

 &  & Oracle & 12 & -- & -- & -- & -- & -- & -- & 1.496 & 12/0/0 \\

\addlinespace[3pt]

PROTEINS & Ext. motif & Random & 55 & 100.0 & $43.3\pm3.9$ & 5 & $39.2\pm6.9$ & $0.0\pm0.0$ & 124.0 & 3.12 & -- \\

 &  & Delete &  & 100.0 & $0.0\pm0.0$ & 5 & $66.3\pm0.0$ & $-27.2\pm6.9$ & 12.0 & 0.31 & -- \\

 &  & Add &  & 0.0 & \textemdash & 0 & \textemdash & \textemdash & 2.0 & 0.03 & -- \\

 &  & Uniform &  & 100.0 & $42.2\pm12.0$ & 5 & $37.9\pm4.7$ & $1.3\pm10.4$ & 124.0 & 3.14 & -- \\

 &  & Restart &  & 100.0 & $44.7\pm4.0$ & 5 & $39.2\pm6.9$ & $0.0\pm0.0$ & 141.2 & 3.60 & -- \\

 &  & Greedy &  & 100.0 & $96.4\pm2.9$ & 5 & $25.0\pm0.0$ & $14.2\pm6.9$ & 130.0 & 3.33 & -- \\

 &  & Neutral &  & 100.0 & $47.3\pm3.9$ & 5 & $39.4\pm10.3$ & $-0.2\pm5.9$ & 152.6 & 3.92 & -- \\

 &  & \textbf{Sched.} &  & 100.0 & $91.3\pm3.5$ & 5 & $25.0\pm0.0$ & $14.2\pm6.9$ & 164.6 & 4.25 & -- \\

 &  & Alias &  & 100.0 & $64.7\pm4.9$ & 5 & $31.0\pm5.8$ & $8.2\pm6.8$ & 152.5 & 3.93 & -- \\

 &  & Learned &  & 100.0 & $40.0\pm10.8$ & 5 & $48.2\pm9.5$ & $-9.1\pm11.7$ & 124.0 & 5.80 & -- \\

 &  & Oracle & 12 & -- & -- & -- & -- & -- & -- & 28.874 & 12/0/0 \\

\addlinespace[3pt]

\midrule

COLLAB & Marker-free & Random & 47 & 100.0 & $8.9\pm2.3$ & 5 & $60.7\pm8.0$ & $0.0\pm0.0$ & 36.0 & 6.12 & -- \\

 &  & Delete &  & 100.0 & $0.0\pm0.0$ & 5 & $62.5\pm0.0$ & $-1.8\pm8.0$ & 8.0 & 0.32 & -- \\

 &  & Add &  & 0.0 & \textemdash & 0 & \textemdash & \textemdash & 2.0 & 0.18 & -- \\

 &  & Uniform &  & 100.0 & $9.8\pm5.7$ & 5 & $60.1\pm5.0$ & $0.7\pm11.4$ & 36.0 & 6.09 & -- \\

 &  & Restart &  & 100.0 & $10.2\pm3.2$ & 5 & $60.8\pm6.9$ & $-0.1\pm3.0$ & 50.8 & 6.91 & -- \\

 &  & Greedy &  & 100.0 & $86.8\pm4.6$ & 5 & $29.7\pm2.8$ & $31.1\pm8.8$ & 40.0 & 6.74 & -- \\

 &  & Neutral &  & 100.0 & $8.9\pm2.3$ & 5 & $65.3\pm4.0$ & $-4.6\pm8.8$ & 57.8 & 7.17 & -- \\

 &  & \textbf{Sched.} &  & 100.0 & $62.6\pm7.0$ & 5 & $47.5\pm4.9$ & $13.2\pm10.3$ & 65.9 & 7.84 & -- \\

 &  & Alias &  & 100.0 & $26.4\pm10.0$ & 5 & $54.1\pm5.2$ & $6.6\pm4.7$ & 57.6 & 7.62 & -- \\

 &  & Learned &  & 100.0 & $8.1\pm4.1$ & 5 & $57.5\pm5.5$ & $3.3\pm12.7$ & 36.0 & 14.99 & -- \\

 &  & Oracle & 12 & -- & -- & -- & -- & -- & -- & 10.858 & 12/0/0 \\

\addlinespace[3pt]

COLLAB & Motif repair & Random & 16 & 100.0 & $10.0\pm7.1$ & 5 & $63.0\pm5.5$ & $0.0\pm0.0$ & 36.0 & 2.37 & -- \\

 &  & Delete &  & 100.0 & $0.0\pm0.0$ & 5 & $64.3\pm0.0$ & $-1.3\pm5.5$ & 8.0 & 0.56 & -- \\

 &  & Add &  & 0.0 & \textemdash & 0 & \textemdash & \textemdash & 2.0 & 0.05 & -- \\

 &  & Uniform &  & 100.0 & $6.2\pm4.4$ & 5 & $57.0\pm3.3$ & $6.0\pm8.3$ & 36.0 & 2.37 & -- \\

 &  & Restart &  & 100.0 & $10.0\pm7.1$ & 5 & $61.6\pm4.4$ & $1.4\pm2.3$ & 51.5 & 3.44 & -- \\

 &  & Greedy &  & 100.0 & $95.0\pm5.2$ & 5 & $25.0\pm0.0$ & $38.0\pm5.5$ & 40.0 & 2.75 & -- \\

 &  & Neutral &  & 100.0 & $10.0\pm7.1$ & 5 & $64.5\pm4.4$ & $-1.5\pm4.0$ & 58.3 & 4.14 & -- \\

 &  & \textbf{Sched.} &  & 100.0 & $65.0\pm8.4$ & 5 & $40.4\pm7.1$ & $22.6\pm8.5$ & 66.4 & 4.51 & -- \\

 &  & Alias &  & 100.0 & $33.8\pm8.4$ & 5 & $44.6\pm2.2$ & $18.4\pm7.1$ & 58.1 & 4.01 & -- \\

 &  & Learned &  & 100.0 & $8.8\pm5.6$ & 5 & $57.6\pm6.3$ & $5.3\pm9.3$ & 36.0 & 5.22 & -- \\

 &  & Oracle & 12 & -- & -- & -- & -- & -- & -- & 7.726 & 12/0/0 \\

\addlinespace[3pt]

COLLAB & Tight cap. & Random & 24 & 100.0 & $0.0\pm0.0$ & 5 & $62.5\pm0.0$ & $0.0\pm0.0$ & 10.0 & 0.68 & -- \\

 &  & Delete &  & 100.0 & $0.0\pm0.0$ & 5 & $62.5\pm0.0$ & $0.0\pm0.0$ & 4.0 & 0.09 & -- \\

 &  & Add &  & 0.0 & \textemdash & 0 & \textemdash & \textemdash & 2.0 & 0.05 & -- \\

 &  & Uniform &  & 100.0 & $0.0\pm0.0$ & 5 & $62.5\pm0.0$ & $0.0\pm0.0$ & 10.0 & 0.67 & -- \\

 &  & Restart &  & 100.0 & $0.0\pm0.0$ & 5 & $62.5\pm0.0$ & $0.0\pm0.0$ & 16.0 & 0.88 & -- \\

 &  & Greedy &  & 100.0 & $0.0\pm0.0$ & 5 & $62.5\pm0.0$ & $0.0\pm0.0$ & 12.0 & 0.77 & -- \\

 &  & Neutral &  & 100.0 & $0.0\pm0.0$ & 5 & $62.5\pm0.0$ & $0.0\pm0.0$ & 20.0 & 0.94 & -- \\

 &  & \textbf{Sched.} &  & 100.0 & $0.0\pm0.0$ & 5 & $62.5\pm0.0$ & $0.0\pm0.0$ & 24.0 & 1.06 & -- \\

 &  & Alias &  & 100.0 & $0.0\pm0.0$ & 5 & $62.5\pm0.0$ & $0.0\pm0.0$ & 20.0 & 0.94 & -- \\

 &  & Learned &  & 100.0 & $0.0\pm0.0$ & 5 & $62.5\pm0.0$ & $0.0\pm0.0$ & 10.0 & 2.02 & -- \\

 &  & Oracle & 12 & -- & -- & -- & -- & -- & -- & 1.107 & 0/12/0 \\

\addlinespace[3pt]

COLLAB & Relaxed repl. & Random & 24 & 100.0 & $73.3\pm5.6$ & 5 & $37.0\pm6.7$ & $0.0\pm0.0$ & 10.0 & 0.84 & -- \\

 &  & Delete &  & 100.0 & $0.0\pm0.0$ & 5 & $62.5\pm0.0$ & $-25.5\pm6.7$ & 4.0 & 0.09 & -- \\

 &  & Add &  & 100.0 & $100.0\pm0.0$ & 5 & $25.0\pm0.0$ & $12.0\pm6.7$ & 4.0 & 0.08 & -- \\

 &  & Uniform &  & 100.0 & $70.0\pm12.3$ & 5 & $37.0\pm4.1$ & $0.0\pm5.3$ & 10.0 & 0.83 & -- \\

 &  & Restart &  & 100.0 & $73.3\pm5.6$ & 5 & $37.0\pm6.7$ & $0.0\pm0.0$ & 14.3 & 0.99 & -- \\

 &  & Greedy &  & 100.0 & $100.0\pm0.0$ & 5 & $25.0\pm0.0$ & $12.0\pm6.7$ & 12.0 & 0.90 & -- \\

 &  & Neutral &  & 100.0 & $73.3\pm5.6$ & 5 & $37.0\pm6.7$ & $0.0\pm0.0$ & 18.3 & 1.08 & -- \\

 &  & \textbf{Sched.} &  & 100.0 & $100.0\pm0.0$ & 5 & $25.0\pm0.0$ & $12.0\pm6.7$ & 22.3 & 1.18 & -- \\

 &  & Alias &  & 100.0 & $100.0\pm0.0$ & 5 & $25.0\pm0.0$ & $12.0\pm6.7$ & 18.3 & 1.08 & -- \\

 &  & Learned &  & 100.0 & $99.2\pm1.9$ & 5 & $25.0\pm0.0$ & $12.0\pm6.7$ & 10.0 & 2.39 & -- \\

 &  & Oracle & 12 & -- & -- & -- & -- & -- & -- & 1.787 & 12/0/0 \\

\addlinespace[3pt]

COLLAB & Ext. motif & Random & 20 & 100.0 & $35.0\pm7.9$ & 5 & $43.7\pm7.4$ & $0.0\pm0.0$ & 124.0 & 8.21 & -- \\

 &  & Delete &  & 100.0 & $0.0\pm0.0$ & 5 & $62.5\pm0.0$ & $-18.8\pm7.4$ & 12.0 & 0.86 & -- \\

 &  & Add &  & 0.0 & \textemdash & 0 & \textemdash & \textemdash & 2.0 & 0.06 & -- \\

 &  & Uniform &  & 100.0 & $35.0\pm10.0$ & 5 & $37.2\pm6.7$ & $6.5\pm4.7$ & 124.0 & 8.21 & -- \\

 &  & Restart &  & 100.0 & $37.0\pm6.7$ & 5 & $39.8\pm3.9$ & $3.8\pm5.3$ & 141.8 & 9.51 & -- \\

 &  & Greedy &  & 100.0 & $92.0\pm4.5$ & 5 & $26.2\pm2.6$ & $17.5\pm5.9$ & 130.0 & 8.76 & -- \\

 &  & Neutral &  & 100.0 & $38.0\pm5.7$ & 5 & $44.0\pm8.2$ & $-0.3\pm6.4$ & 153.0 & 10.67 & -- \\

 &  & \textbf{Sched.} &  & 100.0 & $90.0\pm3.5$ & 5 & $25.6\pm1.3$ & $18.1\pm7.4$ & 165.0 & 11.26 & -- \\

 &  & Alias &  & 100.0 & $55.0\pm5.0$ & 5 & $32.3\pm4.7$ & $11.4\pm6.5$ & 152.9 & 10.39 & -- \\

 &  & Learned &  & 100.0 & $38.0\pm11.5$ & 5 & $37.7\pm7.9$ & $6.0\pm2.5$ & 124.0 & 14.96 & -- \\

 &  & Oracle & 12 & -- & -- & -- & -- & -- & -- & 33.390 & 12/0/0 \\

\addlinespace[3pt]

\end{longtable}
\endgroup

\subsection{Structural and Computational Comparisons}
\label{app:subsection-21}

The ordering controls separate edit feasibility from endpoint fidelity.
Across the marker-free cases, greedy construction achieves greater
zero-conflict coverage than strict adjacent descent in several
settings. On Email-Eu, for example, their zero-conflict rates are
$91.6\%$ and $63.2\%$, with exact projected paired errors of
$26.5\%$ and $41.7\%$. Wiki-Talk exhibits the same ordering:
$91.4\%$ versus $65.2\%$ zero-conflict coverage and
$27.1\%$ versus $35.6\%$ paired error. These differences show
that the choice of legal order matters even before fitting a student.
They do not make local conflict descent a global optimizer.

Under tight capacity, all feasible ordering controls retain
positive cross-target conflict, with a measured paired error
of $62.5\%$, above the $25\%$ clock floor. Relaxing the replacement capacity permits
conflict-free orders; greedy, scheduled, and aliasing-based routes
then reach the clock floor across the five sources. This mirrors
the capacity obstruction in Theorem~\ref{thm:capacity} and the
feasibility diagnostics in Table~\ref{tab:task9-r2}.
At positive conflict, the aliasing objective and structural count
need not rank orders by paired endpoint risk in the same way,
as discussed in Appendix~\ref{app:numerical-verification}.

Computation depends on the candidate search space and the
legal-state predicate. Legality calls and CPU timings in the
matched tables measure route construction, including candidate
generation, legality checks, and ordering updates, rather than
neural training or endpoint sampling. These costs reflect the
search procedures used by each baseline and may vary with
edit depth, structural constraints, and the number of admissible
orders. All ordering controls operate on fixed paired endpoints
and mandatory edits, so computational differences arise from
their route-selection procedures.

The learned-order control uses 300 train-only REINFORCE
steps per regime and seed to optimize edit ordering.
Its additional training costs, reported separately in
Table~\ref{tab:task11-learned-training}, are not included
in the route-construction timings. This policy selects
edit sequences rather than predicting graph-edit actions
during generation.

Oracle feasibility is assessed by complete legal-order
enumeration on tractable instances and bounded CP-SAT
search on selected cohorts. Exhaustive search establishes
feasibility or infeasibility, whereas bounded search may
leave instances unresolved. An unresolved solver status
is not an infeasibility certificate. Oracle results serve
as structural feasibility references, not compute-matched
ordering comparisons.

\begin{table}[htbp]
\centering
\caption{Training cost of the learned-order baseline (seed means).}
\label{tab:task11-learned-training}
\scriptsize
\setlength{\tabcolsep}{3.0pt}
\renewcommand{\arraystretch}{1.06}
\begin{tabular}{@{}llrr@{}}
\toprule
Dataset & Regime & CPU (s) & Queries \\
\midrule
Email-Eu & Marker-free & 0.49 & 12,237 \\
 & Motif repair & 0.55 & 10,620 \\
 & Tight cap. & 0.18 & 3,954 \\
 & Relaxed repl. & 0.22 & 3,974 \\
 & Ext. motif & 1.17 & 31,762 \\
\addlinespace[1pt]
Wiki-Talk & Marker-free & 0.33 & 11,122 \\
 & Motif repair & 0.31 & 5,743 \\
 & Tight cap. & 0.13 & 3,175 \\
 & Relaxed repl. & 0.15 & 3,164 \\
 & Ext. motif & 0.50 & 10,608 \\
\addlinespace[1pt]
Reddit body & Marker-free & 0.53 & 12,522 \\
 & Motif repair & 0.35 & 8,792 \\
 & Tight cap. & 0.15 & 3,996 \\
 & Relaxed repl. & 0.17 & 4,014 \\
 & Ext. motif & 0.75 & 23,942 \\
\addlinespace[1pt]
PROTEINS & Marker-free & 0.87 & 8,586 \\
 & Motif repair & 0.88 & 8,345 \\
 & Tight cap. & 0.28 & 3,229 \\
 & Relaxed repl. & 0.35 & 3,216 \\
 & Ext. motif & 1.97 & 22,678 \\
\addlinespace[1pt]
COLLAB & Marker-free & 6.32 & 8,110 \\
 & Motif repair & 3.27 & 5,836 \\
 & Tight cap. & 0.95 & 2,780 \\
 & Relaxed repl. & 1.29 & 2,783 \\
 & Ext. motif & 6.84 & 13,114 \\
\bottomrule
\end{tabular}
\end{table}

\subsection{Observed-Pair Generalization Protocol}
\label{app:r1-generalization}

The R1 evaluation uses observed-window source--target pairs from
Email-Eu, Wiki-Talk, and Reddit body, as screened in
Appendix~\ref{app:benchmark-selection}. The train/validation/test
partitions and pair construction are frozen across ordering
strategies; validation groups are disjoint from training groups.
Each comparison preserves the observed endpoints, mandatory edits,
shortest-route legality, and the Poisson budget calibrated from
training edit lengths. Held-out targets define evaluation pairs and paired-error
measurements but are not supplied as student inference inputs.

For each R1 ordering strategy, we evaluate the source-blind
population Markov projection with the exact conditional forward
system~\eqref{eq:ode}, then train a graph-only student on
the corresponding teacher routes. At inference, the student uses the current graph,
clock time, and public legality constraints, but no source or
target identity or prescribed edit sequence. The projection
measures the structural effect of order selection; the student
additionally reflects action-learning and held-out generalization.
The paired-error results for both are already reported in
Table~\ref{tab:exp-r1} and are not repeated here.

Scheduling improves the mean paired error of both the ideal
projection and student on all three observed datasets, with the
largest gains on Wiki-Talk and Reddit body. The student remains
above the ideal projection on Email-Eu and Reddit body, separating
structural recoverability from neural realization. R1 concerns
naturally observed pairs, whereas the unified four-edit experiments
in Table~\ref{tab:exp-r2} use controlled targets on real graph
topologies. The two evaluations therefore test complementary
aspects of edit ordering.

\clearpage
\section{Unified Controlled Benchmark and Training Protocol}
\label{app:controlled-five-student}

This section specifies the controlled four-edit benchmark underlying
Table~\ref{tab:exp-r2}. We separate the construction and evaluation of
fixed teacher routes (Stage I) from fitting the graph-only student
(Stage II). The observed-pair R1 protocol is described in
Appendix~\ref{app:r1-generalization}, while earlier multi-regime
structural tests appear in Appendix~\ref{app:r2-controlled}.

\subsection{Four-Edit Construction and Benchmark Selection}
\label{app:subsection-22}

We use five real graph data sources: TGBL-Coin-v2
\citep{huang2023temporal}, CollegeMsg \citep{panzarasa2009patterns},
MathOverflow \citep{paranjape2017motifs}, and PROTEINS and COLLAB
\citep{morris2020tudataset}. The first three supply temporal
interaction topologies (Temporal-C); the latter two supply native
TU graph topologies (TU-C). These categories describe the source
topologies, not distinct target-generation mechanisms.

Each cohort contains two controlled pairs generated by the
four-edit crossed construction in~\eqref{eq:cross}. Every task
requires two deletions and two additions, with $K=4$ and no
added identifying edges. Targets are constructed from the
source topology rather than observed as future snapshots.
The fixed pairing and legal edit orders follow the setting of
Section~\ref{sec:setting}. Table~\ref{tab:controlled-five-data}
reports cohort coverage and initial structural diagnostics.

\begin{table}[htbp]
\centering
\caption{Topology sources, controlled-pair coverage, and pre-student validation diagnostics. Panel B compares random (R) and conflict-descent scheduled (S) legal orders.}
\label{tab:controlled-five-data}
\begingroup
\footnotesize
\setlength{\tabcolsep}{3.2pt}
\renewcommand{\arraystretch}{1.10}
\textbf{A. Topology provenance and cohort coverage}\par\vspace{3pt}
\begin{tabular*}{\linewidth}{@{\extracolsep{\fill}}llrrrr@{}}
\toprule
Dataset & Type & Native units & Source T/V & Cohorts T/V & Pairs T/V \\
\midrule
TGBL-Coin-v2 & Temporal-C & 22,796,409 & 144/40 & 576/160 & 1,152/320 \\
CollegeMsg & Temporal-C & 59,835 & 37/5 & 148/20 & 296/40 \\
MathOverflow & Temporal-C & 390,441 & 144/40 & 576/160 & 1,152/320 \\
PROTEINS & TU-C & 1,113 & 792/165 & 256/80 & 512/160 \\
COLLAB & TU-C & 5,000 & 3,548/753 & 256/80 & 512/160 \\
\midrule
\textbf{Total} & -- & -- & -- & \textbf{1,812/500} & \textbf{3,624/1,000} \\
\bottomrule
\end{tabular*}

\vspace{9pt}
\textbf{B. Validation structural and projected-risk diagnostics}\par\vspace{3pt}
\begin{tabular*}{\linewidth}{@{\extracolsep{\fill}}lrrrrr@{}}
\toprule
Dataset & $C_{\rm R}/C_{\rm S}$ & $C_{\rm R}^{+}/C_{\rm S}^{+}$ (\%) & $D_{\rm R}^{\rm proj}/D_{\rm S}^{\rm proj}$ (\%) & $\Delta D$ (pp) & $\mathfrak A_{\rm R}/\mathfrak A_{\rm S}$ \\
\midrule
TGBL-Coin-v2 & 1.775/0.730 & 91.0/36.5 & 58.70/38.13 & +20.56 & 0.516/0.206 \\
CollegeMsg & 1.730/0.620 & 88.0/31.0 & 58.89/37.26 & +21.63 & 0.521/0.199 \\
MathOverflow & 1.751/0.680 & 90.0/34.0 & 58.91/37.36 & +21.55 & 0.529/0.195 \\
PROTEINS & 1.793/0.840 & 90.2/42.0 & 59.12/40.34 & +18.78 & 0.537/0.242 \\
COLLAB & 1.748/0.765 & 89.8/38.2 & 58.76/39.11 & +19.65 & 0.516/0.222 \\
\bottomrule
\end{tabular*}
\endgroup
\end{table}

Native units count non-self timestamped interactions for Temporal-C
or original graphs for TU-C. Source T/V counts sampled chronological
windows for Temporal-C and source graphs in the fixed split pools for
TU-C; these are dataset-specific units. Each cohort contains two
paired tasks. In Table~\ref{tab:controlled-five-data}, $C$ is the
cross-target conflict score in~\eqref{eq:collision}; $C^+$ is the
percentage of conflict-positive cohorts; and
$\Delta D=D_{\rm R}^{\rm proj}-D_{\rm S}^{\rm proj}$.
The $\mathfrak A$ columns report the target-conditioned ambiguity
budget in~\eqref{eq:qy}. Panel B uses five route seeds and the ideal
Markov projection, before student fitting. All cohorts pass
legality and shortest-route checks; no test targets are constructed.

The sources vary in interaction scale, graph size, and structure
while retaining the same controlled target rule. The benchmark
was assembled after preliminary structural screening in
Appendix~\ref{app:benchmark-selection}; its purpose is to compare
edit-ordering mechanisms on real graph topologies under a common
protocol. The diagnostic gains in Table~\ref{tab:controlled-five-data}
are not learned-student performance measures.

We fix $K=4$ for student training. Each task has eight legal
shortest routes under the prescribed edge-capacity constraint,
so a two-task cohort admits 64 joint route selections. The
four-edge edit space has 11 legal states, permitting exhaustive
route enumeration and finite-state endpoint evaluation.
The choice isolates ordering at a fixed edit depth and clock
budget. For context, Table~\ref{tab:controlled-edit-depth}
condenses the earlier structural tests from
Appendix~\ref{app:r2-controlled}.

\begin{table}[htbp]
\centering
\caption{Edit-depth coverage in the earlier controlled structural audit. Certificate counts and validation pairs are aggregated across five datasets; percentage ranges span dataset means.}
\label{tab:controlled-edit-depth}
\begingroup
\footnotesize
\setlength{\tabcolsep}{3pt}
\renewcommand{\arraystretch}{1.12}
\begin{tabular*}{\linewidth}{@{\extracolsep{\fill}}clrrrrr@{}}
\toprule
$K$ & Controlled regime & Val. pairs & Cert./No/U & $C_{\rm init}^{+}$ (\%) & $\Delta C$ (\%) & Sched.-0 (\%) \\
\midrule
2 & Tight capacity & 1,680 & 0/840/0 & 100.0 & 0.0 & 0.0 \\
2 & Relaxed replacement & 1,680 & 840/0/0 & 23.3--26.7 & 100.0 & 100.0 \\
4 & Marker-free control & 2,460 & 1,230/0/0 & 88.0--90.2 & 57.5--62.8 & 60.9--66.5 \\
4 & Motif repair & 522 & 194/67/0 & 88.4--96.0 & 14.3--68.0 & 20.4--71.2 \\
6 & Extended motif & 630 & 306/9/0 & 54.0--76.5 & 21.0--91.1 & 39.4--94.5 \\
\bottomrule
\end{tabular*}
\endgroup
\end{table}

The earlier audit uses Email-Eu, Wiki-Talk, Reddit body,
PROTEINS, and COLLAB, with two pairs per cohort and five route
seeds. Cert./No/U counts recoverable, infeasible, and unresolved
cohorts; $\Delta C$ measures relative conflict reduction and
Sched.-0 the zero-conflict rate after scheduling. Each regime
uses a training-calibrated Poisson budget with
$\varepsilon_B=0.25$. Both edit depth and legal-state constraints
change between these regimes: tight capacity produces unavoidable
conflicts at $K=2$, whereas relaxed replacement permits recovery.
The table therefore describes structural coverage, not an isolated
edit-depth ablation or student performance.

\subsection{Data Splits, Legality, and Clock Calibration}
\label{app:subsection-23}

\paragraph{Stage I: controlled-pair construction.}
Native units are split into training, validation, and reserved
test groups before controlled endpoints are generated. Temporal-C
uses chronological source windows; TU-C uses disjoint source-graph
groups. Teacher construction and validation operate on the frozen
training and validation partitions, respectively.

For every eligible source, a source-only structural rule selects
the four publicly editable edges and an immutable scaffold.
Node identities index graph edges but do not resolve structural
ties or provide identifying features to the student. The legal-state
predicate and public edit domain are fixed for every ordering
control within a cohort. For cohort $g$, the paired law is
\(
\pi_g=\tfrac12\delta_{(s_{g,1},y_{g,1})}
     +\tfrac12\delta_{(s_{g,2},y_{g,2})},
\)
with endpoints constructed by~\eqref{eq:cross} and both edit
lengths equal to four.

The Poisson clock is source-independent and shared across all
comparisons. Its budget is calibrated on the training protocol
to $B=5.1094$, giving $\varepsilon_B=0.25$ from
\eqref{eq:floor}. Each cohort has its own projection scope
$(E_{*,g},\mathcal X_g,\pi_g)$; conditional teacher actions
are pooled within that scope, not across unrelated graph
universes. The legal action mask depends only on the current
graph and public constraints.

\subsection{Teacher Schedules and Ideal Projection}
\label{app:subsection-24}

\paragraph{Stage I: fixed ordering controls.}
We compare ten legal shortest-route controls: Random,
Scheduled Random, Delete-first, Scheduled Delete-first,
Zero-$C$ witness, Canonical C0-first, Canonical C0-last,
and representatives C1, C2, and C3. The two scheduled
controls apply legality-preserving strict conflict descent
from their corresponding initial routes, as in
Algorithm~\ref{alg:swap}. A zero-conflict witness certifies
clock-floor recoverability in each eligible cohort.

The canonical controls come from exhaustive enumeration of
the 64 joint legal route selections. C0-first and C0-last
are distinct zero-conflict representatives; C1--C3 select
one representative from each positive-conflict class.
Selections use fixed lexicographic rules with duplicate
frozen controls excluded. Teacher routes are fixed before
student optimization, without using validation performance
to choose canonical orders.

The primary matched comparison changes only the legal
ordering of the same mandatory edits. Endpoints, source
weights, graph constraints, and Poisson budget are held
fixed. For each frozen schedule, we record $C(\gamma)$,
teacher action entropy, the ambiguity budget
$\mathfrak A_B^Y$, and the ideal projected paired error
$D_{\mathrm{pair}}^{\mathrm{proj}}$. The latter is obtained
from the conditional forward system~\eqref{eq:ode} and
\eqref{eq:exactrisk}. The structural criterion in
Theorem~\ref{thm:endpoint} separates zero-conflict
certificates from positive-conflict risk comparisons.

Stage I diagnostics in Table~\ref{tab:controlled-five-data}
use five route seeds, whereas the matched neural comparisons
in Table~\ref{tab:exp-r2} use three route/model seeds. The two
tables serve different evaluation stages.

\subsection{GraphSAGE Student and Empirical Objectives}
\label{app:subsection-25}

\paragraph{Stage II: source-blind action model.}
For each dataset and frozen ordering control, we fit a
graph-only action policy. Model weights are shared among
training cohorts within a dataset, while teacher pooling
remains cohort-specific. The architecture and optimization
budget are matched across ordering controls. The student
receives the current graph, clock time, and public legality
information, but not the source or target identity, route
index, remaining edits, or hidden edge-role labels.

The encoder is a three-layer GraphSAGE \citep{hamilton2017inductive}
with hidden width 128 and separate edit-action and null-action
heads, implemented in PyTorch \citep{paszke2019pytorch}.
Four editable edges are canonically ranked using immutable
scaffold statistics. Their public rank encodings and current
presence indicators enter the action heads; no target-derived
edge roles are supplied. Before masking illegal edits, the
logits are
\(
\ell_\theta(a\mid x,u;\mathcal P_g)
=\ell_{\mathrm{pub},\theta}(a\mid x,u;\mathcal P_g)
+\alpha\tanh\!\left(
\ell_{\mathrm{graph},\theta}(a\mid x,u;\mathcal P_g)
\right).
\)
The graph residual is switched off during public-head
pretraining and enabled with $\alpha=0.8$ during joint
fine-tuning. The action space includes the terminal null
edit $\bot$; masking uses only the public legality predicate.

Optimization has three stages. First, the public heads are
pretrained for 35 epochs with pooled-teacher cross-entropy
\eqref{eq:empirical-pooled-ce}. Next, the public heads are
optimized for 55 epochs with
\(
\mathcal L_{\mathrm{joint}}(\theta)
=0.2\,\widehat{\mathcal L}_{\mathrm{CE}}^{\mathrm{tr}}(\theta)
+\widehat D_{\mathrm{pair}}^{\mathrm{tr}}(\theta),
\)
where the paired-endpoint term is computed on sampled
training cohorts with a differentiable finite-state RK4
solver. Finally, up to 60 epochs jointly fine-tune the
GraphSAGE residual and public heads under the same objective,
using stochastic teacher-imitation minibatches and separate
learning rates for the two branches. Public heads remain
trainable throughout the last stage. The training objective
and its pooled-kernel interpretation follow
Section~\ref{sec:learning}.

\subsection{Model Selection and Endpoint Evaluation}
\label{app:subsection-26}

Checkpoint selection uses validation paired-endpoint error,
with action KL as a small tie-breaker, on a fixed, evenly
spaced subset of at most 16 validation cohorts. Selection
covers all training stages; the selected checkpoint need not
come from the final GraphSAGE fine-tuning stage. This subset is included in the complete validation split
used for reported endpoint risks.
After selection, we evaluate the model over all validation
cohorts using the same clock and cohort weights as the
matched teacher. Student endpoint risks are computed using the same
forward-equation form on the full reachable legal state
space, with student action probabilities replacing the
ideal pooled kernel. Independent event-driven Monte Carlo
provides a numerical check. Metrics include held-out action
NLL, teacher--student KL, null-action calibration,
edit-conditioned top-1 accuracy, permutation and expressivity
diagnostics, and paired endpoint risk $D_{\mathrm{pair}}^\theta$.

We compare $D_{\mathrm{pair}}^\theta$ against the matched
ideal projected risk $D_{\mathrm{pair}}^{\mathrm{proj}}$
and the common $25\%$ clock floor. For positive-conflict controls, endpoint-trained students
may outperform their matched teacher-induced projections
by deviating from the pooled action law. The complete ten-control
comparisons appear in Table~\ref{tab:exp-r2}, with richer
per-dataset diagnostics in Appendix~\ref{app:extended-experiments}.
Training teacher occupancies use only training pairs; held-out
targets enter offline evaluation and not student inputs.
No test endpoints are used in model fitting or selection.

\clearpage

\section{Extended Student Experiments}
\label{app:extended-experiments}

This section extends the controlled graph-only evaluation in
Table~\ref{tab:exp-r2}. Benchmark construction, legal ordering
controls, and student training are specified in Appendix~\ref{app:controlled-five-student}.

\subsection{Full-Validation Diagnostics and Training Curves}
\label{app:subsection-27}

Table~\ref{tab:exp-r2} reports projected and student paired
error and action KL. Table~\ref{tab:step5-v2-full} reports
additional full-validation diagnostics omitted from the main text:
edit top-1 accuracy and predicted null-action probability.

\begingroup
\footnotesize
\setlength{\LTleft}{\fill}
\setlength{\LTright}{\fill}
\setlength{\LTcapwidth}{\linewidth}
\setlength{\tabcolsep}{4pt}
\renewcommand{\arraystretch}{1.08}
\begin{longtable}{@{}lrr@{}}
\caption{Full-validation edit top-1 accuracy and predicted null-action probability on five controlled $K=4$ benchmarks.}
\label{tab:step5-v2-full}\\
\toprule
Method & Edit Top-1 (\%) & $P(\bot)$ \\
\midrule
\endfirsthead
\multicolumn{3}{@{}l}{\emph{Full-validation diagnostics (continued)}}\\
\toprule
Method & Edit Top-1 (\%) & $P(\bot)$ \\
\midrule
\endhead
\midrule
\multicolumn{3}{r@{}}{\emph{Continued on next page}}\\
\endfoot
\bottomrule
\endlastfoot
\multicolumn{3}{@{}l}{\textbf{TGBL-Coin-v2}} \\
Random & 51.72 $\pm$ 3.37 & 0.291 $\pm$ 0.007 \\
Sched. Random & 52.40 $\pm$ 0.77 & 0.293 $\pm$ 0.004 \\
Delete-first & 70.75 $\pm$ 1.09 & 0.296 $\pm$ 0.000 \\
Sched. Delete-first & 58.02 $\pm$ 4.29 & 0.298 $\pm$ 0.000 \\
Zero-C witness & 100.00 $\pm$ 0.00 & 0.297 $\pm$ 0.000 \\
C0-first & 100.00 $\pm$ 0.00 & 0.297 $\pm$ 0.000 \\
C0-last & 100.00 $\pm$ 0.00 & 0.297 $\pm$ 0.000 \\
C1 & 74.25 $\pm$ 0.00 & 0.297 $\pm$ 0.000 \\
C2 & 85.22 $\pm$ 0.64 & 0.297 $\pm$ 0.000 \\
C3 & 72.06 $\pm$ 0.96 & 0.297 $\pm$ 0.000 \\
\midrule
\multicolumn{3}{@{}l}{\textbf{CollegeMsg}} \\
Random & 51.86 $\pm$ 3.42 & 0.290 $\pm$ 0.011 \\
Sched. Random & 53.93 $\pm$ 2.08 & 0.292 $\pm$ 0.007 \\
Delete-first & 70.79 $\pm$ 1.07 & 0.296 $\pm$ 0.000 \\
Sched. Delete-first & 66.32 $\pm$ 11.23 & 0.297 $\pm$ 0.000 \\
Zero-C witness & 100.00 $\pm$ 0.00 & 0.297 $\pm$ 0.000 \\
C0-first & 100.00 $\pm$ 0.00 & 0.297 $\pm$ 0.000 \\
C0-last & 100.00 $\pm$ 0.00 & 0.297 $\pm$ 0.000 \\
C1 & 74.25 $\pm$ 0.00 & 0.297 $\pm$ 0.000 \\
C2 & 85.27 $\pm$ 0.58 & 0.297 $\pm$ 0.000 \\
C3 & 72.13 $\pm$ 1.06 & 0.297 $\pm$ 0.000 \\
\midrule
\multicolumn{3}{@{}l}{\textbf{MathOverflow}} \\
Random & 50.30 $\pm$ 1.01 & 0.294 $\pm$ 0.001 \\
Sched. Random & 51.79 $\pm$ 0.74 & 0.295 $\pm$ 0.006 \\
Delete-first & 70.80 $\pm$ 1.16 & 0.296 $\pm$ 0.000 \\
Sched. Delete-first & 58.48 $\pm$ 4.41 & 0.298 $\pm$ 0.001 \\
Zero-C witness & 100.00 $\pm$ 0.00 & 0.297 $\pm$ 0.000 \\
C0-first & 100.00 $\pm$ 0.00 & 0.297 $\pm$ 0.000 \\
C0-last & 100.00 $\pm$ 0.00 & 0.297 $\pm$ 0.000 \\
C1 & 74.25 $\pm$ 0.00 & 0.297 $\pm$ 0.000 \\
C2 & 85.28 $\pm$ 0.58 & 0.297 $\pm$ 0.000 \\
C3 & 72.02 $\pm$ 1.03 & 0.297 $\pm$ 0.000 \\
\midrule
\multicolumn{3}{@{}l}{\textbf{PROTEINS}} \\
Random & 53.04 $\pm$ 1.16 & 0.289 $\pm$ 0.007 \\
Sched. Random & 52.44 $\pm$ 1.10 & 0.293 $\pm$ 0.005 \\
Delete-first & 70.74 $\pm$ 1.10 & 0.296 $\pm$ 0.000 \\
Sched. Delete-first & 52.56 $\pm$ 4.64 & 0.297 $\pm$ 0.000 \\
Zero-C witness & 100.00 $\pm$ 0.00 & 0.297 $\pm$ 0.000 \\
C0-first & 100.00 $\pm$ 0.00 & 0.297 $\pm$ 0.000 \\
C0-last & 100.00 $\pm$ 0.00 & 0.297 $\pm$ 0.000 \\
C1 & 74.25 $\pm$ 0.00 & 0.297 $\pm$ 0.000 \\
C2 & 85.25 $\pm$ 0.59 & 0.297 $\pm$ 0.000 \\
C3 & 72.02 $\pm$ 1.03 & 0.297 $\pm$ 0.000 \\
\midrule
\multicolumn{3}{@{}l}{\textbf{COLLAB}} \\
Random & 50.89 $\pm$ 0.55 & 0.293 $\pm$ 0.003 \\
Sched. Random & 51.95 $\pm$ 2.02 & 0.297 $\pm$ 0.002 \\
Delete-first & 70.70 $\pm$ 0.94 & 0.296 $\pm$ 0.000 \\
Sched. Delete-first & 58.67 $\pm$ 2.76 & 0.297 $\pm$ 0.000 \\
Zero-C witness & 100.00 $\pm$ 0.00 & 0.297 $\pm$ 0.000 \\
C0-first & 100.00 $\pm$ 0.00 & 0.297 $\pm$ 0.000 \\
C0-last & 100.00 $\pm$ 0.00 & 0.297 $\pm$ 0.000 \\
C1 & 74.25 $\pm$ 0.00 & 0.297 $\pm$ 0.000 \\
C2 & 85.25 $\pm$ 0.59 & 0.297 $\pm$ 0.000 \\
C3 & 72.02 $\pm$ 1.03 & 0.297 $\pm$ 0.000 \\
\end{longtable}
\endgroup

The following five figures provide dataset-specific details
of the training curves in Figure~\ref{fig:exp-training}.

\begin{figure}[htbp]
\centering
\includegraphics[
  width=0.96\linewidth,
  height=0.75\textheight,
  keepaspectratio
]{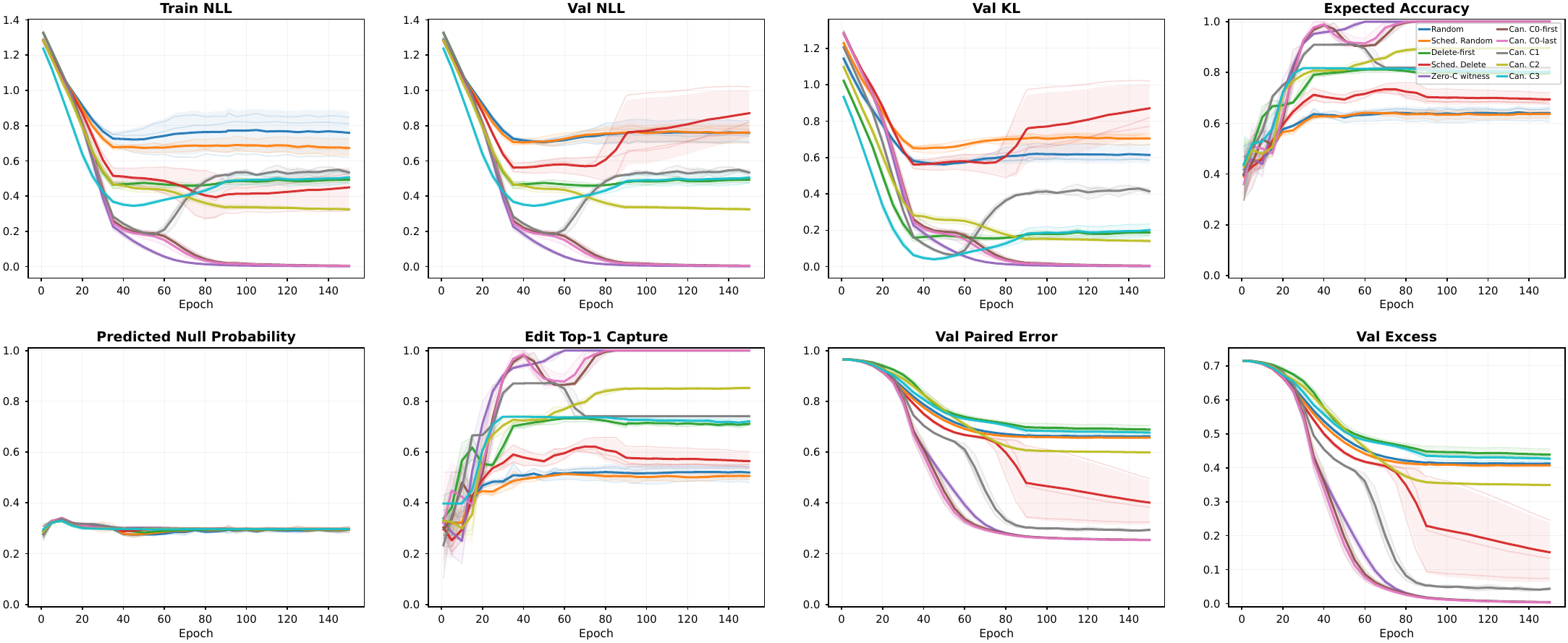}
\caption{Student training trajectories on TGBL-Coin-v2.}
\label{fig:step5-v2-training}
\end{figure}

\begin{figure}[htbp]
\centering
\includegraphics[
  width=0.96\linewidth,
  height=0.75\textheight,
  keepaspectratio
]{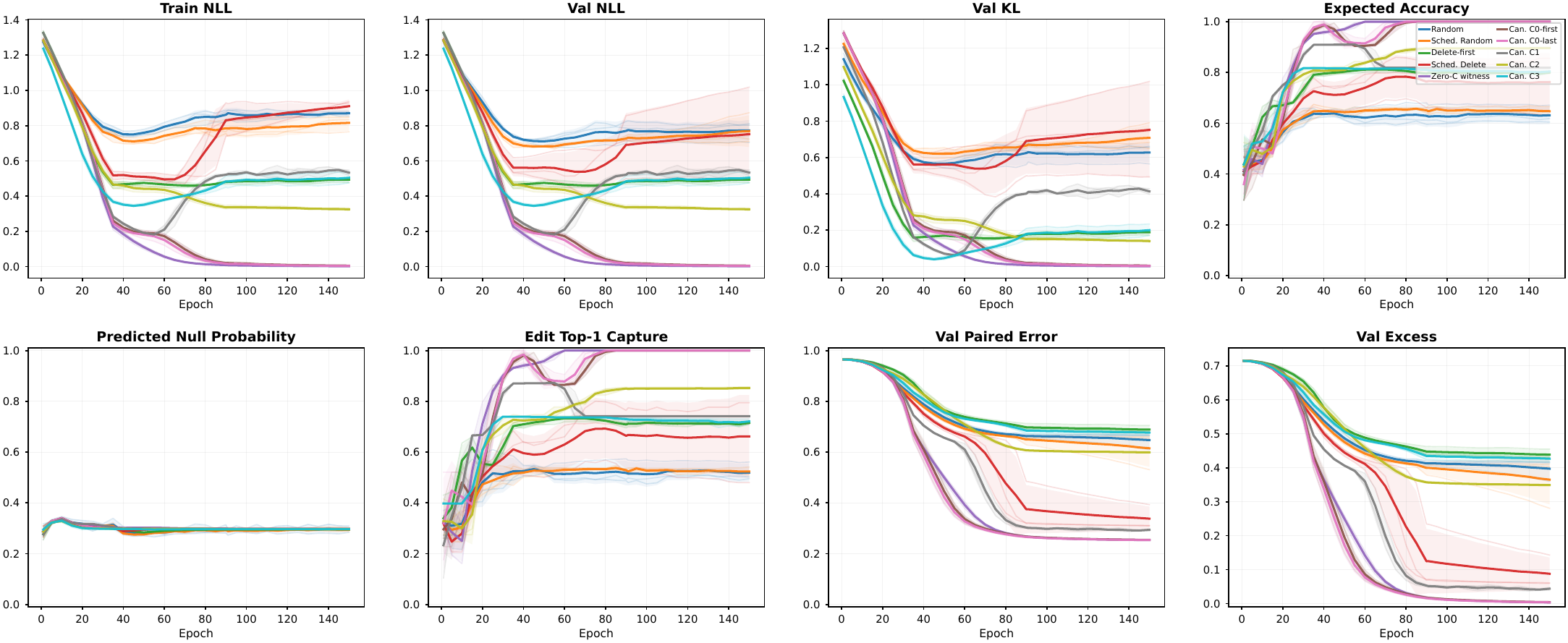}
\caption{Student training trajectories on CollegeMsg.}
\label{fig:step5-v2-training-college}
\end{figure}

\begin{figure}[htbp]
\centering
\includegraphics[
  width=0.96\linewidth,
  height=0.75\textheight,
  keepaspectratio
]{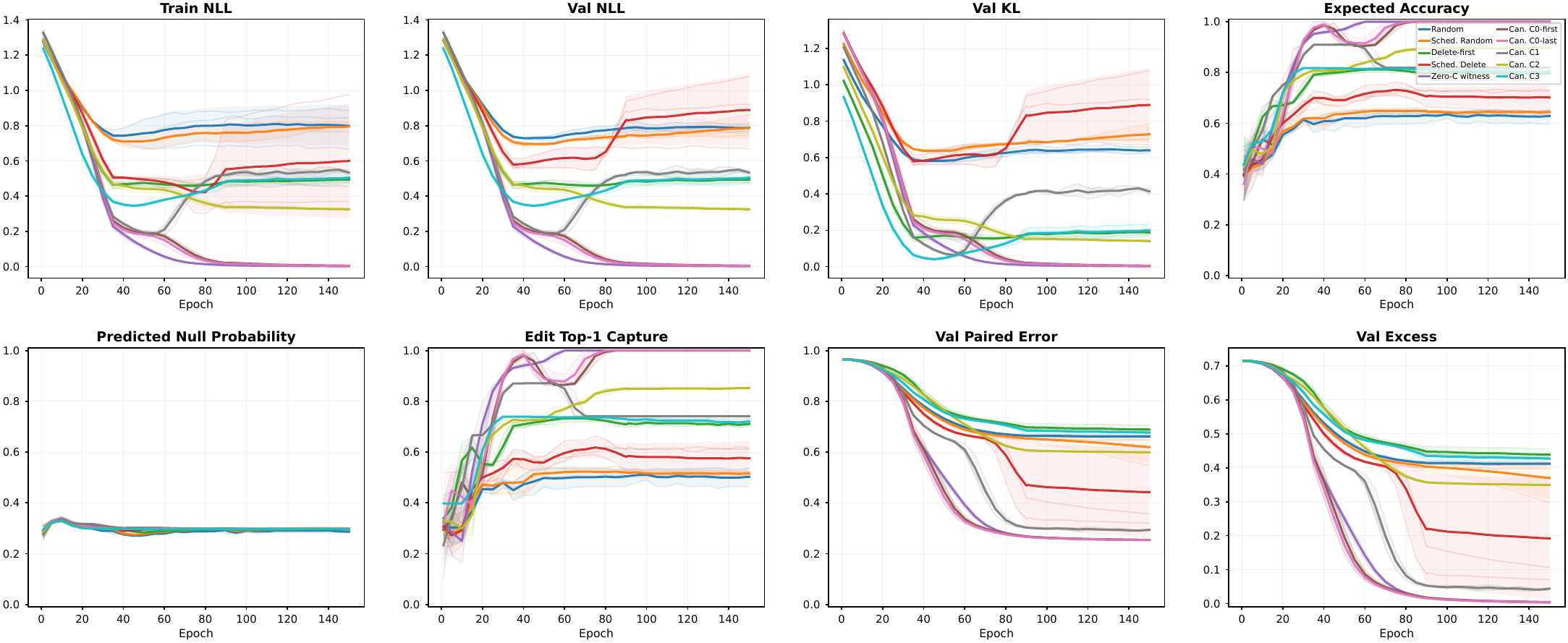}
\caption{Student training trajectories on MathOverflow.}
\label{fig:step5-v2-training-cont}
\end{figure}

\begin{figure}[htbp]
\centering
\includegraphics[
  width=0.96\linewidth,
  height=0.75\textheight,
  keepaspectratio
]{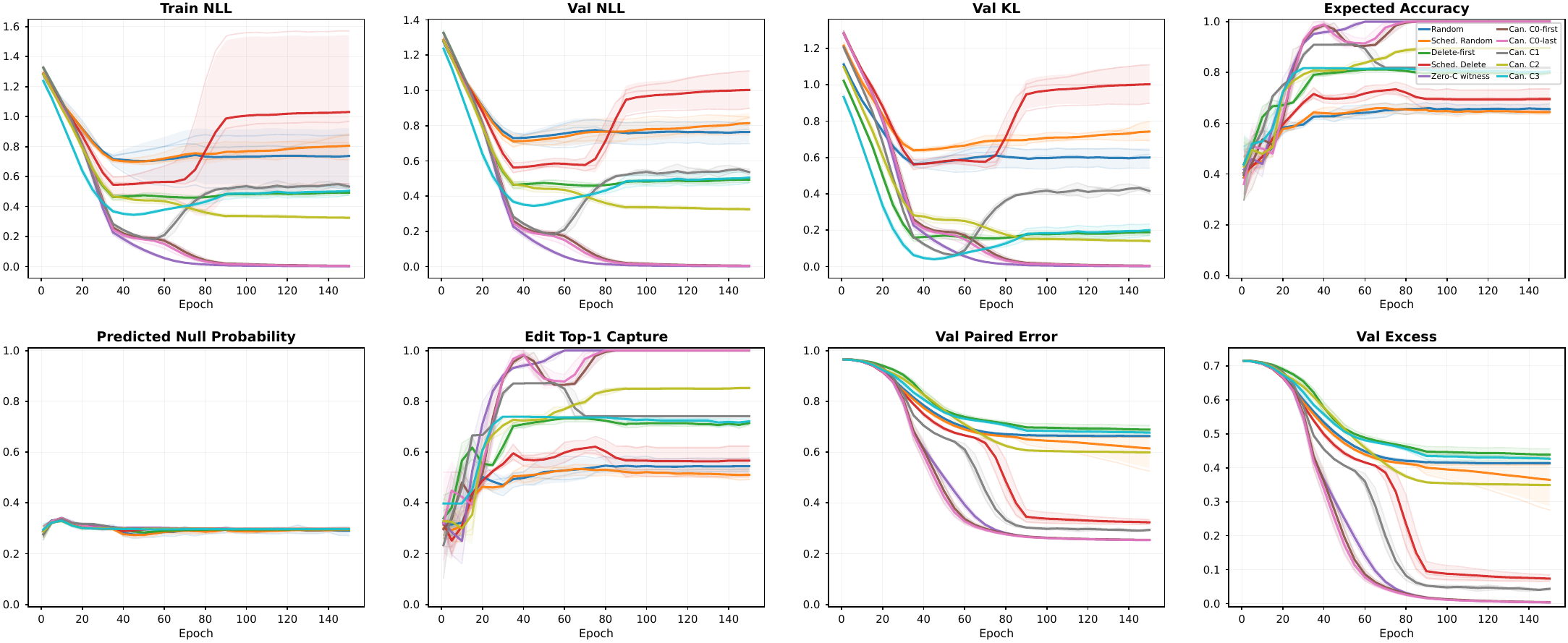}
\caption{Student training trajectories on PROTEINS.}
\label{fig:step5-v2-training-proteins}
\end{figure}

\begin{figure}[htbp]
\centering
\includegraphics[
  width=0.96\linewidth,
  height=0.75\textheight,
  keepaspectratio
]{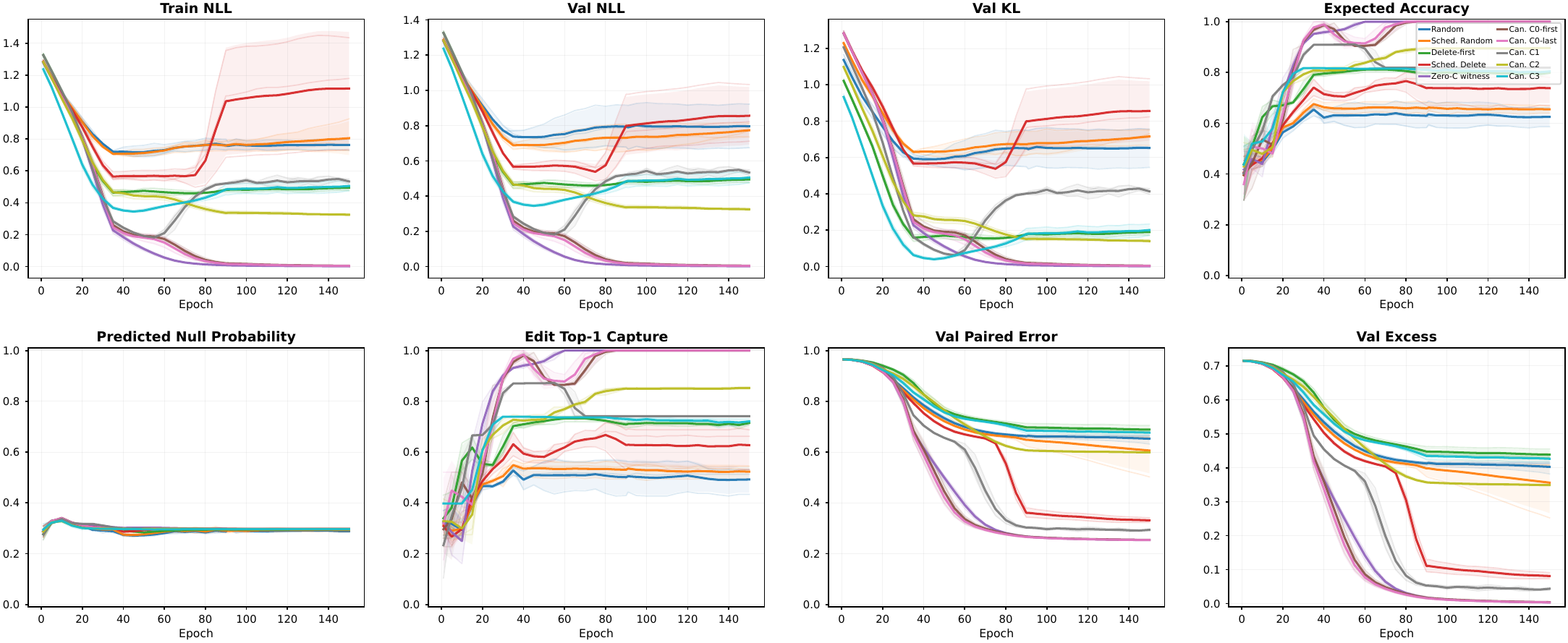}
\caption{Student training trajectories on COLLAB.}
\label{fig:step5-v2-training-tu}
\end{figure}

\subsection{Five-Axis Scaling Sensitivity}
\label{app:subsection-28}

On PROTEINS, we vary the number of training pairs ($m$), validation
graph size ($n$), hidden width, MLP depth, and GraphSAGE depth using
the zero-$C$ witness. Training and checkpoint
selection follow Appendix~\ref{app:subsection-25} and
Appendix~\ref{app:subsection-26}.
Table~\ref{tab:q1-compact-scaling} reports paired error and compute cost.
Near-floor endpoint errors are stable across the tested settings,
while model size and memory change with architecture.

\begin{table}[htbp]
\centering
\caption{Five-axis sensitivity on PROTEINS with the zero-$C$ witness (mean $\pm$ SD). Graph-size rows reuse the reference model: training time
and parameter count are inherited, while CUDA reports
inference memory.}
\label{tab:q1-compact-scaling}
\scriptsize
\setlength{\tabcolsep}{3.1pt}
\resizebox{0.8\linewidth}{!}{%
\begin{tabular}{llrrrrrrrr}
\toprule
Axis & Level & Train pairs & $n_{\rm val}$
& $D_{\rm pair}$ (\%) & Excess (pp) & Fit (min)
& CUDA (MiB) & Params (M) \\
\midrule
base & reference & 512.00 & 41.5
& 25.35 $\pm$ 0.07 & 0.35 $\pm$ 0.07
& 8.49 $\pm$ 0.19 & 95.35 $\pm$ 0.49 & 0.247 \\

gnn & 2 & 512.00 & 41.5
& 25.35 $\pm$ 0.07 & 0.35 $\pm$ 0.07
& 8.17 $\pm$ 0.75 & 75.54 $\pm$ 0.40 & 0.198 \\

gnn & 4 & 512.00 & 41.5
& 25.35 $\pm$ 0.07 & 0.35 $\pm$ 0.07
& 8.55 $\pm$ 0.07 & 114.99 $\pm$ 0.73 & 0.297 \\

hidden & 256 & 512.00 & 41.5
& 25.35 $\pm$ 0.06 & 0.35 $\pm$ 0.06
& 8.61 $\pm$ 0.06 & 176.47 $\pm$ 1.06 & 0.942 \\

hidden & 64 & 512.00 & 41.5
& 25.35 $\pm$ 0.06 & 0.35 $\pm$ 0.06
& 8.52 $\pm$ 0.18 & 58.06 $\pm$ 0.49 & 0.072 \\

m & 50\% & 256.00 & 41.5
& 25.35 $\pm$ 0.06 & 0.35 $\pm$ 0.06
& 8.64 $\pm$ 0.11 & 98.74 $\pm$ 2.44 & 0.247 \\

mlp & 1 & 512.00 & 41.5
& 25.35 $\pm$ 0.07 & 0.35 $\pm$ 0.07
& 7.78 $\pm$ 1.07 & 94.16 $\pm$ 0.49 & 0.239 \\

mlp & 4 & 512.00 & 41.5
& 25.35 $\pm$ 0.06 & 0.35 $\pm$ 0.06
& 6.05 $\pm$ 4.20 & 97.64 $\pm$ 0.49 & 0.256 \\

n & large & 512.00 & 60.7
& 25.35 $\pm$ 0.07 & 0.35 $\pm$ 0.07
& 8.49 $\pm$ 0.19 & 31.92 $\pm$ 0.00 & 0.247 \\

n & small & 512.00 & 20.3
& 25.35 $\pm$ 0.07 & 0.35 $\pm$ 0.07
& 8.49 $\pm$ 0.19 & 31.92 $\pm$ 0.00 & 0.247 \\
\bottomrule
\end{tabular}}

\end{table}

\subsection{Edit-Depth Sensitivity}
\label{app:subsection-29}

An independent depth study compares random uniform routing with
zero-conflict scheduling ($\texttt{sched\_min}$) for $K\in\{2,4,6\}$
on the same five topology sources. Each depth uses a separately
calibrated Poisson budget with a $25\%$ clock floor and three paired
route/model seeds (90 fits). The depth-specific student uses a
public finite-state/time action table initialized from training
teacher occupancies, followed by CE and endpoint-risk optimization
with GraphSAGE residual fine-tuning. The $K=6$ regime also changes
the editable-action universe and capacity constraints,
so the comparison does not isolate edit depth alone.
Table~\ref{tab:depth-v3} gives the full endpoint and action-KL results.

\begin{table}[htbp]
\centering
\caption{Edit-depth sensitivity across five datasets
(mean $\pm$ SD).
R/S denote \texttt{random\_uniform}/\texttt{sched\_min}.
All clock floors are $25\%$; $K=6$ uses different constraints.
States counts enumerated legal graph states,
rather than states visited by selected routes.}
\label{tab:depth-v3}
\begingroup
\scriptsize
\setlength{\tabcolsep}{3pt}
\renewcommand{\arraystretch}{0.95}
\resizebox{0.97\linewidth}{!}{%
\begin{tabular}{@{}lrrcccccc@{}}
\toprule
& & &
\multicolumn{2}{c}{$D_{\rm pair}^{\rm proj}$ (\%)} &
\multicolumn{2}{c}{$D_{\rm pair}^{\theta}$ (\%)} &
\multicolumn{2}{c}{KL (nats)} \\
\cmidrule(lr){4-5}
\cmidrule(lr){6-7}
\cmidrule(l){8-9}
Dataset & $K$ & States &
R & S & R & S & R & S \\
\midrule

COLLAB
& 2 & 11
& $28.12\pm2.71$ & $25.00\pm0.00$
& $33.05\pm0.77$ & $25.28\pm0.01$
& $0.383\pm0.061$ & $0.012\pm0.000$ \\
& 4 & 11
& $59.30\pm3.73$ & $25.00\pm0.00$
& $69.11\pm1.09$ & $25.70\pm0.03$
& $0.715\pm0.121$ & $0.021\pm0.001$ \\
& 6 & 42
& $41.00\pm9.38$ & $25.00\pm0.00$
& $68.74\pm0.50$ & $26.55\pm0.12$
& $1.640\pm0.156$ & $0.032\pm0.003$ \\

\midrule
CollegeMsg
& 2 & 11
& $35.94\pm15.07$ & $25.00\pm0.00$
& $35.81\pm2.16$ & $25.30\pm0.01$
& $0.382\pm0.121$ & $0.015\pm0.000$ \\
& 4 & 11
& $56.46\pm2.10$ & $25.00\pm0.00$
& $69.25\pm0.53$ & $25.73\pm0.03$
& $0.634\pm0.056$ & $0.023\pm0.001$ \\
& 6 & 42
& $50.64\pm6.77$ & $25.00\pm0.00$
& $68.26\pm0.62$ & $26.70\pm0.07$
& $1.531\pm0.173$ & $0.038\pm0.005$ \\

\midrule
PROTEINS
& 2 & 11
& $31.25\pm7.16$ & $25.00\pm0.00$
& $31.93\pm0.79$ & $25.27\pm0.00$
& $0.480\pm0.072$ & $0.011\pm0.000$ \\
& 4 & 11
& $60.82\pm3.18$ & $25.00\pm0.00$
& $68.09\pm1.09$ & $25.69\pm0.02$
& $0.575\pm0.037$ & $0.022\pm0.002$ \\
& 6 & 42
& $48.47\pm5.13$ & $25.00\pm0.00$
& $63.18\pm4.42$ & $26.45\pm0.01$
& $1.532\pm0.475$ & $0.029\pm0.000$ \\

\midrule
MathOverflow
& 2 & 11
& $34.37\pm9.38$ & $25.00\pm0.00$
& $33.39\pm1.78$ & $25.29\pm0.00$
& $0.433\pm0.123$ & $0.012\pm0.000$ \\
& 4 & 11
& $63.68\pm1.01$ & $25.00\pm0.00$
& $69.09\pm1.21$ & $25.69\pm0.02$
& $0.623\pm0.097$ & $0.026\pm0.001$ \\
& 6 & 42
& $46.46\pm4.45$ & $25.00\pm0.00$
& $67.97\pm1.98$ & $26.57\pm0.08$
& $1.532\pm0.102$ & $0.034\pm0.002$ \\

\midrule
TGBL-Coin-v2
& 2 & 11
& $39.06\pm4.69$ & $25.00\pm0.00$
& $33.61\pm0.33$ & $25.29\pm0.00$
& $0.442\pm0.051$ & $0.013\pm0.001$ \\
& 4 & 11
& $61.64\pm1.34$ & $25.00\pm0.00$
& $68.75\pm2.27$ & $25.75\pm0.02$
& $0.792\pm0.261$ & $0.022\pm0.001$ \\
& 6 & 42
& $50.13\pm7.22$ & $25.00\pm0.00$
& $67.07\pm0.85$ & $26.67\pm0.04$
& $1.613\pm0.159$ & $0.039\pm0.004$ \\

\bottomrule
\end{tabular}%
}
\endgroup
\end{table}

The scheduled teacher reaches the projection floor at all three
depths; mean student paired errors averaged over datasets are
$25.29\%$, $25.71\%$, and $26.59\%$ for $K=2,4,6$.
Random-order students remain farther from the floor.
The depth-specific training protocol differs from the main
$K=4$ GraphSAGE comparison, so these values are reported separately.

\subsection{Checkpoint-Selection-Disjoint Validation}
\label{app:subsection-31}

We reevaluate the 150 saved $K=4$ student checkpoints without
retraining. In-support risk uses up to 24 training cohorts per run;
disjoint validation risk excludes the 16 cohorts used for checkpoint
selection. All risks share the same $25\%$ Poisson floor and
finite-state endpoint computation. The full cohorts are split into
training, selection, and selection-disjoint evaluation groups
as specified in Appendix~\ref{app:subsection-23}.

\begingroup
\scriptsize
\setlength{\LTleft}{\fill}
\setlength{\LTright}{\fill}
\setlength{\LTcapwidth}{\linewidth}
\setlength{\tabcolsep}{4pt}
\renewcommand{\arraystretch}{1.08}
\begin{longtable}{@{}lrrrr@{}}
\caption{Training-support and checkpoint-selection-disjoint paired errors (\%, mean $\pm$ SD). $\Delta=D_{\rm free}^{\theta}-D_{\rm tr}^{\theta}$ (pp).}
\label{tab:task13-protocol}\\
\toprule
Teacher & $D_{\rm tr}^{\theta}$ (\%) & $D_{\rm free}^{\theta}$ (\%) & $\Delta$ (pp) & $D_{\rm free}^{\rm proj}$ (\%) \\
\midrule
\endfirsthead
\multicolumn{5}{@{}l}{\emph{Selection-disjoint validation (continued)}}\\
\toprule
Teacher & $D_{\rm tr}^{\theta}$ (\%) & $D_{\rm free}^{\theta}$ (\%) & $\Delta$ (pp) & $D_{\rm free}^{\rm proj}$ (\%) \\
\midrule
\endhead
\midrule
\multicolumn{5}{r@{}}{\emph{Continued on next page}}\\
\endfoot
\bottomrule
\endlastfoot
\multicolumn{5}{@{}l}{\textbf{TGBL-Coin-v2} (select/free 16/144)} \\
Random & 66.16 $\pm$ 0.08 & 66.16 $\pm$ 0.07 & 0.00 $\pm$ 0.01 & 58.63 $\pm$ 0.74 \\
Sched. Random & 65.64 $\pm$ 0.25 & 65.63 $\pm$ 0.24 & -0.00 $\pm$ 0.01 & 38.19 $\pm$ 1.04 \\
Delete-first & 68.91 $\pm$ 1.41 & 68.91 $\pm$ 1.41 & 0.00 $\pm$ 0.00 & 78.53 $\pm$ 0.00 \\
Sched. Delete-first & 40.09 $\pm$ 8.73 & 40.09 $\pm$ 8.74 & -0.01 $\pm$ 0.02 & 25.00 $\pm$ 0.00 \\
Zero-C witness & 25.36 $\pm$ 0.07 & 25.37 $\pm$ 0.07 & 0.00 $\pm$ 0.00 & 25.00 $\pm$ 0.00 \\
C0-first & 25.43 $\pm$ 0.12 & 25.43 $\pm$ 0.12 & 0.00 $\pm$ 0.00 & 25.00 $\pm$ 0.00 \\
C0-last & 25.36 $\pm$ 0.05 & 25.36 $\pm$ 0.05 & -0.00 $\pm$ 0.00 & 25.00 $\pm$ 0.00 \\
C1 & 29.03 $\pm$ 0.34 & 29.03 $\pm$ 0.34 & -0.00 $\pm$ 0.00 & 62.50 $\pm$ 0.00 \\
C2 & 59.91 $\pm$ 0.14 & 59.91 $\pm$ 0.14 & 0.00 $\pm$ 0.00 & 68.31 $\pm$ 0.00 \\
C3 & 67.65 $\pm$ 0.98 & 67.66 $\pm$ 0.98 & 0.00 $\pm$ 0.01 & 78.53 $\pm$ 0.00 \\
\midrule
\multicolumn{5}{@{}l}{\textbf{CollegeMsg} (select/free 16/4)} \\
Random & 64.70 $\pm$ 3.21 & 64.70 $\pm$ 3.20 & -0.00 $\pm$ 0.02 & 60.38 $\pm$ 4.82 \\
Sched. Random & 61.53 $\pm$ 7.19 & 61.57 $\pm$ 7.05 & 0.04 $\pm$ 0.14 & 36.98 $\pm$ 7.57 \\
Delete-first & 68.88 $\pm$ 1.41 & 68.87 $\pm$ 1.43 & -0.00 $\pm$ 0.02 & 78.53 $\pm$ 0.00 \\
Sched. Delete-first & 33.78 $\pm$ 4.78 & 33.79 $\pm$ 4.77 & 0.00 $\pm$ 0.01 & 25.00 $\pm$ 0.00 \\
Zero-C witness & 25.36 $\pm$ 0.06 & 25.37 $\pm$ 0.06 & 0.01 $\pm$ 0.00 & 25.00 $\pm$ 0.00 \\
C0-first & 25.43 $\pm$ 0.12 & 25.42 $\pm$ 0.11 & -0.01 $\pm$ 0.01 & 25.00 $\pm$ 0.00 \\
C0-last & 25.36 $\pm$ 0.04 & 25.36 $\pm$ 0.04 & -0.00 $\pm$ 0.01 & 25.00 $\pm$ 0.00 \\
C1 & 29.04 $\pm$ 0.32 & 29.05 $\pm$ 0.32 & 0.00 $\pm$ 0.00 & 62.50 $\pm$ 0.00 \\
C2 & 59.91 $\pm$ 0.14 & 59.91 $\pm$ 0.14 & 0.00 $\pm$ 0.01 & 68.31 $\pm$ 0.00 \\
C3 & 67.67 $\pm$ 0.98 & 67.67 $\pm$ 0.98 & 0.00 $\pm$ 0.01 & 78.53 $\pm$ 0.00 \\
\midrule
\multicolumn{5}{@{}l}{\textbf{MathOverflow} (select/free 16/144)} \\
Random & 66.11 $\pm$ 0.33 & 66.10 $\pm$ 0.31 & -0.01 $\pm$ 0.02 & 58.74 $\pm$ 0.69 \\
Sched. Random & 61.96 $\pm$ 6.34 & 61.96 $\pm$ 6.34 & 0.00 $\pm$ 0.00 & 38.03 $\pm$ 1.75 \\
Delete-first & 68.93 $\pm$ 1.39 & 68.92 $\pm$ 1.41 & -0.02 $\pm$ 0.02 & 78.53 $\pm$ 0.00 \\
Sched. Delete-first & 44.16 $\pm$ 17.90 & 44.16 $\pm$ 17.90 & 0.00 $\pm$ 0.00 & 25.00 $\pm$ 0.00 \\
Zero-C witness & 25.34 $\pm$ 0.07 & 25.35 $\pm$ 0.06 & 0.00 $\pm$ 0.00 & 25.00 $\pm$ 0.00 \\
C0-first & 25.42 $\pm$ 0.12 & 25.42 $\pm$ 0.12 & 0.00 $\pm$ 0.00 & 25.00 $\pm$ 0.00 \\
C0-last & 25.36 $\pm$ 0.04 & 25.36 $\pm$ 0.04 & 0.00 $\pm$ 0.00 & 25.00 $\pm$ 0.00 \\
C1 & 29.03 $\pm$ 0.33 & 29.03 $\pm$ 0.33 & -0.00 $\pm$ 0.01 & 62.50 $\pm$ 0.00 \\
C2 & 59.91 $\pm$ 0.15 & 59.91 $\pm$ 0.15 & 0.00 $\pm$ 0.00 & 68.31 $\pm$ 0.00 \\
C3 & 67.66 $\pm$ 0.97 & 67.66 $\pm$ 0.97 & 0.00 $\pm$ 0.00 & 78.53 $\pm$ 0.00 \\
\midrule
\multicolumn{5}{@{}l}{\textbf{PROTEINS} (select/free 16/64)} \\
Random & 66.30 $\pm$ 0.22 & 66.30 $\pm$ 0.22 & -0.00 $\pm$ 0.00 & 58.98 $\pm$ 1.84 \\
Sched. Random & 61.40 $\pm$ 7.62 & 61.40 $\pm$ 7.62 & -0.00 $\pm$ 0.00 & 40.20 $\pm$ 1.54 \\
Delete-first & 68.88 $\pm$ 1.40 & 68.88 $\pm$ 1.41 & 0.00 $\pm$ 0.01 & 78.53 $\pm$ 0.00 \\
Sched. Delete-first & 32.35 $\pm$ 0.96 & 32.35 $\pm$ 0.96 & 0.00 $\pm$ 0.00 & 25.00 $\pm$ 0.00 \\
Zero-C witness & 25.35 $\pm$ 0.07 & 25.35 $\pm$ 0.07 & -0.00 $\pm$ 0.00 & 25.00 $\pm$ 0.00 \\
C0-first & 25.42 $\pm$ 0.12 & 25.42 $\pm$ 0.12 & -0.00 $\pm$ 0.00 & 25.00 $\pm$ 0.00 \\
C0-last & 25.36 $\pm$ 0.04 & 25.36 $\pm$ 0.04 & 0.00 $\pm$ 0.00 & 25.00 $\pm$ 0.00 \\
C1 & 29.03 $\pm$ 0.34 & 29.03 $\pm$ 0.34 & 0.00 $\pm$ 0.00 & 62.50 $\pm$ 0.00 \\
C2 & 59.91 $\pm$ 0.15 & 59.91 $\pm$ 0.15 & 0.00 $\pm$ 0.00 & 68.31 $\pm$ 0.00 \\
C3 & 67.65 $\pm$ 0.97 & 67.65 $\pm$ 0.97 & 0.00 $\pm$ 0.00 & 78.53 $\pm$ 0.00 \\
\midrule
\multicolumn{5}{@{}l}{\textbf{COLLAB} (select/free 16/64)} \\
Random & 65.26 $\pm$ 1.79 & 65.26 $\pm$ 1.79 & -0.00 $\pm$ 0.00 & 58.76 $\pm$ 1.53 \\
Sched. Random & 60.58 $\pm$ 8.95 & 60.58 $\pm$ 8.95 & -0.00 $\pm$ 0.00 & 38.54 $\pm$ 0.86 \\
Delete-first & 68.87 $\pm$ 1.42 & 68.86 $\pm$ 1.43 & -0.01 $\pm$ 0.02 & 78.53 $\pm$ 0.00 \\
Sched. Delete-first & 33.11 $\pm$ 0.93 & 33.11 $\pm$ 0.93 & -0.00 $\pm$ 0.00 & 25.00 $\pm$ 0.00 \\
Zero-C witness & 25.35 $\pm$ 0.06 & 25.35 $\pm$ 0.06 & -0.00 $\pm$ 0.00 & 25.00 $\pm$ 0.00 \\
C0-first & 25.42 $\pm$ 0.11 & 25.42 $\pm$ 0.12 & -0.00 $\pm$ 0.00 & 25.00 $\pm$ 0.00 \\
C0-last & 25.35 $\pm$ 0.04 & 25.35 $\pm$ 0.04 & 0.00 $\pm$ 0.00 & 25.00 $\pm$ 0.00 \\
C1 & 29.03 $\pm$ 0.34 & 29.03 $\pm$ 0.34 & -0.00 $\pm$ 0.00 & 62.50 $\pm$ 0.00 \\
C2 & 59.91 $\pm$ 0.15 & 59.91 $\pm$ 0.15 & 0.00 $\pm$ 0.00 & 68.31 $\pm$ 0.00 \\
C3 & 67.65 $\pm$ 0.97 & 67.65 $\pm$ 0.97 & -0.00 $\pm$ 0.00 & 78.53 $\pm$ 0.00 \\
\end{longtable}
\endgroup

Table~\ref{tab:task13-protocol} shows similar in-support and selection-disjoint risks at the
reported precision. Exact source overlap is zero, while the
controlled edit template is shared across splits.
Unlike the full-validation results in Table~\ref{tab:exp-r2},
$D_{\rm free}^{\theta}$ excludes the checkpoint-selection
cohorts. For CollegeMsg, this leaves four of the 20
validation cohorts.

\subsection{External-Backbone Ordering Comparison}
\label{app:subsection-32}

We evaluate five task-adapted graph generative backbones:
DDSBM~\citep{kim2025ddsbm}, DeFoG~\citep{qin2025defog},
SimGFM~\citep{luo2026simgfm}, ConStruct~\citep{madeira2024generative},
and SparseDiff~\citep{qin2025sparsediff}.
On controlled four-edit pairs from PROTEINS and COLLAB,
the baseline uses random legal ordering, whereas the
scheduled teacher uses coupling-aware legal ordering.
Within each backbone, both arms share the architecture,
legal-action decoder, training objective, Poisson clock,
and evaluation protocol; only the teacher ordering differs.
Training uses matched reduced budgets, with the paired-endpoint
risk term activated after a 100-step warmup.
We report checkpoint-selection-disjoint validation results. These are task adaptations for
source-blind paired editing, rather than native
graph-generation benchmarks.

\begin{table}[htbp]
\centering
\small
\setlength{\tabcolsep}{5pt}
\caption{Ordering-only comparison on PROTEINS and COLLAB
(mean $\pm$ SD).
Values are checkpoint-selection-disjoint paired endpoint
errors (\%); gain is random baseline minus scheduled (pp).}
\label{tab:q3-native-external}
\begin{tabular}{@{}llrrr@{}}
\toprule
Dataset & Method & Baseline (\%) & Sched. (\%) & Gain (pp) \\
\midrule
PROTEINS & DDSBM
& $73.76\pm1.44$ & $26.42\pm1.82$ & $+47.34\pm2.75$ \\
PROTEINS & DeFoG
& $75.02\pm1.26$ & $27.97\pm0.94$ & $+47.05\pm0.90$ \\
PROTEINS & SimGFM
& $76.10\pm0.91$ & $26.87\pm0.21$ & $+49.23\pm0.81$ \\
PROTEINS & ConStruct
& $76.55\pm0.09$ & $25.81\pm1.36$ & $+50.74\pm1.27$ \\
PROTEINS & SparseDiff
& $76.96\pm0.73$ & $25.84\pm0.29$ & $+51.12\pm0.49$ \\
\midrule
COLLAB & DDSBM
& $74.41\pm0.81$ & $28.76\pm2.67$ & $+45.65\pm3.19$ \\
COLLAB & DeFoG
& $75.35\pm1.63$ & $25.90\pm1.12$ & $+49.45\pm0.56$ \\
COLLAB & SimGFM
& $77.68\pm0.83$ & $26.68\pm0.83$ & $+51.00\pm0.68$ \\
COLLAB & ConStruct
& $76.86\pm0.13$ & $26.61\pm0.63$ & $+50.25\pm0.56$ \\
COLLAB & SparseDiff
& $76.55\pm0.76$ & $27.64\pm0.43$ & $+48.91\pm0.33$ \\
\bottomrule
\end{tabular}
\end{table}

\begin{figure}[htbp]
\centering
\includegraphics[width=\linewidth]{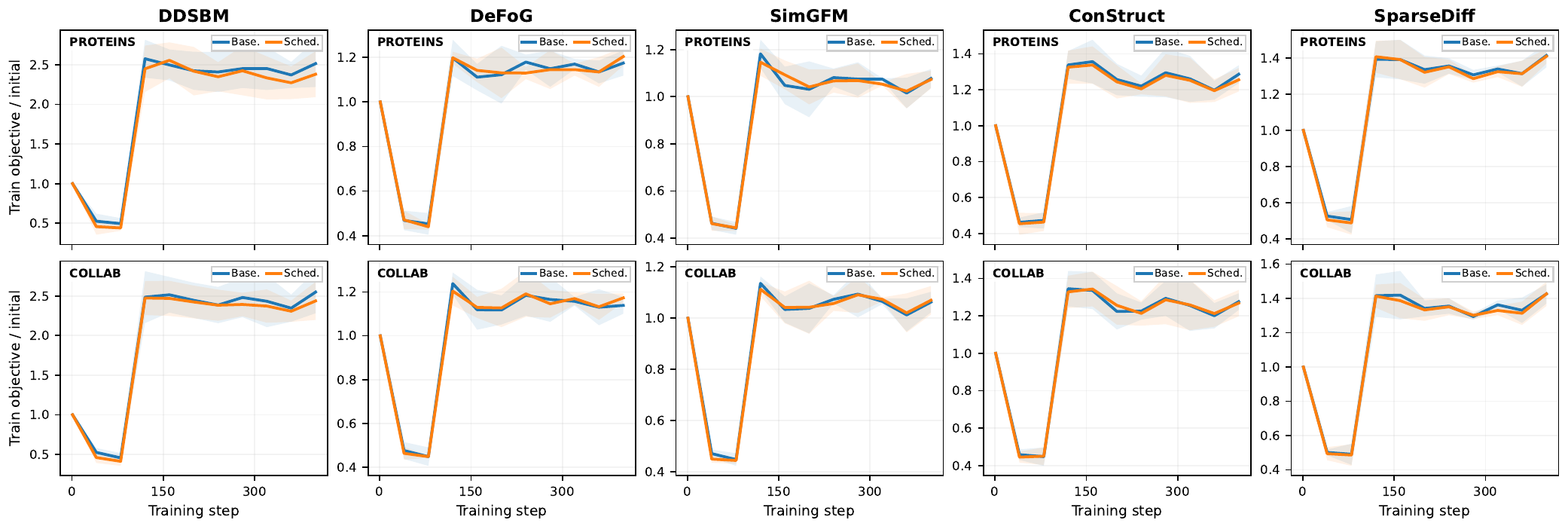}
\caption{Training objectives for five task-adapted
backbones on PROTEINS (top) and COLLAB (bottom).
Each panel compares random and coupling-aware legal
ordering under matched settings. Curves are normalized
by their initial objective values; bands show seed
variation. The paired-endpoint risk term is activated
after a 100-step warmup.}
\label{fig:q3-native-external-training}
\end{figure}

Scheduling reduces paired error across all ten matched
backbone--dataset comparisons, with gains of
$45.65$--$51.12$ percentage points.
The high baseline errors should be interpreted within
this reduced-budget, source-blind paired-editing setting,
rather than as native graph-generation performance.
These matched comparisons provide supplementary evidence
that legal edit ordering improves paired endpoint fidelity
across different model parameterizations.
Figure~\ref{fig:q3-native-external-training}
shows the corresponding training trajectories.

\end{document}